%% file: Doctoral_Dissertation.tex
\documentclass[dvipdfmx,10pt,openright]{book}
\AtBeginDocument{\addtocontents{toc}{\protect\thispagestyle{empty}}} 
\usepackage[centering]{geometry}
\usepackage{latexsym}
\usepackage{fancyhdr}
\usepackage{pdfpages}
\usepackage{mathrsfs}
\usepackage{amsmath,amssymb}
\usepackage{amsthm}
\usepackage{mathtools}
\usepackage{bm}
\usepackage{bbm}
\usepackage{ascmac}
\usepackage{braket}
\usepackage{geometry}
\usepackage{chngcntr}
\usepackage{docmute}
\usepackage{here}
\usepackage[unicode,bookmarks=true,colorlinks=true]{hyperref}
\counterwithin*{section}{part}
\theoremstyle{plain}
\newtheorem{thm}{Theorem}[section]
\newtheorem{lem}[thm]{Lemma}
\newtheorem{prop}[thm]{Proposition}
\newtheorem{cor}[thm]{Corollary}
\newtheorem{dfn}[thm]{Definition}
\newtheorem{asm}[thm]{Assumption}
\newtheorem{rem}[thm]{Remark}
\newtheorem{eg}[thm]{Example}
\mathtoolsset{showonlyrefs=true}
\newcommand{\esssup}{\mathop{\rm ess~sup}\limits}

\numberwithin{equation}{section}
\providecommand{\keywords}[1]
{
  \small
  \textbf{{Keywords~}} Asset pricing, Backward stochastic differential equation, Exponential utility, Market clearing, Mean field game, Stochastic optimal control
  \normalsize
}
\providecommand{\MCS}[1]
{
  \small
  \textbf{{Mathematics Subject Classification~}} 49N80, 60H10, 91B51, 91G10, 93E11
  \normalsize
}
\providecommand{\JEL}[1]
{
  \small
  \textbf{{JEL Classification~}} C61, C73, D52, G11, G12
  \normalsize
}

\title{
    \vspace*{-30mm}
    \textsc{\Large{Doctoral Dissertation}}\vspace{5mm}\\ 
    \hrulefill\\
    \textbf{\LARGE{Mean Field Equilibrium Asset Pricing Models With Exponential Utility}}\\
    \hrulefill
}

\author{
    \vspace{5mm}\\ 
    \textit{Author:} \href{https://masashi-sekine.github.io/En/index.html}{Masashi Sekine} \vspace{10mm} \\
    \textsc{Department of Quantitative Finance}\\
    \textsc{Division of Management}\\
    \textsc{Graduate School of Economics}\\
    \textsc{The University of Tokyo} \vspace{10mm}\\
    A Dissertation Submitted for the Degree of \vspace{10mm}\\
    \Large{\textit{Doctor of Philosophy in Management}} \vspace{8mm}\\
}

\date{November 2025}
\begin{document}
\maketitle
\newpage
~
\newpage
\vspace{20mm}
\centerline{\LARGE{\textbf{Acknowledgement}}}\vspace{10mm}
First and foremost, I would like to express my deepest gratitude to Associate Professor Masaaki Fujii for his long-term guidance and support. 
His passion has taught me the depth and joy of research and his earnest attitude has been and always will be a guiding principle in my life.
Over the past seven years, he has devoted countless hours to mentoring me and providing valuable feedback on my papers, which has encouraged me to think critically and pushed me to achieve my best.

I would also like to express my sincere gratitude to my parents, Tsuyoshi and Shizuka, for their thoughtful advice and understanding, which enabled me to enjoy my research life and keep challenging myself.
I extend my heartfelt thanks to my dear wife, Mina, for her constant love and support since my undergraduate days.
I would like to thank my dear friends and relatives for their warm support and encouragement.
I am truly fortunate to have such wonderful people by my side.

This research is supported by the Grant-in-Aid for JSPS Fellows (Grant Number JP23KJ0648).

\begin{flushright}
    Tokyo, Japan\\
    November 2025\\
    Masashi Sekine
\end{flushright}
\newpage
~
\newpage
\vspace{10mm}
\centerline{\LARGE{\textbf{Abstract}}}
\vspace{10mm}
This thesis develops equilibrium asset pricing models in incomplete markets with a large number of heterogeneous agents, using mean field game theory. 
The contributions of the thesis are summarized below:
\begin{enumerate}
    \item It proposes novel models in asset pricing theory and characterizes the equilibrium using a novel form of equation.
    \item It examines the properties of the relevant equation and derives the risk premium that achieves market clearing under certain conditions.
    \item It provides several extensions for general settings.
\end{enumerate} 
The thesis begins with a brief introduction and review of the field, followed by three models of mean field asset pricing theory.

The first model is a theoretical model that endogenously derives the equilibrium risk premium of securities in a market with heterogeneous agents. 
The agents are assumed to have exponential-type preferences and to be heterogeneous in the initial wealth, the coefficient of absolute risk aversion, and the stochastic terminal liability, which is unspanned by the security prices. 
The study begins by solving the optimal investment problem using the optimal martingale principle. 
The market equilibrium is characterized by the solution of a novel mean field backward stochastic differential equation (BSDE), whose driver has quadratic growth both in the stochastic integrands and in their conditional expectations. 
The existence of a solution to such an equation is proved under several conditions. The equilibrium risk-premium process is characterized in terms of the solution to this mean field BSDE and is shown to clear the market in the large population limit.

The second model is an extension of the first model; it considers agents' consumption and incorporates habit formation. 
The incorporation of habit formation relaxes the assumption of time-separable utility functions by making utility depend not only on the current level of consumption but also on the agents' accumulated stock of past consumption. 
In order to characterize the market-clearing equilibrium, a mean field BSDE, which has a similar form to that of the first model, is derived.
Its well-posedness and asymptotic behavior are then examined. Additionally, an exponential quadratic Gaussian (EQG) reformulation of the asset pricing model is introduced, in which an equilibrium solution is obtained in a semi-analytic form.

The third model extends the first model to the problem of a partially observable market with a large number of heterogeneous agents. 
This model assumes that agents can only observe stock prices and must infer the risk premium from these observations when determining trading strategies. 
The equilibrium risk premium in such a market is characterized in terms of a solution to the mean field BSDE. 
Specifically, the solution to the mean field BSDE can be expressed semi-analytically by employing the EQG framework. 
The equilibrium risk-premium process is then constructed endogenously using Kalman-Bucy filtering theory.
In addition, a simple numerical simulation is included to visualize the dynamics of the market model.\\
~\\
\keywords~\\
\MCS~\\
\JEL

\newpage
\small
\tableofcontents
\normalsize
\newpage

\pagestyle{fancy}
\fancyhead{}
\fancyhead[RE]{\nouppercase{\rightmark}}
\fancyhead[LO]{\nouppercase{\leftmark}}
\fancyhead[LE,RO]{\thepage}
\fancyfoot{}
\setcounter{page}{1}
\chapter{Introduction}

\input{introduction.tex}

\newpage

\chapter{Notations}

\input{notation.tex}

\newpage
\chapter{A Review of the Recent Research}

\input{review.tex}

\newpage
\chapter{Mean Field Equilibrium Asset Pricing Model in an Incomplete Market}

\input{part1.tex}

\newpage
\chapter{Mean Field Equilibrium Asset Pricing Model With Habit Formation}
\input{part2.tex}

\newpage
\chapter{Mean Field Equilibrium Asset Pricing Model Under Partial Observation}

\input{part3.tex}

\newpage
\chapter{Concluding Remarks}

\input{conclusion.tex}

\newpage

\input{references.tex}
\end{document}

%% file: introduction.tex
\vspace{0mm}

Financial markets have enabled smooth funding and effective risk allocation among participants in contemporary society. 
Mathematical finance and financial economics are fields of study that mathematically analyze various phenomena in financial markets, and they have developed rapidly over the past half-century commencing with pioneering work such as Markowitz \cite{portfolio_selection} and Black \& Scholes \cite{black-scholes}. 
Specific areas of this study include the valuation of financial derivatives, risk management of investment assets, and the formulation of optimal investment strategies, all of which are widely applied in financial practice.
Asset pricing theory also plays a crucial role in this field.
The fundamental objective of the theory is to establish the equilibrium price at which the supply of assets matches the demand. 
The capital asset pricing model (CAPM), proposed by Sharpe \cite{w_sharpe}, Lintner \cite{j_lintner}, and Mossin \cite{mossin}, is the most commonly known model, and these theories provide deep insights into how asset prices are determined in the financial market. 

Various economic phenomena can be explained by the interaction of a large number of individuals who make decisions in an attempt to maximize their own profits.
This situation can be analyzed within the framework of a ``multiplayer game.'' Financial markets are, of course, no exception, since they are places where a large number of investors trade on a daily basis, making tactical decisions to execute their long-term strategic goals.
In the classical literature of mathematical finance, particularly in models of derivative pricing, optimal investment problems, and portfolio risk management, market environments are typically assumed to be exogenous.
These models can describe how investors behave in a given market environment.
However, market environments, especially security prices, are essentially determined by the balance of investors' buy and sell orders, and it is therefore desirable to consider a theoretical model that can express the market environments endogenously. 
Moreover, standard asset pricing models often assume the existence of representative investors.
While such an assumption has the advantage of analytical simplicity, it cannot describe situations in which each investor adopts a different investment strategy based on their own information, and equilibrium prices are determined through their interactions.

To overcome these issues, we need a theoretical model that can capture the financial market endogenously within the model and take into account the competitive or cooperative behavior of heterogeneous investors.
In other words, a model that views various phenomena of financial markets as a multiplayer game is necessary.
If such a model is developed, it becomes possible to analyze how investor interactions influence the formation of security prices and how market parameters affect the equilibrium state.
In particular, it can analyze how the existence of agents' idiosyncrasies affects the price formation as a whole. 
However, a typical approach to multiplayer games usually becomes mathematically intractable since it leads to a strongly coupled optimization problem.
The solution of such a problem can hardly be obtained analytically or even numerically, especially when the number of players is large, such as participants in financial markets.
Therefore, research on models that describe securities price formation through the interaction of a large number of heterogeneous investors is very limited, and many unresolved issues remain.

Mean field game theory provides a powerful framework to address this type of issue. The theory was independently proposed by Lasry \& Lions \cite{lasryMeanFieldGames2007} and Huang, Malham\'{e} \& Caines \cite{huangLargePopulationStochastic2006}. 
The advantage of this theory is that it replaces an intractable multiplayer-game problem with a simpler optimal control problem for a single player and a fixed-point problem. 
It can also be shown to provide an approximate Nash equilibrium, often referred to as an ``$\varepsilon$-Nash equilibrium'', for stochastic differential games involving many agents.
\cite{lasryMeanFieldGames2007, huangLargePopulationStochastic2006} showed that the multiplayer game problem results in a system of two highly coupled nonlinear partial differential equations (PDEs) in the large population limit.
One is called the Hamilton-Jacobi-Bellman (HJB) equation, which corresponds to the stochastic optimal control problem, and the other is called the Fokker-Planck (FP) equation, which corresponds to the self-consistency problem of the state of the whole system.
Meanwhile, Carmona \& Delarue \cite{carmona-mffbsde,carmonaProbabilisticAnalysisMeanField2013,carmonaForwardBackwardStochastic2015} developed a probabilistic approach to mean field game theory.
They formulated the mean field game using forward-backward stochastic differential equations (FBSDEs), instead of a system of PDEs.
Their approach is extensively covered in two volumes of monographs by Carmona \& Delarue \cite{carmonaProbabilisticTheoryMean2018,carmonaProbabilisticTheoryMean2018a}, offering thorough details and applications.

Along with its theoretical development, there has been an increase in research applying mean field game theory to various fields of economics.
For example, there are applications to optimal investment problems and optimal liquidation problems; see \cite{Espinosa-Touzi, Frei-DosReis, Fu2020MeanFE, fu_Mean_field_portfolio_games, Lacker-Z}.
In macroeconomics, earlier studies include \cite{PDE_macro,BM-2,Bayraktar,Gomes_econ}.
Moreover, \cite{Gomes_MFG, Gomes_rand_MFG, Gomes_MFG_noise, SREC, Firoozi_et_al} have studied price formation in electricity markets and energy markets.
\cite{evangelista2022price, fujiiMeanFieldGame2022, Fujii-Takahashi_strong, fujii2022equilibrium, Fujii-equilibrium-pricing} have developed price formation models in financial markets.
Details of the earlier research and applications will be reviewed in Chapter 3.

As will be explained in detail later, this thesis studies asset pricing models using mean field game theory.
The existing work on mean field price formation models employs the probabilistic approach of Carmona \& Delarue \cite{carmona-mffbsde,carmonaProbabilisticAnalysisMeanField2013,carmonaForwardBackwardStochastic2015},
and this approach adopts Pontryagin's maximum principle to solve the investors' stochastic optimal control problems.
As is well known, Pontryagin's maximum principle for dynamic stochastic systems involves formulating a Hamiltonian function that incorporates the dynamics of the state variables and the cost function.
The optimal control as well as the optimal state trajectory are characterized by the minimizer (or maximizer) of the Hamiltonian and the solution of the associated FBSDE.
Then, in solving the fixed-point problem, the theory employs FBSDEs of McKean-Vlasov type: FBSDEs where the distribution of the solution itself enters the coefficient function.
Therefore, the solvability of these FBSDEs is a crucial key in their theory. 
However, it is well known that the well-posedness of FBSDEs on general time intervals can be shown in very limited cases, such as \cite{pengFullyCoupledForwardBackward1999}. 
Thus, their theory necessitates the cost functions and the coefficient functions of state dynamics to satisfy some strong assumptions in order to guarantee the convexity of the Hamiltonian and the Lipschitz regularity of the associated FBSDEs.
In particular, utility functions that are commonly used in economic literature, such as exponential, power, and logarithmic-type utility functions, are not covered.
This implies that their framework is not suitable for the problem of optimizing CARA or CRRA-class utilities subject to the dynamics of a self-financing portfolio.

Hu, Imkeller \& M\"{u}ller \cite{huUtilityMaximizationIncomplete2005a} considers the utility maximization problem in incomplete markets.
This research develops the optimal martingale method, which can be applied to standard classes of utility functions, including exponential, power, and logarithmic types. 
Recently, Xing \cite{Hao-Xing} has solved the optimization problem for Epstein-Zin recursive utility using the same method.
The relevant equation that characterizes the optimality is given by a quadratic-growth BSDE (qg-BSDE).
The well-posedness of qg-BSDEs (See Kobylanski \cite{Kobylanski2000BackwardSD}) is well known and still more manageable than that of typical FBSDEs.
This result provides an important basis for constructing a multiplayer game model where a large number of agents interact and optimize their exponential-type utility functions subject to self-financing strategies.

Motivated by this background, this thesis aims to construct theoretical models of asset pricing in incomplete markets, using mean field game theory.
Specifically, the market involves a large number of heterogeneous agents with exponential-type preferences.
Each agent attempts to hedge the stochastic terminal liability, which differs for each agent, through security trading.
To solve the optimal investment problem of each agent, this thesis adopts the aforementioned optimal martingale method developed by \cite{huUtilityMaximizationIncomplete2005a} and proves the well-posedness of the associated qg-BSDE.
It then analyzes the market risk-premium process that achieves the market-clearing condition in the large population limit when all agents adopt their optimal strategies.
By considering the large population limit, the mean field equilibria are characterized by novel qg-BSDEs, whose driver has quadratic growth both in the stochastic integrands and in their conditional expectations.
The main contributions of this thesis are as follows:
\begin{enumerate}
    \item It provides an economic meaning to the novel mean field qg-BSDEs and proves the existence of a solution to such equations under several conditions.
    \item It proves that the risk-premium process expressed by the solution of the mean field qg-BSDE clears the market in the large population limit.
    \item It provides a special framework in which a solution to the mean field qg-BSDE is obtained in a semi-analytic form, allowing for numerical analyses.
    \item It extends the model to more general settings: a model with consumption habit formation and a model under partial observation.
\end{enumerate} 

The thesis consists of seven chapters, providing three novel models of equilibrium asset pricing in incomplete markets. 
Each chapter covers the following topics:
In Chapter 2, we introduce the notation for frequently used sets and function spaces. Chapter 3 provides an overview of recent research in this field.
In particular, Fujii \& Takahashi \cite{fujiiMeanFieldGame2022} will be reviewed in detail. 
While the aim of this thesis is similar to theirs, as it also deals with the theory of price formation using mean field games, there are significant differences in the setting, method, and results.
By explaining the similarities and differences, it is expected to help readers understand the motivation and methodology of this study.

The content of Chapter 4 is based on Fujii \& Sekine \cite{fujiiMeanFieldEquilibriumPrice2023a}, providing a fundamental model of equilibrium price formation using mean field game theory. 
In this model, agents are assumed to have exponential-type preferences and to be heterogeneous in their initial wealth, the coefficient of absolute risk aversion, and the stochastic liability at the terminal time, which is unspanned by the security prices. 
The optimal investment problem is solved using the optimal martingale principle developed by Hu, Imkeller \& M\"{u}ller \cite{huUtilityMaximizationIncomplete2005a}. 
While their study assumes the risk-premium process to be bounded, this chapter relaxes this assumption in order to deal with the equilibrium model.
This necessitates a slight extension of the result in \cite{Kobylanski2000BackwardSD} to address the well-posedness of the associated qg-BSDE.
As noted above, the goal is to construct the risk-premium process endogenously, where the demand and supply of the associated securities always balance among a large number of financial agents who deploy general self-financing trading strategies. 
The market equilibrium is characterized in terms of the solution of a novel mean field BSDE, as mentioned above. 
The existence of a solution to the mean field BSDE is shown under certain conditions on the size of the terminal liability. 
In addition, it verifies that the risk-premium process characterized by the mean field BSDE satisfies the market-clearing condition in the large population limit.

Chapter 5 is based on Fujii \& Sekine \cite{fujiiMeanFieldEquilibriumAsset2024}, which extends the model of Chapter 4 by considering agents' consumption behavior and habit formation.
Research on consumption habit formation has been a fundamental and classical subject in financial economics. 
The existence of habit formation relaxes the assumption of time-separable utility functions by making the utility dependent not only on the current level of consumption but also on the agent's accumulated stock of past consumption. 
When consumption is taken into account as a control variable, the method of \cite{huUtilityMaximizationIncomplete2005a} cannot be applied directly and requires modification. 
This also leads to a similar type of mean field BSDE, whose well-posedness is proved under several assumptions.
It further offers an exponential quadratic Gaussian (EQG) formulation of the model. 
In this formulation, a solution to the mean field BSDE can be characterized by a system of Riccati-type ordinary differential equations by assuming that the terminal liability is a quadratic form of Gaussian factor processes.

Chapter 6 is based on Sekine \cite{sekine2024MFE_PO}, which extends the model of Chapter 4 to a partially observable market.
In this model, it is assumed that agents do not have full access to market information, and they must infer the risk premium from stock price observations when determining trading strategies. 
The equilibrium risk premium in such a market is characterized in terms of a solution to a mean field BSDE of a similar type. 
Specifically, the solution to the mean field BSDE can be expressed semi-analytically using the EQG framework. 
It then constructs the risk-premium process endogenously using Kalman-Bucy filtering theory.
This chapter includes a simple numerical simulation that visualizes the dynamics of the market model.
The thesis concludes with Chapter 7, which summarizes the entire work and explains the advantages of these models. It also presents the limitations and challenges of this research and makes suggestions for future studies.

%% file: notation.tex
\vspace{0mm}

In this thesis, we shall work on a finite time interval $[0,T]$ for some $T>0$. 
For a given filtered probability space satisfying the usual conditions $(\Omega,\mathcal{F},\mathbb{P},\mathbb{F}~(:=(\mathcal{F}_t)_{t\in[0,T]}))$ and a vector space $E$ over $\mathbb{R}$, we use the following notation to describe frequently used sets and function spaces.\\

\noindent
(1) $\mathcal{T}(\mathbb{F})$ is the set of all $\mathbb{F}$-stopping times with values in $[0,T]$.\\

\noindent
(2) $\mathbb{L}^0(\mathcal{F},E)$ is the set of $E$-valued $\mathcal{F}$-measurable random variables. \\

\noindent
(3) $\mathbb{L}^2(\mathbb{P},\mathcal{F},E)$ is the set of $E$-valued $\mathcal{F}$-measurable random variables $\xi$ satisfying 
\[
    \|\xi\|_2:=\mathbb{E}^{\mathbb{P}}[|\xi|^2]^{\frac{1}{2}}<\infty.
\]

\noindent
(4) $\mathbb{L}^\infty(\mathbb{P},\mathcal{F},E)$ is the set of $E$-valued $\mathcal{F}$-measurable random variables $\xi$ satisfying 
\[
    \|\xi\|_\infty:=\esssup_{\omega\in\Omega}|\xi(\omega)|<\infty.
\]

\noindent
(5) $\mathbb{L}^0(\mathbb{F},E)$ is the set of $E$-valued $\mathbb{F}$-progressively measurable stochastic processes. \\

\noindent
(6) $\mathbb{H}^2(\mathbb{P},\mathbb{F},E)$ is the set of $E$-valued $\mathbb{F}$-progressively measurable stochastic processes $X$ satisfying
\[
\|X\|_{\mathbb{H}^2}:=\mathbb{E}^{\mathbb{P}}\Bigl[\int_0^T |X_t|^2 dt\Bigr]^{\frac{1}{2}}<\infty.
\]

\noindent
(7) $\mathbb{L}^\infty(\mathbb{P},\mathbb{F},E)$ is the set of $E$-valued $\mathbb{F}$-progressively measurable stochastic processes $X$ satisfying
\[
\|X\|_{\mathbb{L}^\infty}:=\esssup_{(t,\omega)\in[0,T]\times\Omega}|X_t(\omega)| ~<\infty.
\]

\noindent
(8) $\mathbb{H}^2_{\mathrm{BMO}}(\mathbb{P},\mathbb{F},E)$ is the set of $E$-valued $\mathbb{F}$-progressively measurable stochastic processes $X$ satisfying
\begin{equation}
\|X\|_{\mathbb{H}^2_{\mathrm{BMO}}}:=\sup_{\tau\in\mathcal{T}(\mathbb{F})} \Bigl\| \mathbb{E}^{\mathbb{P}}\Bigl[\int_\tau^T |X_t|^2dt | \mathcal{F}_\tau\Bigr]^{\frac{1}{2}}\Bigr\|_\infty <\infty,
\label{eq-h2bmo-norm}
\end{equation}
where $\|\cdot\|_{\infty}$ denotes the $\mathbb{P}$-essential supremum as in (4). Moreover, for $X\in \mathbb{H}^2_{\mathrm{BMO}}(\mathbb{P},\mathbb{F},E)$, the result of Kazamaki~\cite{kazamakiContinuousExponentialMartingales1994} shows that
the Dol\'{e}ans-Dade exponential $\Bigl\{\mathcal{E}\Bigl(\displaystyle\int_0^\cdot X_s dW_s\Bigr)_\tau;~\tau\in\mathcal{T}(\mathbb{F})\Bigr\}$ is uniformly integrable (i.e. of class $\mathcal{D}$), and that
$\Bigl(\displaystyle\int_0^t X_s dW_s\Bigr)_{t\in[0,T]}$ satisfies the BMO-martingale property. Here, $W$ is a $d$-dimensional standard Brownian motion defined on $(\Omega, \mathcal{F}, \mathbb{P}, \mathbb{F})$. 
As an important property of $X\in \mathbb{H}^2_{\rm BMO}$, let us mention the so-called energy inequality. For every $n\in \mathbb{N}$, the following inequality holds:
\begin{equation}
\mathbb{E}^{\mathbb{P}}\Bigl[\Bigl(\int_0^T |X_s|^2 ds\Bigr)^n\Bigr]\leq n!\bigl(\|X\|^2_{\mathbb{H}^2_{\mathrm BMO}}\bigr)^n.
\label{ineq-energy}
\end{equation}
For the proof, see \cite{Cvitanic-Zhang} [Lemma 9.6.5]. \\

\noindent
(9) $\mathbb{S}^2(\mathbb{P},\mathbb{F},E)$ is the set of $E$-valued $\mathbb{F}$-adapted continuous stochastic processes $X$ satisfying
\[
\|X\|_{\mathbb{S}^2}:=\mathbb{E}^{\mathbb{P}}\Bigl[\sup_{t\in[0,T]}|X_t|^2\Bigr]^{\frac{1}{2}}<\infty.
\]

\noindent
(10) $\mathbb{S}^\infty(\mathbb{P},\mathbb{F},E)$ is the set of $E$-valued $\mathbb{F}$-adapted continuous stochastic processes $X$ satisfying
\[
\|X\|_{\mathbb{S}^\infty}:=\esssup_{(t,\omega)\in[0,T]\times\Omega}|X_t(\omega)| ~< \infty.
\]

\noindent
(11) $\mathcal{C}([0,T],E)$ is the set of continuous functions $f:[0,T]\to E$. \\

\noindent
(12) $\mathcal{C}^1([0,T],E)$ is the set of once continuously differentiable functions $f:[0,T]\to E$. \\

\noindent
(13) $\mathbb{R}^n_+:=\{x\in\mathbb{R}^n; x\geq 0\}$ and $\mathbb{R}^n_{++}:=\{x\in\mathbb{R}^n; x> 0\}$ for $n\in\mathbb{N}$, where $x\geq 0$ and $x> 0$ mean that all elements of $x$ are nonnegative and strictly positive, respectively. Also, $\mathbb{M}_n$ is the set of real symmetric matrices of size $n\times n$.\par

For (1) to (12), we may omit the arguments such as $(\mathbb{P},\mathcal{F},\mathbb{F},E)$ if they are obvious. Throughout the thesis, the symbol $C$ represents a generic nonnegative constant used in various estimates. Also, the argument $\omega\in\Omega$ is usually omitted when there is no risk of misinterpretation.

%% file: review.tex
\vspace{0mm}

\section{Preliminary}
Since the pioneering research of Lasry \& Lions \cite{lasryMeanFieldGames2007} and Huang, Malham\'{e} \& Caines \cite{huangLargePopulationStochastic2006}, the mean field game theory has become one of the most actively studied topics in probability theory, applied mathematics, finance, and economics.
As mentioned in Chapter 1, there are many applications of mean field games in macroeconomics.
Achdou et al. \cite{PDE_macro} focuses on PDEs related to some macroeconomic models. Achdou et al. \cite{BM-2} reformulates the heterogeneous agent model in continuous time using mean field game theory.
Bayraktar, Mitra \& Zhang \cite{Bayraktar} analyzes the effect of countercyclical unemployment benefit policies using the general equilibrium model.
In the realm of finance, \cite{Espinosa-Touzi,Lacker-Z,Frei-DosReis} study optimal investment problems under relative performance criteria.
\cite{Fu2020MeanFE, fu_Mean_field_portfolio_games} study mean field portfolio games among a large number of agents.

In recent years, there has also been an increasing number of studies of equilibrium price formation models that use the mean field game theory.
Gu\'{e}ant, Lasry \& Lions \cite{Gueant-Oil} models the balance of demand and supply in the oil market.
Shrivats, Firoozi \& Jaimungal \cite{SREC} studies the price formation model in Solar Renewable Energy Certificate (SREC) markets by solving McKean-Vlasov FBSDEs. 
Gomes \& S\'{a}ude \cite{Gomes_MFG} develops an endogenous price formation model in the electricity market by considering the market-clearing condition, using the mean field game theory.
Due to the absence of common noise, the resultant price process is deterministic in their model. 
Gomes, Gutierrez \& Ribeiro \cite{Gomes_rand_MFG} extends this model by considering random supply.
In financial economics, Evangelista, Saporito \& Thamsten \cite{evangelista2022price} develops a mean field game theoretic model of the price formation of assets traded in a limit order book. 
Fujii \& Takahashi \cite{fujiiMeanFieldGame2022, Fujii-Takahashi_strong} present a mean field pricing model for securities under stochastic order flows and \cite{fujii2022equilibrium} provides its extension with a major player.  
Fujii \cite{Fujii-equilibrium-pricing} develops a model that allows the co-presence of cooperative and non-cooperative populations and shows how the price process is formed when a group of agents act in a coordinated manner.

Among these earlier studies, this chapter specifically focuses on Fujii \& Takahashi \cite{fujiiMeanFieldGame2022} and provides a thorough overview of their model and theory. 
Their study uses the mean field game theory to analyze the price formation of securities, which is of common interest to this thesis.
On the other hand, the way the model is set up, the mathematical tools used, and the conclusions are quite different.
It is expected that contrasting their study with the present study will help the reader to better understand the models of this thesis that are to be presented in the next three chapters.

\section{Mean Field Game Approach to Equilibrium Pricing With Market Clearing Condition}
This section offers a detailed review of Fujii \& Takahashi \cite{fujiiMeanFieldGame2022} [Section 1-5]. 
They seek to derive a security price process endogenously by requiring the price to balance the flow of orders in a market with a large number of agents.
After offering a review, this section makes a comparison with the content of this thesis. 
As we shall see, both studies aim to construct equilibrium models of asset pricing using the mean field game theory, but the set-up of the model, the mathematical methodologies, and the results are quite different. 
\subsection{Overview}
For $N\in\mathbb{N}$, $(\Omega^i,\mathcal{F}^i,\mathbb{P}^i)_{i=0}^N$ are complete probability spaces endowed with filtrations $\mathbb{F}^i:=(\mathcal{F}_t^i)_{t\in[0,T]}$ for $i\in\{0,\ldots, N\}$.
$\mathbb{F}^0$ is a complete and right-continuous filtration generated by a $d_0$-dimensional Brownian motion $W^0$. 
For each $i\in\{1,\ldots, N\}$,  $\mathbb{F}^i$ is a complete and right-continuous augmentation of the filtration generated by a $d$-dimensional Brownian motion $W^i$ and a $W^i$-independent $n$-dimensional square-integrable random variable $\xi^i$. 
$(\xi^i)_{i=1}^N$ are assumed to follow the same distribution.

For each $i\in\{1,\ldots, N\}$, $(\Omega^{0,i},\mathcal{F}^{0,i},\mathbb{P}^{0,i})$ is a product probability space over $\Omega^{0,i}:=\Omega^0\times \Omega^i$ with $(\mathcal{F}^{0,i},\mathbb{P}^{0,i})$, which is the completion of $(\mathcal{F}^0\otimes \mathcal{F}^i, \mathbb{P}^0\otimes \mathbb{P}^i)$.
The associated filtration $\mathbb{F}^{0,i}:=(\mathcal{F}^{0,i}_t)_{t\in[0,T]}$ is the complete and right-continuous augmentation of $(\mathcal{F}^0_t\otimes \mathcal{F}_t^i)_{t\in[0,T]}$.
$(\Omega,\mathcal{F},\mathbb{P})$ is a complete probability space endowed with a filtration $\mathbb{F}:=(\mathcal{F}_t)_{t\in[0,T]}$, which satisfies the usual conditions and is defined as a product probability space of $(\Omega^i,\mathcal{F}^i,\mathbb{P}^i,\mathbb{F}^i)_{i=0}^N$. 

Suppose that there are $N$ agents, each of which is indexed by ``$i$'' for $i\in\{1,\ldots,N\}$. Agent-$i$ solves the cost-minimization problem
\begin{equation}
    \begin{split} 
        \inf_{\alpha^i\in\mathbb{A}^i}J^i(\alpha^i)
    \end{split}
\end{equation}
subject to the state dynamics:
\begin{equation}
    \begin{split} 
        dX_t^i=\bigl(\alpha_t^i+l(t,\varpi_t,c_t^0, c_t^i)\bigr)dt+\sigma_0(t,\varpi_t,c_t^0, c_t^i)dW_t^0+\sigma(t,\varpi_t,c_t^0,c_t^i)dW_t^i,~~~t\in[0,T],~~~X^i_0=\xi^i.
    \end{split}
\end{equation}
Here, $l:[0,T]\times (\mathbb{R}^n)^3\rightarrow \mathbb{R}^n$, $\sigma_0:[0,T]\times (\mathbb{R}^n)^3\rightarrow \mathbb{R}^{n\times d_0}$, and $\sigma:[0,T]\times (\mathbb{R}^n)^3\rightarrow \mathbb{R}^{n\times d}$ are measurable functions. 
$(\varpi_t)_{t\in[0,T]}$ is an $\mathbb{R}^n$-valued $\mathbb{F}^0$-adapted process and represents the price process of securities.
$(c_t^0)_{t\in[0,T]}\in \mathbb{H}^2(\mathbb{F}^0,\mathbb{R}^n)$ with $c_T^0\in \mathbb{L}^2(\mathcal{F}^0_T,\mathbb{R}^n)$ denotes the cash flow from the securities and the market news commonly available to all the agents, 
while $(c_t^i)_{t\in[0,T]}\in \mathbb{H}^2(\mathbb{F}^i,\mathbb{R}^n)$ with $c_T^i\in \mathbb{L}^2(\mathcal{F}_T^i,\mathbb{R}^n)$ denotes some independent factors and news affecting only agent-$i$.
$(X_t^i)^k$ denotes the \textit{number of shares} of the $k$-th equity possessed by agent-$i$ at time $t$.
If it is negative, it means that the agent is taking a short position in the $k$-th equity at time $t$.
$\alpha^i:=(\alpha_t^i)_{t\in[0,T]}$ is the control variable of agent-$i$, and denotes the trading speed of the securities.
That is, $(\alpha^i_t)^k dt$ for $k\in\{1,\ldots,n\}$ denotes the number of shares of the $k$-th equity bought or sold by agent-$i$ within the time interval $[t,t+dt]$.
The admissible space $\mathbb{A}^i$ is the set of processes $\alpha^i$ such that
\begin{equation}
    \mathbb{E}\int_0^T |\alpha_t^i|^2 dt<\infty.
\end{equation}
The cost functional $J^i$ is defined by
\begin{equation}
    J^i(\alpha^i):=\mathbb{E}\Bigl[\int_0^T f(t,X_t^i,\alpha_t^i,\varpi_t,c_t^0, c_t^i)dt+g(X_T^i,\varpi_T,c_T^0,c_T^i)\Bigr].
\end{equation}

The functions $f:[0,T]\times (\mathbb{R}^n)^5\to \mathbb{R}$ and $g:(\mathbb{R}^n)^4\to \mathbb{R}$, denoting the running and terminal costs, respectively, are introduced as follows.
\begin{equation}
    \begin{split}
    f(t,x,\alpha,\varpi,c^0,c)&:=\langle \varpi,\alpha\rangle+\frac{1}{2}\langle \alpha,\Lambda\alpha\rangle+\overline{f}(t,x,\varpi,c^0,c), \\
    g(x,\varpi,c^0,c)&:=-b \langle \varpi,x\rangle+\overline{g}(x,c^0,c),
\end{split}
\end{equation}
where $\overline{f}:[0,T]\times (\mathbb{R}^n)^4\to \mathbb{R}$ and $\overline{g}:(\mathbb{R}^n)^3\to \mathbb{R}$ are measurable functions.
The economic interpretation of each term in $f$ is as follows. The term $\langle \varpi,\alpha\rangle$ represents the amount of cash each agent pays (resp. receives) to buy (resp. sell) shares at price $\varpi$ with trading speed $\alpha$.
The term $\langle \alpha,\Lambda\alpha\rangle$ represents the transaction costs that the agents must pay. Here, $\Lambda$ is an $n\times n$ positive definite matrix.
The function $\overline{f}$ represents the costs incurred by financial risks and proper inventory management.
Moreover, the term $\langle \varpi,x\rangle$ in $g$ is the mark-to-market value of the portfolio at time $T$, and $b<1$ denotes the discount factor. 
The cost $\overline{g}$ is a penalty imposed on the terminal position size.

The purpose of their work is to construct the equilibrium price process $(\varpi_t)_{t\in[0,T]}$ endogenously.
They consider the condition where the total number of shares being purchased is equal to the number of shares being sold at each time.
In other words, the buy and sell orders for the shares must balance at any point in time.
This condition is called the market-clearing condition and is expressed as follows.
\begin{equation}
    \label{rev-market-clearing}
    \sum_{i=1}^N \widehat\alpha_t^i=0~, ~~~ dt\otimes \mathbb{P}\text{-}\mathrm{a.e.},
\end{equation}
where $\widehat{\alpha}^i$ is the optimal control of agent-$i$. They make the following assumptions on the functions.

\begin{asm} (Cited from \cite{fujiiMeanFieldGame2022} [Assumption 3.1])\\ 
    \label{rev-asm-1} 
{\rm (i)} $\Lambda\in\mathbb{M}_n$ is a positive definite matrix satisfying $\underline{\lambda}I_{n}\leq \Lambda \leq \overline{\lambda}I_{n}$ for some constants $0<\underline{\lambda}\leq\overline{\lambda}$. 
Here, $I_n$ denotes the $n\times n$ identity matrix.\\
{\rm (ii)} For any $(t,x,\varpi,c^0,c)\in[0,T]\times (\mathbb{R}^n)^4$, the functions $\overline{f}$ and $\overline{g}$ satisfy
\begin{equation}
    |\overline{f}(t,x,\varpi,c^0,c)|+|\overline{g}(x,c^0,c)|\leq L(1+|x|^2+|\varpi|^2+|c^0|^2+|c|^2)
\end{equation}
for some constant $L\geq 0$.\\
{\rm (iii)} The functions $\overline{f}$ and $\overline{g}$ are continuously differentiable in $x$ and the derivatives satisfy
\begin{equation}
    \begin{split}
    &|\partial_x \overline{f}(t,x',\varpi,c^0,c)-\partial_x\overline{f}(t,x,\varpi,c^0,c)|+|\partial_x \overline{g}(x',c^0,c)-\partial_x \overline{g}(x,c^0,c)|\leq L |x'-x|,\\
    &|\partial_x \overline{f}(t,x,\varpi,c^0,c)|+|\partial_x \overline{g}(x,c^0,c)|\leq L(1+|x|+|\varpi|+|c^0|+|c|)
    \end{split}
\end{equation}
for any $(t,x,x',\varpi,c^0,c)\in[0,T]\times (\mathbb{R}^n)^5$. \\
{\rm (iv)} The functions $\overline{f}$ and $\overline{g}$ are strongly convex in $x$, i.e., for any $(t,x,x',\varpi,c^0,c)\in[0,T]\times (\mathbb{R}^n)^5$, they satisfy
\begin{equation}
    \begin{split}
    \overline{f}(t,x',\varpi,c^0,c)-\overline{f}(t,x,\varpi,c^0,c)-\langle x'-x, \partial_x \overline{f}(t,x,\varpi,c^0,c)\rangle &\geq \frac{\gamma^f}{2}|x'-x|^2, \\
    \overline{g}(x',c^0,c)-\overline{g}(x,c^0,c)-\langle x'-x,\partial_x \overline{g}(x,c^0,c)\rangle &\geq \frac{\gamma^g}{2}|x'-x|^2
    \end{split}
\end{equation}
for some constants $\gamma^f,\gamma^g\geq 0$. \\
{\rm (v)} The functions $l,\sigma_0$ and $\sigma$ satisfy the linear growth condition
\begin{equation}
    |(l,\sigma_0,\sigma)(t,\varpi,c^0,c)|\leq L(1+|\varpi|+|c^0|+|c|) 
\end{equation}
for any $(t,\varpi,c^0,c)\in[0,T]\times (\mathbb{R}^n)^3$. \\
{\rm (vi)} $b \in [0,1)$ is a given constant.
\end{asm}

The Hamiltonian for this problem is
\begin{equation}
    H(t,x,y,\alpha,\varpi,c^0,c)=\langle y, \alpha+l(t,\varpi,c^0,c)\rangle+f(t,x,\alpha,\varpi,c^0,c). 
\end{equation}
Since $\partial_\alpha H(t,x,y,\alpha,\varpi,c^0,c)=y+\varpi+\Lambda \alpha$ and the Hamiltonian $H$ is convex in $\alpha$, the minimizer of the Hamiltonian is given by
\begin{equation}
    \widehat\alpha(y,\varpi):=-\overline{\Lambda}(y+\varpi),
\label{def-ha}
\end{equation}
where $\overline{ \Lambda}:= \Lambda^{-1}$. The FBSDE associated with the stochastic maximum principle is given by, for $i\in\{1,\ldots,N\}$,
\begin{equation}
    \begin{split}
        \label{agent-FBSDE}
        dX_t^i&=\Bigl( \widehat\alpha(Y_t^i,\varpi_t)+l(t,\varpi_t,c_t^0, c_t^i)\Bigr)dt+\sigma_0(t,\varpi_t,c_t^0, c_t^i)dW_t^0+\sigma(t,\varpi_t,c_t^0,c_t^i)dW_t^i, \\
        dY_t^i&=-\partial_x \overline{f}(t,X_t^i,\varpi_t,c_t^0,c_t^i)dt+Z_t^{i,0}dW_t^0+Z_t^{i}dW_t^i~,
\end{split}
\end{equation}
for $t\in[0,T]$ with $X_0^i=\xi^i$ and $Y_T^i=\partial_x g(X_T^i, \varpi_T, c_T^0, c_T^i)$. Under Assumption \ref{rev-asm-1}, the FBSDE \eqref{agent-FBSDE} has a unique solution $(X^i,Y^i,Z^{i,0}, Z^{i})\in  \mathbb{S}^2( \mathbb{F}^{0,i}, \mathbb{R}^n)\times 
 \mathbb{S}^2( \mathbb{F}^{0,i}, \mathbb{R}^n)\times  \mathbb{H}^2( \mathbb{F}^{0,i}, \mathbb{R}^{n\times d_0})\times  \mathbb{H}^2( \mathbb{F}^{0,i}, \mathbb{R}^{n\times d})$ for given 
$(\varpi_t)_{t\in[0,T]}\in  \mathbb{H}^2( \mathbb{F}^0, \mathbb{R}^n)$ with $ \varpi_T\in \mathbb{L}^2(\mathcal{F}^0_T, \mathbb{R}^n)$. See \cite{fujiiMeanFieldGame2022} [Theorem 3.1].

Given the security price $(\varpi_t)_{t\in[0,T]}$, the optimal strategy of agent-$i$ is
\begin{equation}
    \widehat\alpha^i_t:=-\overline{\Lambda}(Y^i_t+\varpi_t),~~~t\in[0,T].
\end{equation}
Under the market clearing condition \eqref{rev-market-clearing}, this becomes
\begin{equation}
    \varpi_t=-\frac{1}{N}\sum_{i=1}^N Y_t^i
\end{equation}
for each $t\in[0,T]$. This is not consistent with the assumption that the price process $\varpi$ is $\mathbb{F}^0$-adapted. However, since the idiosyncratic noises are assumed to be i.i.d. and interactions among agents are symmetric,
the random variables $(Y^i_t)_{1\leq i\leq N}$ are exchangeable for each $t\in[0,T]$. Then, the De Finetti's theorem for exchangeable sequences of random variables yields
\begin{equation}
    \frac{1}{N}\sum_{i=1}^N Y_t^i \to \mathbb{E}\Bigl[Y^1_t | \bigcap_{k=1}^\infty \sigma(Y_t^j,j\geq k)\Bigr],~~~N\to\infty,~~~\mathrm{a.s.}
\end{equation}
It is natural to expect that the tail $\sigma$-algebra $\bigcap_{k=1}^\infty \sigma(Y_t^j,j\geq k)$ reduces to $\mathcal{F}^0_t$ for each $t\in[0,T]$.
Therefore, in the large population limit, namely as $N\to\infty$, the equilibrium price is expected to be given by
\begin{equation}
    \label{heuristic}
    \varpi_t=-\mathbb{E}[Y_t^1|\mathcal{F}^0_t],~~~t\in[0,T],
\end{equation}
which motivates us to study the McKean-Vlasov FBSDE: for $t\in[0,T]$,
\begin{equation}
    \begin{split}
        \label{MKV-FBSDE}
        dX_t&=\Bigl(\widehat\alpha\bigl(Y_t,- \mathbb{E}[Y_t|\mathcal{F}^0_t]\bigr)+l\bigl(t,- \mathbb{E}[Y_t|\mathcal{F}^0_t],c_t^0, c_t\bigr)\Bigr)dt \\
         &~~~+\sigma_0\bigl(t,- \mathbb{E}[Y_t|\mathcal{F}^0_t],c_t^0, c_t\bigr)dW_t^0+\sigma\bigl(t,- \mathbb{E}[Y_t|\mathcal{F}^0_t],c_t^0,c_t\bigr)dW_t^1~, \\
        dY_t&=-\partial_x \overline{f}\bigl(t,X_t,- \mathbb{E}[Y_t|\mathcal{F}^0_t],c_t^0,c_t\bigr)dt+Z_t^{0}dW_t^0+Z_tdW_t^1~,
\end{split}
\end{equation}
with $X_0=\xi$ and 
\begin{equation}
    Y_T=\frac{b}{1-b}\mathbb{E}\bigl[\partial_x \overline{g}(X_T,c_T^0,c_T)|\mathcal{F}^0_T\bigr]+\partial_x \overline{g}(X_T,c_T^0,c_T).
\end{equation}
To simplify the notation, we omit the superscript ``$1$'' from $Y^1$, $X^1$, $\xi^1$ and $c^1$.

To show the solvability of \eqref{MKV-FBSDE}, they make the following assumptions:
\begin{asm} (Cited from \cite{fujiiMeanFieldGame2022} [Assumption 4.1] and [Assumption 4.2]) \\
\label{rev-asm-2}
{\rm (i)} The functions $\sigma_0$ and $\sigma$ are independent of the argument $\varpi$. \\
{\rm (ii)} For any $(t,x,c^0,c)\in[0,T]\times (\mathbb{R}^n)^3$ and any $\varpi,\varpi' \in \mathbb{R}^n$, the functions $l$ and $\overline{f}$ satisfy
    \begin{equation}
    \begin{split}
        &|l(t,\varpi,c^0,c)-l(t,\varpi',c^0,c)|\leq L |\varpi-\varpi'|, \\
        &|\partial_x \overline{f}(t,x,\varpi,c^0,c)-\partial_x\overline{f}(t,x,\varpi',c^0,c)|\leq L_{\varpi} |\varpi-\varpi'|
    \end{split}
    \end{equation}
for some constant $L\geq 0$ and another constant $L_\varpi\geq 0$.\\
{\rm (iii)} For any $t\in[0,T]$, any random variables $x,x',c^0,c\in \mathbb{L}^2(\mathcal{F},\mathbb{R}^n)$, and any sub $\sigma$-algebra $\mathcal{G}\subset \mathcal{F}$, 
the function $l$ satisfies the monotone condition
\begin{equation}
    \mathbb{E}\Bigl[\langle l(t,\mathbb{E}[x|\mathcal{G}],c^0,c)-l(t,\mathbb{E}[x'|\mathcal{G}],c^0,c),x-x'\rangle\Bigr]\geq \gamma^l  \mathbbm{1}_{\{L_{\varpi}>0\}} \mathbb{E}\bigl[|\mathbb{E}[x-x'|\mathcal{G}]|^2\bigr]
\end{equation}
for some constant $\gamma^l>0$.\\
{\rm (iv)} There exists a constant $\gamma$ satisfying $0<\gamma\leq \Bigl(\gamma^f-\dfrac{L_\varpi^2}{4\gamma^l}\Bigr)\wedge \gamma^g$. Moreover, the function $\overline{g}$ satisfies
\begin{equation}
    \gamma^g\mathbb{E}[|x-x'|^2]+\frac{b}{1-b}\mathbb{E}\Bigl[\langle \mathbb{E}\bigl[\partial_x\overline{g}(x,c^0,c)-\partial_x\overline{g}(x',c^0,c)|\mathcal{G}\bigr],x-x'\rangle\Bigr]\geq \gamma \mathbb{E}[|x-x'|^2]
\end{equation}
for any random variables $x,x', c^0,c\in \mathbb{L}^2(\mathcal{F},\mathbb{R}^n)$ and any sub $\sigma$-algebra $\mathcal{G}\subset \mathcal{F}$.
\end{asm}

This is the first main result of their paper.
\begin{thm} (Cited from \cite{fujiiMeanFieldGame2022} [Theorem 4.2])\\
Under Assumptions~\ref{rev-asm-1} and \ref{rev-asm-2}, there exists a unique strong solution $(X,Y,Z^0,Z)\in \mathbb{S}^2(\mathbb{F}^{0,1},\mathbb{R}^n)
\times \mathbb{S}^2(\mathbb{F}^{0,1},\mathbb{R}^n)\times \mathbb{H}^2(\mathbb{F}^{0,1},\mathbb{R}^{n\times d_0})\times \mathbb{H}^2(\mathbb{F}^{0,1},\mathbb{R}^{n\times d})$
to the FBSDE \eqref{MKV-FBSDE}.
\end{thm}

After proving the well-posedness of the FBSDE \eqref{MKV-FBSDE}, they show that the price process characterized by \eqref{heuristic} satisfies the market clearing condition in the large population limit $N\to\infty$, which is the second main result of their paper.
\begin{thm} (Cited from \cite{fujiiMeanFieldGame2022} [Theorem 5.1])\\
\label{th-clearing}
Under Assumptions~\ref{rev-asm-1} and \ref{rev-asm-2}, it holds that
\begin{equation}
    \lim_{N\to \infty} \mathbb{E}\int_0^T \Bigl|\frac{1}{N}\sum_{i=1}^N\widehat\alpha^i_{\rm{mf}}(t)\Bigr|^2 dt=0.
\end{equation}
Moreover, if there exists some constant $\Gamma$ such that $\sup_{t\in[0,T]}\mathbb{E}\bigl[|Y_t|^q\bigr]^\frac{1}{q}\leq \Gamma< \infty$ for some $q>4$, 
then there exists some constant $C$ independent of $N$ such that
\begin{equation}
    \mathbb{E}\int_0^T \Bigl|\frac{1}{N}\sum_{i=1}^N\widehat\alpha^i_{\rm{mf}}(t)\Bigr|^2 dt\leq C\Gamma^2 \epsilon_N,
\end{equation}
where $\epsilon_N:=N^{-2/\max(n,4)}\bigl(1+\log(N)\mathbbm{1}_{\{n=4\}}\bigr)$.
\end{thm}

In conclusion, Fujii \& Takahashi \cite{fujiiMeanFieldGame2022} [Section 1-5] is summarized as follows:
\begin{enumerate}
    \item[(1)] They formulated the financial market in which a large number of agents minimize their cost functionals by controlling their inventory size. 
    In particular, each agent controls the ``trading speed'' $\alpha^i$, and the cost functionals consist of transaction costs and penalty on the terminal inventory size. 
    \item[(2)] The stochastic optimal control problem was solved by the framework of Pontryagin's maximum principle. Its solution was characterized by an FBSDE \eqref{agent-FBSDE}, which can be shown to have a unique solution.
    \item[(3)] They introduced the notion of the market clearing condition \eqref{rev-market-clearing} and provided a heuristic derivation \eqref{heuristic} for the equilibrium price in the large population limit.
    \item[(4)] By plugging \eqref{heuristic} into the FBSDE \eqref{agent-FBSDE}, they proposed an FBSDE of McKean-Vlasov type \eqref{MKV-FBSDE}. They proved the well-posedness of \eqref{MKV-FBSDE} using the technique of Peng \& Wu \cite{pengFullyCoupledForwardBackward1999}.
    \item[(5)] Based on the above results, they also showed that the price process \eqref{heuristic} indeed achieved the market clearing condition when $N\to\infty$.
\end{enumerate}

\subsection{Comparison With This Research}
This section outlines the differences between the model presented in Fujii \& Takahashi \cite{fujiiMeanFieldGame2022} and that of this thesis.
In \cite{fujiiMeanFieldGame2022}, agents control the ``trading speed,'' and the state variable is the number of shares they hold. 
In contrast, in the models of this thesis, agents directly control the amount of money invested in securities, with the value of self-financing portfolio $(\mathcal{W}^{\pi}_t)_{t\in[0,T]}$ being the state variable:
\begin{equation}
    \begin{split}
        \mathcal{W}^{\pi}_t 
        &:= 
        \xi + \int_0^t \pi_r^\top \mathrm{diag}(S_r)^{-1}dS_r,~~~t\in[0,T].
    \end{split}
\end{equation}
Here, $\xi$ is the initial wealth, $\pi$ is the control variable representing the amount of money invested in securities, and $S$ is the security price process, whose dynamics is given by
\begin{equation}
		S_t=S_0+\int_0^t {\rm diag}(S_s)(\mu_s ds+\sigma_s dW_s^0), ~~~t\in[0,T].
\end{equation}
The utility function (the negative of the cost function) in this thesis is given by an exponential-type preference, which has been widely used in economic literature as a typical CARA utility function:
\begin{equation}
    U(\pi):=\mathbb{E}\Bigl[-\exp\Bigl(-\gamma (\mathcal{W}^{\pi}_T-F)\Bigr)\Bigr],
\end{equation}
where $F$ denotes the terminal liability. 
Consequently, the problem is formulated in a way that is closer to the standard setting of financial economics. 
In this setting, the agents' optimal strategy is characterized by a quadratic-growth BSDE (qg-BSDE) instead of an FBSDE derived from Pontryagin's maximum principle.
The equilibrium state is characterized by a mean field quadratic growth BSDE instead of an FBSDE of McKean-Vlasov type.

Furthermore, this thesis derives the equilibrium market risk-premium process whereas \cite{fujiiMeanFieldGame2022} derived the price process itself.
That is, this thesis does not directly derive the process $(S_t)_{t\in[0,T]}$ but derives the risk premium $(\theta_t)_{t\in[0,T]}$ defined by $\theta_t:=\sigma_t^\top(\sigma_t\sigma_t^\top)^{-1}\mu_t$ instead.
Therefore, this model cannot endogenously derive the return $\mu$ or the volatility $\sigma$.
The same observation is made in the case of a complete market.
Karatzas \& Shreve~\cite{karatzas_methods_1998} points out that this is because of the possibility of replacing stocks by self-financing portfolios (see \cite{karatzas_methods_1998} [Section 4.6]).
On the other hand, \cite{fujiiMeanFieldGame2022} imposes transaction costs via the term $\langle \alpha,\Lambda\alpha\rangle$, which can essentially work as a restriction on stock liquidity.

In summary, the differences between the two studies are presented in the following table.\\

\small
\noindent
\begin{tabular}{|l|l|l|} \hline
     & Fujii \& Takahashi \cite{fujiiMeanFieldGame2022} & This thesis \\ \hline
    Set-up & Optimal inventory management & Exponential utility maximization \\ \hline
    Methodology for optimal control & Pontryagin's maximum principle & Optimal martingale method \\  \hline
    Interaction among agents & \multicolumn{2}{|c|}{Market clearing condition} \\ \hline 
    Liquidity restrictions & Transaction cost & None \\ \hline
    Relevant equation & McKean-Vlasov FBSDE & Mean field qg-BSDE \\ \hline
    Result & Equilibrium price & Equilibrium risk premium \\ \hline
\end{tabular}
\normalsize

%% file: part1.tex
\vspace{0mm}

\section{Preliminary}
The key objective of theories of equilibrium price formation is to find asset prices for which the financial market is in equilibrium.  
Equilibrium prices are prices for which the optimal decisions of all the agents, given these asset prices, are such that the aggregate demand of each asset equals its aggregate supply.
Constructing and investigating the properties of the equilibrium prices are one of the major issues in financial economics.
See, for example,  Back~\cite{backAssetPricingPortfolio2017} and Munk~\cite{Munk} for details on this topic. 
The existence of market equilibrium in complete markets is well known, the details of which can be found in Karatzas \& Shreve~\cite{karatzas_methods_1998}.
However, the situation in incomplete markets has been unclear except for finite dimensional economic models, such as the one in Cuoco \& He~\cite{CuocoHe}.
Although a set of sufficient conditions for the existence of a representative agent has been derived, confirming these conditions for given models involves difficult mathematical problems and still remains open.
See discussions in Jarrow~\cite[Section 14.5]{Jarrow} and references therein. 
In particular, it has been known that obtaining a concrete characterization of the equilibrium prices is very difficult when there are idiosyncratic liabilities, labor incomes, or endowments that are unspanned by the security prices. 
This is a quite unsatisfactory situation, because we can imagine that the existence of a variety of idiosyncratic shocks among the agents induced from these terms is the very reason why we observe lively transactions in financial markets.
In this chapter, with the help of the recent developments of the mean-field game (MFG) theory, we propose a new concise technique to tackle the problem of equilibrium price formation in an incomplete market under an exponential utility.

Since the pioneering works of Lasry \& Lions~\cite{Lions-1, Lions-2, lasryMeanFieldGames2007} and 
Huang et al.~\cite{huangLargePopulationStochastic2006, Caines-Huang-2, Caines-Huang-3}, significant developments in the mean-field game 
theory have enabled us to understand some of the long-standing issues of multi-agent games.
If the interactions among the agents are symmetric, then
the MFG techniques can render, in the large population limit, a very complex problem of solving a large coupled system of equations
that characterizes a Nash equilibrium feasible by 
transforming it into separate and simpler problems of the optimization for a representative agent
and of finding a fixed point. 
The resultant solution in the mean field equilibrium is known to provide an $\epsilon$-Nash equilibrium for the original game with a large but finite number of agents.
The details of the MFG theory and many applications can be found, for example, 
in two volumes by Carmona \& Delarue~\cite{carmonaProbabilisticTheoryMean2018, carmonaProbabilisticTheoryMean2018a} and 
in a lecture note by Cardaliaguet~\cite{Cardaliaguet-note}.

We want to construct the price process of risky stocks
(more precisely, the associated risk-premium process) {\it endogenously} 
so that it balances the demand and supply among a large number of financial firms facing the market-wide common noise as well as
their own idiosyncratic noise.  
Unfortunately, this {\it market-clearing condition} does not fit well with the concept of Nash equilibrium.
Actually, if we change a trading strategy of one agent away from her equilibrium solution while keeping the other agents'
strategies unchanged, then the balance of demand and supply will inevitably be broken down.
Since the MFG theory has been developed primarily for the Nash games, 
most of its applications have not treated the market-clearing equilibrium.
In fact, in many of the existing examples, their primary interests are not in the price formation 
and the asset price process is typically assumed exogenously. 
For example, let us refer to \cite{Espinosa-Touzi, Frei-DosReis, Fu2020MeanFE, fu_Mean_field_portfolio_games, Lacker-Z} 
as related works dealing with optimal investment problems with exponential and power utilities.
All of these works are concerned with the Nash equilibrium among the investors competing in a relative performance criterion,
while the relevant price processes are given exogenously.

Recently, there also has been progress in the MFG theory for the problem of equilibrium price formation under the market clearing condition. 
Gomes \& Sa\'ude~\cite{Gomes_MFG} present a deterministic model of electricity price.
Its extension with random supply is given by Gomes et al.~\cite{Gomes_rand_MFG}.
The same authors also study, in \cite{Gomes_MFG_noise}, a price formation problem of a commodity whose production is subject to random fluctuations.
Evangelista et al.~\cite{evangelista2022price} investigate the price formation of an asset being traded in a limit order book and show promising numerical results using the actual high-frequency data of the listed stocks in several exchanges.
Shrivats et al.~\cite{SREC} deal with a price formation problem for the solar renewable energy certificate (SREC) by solving forward-backward stochastic differential equations (FBSDEs) of McKean-Vlasov type,
and Firoozi et al.~\cite{Firoozi_et_al} deal with a principal-agent problem in the associated emission market.
Fujii \& Takahashi~\cite{fujiiMeanFieldGame2022} solve a stochastic mean-field price model of securities in the presence of stochastic order flows, and in \cite{Fujii-Takahashi_strong},  the same authors prove the strong convergence to
the mean-field limit from the setup with finite number of agents. 
Fujii \& Takahashi~\cite{fujii2022equilibrium} extend the above model to the presence of a major player.
Recently, Fujii~\cite{Fujii-equilibrium-pricing} develops a model that allows the co-presence of cooperative and non-cooperative populations to learn how the price process is formed when a group of agents act in a coordinated manner.

But still, there exist two important restrictions in all of these papers mentioned above:
firstly,  the relevant control of each agent is interpreted as the {\it trading rate}
that is absolutely continuous with respect to the Lebesgue measure $dt$; secondly, the  cost function of each agent consists
of  terms representing some penalties on the trading speed and on the inventory size of the assets.
In other words, from the perspective of financial applications, the existing results cannot deal with the general self-financing trading strategies or the utility (or cost) functions given directly in terms of the associated wealth process of the portfolio.\footnote{See also a very recent work by Lavigne \& Tankov~\cite{Lavigne-Tankov}
which adopts a very different approach to investigate the mean-field equilibrium of the carbon emissions among the firms.}
Hence the question of equilibrium price formation, at least in the standard and traditional setup in financial economics, remains unanswered.
Here, a major problem in dealing with a utility function of 
wealth has been the difficulties in guaranteeing the convexity of the Hamiltonian 
associated with the Pontryagin's maximum principle and in obtaining enough regularity to 
prove the well-posedness of the associated FBSDEs.

In this chapter, we investigate the equilibrium price formation of the risky stocks to address the above two concerns by using the method of Hu, Imkeller \& M\"uller \cite{huUtilityMaximizationIncomplete2005a} and the mean field game theory.
In contrast to the existing works using the method of \cite{huUtilityMaximizationIncomplete2005a}, such as \cite{Fu2020MeanFE,fu_Mean_field_portfolio_games}, which found an MFG equilibrium with relative performance criterion, we establish a market-clearing equilibrium for the price-formation problem.
Our goal is to construct the risk-premium process endogenously so that the demand and supply of the associated stocks always balance among a large number of financial firms (agents) who are allowed to deploy
general self-financing trading strategies. 
We assume that the agents have a common type of preference specified by an exponential utility with respect to their wealth.
The agents are heterogeneous in their initial wealth, the coefficients of risk aversion, as well as the stochastic liability at the terminal time.
We allow the liabilities (not restricted to be positive) to depend both on the common noise
and on the idiosyncratic noise so that
they can describe idiosyncratic shocks of each agent in addition to the market-wide financial shocks.
For example, since our agents are financial firms, it is realistic to assume that they have significant derivative liabilities.
Although payoffs specified by financial contracts usually depend only on the common  information, the volume of the contracts
can naturally depend also on idiosyncratic information of each agent, such as her corporate size, corporate culture, and her client base. 
Personnel expenses of each agent can also be included in her liability.

To solve the optimization problem of each agent, we adopt the optimal martingale principle developed by Hu, Imkeller \& M\"uller \cite{huUtilityMaximizationIncomplete2005a}, instead of the Pontryagin's maximum principle. 
Their method can be applied to many popular classes of utility functions, such as exponential, power and logarithmic types. 
Recently, Xing~\cite{Hao-Xing} solved the optimization problem for  Epstein-Zin recursive utility by the same method.
In every case, the relevant equation characterizing the optimality is given by a quadratic-growth BSDE (qg-BSDE)~\cite{Kobylanski2000BackwardSD}.
In contrast to the existing literature, 
we cannot simply assume the risk-premium to be a bounded process, since this condition may be too stringent to achieve a market-clearing equilibrium.
In fact, we will see that it is necessary to relax the boundedness assumption on the risk-premium process in general. 
Since the standard results on a qg-BSDE cannot handle the existence of such a process in its driver, we need a special treatment to show its well-posedness. 
Fortunately, the special structure of the BSDE inherited from the exponential utility allows us to solve the problem by a standard application of the comparison principle.

After we obtain the optimal strategy of each agent, we move on to characterize the mean-field equilibrium in terms of a novel mean-field BSDE whose driver has quadratic growth both in the stochastic integrands and in their conditional expectations. 
The new BSDE is interesting in its own right since the same type of equations may become relevant for similar applications of mean-field equilibrium to other utility functions. 
As the main contributions of this work, we show the existence of a solution to this novel mean-field BSDE under several conditions, and then prove that the risk-premium process expressed by its solution actually clears the market in the large population limit.

The organization of this chapter is as follows. We solve the optimization problem of each agent with exponential utility in Section~\ref{sec-each-optimization}.
In particular, special attention is paid to allow the unbounded risk-premium process.
In Section~\ref{sec-mean-field-BSDE}, we derive a novel mean-field BSDE and prove the existence of a solution under several conditions. 
We show that the risk-premium process characterized by the solution of this BSDE actually clears the market in the large population limit in Section~\ref{sec-market-clearing}.
We conclude the chapter by Section~\ref{sec-conclusion}.

\section{Exponential Utility Optimization for a Given Agent}
\label{sec-each-optimization}
The relevant probability spaces used in the first part of this work are given below.\\

\noindent
(1) $(\Omega^0, \mathcal{F}^0, \mathbb{P}^0)$ is a complete probability space with a complete and right-continuous filtration 
$\mathbb{F}^0:=(\mathcal{F}_t^0)_{t\in[0,T]}$ generated by a $d_0$-dimensional standard Brownian motion $W^0:=(W_t^0)_{t \in[0,T]}$.
We set $\mathcal{F}^0:=\mathcal{F}^0_T$.
This space is used to model market-wide noise and information common to all the agents.\\

\noindent
(2) $(\Omega^1, \mathcal{F}^1, \mathbb{P}^1)$ is a complete probability space with a complete and right-continuous filtration
$\mathbb{F}^1:=(\mathcal{F}_t^1)_{t\in [0,T]}$ generated by a $d$-dimensional standard Brownian motion $W^1:=(W^1_t)_{t\in [0,T]}$
and a $\sigma$-algebra $\sigma(\xi^1,\gamma^1)$, which defines $\mathcal{F}_0^1$. The $\sigma$-algebra $\mathcal{F}_0^1$ is generated by a bounded $\mathbb{R}$-valued 
random variable $\xi^1$ and a strictly positive bounded random variable $\gamma^1$. We set $\mathcal{F}^1:=\mathcal{F}^1_T$.
This space is used to model idiosyncratic noise and information for an agent (agent-1).
In later sections, we create independent copies $(\Omega^i,\mathcal{F}^i, \mathbb{P}^i)_{i\in \mathbb{N}}$ of this  space endowed with 
$\mathbb{F}^i:=(\mathcal{F}_t^i)_{t \in[0,T]}$ to model idiosyncratic information for a large number of agents, (agent-$i$, $i\in\mathbb{N}$).\\

\noindent
(3) $(\Omega^{0,1}, \mathcal{F}^{0,1}, \mathbb{P}^{0,1})$ is a probability space defined on the product set $\Omega^{0,1}:=\Omega^0\times \Omega^1$
with $(\mathcal{F}^{0,1},\mathbb{P}^{0,1})$ the completion of $(\mathcal{F}^0\otimes \mathcal{F}^1, \mathbb{P}^0\otimes \mathbb{P}^1)$.
$\mathbb{F}^{0,1}:=(\mathcal{F}^{0,1}_t)_{t\in [0,T]}$ denotes the complete and right-continuous augmentation of $(\mathcal{F}_t^0\otimes \mathcal{F}_t^1)_{t\in[0,T]}$.
A generic element of $\Omega^{0,1}$
is denoted by $\omega:=(\omega^0,\omega^1)\in \Omega^0\times \Omega^1$.\\

We set $\mathcal{T}^{0,1}:=\mathcal{T}(\mathbb{F}^{0,1})$ and $\mathcal{T}^{0}:=\mathcal{T}(\mathbb{F}^{0})$. 
Throughout the chapter, we do not distinguish a random variable defined on a marginal probability space from its trivial extension to a product space for notational simplicity.
For example, we will use a same symbol $X$ for a random variable $X(\omega^0)$ defined on 
the space $(\Omega^0,\mathcal{F}^0,\mathbb{P}^0)$  and for its trivial extension $X(\omega^0,\omega^1):=X(\omega^0)$
defined on  the product space $(\Omega^{0,1}, \mathcal{F}^{0,1},\mathbb{P}^{0,1})$.

In this section, we consider the optimization problem for an agent (agent-1) whose preference is given by an exponential 
utility. We characterize her optimal trading strategy in terms of the quadratic-growth BSDE 
by the approach proposed by Hu, Imkeller \& M\"uller~\cite{huUtilityMaximizationIncomplete2005a}.
In particular, however, in order to deal with the mean-field price formation as in \cite{fujiiMeanFieldGame2022},
we need to relax their assumption on the boundedness of the risk-premium process  $\theta:=(\theta_t)_{t\in[0,T]}$ to the unbounded one in $\mathbb{H}^2_{\rm BMO}$.

\subsection{The Market and the Utility Function}
The market dynamics and the (agent-1)'s idiosyncratic environment are modeled on the filtered probability space $(\Omega^{0,1},\mathcal{F}^{0,1},\mathbb{P}^{0,1},\mathbb{F}^{0,1})$ defined in the previous section. In this section, 
the expectation with respect to $\mathbb{P}^{0,1}$ is simply denoted by $\mathbb{E}[\cdot]$. The financial market is specified as follows.

\begin{asm} ~~\\
\label{assumption-market}
{\rm (i)} The risk-free interest rate is zero.\\
{\rm (ii)} There are $n\in \mathbb{N}$ non-dividend paying risky stocks whose price dynamics is given by
\begin{equation}
		S_t=S_0+\int_0^t {\rm diag}(S_s)(\mu_s ds+\sigma_s dW_s^0), ~~~t\in[0,T],
\label{eq-stock-price}
\end{equation}
where $S_0\in\mathbb{R}^n_{++}$ denotes the initial stock prices, $\mu:=(\mu_t)_{t\in[0,T]}$ is an $\mathbb{R}^n$-valued, $\mathbb{F}^0$-progressively
measurable process belonging to $\mathbb{H}^2_{\rm BMO}(\mathbb{P}^{0,1},\mathbb{F}^0)$~\footnote{Clearly, $\mu$ 
is also in $\mathbb{H}^2_{\rm BMO}(\mathbb{P}^{0,1},\mathbb{F}^{0,1})$. In fact,  additional information from $\mathbb{F}^1$
cannot increase the $\mathbb{H}^2_{\rm BMO}$-norm $\eqref{eq-h2bmo-norm}$ 
since $\mu$ is $\mathbb{F}^0$-adapted which is independent from $\mathbb{F}^1$.}.
$\sigma:=(\sigma_t)_{t\in[0,T]}$ is an $\mathbb{R}^{n\times d_0}$-valued, bounded, and $\mathbb{F}^0$-progressively measurable process
 such that  there exist positive constants $0<\underline{\lambda}<\overline{\lambda}$ satisfying 
\begin{equation}
\underline{\lambda}I_n\leq \sigma_t\sigma_t^\top \leq \overline{\lambda}I_n, \quad {\text{$dt\otimes \mathbb{P}^0$-a.e.}} \nonumber
\end{equation} 
Here, $I_n$ denotes $n\times n$ identity matrix.
\end{asm}

\begin{rem}
	$\mu\in \mathbb{H}^2_{\rm BMO}$ is a very strong assumption. 
    In fact, typical Gaussian processes such as the Ornstein-Uhlenbeck process are not contained in this class. 
    This is unfortunate, however, such an assumption is essential to deal with the quadratic-growth BSDE, which will be introduced later.
\end{rem}

Since the interest rate is zero, the risk-premium process  $\theta:=(\theta_t)_{t\in[0,T]}$ is defined by $\theta_t:=\sigma_t^\top
(\sigma_t\sigma_t^\top)^{-1}\mu_t$. Hence, for any $t\geq 0$, $\theta_t\in {\rm Range}(\sigma_t^\top)={\rm Ker}(\sigma_t)^\perp$.
Here, $\top$ denotes the transposition, and ${\rm Ker} (\sigma_t)^\perp$ is the orthogonal complement of ${\rm Ker}(\sigma_t)$ in $\mathbb{R}^{d_0}$.
Note that $\theta$ is in  $\mathbb{H}^2_{\rm BMO}$ due to the boundedness of the process $\sigma$. 
By the regularity of $(\sigma\sigma^\top)$, we have ${\rm rank}(\sigma)=n$. 
The financial market is incomplete even within the common information
$\mathbb{F}^0$ in general, since we have $n\leq d_0$. Recall that we have additional noise associated with $\mathbb{F}^1$ in the liability given below.

\begin{rem}
We emphasize that, although $\mu$ is unbounded, the dynamics of stock price $\eqref{eq-stock-price}$ 
is well defined. The easiest way to check this is to change the probability measure to the
risk-neutral one, which is possible since $\theta\in \mathbb{H}^2_{\rm BMO}$. Under this equivalent probability measure,
the stock price process is a uniquely specified as a martingale.
\end{rem}

\begin{dfn}
For each $s\in[0,T]$, let us denote by
\begin{equation}
L_s:=\{u^\top \sigma_s; u\in \mathbb{R}^n\} \nonumber
\end{equation}
the  linear subspace of $\mathbb{R}^{1\times d_0}$ spanned by the $n$ row vectors of $\sigma_s$.
For any $z\in \mathbb{R}^{1\times d_0}$,  $\Pi_s(z)$ denotes the orthogonal projection of $z$ onto the 
linear subspace $L_s$.
\end{dfn}

Notice that $\theta_s^\top\in L_s$ for every $s\in[0,T]$ by its construction.
\begin{rem}
The projections to the range as well as the kernel of a finite dimensional matrix are Borel measurable.
It then follows that the process $(\Pi_s(z_s))_{s\in[0,T]}$ is progressively measurable if so is the process $(\sigma_s, z_s)_{s\in[0,T]}$.
See \cite[Chapter 1; Lemma 4.4, Corollary 4.5]{karatzas_methods_1998} for the proof.
\end{rem}

The idiosyncratic environment for the agent-1 is modelled by a triple $(\xi^1, \gamma^1, F^1)$.
\begin{asm}~\\
\label{assumption-agent}
{\rm (i)} $\xi^1$ is an $\mathbb{R}$-valued, bounded, and $\mathcal{F}_0^1$-measurable random variable denoting the initial wealth for the agent-1.\\
{\rm (ii)} $\gamma^1$ is an $\mathbb{R}$-valued, bounded, and $\mathcal{F}_0^1$-measurable random variable, satisfying
\begin{equation}
\underline{\gamma}\leq \gamma^1\leq \overline{\gamma}, \nonumber
\end{equation}
with some positive constants $0<\underline{\gamma}\leq \overline{\gamma}$. $\gamma^1$ denotes the coefficient of (absolute) risk aversion of the agent-1. \\
{\rm (iii)} $F^1$ is an $\mathbb{R}$-valued, bounded, and $\mathcal{F}^{0,1}_T$-measurable random variable denoting the liability of the agent-1 at time $T$. \\
{\rm (iv)} The agent-1 has a negligible market share and hence his/her trading activities have no impact on the stock prices, i.e., 
he/she is a price taker. 
\end{asm}
\begin{rem}
Notice that the liability $F^1$ is subject to common shocks from $\mathcal{F}_T^0$ as well as idiosyncratic shocks from $\mathcal{F}_T^1$.
\end{rem}

The wealth process of the agent-$1$ under the self-financing trading strategy $\pi$ is given by
\begin{equation}
\mathcal{W}_t^{1,\pi}=\xi^1+\sum_{j=1}^n\int_0^t\frac{\pi_{j,s}}{S^j_s}dS^j_s=\xi^1+\int_0^t \pi_s^\top \sigma_s(dW_s^0+\theta_s ds),~~~t\in[0,T]. \nonumber
\end{equation}
Here, $\pi:=(\pi_t)_{t\in[0,T]}$ is an $\mathbb{R}^n$-valued, $\mathbb{F}^{0,1}$-progressively measurable process
representing the invested amount of money in each of the $n$ stocks.
The problem of the agent-$1$ is to solve
\begin{equation}
\sup_{\pi \in \mathbb{A}^1} U^1(\pi), \nonumber
\end{equation}
where the functional $U^1$ is called exponential utility (for the agent-$1$). It is defined by
\begin{equation}
U^1(\pi):=\mathbb{E}\Bigl[-\exp\Bigl(-\gamma^1 \Bigl(\xi^1+\int_0^T \pi_s^\top \sigma_s (dW_s^0+\theta_s ds)-F^1\Bigr)\Bigr)\Bigr]. 
\label{eq-utility-functional}
\end{equation}
It means that the low performance in the sense of $\mathcal{W}^{1,\pi}_T-F^1<0$ is punished heavily and
the high performance $\mathcal{W}^{1,\pi}_T-F^1>0$ is only weakly valued. 

In this thesis, we will not delve into the concrete modelling of the liability, which can include common as well as idiosyncratic shocks associated with,
for example, financial market,  endowment, consumption, local price of commodities,  and/or budgetary target imposed on the agent by his/her manager.
\begin{dfn}
\label{def-admissible-space}
The admissible space $\mathbb{A}^1$ is the set of all $\mathbb{R}^n$-valued, $\mathbb{F}^{0,1}$-progressively measurable trading strategies $\pi$
that satisfy $\mathbb{E}\Bigl[ \displaystyle\int_0^T |\pi_s^\top \sigma_s|^2ds \Bigr]<\infty$ and such that
\begin{equation}
\bigl\{\exp(-\gamma^1 \mathcal{W}_\tau^{1,\pi}); \tau \in \mathcal{T}^{0,1}\bigr\} \nonumber
\end{equation}
is uniformly integrable (i.e. of class $\mathcal{D}$). We also define $\mathcal{A}^1:=\bigl\{p=\pi^\top \sigma; \pi\in \mathbb{A}^1\bigr\}.$
\end{dfn}

\noindent
Note that $p$ is an $\mathbb{R}^{1\times d_0}$-valued process with $p_s\in L_s$ for any $s\in[0,T]$.
The problem for the agent-1 can be equivalently said to find the value function:
\begin{equation}
V^{1,*}:=\sup_{p\in \mathcal{A}^1}\mathbb{E}\Bigl[-\exp\Bigl(-\gamma^1\Bigl(\xi^1+\int_0^T p_s (dW_s^0+\theta_sds)-F^1\Bigr)\Bigr)\Bigr]. \nonumber
\end{equation}

\begin{rem}
	\label{mu-theta} Before closing this section, let us give some remarks on the choice of $\mathbb{H}^2_{\rm BMO}$ 
	for the class of $\theta$. In this work, our goal is  to find an appropriate $\theta$ (and hence $\mu$)
	to achieve an equilibrium. If we put an agent in a stochastic environment as we did, his/her demand/supply of goods becomes inevitably stochastic. 
	In order to take balance among many agents, the risk-premium process $\theta$ must also  be stochastic in general.
	This is already the case in the classical results for the complete market with finite number of agents (See Chapter 4 in Karatzas~\&~Shreve~\cite{karatzas_methods_1998}).
	In order not to miss a valuable candidate of equilibrium, 
	it is of course desirable to weaken assumptions on $\theta $ as much as possible,  but we need mathematical tractability at the same time.
	We have chosen the class of $\mathbb{H}^2_{\rm BMO}$ particularly because that it makes Dol\'{e}ans-Dade exponential uniformly integrable for {\it every finite time interval} $[0,T]$, which is necessary to justify the measure change we need. As we will see below, it also turns out to be the class of 
	the stochastic integrands $Z$ of qg-BSDEs with solutions in $\mathbb{S}^\infty$.  This plays a crucial role in maintaining the consistency of our logic. See Sections~\ref{sec-mean-field-BSDE} and \ref{sec-market-clearing}.
\end{rem}

\subsection{Characterization of the Optimal Trading Strategy}
Thanks to the work~\cite{huUtilityMaximizationIncomplete2005a}, we can characterize the optimal trading strategy by using a solution to a certain BSDE (instead of FBSDEs)
in a rather straightforward way. When a utility (or equivalently cost) function has a special homothetic form as in the current case,
their method often provides much simpler description of the optimal strategy than in the case where the  Pontryagin's maximum principle is applied.

We try to construct a family of stochastic processes $\{R^p:=(R^p_t)_{t\in[0,T]}, p\in \mathcal{A}^1\}$ satisfying the following properties:
\begin{dfn} 
(Condition-R) \\
{\rm (i)} $R_T^p=-\exp\bigl(-\gamma^1 (\mathcal{W}_T^{1,p}-F^1)\bigr)$ a.s. for all $p\in\mathcal{A}^1$.\\
{\rm (ii)} $R_0^p=R_0$ a.s. for all $p\in \mathcal{A}^1$ and for some $\mathcal{F}_0^1~(=\mathcal{F}^{0,1}_0)$-measurable random variable $R_0$.\\
{\rm (iii)} $R^p$ is an $(\mathbb{F}^{0,1},\mathbb{P}^{0,1})$-supermartingale for all $p\in \mathcal{A}^1$, and there exists some $p^*\in \mathcal{A}^1$
such that $R^{p^*}$ is an $(\mathbb{F}^{0,1},\mathbb{P}^{0,1})$-martingale.
\end{dfn}
\noindent
In fact, if we can find such a family $\{R^p\}$, then for any $p\in \mathcal{A}^1$, we have
\begin{equation}
\mathbb{E}\bigl[-\exp\bigl(-\gamma^1(\mathcal{W}_T^{1,p}-F^1)\bigr)\bigr]\leq \mathbb{E}[R_0]=\mathbb{E}\bigl[-\exp\bigl(-\gamma^1(\mathcal{W}_T^{1,p^*}-F^1)\bigr)\bigr], \nonumber
\end{equation}
and hence $p^*$ is an optimal trading strategy for the agent-1.

In order to construct the family $\{R^p\}$, we try to find an appropriate process $Y:=(Y_t)_{t\in[0,T]}$ with
which the process $R^p$ is given by 
\begin{equation}
R_t^p=-\exp\bigl(-\gamma^1(\mathcal{W}_t^{1,p}-Y_t)\bigr), ~~~ t\in[0,T], ~~~p\in \mathcal{A}^1.
\label{eq-R-hypo}
\end{equation} 
Here, the triple $(Y,Z^0,Z^1)$, which is an $(\mathbb{R},\mathbb{R}^{1\times d_0}, \mathbb{R}^{1\times d})$-valued process, 
is an $\mathbb{F}^{0,1}$-adapted solution to the BSDE
\begin{equation}
Y_t=F^1+\int_t^T f(s,Z_s^0,Z_s^1)ds-\int_t^T Z_s^0 dW_s^0-\int_t^T Z_s^1 dW_s^1,~~~ t\in[0,T]. \nonumber
\end{equation}
The concrete form of  the driver $f$ is to be determined below so that $\{R^p\}$ satisfies the desired properties.  

Under the hypothesis of $\eqref{eq-R-hypo}$, we get, by Ito formula,
\begin{equation}
\begin{split}
dR_t^p&=R_t^p\Bigl(-{\gamma^1} d(\mathcal{W}_t^{1,p}-Y_t)+\frac{(\gamma^1)^2}{2}d\langle \mathcal{W}^{1,p}-Y\rangle_t\Bigr)\\
&=R_t^p\Bigl(-\gamma^1 (p_t\theta_t+f(t,Z_t^0,Z_t^1))+\frac{(\gamma^1)^2}{2}(|p_t-Z_t^0|^2+|Z_t^1|^2)\Bigr)dt\\
&\quad+R_t^p\bigl(-\gamma^1(p_t-Z_t^0)dW_t^0+\gamma^1 Z_t^1 dW_t^1\bigr), \quad t\in[0,T]. \nonumber
\end{split}
\end{equation} 
In order to guess an appropriate form of $f$, let us formally solve it as
\begin{equation}
\begin{split}
R_t^p&=-\exp\bigl(-{\gamma^1}(\xi^1-Y_0)\bigr)\exp\Bigl(\int_0^t\bigl[-{\gamma^1}(p_s\theta_s+f(s,Z_s^0,Z_s^1))+\frac{({\gamma^1})^2}{2}(|p_s-Z_s^0|^2+|Z_s^1|^2)\bigr]ds\Bigr)\\
&\quad \times \mathcal{E}\Bigl(\int_0^\cdot\bigl[-{\gamma^1}(p_s-Z_s^0)dW_s^0+{\gamma^1} Z_s^1 dW_s^1\bigr]\Bigr)_t. \nonumber
\end{split}
\end{equation}
We search for the driver $f(s,Z_s^0, Z_s^1)$ that satisfies
\begin{itemize}
\item $-{\gamma^1}(p_s\theta_s+f(s,Z_s^0,Z_s^1))+\dfrac{({\gamma^1})^2}{2}(|p_s-Z_s^0|^2+|Z_s^1|^2)\geq 0$ for all $p\in \mathcal{A}^1$,
\item There exists $p^*\in \mathcal{A}^1$ such that $-{\gamma^1}(p_s^*\theta_s+f(s,Z_s^0,Z_s^1))+\dfrac{({\gamma^1})^2}{2}(|p_s^*-Z_s^0|^2+|Z_s^1|^2)=0$
\end{itemize}
for all $s\in[0,T]$. The above conditions suggest that
\begin{equation}
\begin{split}
f(s,Z_s^0,Z_s^1)&=\inf_{p_s\in L_s}\Bigl\{-p_s\theta_s+\frac{{\gamma^1}}{2}(|p_s-Z_s^0|^2+|Z_s^1|^2)\Bigr\}\\
&=\inf_{p_s\in L_s}\Bigl\{ \frac{\gamma^1}{2}\Bigl|p_s-\Bigl(Z_s^0+\frac{\theta_s^\top}{\gamma^1}\Bigr)\Bigr|^2-Z_s^0\theta_s-\frac{1}{2\gamma^1}|\theta_s|^2
+\frac{\gamma^1}{2}|Z_s^1|^2\Bigr\}. \nonumber
\end{split}
\end{equation}
This is a special case treated by \cite{huUtilityMaximizationIncomplete2005a}[Section 2] 
with a trading constraint $\pi_t\in \widetilde{C}$ by a general closed subset $\widetilde{C}\subset \mathbb{R}^{n}$,
which is now replaced by the whole space $\mathbb{R}^{n}$. A candidate for the optimal strategy $p^*$ is then given by
\begin{equation}
p^*_t=Z_t^{0\|}+\frac{\theta_t^\top}{{\gamma^1}}, ~~~t\in[0,T].
\label{optimal-p}
\end{equation}
Here, for notational simplicity, we have written $Z_s^{0\|}:=\Pi_s(Z_s^0)$ and $Z_s^{0\perp}:=Z_s^0-\Pi_s(Z_s^0)$.
They are orthogonal to each other and $|Z_s^0|^2=|Z_s^{0\|}|^2+|Z_s^{0\perp}|^2$. Recall that $\Pi_s(\theta_s^\top)=\theta_s^\top$ for every $s\in[0,T]$.
With this convention, we have
\begin{equation}
\begin{split}
f(s,Z_s^0,Z_s^1)&=-Z_s^0\theta_s-\frac{1}{2\gamma^1}|\theta_s|^2+\frac{\gamma^1}{2}(|Z_s^{0\perp}|^2+|Z_s^1|^2 )\\
&=-Z_s^{0\|}\theta_s-\frac{1}{2\gamma^1}|\theta_s|^2+\frac{\gamma^1}{2}(|Z_s^{0\perp}|^2+|Z_s^1|^2 ),
\end{split}
\end{equation}
where we used the fact $\theta_s^\top \in L_s$ and hence $Z_s^{0\perp}\theta_s=0$ in the second equality.

Therefore, the associated qg-BSDE is given by
\begin{equation}
Y_t=F^1+\int_t^T \Bigl(-Z_s^{0\|}\theta_s-\frac{|\theta_s|^2}{2{\gamma^1}}+\frac{{\gamma^1}}{2}(|Z_s^{0\perp}|^2+|Z_s^1|^2)\Bigr)ds-\int_t^T Z_s^0 dW_s^0-\int_t^T Z_s^1 dW_s^1,
~~~t\in[0,T]. 
\label{BSDE-org}
\end{equation}
For simplicity, we rewrite the equation with $G^1:={\gamma^1} F^1$, and $(y, z^0,z^1):=({\gamma^1} Y,{\gamma^1} Z^0,{\gamma^1} Z^1)$.
Then, we can equivalently work on the normalized BSDE,
\begin{equation}
y_t=G^1+\int_t^T\Bigl(-z_s^{0\|}\theta_s-\frac{1}{2}|\theta_s|^2+\frac{1}{2}(|z_s^{0\perp}|^2+|z_s^1|^2)\Bigr)ds-\int_t^T z_s^0 dW_s^0-\int_t^T z_s^1 dW_s^1,~~~t\in[0,T].
\label{BSDE-norm}
\end{equation}
In this case, $p^*$ is given by $\gamma^1 p^*_t=z_t^{0\|}+\theta_t^\top$. There should be no confusion which BSDE is being discussed by checking the terminal function and the presence of $\gamma^1$.
In order to conclude that they are actually  what we want, we need to verify that the resultant family $\{R^p\}$ $\eqref{eq-R-hypo}$ and the process $p^*$ $\eqref{optimal-p}$ satisfy (Condition-R).

\subsection{Solution of the BSDE and Its Verification}

We  emphasize that, in contrast to the work~\cite{huUtilityMaximizationIncomplete2005a}, our risk-premium process $\theta\in \mathbb{H}^2_{\rm BMO}$ is unbounded.
As we will see in later sections, this generalization is necessary to handle the mean-field market clearing equilibrium.
Due to this unbounded risk-premium process,  we cannot apply the  standard results on qg-BSDEs given by Kobylanski~\cite{Kobylanski2000BackwardSD}. 
Moreover, since the exponential integrability of $(|\theta_t|^2, t\in[0,T])$
is not guaranteed in general, we cannot apply the extensions on the qg-BSDE theories such as \cite{briandBSDEQuadraticGrowth2006, briandQuadraticBSDEsConvex2008, Hu-Tang}, either.
In particular, the case $(|\theta_t|^{1+\alpha}, \alpha<1)$ is covered by the result in $\cite{Hu-Tang}$ but not the case where 
$|\theta_t|^2$-term is contained in the driver. 
Fortunately, thanks to the special form of its driver inherited from the exponential utility, 
we can show the existence of a unique solution $(y,z^0,z^1)$ to the BSDE $\eqref{BSDE-norm}$ (and equivalently $\eqref{BSDE-org})$
in the space $\mathbb{S}^\infty\times \mathbb{H}^2_{\rm BMO}\times 
\mathbb{H}^2_{\rm BMO}$ by a simple modification of the standard approach \cite{Kobylanski2000BackwardSD}.

\begin{lem}
\label{lemma-Z-norm}
Let Assumptions~\ref{assumption-market} and \ref{assumption-agent} be in force.
If there exists a bounded (with respect to the $y$-component) solution, i.e. $(y,z^0,z^1)\in \mathbb{S}^\infty\times \mathbb{H}^2 \times \mathbb{H}^2$,
to the BSDE $\eqref{BSDE-norm}$, then $z:=(z^0,z^1)$ is in $\mathbb{H}^2_{\rm BMO}$.  
\end{lem}

\noindent
\textbf{\textit{Proof}}\\
By Ito formula, we have,
\begin{equation}
\begin{split}
de^{2y_t}&=e^{2 y_t}\bigl(2z_t^{0\|}\theta_t+ |\theta_t|^2-\bigl(|z_t^{0\perp}|^2+|z_t^1|^2)+2
(|z_t^0|^2+|z_t^1|^2)\bigr)dt\\
&\quad+2 e^{2y_t}(z_t^0 dW_t^0+z_t^1 dW_t^1)\\
&\geq e^{2y_t}\bigl(-|z_t^{0\|}|^2-\bigl(|z_t^{0\perp}|^2+|z_t^1|^2)+2(|z_t^0|^2+|z_t^1|^2)\bigr)dt\\
&\quad+2 e^{2 y_t}(z_t^0 dW_t^0+z_t^1 dW_t^1)\\
&\geq e^{2y_t}(|z_t^0|^2+|z_t^1|^2)dt+2 e^{2y_t}(z_t^0dW_t^0+z_t^1 dW_t^1)
\end{split}
\nonumber
\end{equation}
and thus, for any $t\in[0,T]$,
\begin{equation}
e^{2 y_T}-e^{2y_t}\geq \int_t^T e^{2 y_s}(|z_s^0|^2+|z_s^1|^2)ds+\int_t^T 2 e^{2 y_s}(z_s^0dW_s^0+z_s^1 dW_s^1). \nonumber
\end{equation}
It is easy to obtain $\|z\|^2_{\mathbb{H}^2_{\rm BMO}}:=\sup_{\tau\in\mathcal{T}^{0,1}}\Bigl\|\mathbb{E}\Bigl[\int_\tau^T (|z^0_s|^2+|z_s^1|^2)ds|\mathcal{F}_\tau^{0,1}\Bigr]\Bigr\|_{\infty}\leq e^{4 \|y\|_{\mathbb{S}^\infty}}$. $\square$

The above lemma is now used to guarantee the uniqueness of solution if $y \in \mathbb{S}^\infty$.
\begin{thm}
\label{th-uniqueness}
Let Assumptions~\ref{assumption-market} and \ref{assumption-agent} be in force.
If the solution to $\eqref{BSDE-norm}$ is bounded, i.e. $(y,z^0,z^1)\in \mathbb{S}^\infty\times \mathbb{H}^2\times \mathbb{H}^2$, then 
such a solution is unique.
\end{thm}

\noindent
\textbf{\textit{Proof}}\\
Suppose that there are two bounded solutions $(y,z^0,z^1)$ and $(\acute{y},\acute{z}^0,\acute{z}^1)$.
By Lemma~\ref{lemma-Z-norm}, we know  that $(z^0,z^1)$ and $(\acute{z}^0, \acute{z}^1)$ are actually in $\mathbb{H}^2_{\rm BMO}$.
Let us put;
\begin{equation}
\Delta y_t:=y_t-\acute{y}_t, \quad \Delta z_t^0:=z_t^0-\acute{z}_t^0, \quad \Delta z_t^1:=z_t^1-\acute{z}_t^1. \nonumber
\end{equation}
From the orthogonality between $z^{0\|}$ and $z^{0\perp}$, we have
\begin{equation}
\begin{split}
\Delta y_t&=\int_t^T\Bigl(-\Delta z_s^{0\|} \theta_s+\frac{1}{2}\Delta (z_s^{0\perp})(z_s^{0\perp}+\acute{z}_s^{0\perp})^\top+
\frac{1}{2}\Delta z_s^1(z_s^1+\acute{z}_s^1)^\top\Bigr)ds-\int_t^T \Delta z_s^0 dW_s^0-\int_t^T \Delta z_s^1 dW_s^1\\
&=\int_t^T\Bigl(-\Delta z_s^0 \Bigl(\theta_s-\frac{1}{2}(z_s^{0\perp}+\acute{z}_s^{0\perp})^\top\Bigr)+
\frac{1}{2}\Delta z_s^1(z_s^1+\acute{z}_s^1)^\top\Bigr)ds-\int_t^T \Delta z_s^0 dW_s^0-\int_t^T \Delta z_s^1 dW_s^1\\
&=-\int_t^T \Delta z_s^0\Bigl(dW_s^0+\bigl(\theta_s-\frac{1}{2}(z_s^{0\perp}+\acute{z}_s^{0\perp})^\top\bigr)ds\Bigr)
-\int_t^T \Delta z_s^1\Bigl(dW_s^1-\frac{1}{2}(z_s^1+\acute{z}_s^1)^\top ds\Bigr)\\
&=-\int_t^T \Delta z_s^0d\widetilde{W}_s^0
-\int_t^T \Delta z_s^1d\widetilde{W}_s^1,  \nonumber
\end{split}
\end{equation}
where we have defined a new measure $\widetilde{\mathbb{P}}$ equivalent to $\mathbb{P}^{0,1}$ by
\begin{equation}
\frac{d\widetilde{\mathbb{P}}}{d\mathbb{P}^{0,1}}\Bigr|_{\mathcal{F}^{0,1}_t}:=M_t:=\mathcal{E}\Bigl(-\int_0^\cdot \bigl(\theta_s^\top-
\frac{1}{2}(z_s^{0\perp}+\acute{z}_s^{0\perp})\bigr) dW_s^0+\int_0^\cdot \frac{1}{2}(z_s^1+\acute{z}_s^1)dW_s^1\Bigr)_t,~~~t\in[0,T]
\nonumber
\end{equation}
and 
$
(\widetilde{W}_t^0, \widetilde{W}_t^1)_{t\in[0,T]}
$
denote the standard Brownian motions under $\widetilde{\mathbb{P}}$. 
This measure change is well-defined since $(\theta^\top, z^{0\perp}+\acute{z}^{0\perp}, z^1+\acute{z}^1)$ are in $\mathbb{H}^2_{\rm BMO}$ and hence $M$ is a uniformly integrable martingale.
By the result of Kazamaki~\cite{kazamaki_sufficient_1979} and \cite{kazamakiContinuousExponentialMartingales1994}[Remark 3.1], the following so-called reverse H\"older inequality holds:
\begin{equation}
\mathbb{E}\bigl[M_T^r|\mathcal{F}_t^{0,1}\bigr]\leq C M_t^r, \nonumber
\end{equation}
where $C>0$ and $r>1$ are some constants depending only on the $\mathbb{H}^2_{\rm BMO}$-norm of $(\theta^\top, z^{0\perp}+\acute{z}^{0\perp}, z^1+\acute{z}^1)$. 
With $q=\dfrac{r}{r-1}>1$ and $j=0,1$, H\"older and the energy inequality \eqref{ineq-energy} imply
\begin{equation}
\begin{split}
\mathbb{E}^{\widetilde{\mathbb{P}}}\Bigl[\int_0^T |\Delta z_s^j|^2ds\Bigr]&=\mathbb{E}\Bigl[M_T\Bigl(\int_0^T |\Delta z_s^j|^2 ds\Bigr)\Bigr]\\
&\leq \mathbb{E}[M_T^r]^\frac{1}{r}\mathbb{E}\Bigl[\Bigl(\int_0^T |\Delta z_s^j|^2 ds\Bigr)^q\Bigr]^\frac{1}{q}<\infty. \nonumber
\end{split}
\end{equation}
Thus $\Delta y$ is an $(\mathbb{F}^{0,1},\widetilde{\mathbb{P}})$-martingale.
Thus we can conclude  that $\Delta y=0$ and so are $(\Delta z^0, \Delta z^1)$. $\square$

Since $\theta$ is in $\mathbb{H}^2_{\rm BMO}$, it is natural to change the measure to absorb the term $(-z^{0\|}\theta)~(=-z^0\theta)$ in the driver of $\eqref{BSDE-norm}$.
Let us define the measure $\mathbb{Q}~(\sim \mathbb{P}^{0,1})$ by
\begin{equation}
\frac{d\mathbb{Q}}{d\mathbb{P}^{0,1}}\Bigr|_{\mathcal{F}^{0,1}_t}:=\mathcal{E}\Bigl(-\int_0^\cdot \theta_s^\top dW_s^0\Bigr)_t, \nonumber
\end{equation}
where the standard Brownian motions under $\mathbb{Q}$ are given by
\begin{equation}
\widetilde{W}_t^0=W_t^0+\int_0^t \theta_s ds, \quad \widetilde{W}_t^1=W_t^1, \quad t\in[0,T]. \nonumber
\end{equation}
Therefore, instead of $\eqref{BSDE-norm}$, we can equivalently work on the BSDE defined on $(\Omega^{0,1},\mathcal{F}^{0,1},\mathbb{Q},\mathbb{F}^{0,1})$
endowed with the Brownian motions $(\widetilde{W}^0, \widetilde{W}^1)$;
\begin{equation}
y_t=G^1+\int_t^T \Bigl(-\frac{1}{2}|\theta_s|^2+\frac{1}{2}(|z_s^{0\perp}|^2+|z_s^1|^2)\Bigr)ds-\int_t^T z_s^0d\widetilde{W}^0_s-\int_t^T z_s^1d\widetilde{W}_s^{1},
~~~t\in[0,T].
\label{BSDE-norm-Q}
\end{equation}
Although in general, the filtration $\mathbb{F}^{0,1}$ is larger than the one generated by $(\widetilde{W}^0,\widetilde{W}^1)$,
we can still apply the standard techniques of BSDEs. This is  due to 
the stability property of the martingale representation under the absolutely continuous 
measure changes. See \cite{HWY}[Theorem 13.12] for general case and \cite{jeanblanc_mathematical_2009}[Section 1.7.7] for Brownian case.
Moreover, by Kazamaki~\cite{kazamakiContinuousExponentialMartingales1994}[Theorem 3.3], $\theta$ is still in $\mathbb{H}^2_{\rm BMO}(\mathbb{Q},\mathbb{F}^{0,1})$.
Obviously, BSDE $\eqref{BSDE-norm}$ (and equivalently $\eqref{BSDE-org})$ has a bounded solution if and only if  BSDE $\eqref{BSDE-norm-Q}$ has a bounded solution.  

\begin{thm}
Let Assumptions~\ref{assumption-market} and \ref{assumption-agent} be in force.
Then there is a unique bounded solution $(y,z^0,z^1)\in \mathbb{S}^\infty(\mathbb{Q},\mathbb{F}^{0,1})\times \mathbb{H}^2_{\rm BMO}(\mathbb{Q},\mathbb{F}^{0,1})\times \mathbb{H}^2_{\rm BMO}(\mathbb{Q},\mathbb{F}^{0,1})$ to the BSDE $\eqref{BSDE-norm-Q}$.
\end{thm}

\noindent
\textbf{\textit{Proof}}\\
Since we work under the measure $\mathbb{Q}$ throughout this proof,  we write $\mathbb{E}[\cdot]$ instead of $\mathbb{E}^{\mathbb{Q}}[\cdot]$
for notational simplicity.
Firstly, for each $n\in \mathbb{N}$, we consider the truncated BSDE;
\begin{equation}
y_t^n=G^1+\int_t^T \Bigl(-\frac{1}{2}(|\theta_s|^2\wedge n)+\frac{1}{2}(|z_s^{n,0\perp}|^2+|z_s^{n,1}|^2)\Bigr)ds-\int_t^T z_s^{n,0}d\widetilde{W}_s^0
-\int_t^T z_s^{n,1}d\widetilde{W}_s^1, ~~~ t\in[0,T].  
\label{BSDE-truncated}
\end{equation}
Clearly, the truncated BSDE $\eqref{BSDE-truncated}$ has a 
unique bounded solution $(y^n,z^{n,0},z^{n,1})\in \mathbb{S}^\infty\times \mathbb{H}^2_{\rm BMO}\times \mathbb{H}^2_{\rm BMO}$ by 
the standard result of \cite{Kobylanski2000BackwardSD}. 
Moreover, by the comparison principle\footnote{The comparison principle tells that the size of the solution $y$
responds monotonically with respect to that of the terminal and driver functions.
See, for example, \cite{ZhangBSDE}[Theorem~7.3.1].} obtained in the same work, we have $y^{n+1}\leq y^n$ for all 
$n\in \mathbb{N}$~\footnote{Here, we use  $|z_s^{n,0\perp}|^2-|z_s^{n+1,0\perp}|^2=\Delta (z_s^{n,0\perp})(z_s^{n,0\perp}+z_s^{n+1,0\perp})^\top
=\Delta z_s^{n,0}(z_s^{n,0\perp}+z_s^{n+1,0\perp})^\top$ 
to absorb it into the stochastic integral.}.
In particular, uniformly in $n\in \mathbb{N}$, the solution $y^n$ is bounded from above as $y^n\leq \overline{y}$, where $\overline{y}$ is the solution to 
another quadratic-growth BSDE;
\begin{equation}
\overline{y}_t=G^1+\int_t^T \frac{1}{2}\bigl(|\overline{z}_s^{0}|^2+|\overline{z}_s^1|^2\bigr)ds-\int_t^T \overline{z}_s^0d\widetilde{W}_s^0-\int_t^T \overline{z}_s^1 d\widetilde{W}_s^1, ~~~ t\in[0,T]. \nonumber
\end{equation}
It also has a unique bounded solution $(\overline{y},\overline{z}^0,\overline{z}^1)\in \mathbb{S}^\infty\times \mathbb{H}^2_{\rm BMO}\times \mathbb{H}^2_{\rm BMO}$ 
with $\|\overline{y}\|_{\mathbb{S}^\infty}\leq \|G^1\|_{\infty}$  by the standard result.

Once again, by the comparison principle, $y^n$  is also bounded from below uniformly in $n\in\mathbb{N}$ as $\underline{y}\leq y^n$,
where $\underline{y}$ is the solution to the next simple BSDE;
\begin{equation}
\underline{y}_t=G^1+\int_t^T \Bigl(-\frac{1}{2}|\theta_s|^2 \Bigr)ds-\int_t^T \underline{z}_s^0 d\widetilde{W}_s^0-\int_t^T \underline{z}_s^1 d\widetilde{W}_s^1, \quad t\in[0,T]. \nonumber
\end{equation}
Obviously, it has a unique solution $(\underline{y}, \underline{z}^0,\underline{z}^1)\in \mathbb{S}^2\times \mathbb{H}^2\times \mathbb{H}^2$.
Moreover, for any $t\in[0,T]$, 
\begin{equation}
\begin{split}
\underline{y}_t&=\mathbb{E}[G^1|\mathcal{F}^{0,1}_t]-\frac{1}{2}\mathbb{E}\Bigl[\int_t^T |\theta_s|^2 ds|\mathcal{F}^{0,1}_t\Bigr] \\
&\geq -\Bigl(\|G^1\|_{\infty}+\frac{1}{2}\|\theta\|^2_{\mathbb{H}^2_{\rm BMO}}\Bigr)>-\infty. \nonumber
\end{split}
\end{equation}
Hence we conclude that, uniformly in $n\in \mathbb{N}$, $y^n$ satisfies the following bound,
\begin{equation}
-\Bigl(\|G^1\|_{\infty}+\frac{1}{2}\|\theta\|^2_{\mathbb{H}^2_{\rm BMO}}\Bigr)\leq y^n\leq \|G^1\|_{\infty}. 
\label{yn-bound}
\end{equation}

Since $\{y^n\}$ is bounded from below and it is monotonically decreasing in $n\in \mathbb{N}$, we can define 
the process $y:=(y_t)_{t\in[0,T]}$ by
\begin{equation}
y=\lim_{n\rightarrow \infty }y^n. \nonumber
\end{equation}
Moreover, by repeating the proof of Lemma~\ref{lemma-Z-norm}, we get from the estimate $\eqref{yn-bound}$
\begin{equation}
\|(z^{n,0},z^{n,1})\|_{\mathbb{H}^2_{\rm BMO}}^2\leq \exp\bigl(4\|G^1\|_{\infty}+2\|\theta\|^2_{\mathbb{H}^2_{\rm BMO}}\bigr)
\end{equation}
for all $n\in \mathbb{N}$. In particular, $(z^{n,0}, z^{n,1})_{n\in \mathbb{N}}$ are weakly relatively compact in $\mathbb{H}^2$ and hence, under an appropriate subsequence (still denoted by $n$), there exists $(z^0,z^1)\in \mathbb{H}^2\times \mathbb{H}^2$, such that
\begin{equation}
z^{n,0}\rightharpoonup z^0, \quad z^{n,1}\rightharpoonup z^1\quad {\text{ weakly in $\mathbb{H}^2$ as $n\rightarrow \infty$}}. \nonumber
\end{equation}
The remaining procedures to show the triple $(y,z^0,z^1)$ actually solves $\eqref{BSDE-norm-Q}$ are the same as those in \cite{Kobylanski2000BackwardSD}.
See also \cite{Cvitanic-Zhang}[Section 9.6]. For readers' convenience, we shall give the details below.

Let us define a smooth convex function $\phi:\mathbb{R}\rightarrow \mathbb{R}_+$ satisfying
\begin{equation}
\phi(0)=0, \quad \phi^\prime(0)=0, \nonumber
\end{equation} 
whose concrete form is to be determined later.
We consider $m, n\in \mathbb{N}$ such that $m\geq n$ and put
\begin{equation}
\Delta y^{n,m}:=y^n-y^m, \quad \Delta z^{n,m;0}:=z^{n,0}-z^{m,0}, \quad \Delta z^{n,m;1}:=z^{n,1}-z^{m,1}. \nonumber
\end{equation}
Note that $\Delta y^{n,m}\geq 0$ since $m\geq n$. From Ito formula, we obtain for any $t\in[0,T]$,
\begin{equation}
\begin{split}
&\phi(\Delta y_t^{n,m})+\int_t^T \frac{1}{2}\phi^{\prime\prime}(\Delta y_s^{n,m})(|\Delta z_s^{n,m;0}|^2+|\Delta z_s^{n,m;1}|^2)ds\\
&=\int_t^T\phi^\prime(\Delta y^{n,m}_s)\bigl[-\frac{1}{2}(|\theta_s|^2\wedge n)+\frac{1}{2}(|z_s^{n,0\perp}|^2+|z^{n,1}_s|^2)+\frac{1}{2}(|\theta_s|^2\wedge m)-\frac{1}{2}(|z_s^{m,0\perp}|^2+|z_s^{m,1}|^2)\bigr]ds\\
&\quad-\int_t^T \phi^\prime(\Delta y_s^{n,m})(\Delta z_s^{n,m;0}d\widetilde{W}_s^0+\Delta z_s^{n,m;1}d\widetilde{W}_s^1).  \nonumber
\end{split}
\end{equation}
Since $\phi(y), \phi^\prime(y)\geq 0$ for all $y\geq 0$, we get
\begin{equation}
\begin{split}
&\mathbb{E}\int_0^T \frac{1}{2}\phi^{\prime\prime}(\Delta y^{n,m}_s)(|\Delta z_s^{n,m;0}|^2+|\Delta z_s^{n,m;1}|^2)ds\\
&\quad \leq \mathbb{E}\int_0^T \frac{1}{2}\phi^\prime(\Delta y_s^{n,m})\bigl(|\theta_s|^2+|z_s^{n,0}|^2+|z_s^{n,1}|^2\bigr)ds\\
&\quad \leq \mathbb{E}\int_0^T \phi^\prime(\Delta y_s^{n,m})\bigl(|\theta_s|^2+|z_s^{n,0}-z_s^0|^2+|z_s^{n,1}-z_s^1|^2+|z_s^0|^2+|z_s^1|^2\bigr)ds. 
\end{split}
\label{phi-prpr}
\end{equation}
We now choose the function $\phi$ as 
\begin{equation}
\phi(y):=\frac{1}{8}[e^{4y}-4y-1], \nonumber
\end{equation}
which gives $\phi^\prime(y)=\frac{1}{2}[e^{4y}-1]$ and $\phi^{\prime\prime}(y)=2e^{4y}$. In particular, 
we have $\phi^{\prime\prime}(y)=4\phi^\prime(y)+2$. This yields, from $\eqref{phi-prpr}$,
\begin{equation}
\begin{split}
&\mathbb{E}\int_0^T[2\phi^\prime(\Delta y_s^{n,m})+1](|\Delta z^{n,m;0}_s|^2+|\Delta z_s^{n,m;1}|^2)ds\\
&\quad\leq \mathbb{E}\int_0^T \phi^\prime(\Delta y^{n,m}_s) (|\theta_s|^2+|z_s^{n,0}-z_s^0|^2+|z_s^{n,1}-z_s^1|^2+|z_s^0|^2+|z_s^1|^2)ds.
\end{split}
\label{phi-pr}
\end{equation}

Note that, since $(\Delta y^{n,m})_{m\geq n}$ are bounded and strongly convergent $\Delta y^{n,m}\rightarrow y^n-y$ as $m\rightarrow \infty$,
we also have, under an appropriate subsequence (still denoted by $m$), the following weak convergence in $\mathbb{H}^2$;
\begin{equation}
\sqrt{2\phi^\prime(\Delta y^{n,m})+1}\Delta z^{n,m;j}\rightharpoonup \sqrt{2\phi^\prime(y^n-y)+1}(z^{n,j}-z^j), \quad {\text{as $m\rightarrow \infty$}}, \nonumber
\end{equation}
with $j=0,1$. Hence, by \cite{Brezis}[Proposition 3.5] and monotone convergence, we obtain from $\eqref{phi-pr}$,
\begin{equation}
\begin{split}
&\mathbb{E}\int_0^T [2\phi^\prime(y^n_s-y_s)+1](|z_s^{n,0}-z_s^0|^2+|z_s^{n,1}-z_s^1|^2)ds\\
&\quad \leq \liminf_{m\rightarrow \infty}\mathbb{E}\int_0^T [2\phi^\prime(\Delta y_s^{n,m})+1](|\Delta z_s^{n,m;0}|^2+|\Delta z_s^{n,m;1}|^2)ds\\
&\quad \leq \mathbb{E}\int_0^T\phi^\prime(y^n_s-y_s)(|\theta_s|^2+|z_s^{n,0}-z_s^0|^2+|z_s^{n,1}-z_s^1|^2+|z_s^0|^2+|z_s^1|^2)ds. \nonumber
\end{split}
\end{equation}
By rearranging the $|z^{n,j}-z^j|^2$-terms $(j=0,1)$, we get
\begin{equation}
\begin{split}
&\mathbb{E}\int_0^T [\phi^\prime(y^n_s-y_s)+1](|z_s^{n,0}-z_s^0|^2+|z_s^{n,1}-z_s^1|^2)ds\\
&\quad \leq  \mathbb{E}\int_0^T\phi^\prime(y^n_s-y_s)(|\theta_s|^2+|z_s^0|^2+|z_s^1|^2)ds. \nonumber
\end{split}
\end{equation}
Since the right-hand side converges to zero as $n\rightarrow \infty$ by the  monotone convergence theorem, we obtain
\begin{equation}
z^{n,0}\rightarrow z^0, \quad z^{n,1}\rightarrow z^1, \quad {\text{strongly in $\mathbb{H}^2$. }} \nonumber
\end{equation}
Then, from Burkholder-Davis-Gundy (BDG) inequality~\footnote{See, for example, \cite{Protter}[Thorem 48 in IV]. }, 
it implies that the following convergence holds for $(j=0,1)$ under an appropriate subsequence,   
\begin{equation}
\sup_{t\in[0,T]}\Bigl|\int_t^T(z_s^{n,j}-z_s^j)d\widetilde{W}_s^j\Bigr|\rightarrow 0, \quad {\text{$\mathbb{Q}$-a.s. as $n\rightarrow \infty$}} \nonumber
\end{equation}
so is the case for $\sup_{t\in[0,T]}|y^n_t-y_t|$. It is now easy to see $(y,z^0,z^1)\in \mathbb{S}^\infty\times \mathbb{H}^2_{\rm BMO}\times 
\mathbb{H}^2_{\rm BMO}$ actually solves the BSDE $\eqref{BSDE-norm-Q}$.
The uniqueness of the solution follows exactly in the same way as in Theorem~\ref{th-uniqueness}. $\square$

\begin{cor}
\label{corollary-existence}
Let Assumptions~\ref{assumption-market} and \ref{assumption-agent} be in force.
Then the BSDE $\eqref{BSDE-org}$ $(\text{resp.} ~\eqref{BSDE-norm} )$ has a
unique bounded solution $(Y,Z^0,Z^1)$ $(\text{resp.}~ (y,z^0,z^1))$ in $\mathbb{S}^\infty(\mathbb{P}^{0,1},\mathbb{F}^{0,1})\times \mathbb{H}^2_{\rm BMO}(\mathbb{P}^{0,1},\mathbb{F}^{0,1})\times  \mathbb{H}^2_{\rm BMO}(\mathbb{P}^{0,1},\mathbb{F}^{0,1})$.
\end{cor}

We are now ready to verify the Condition-R.
\begin{thm}
	\label{verification-part1}
Let Assumptions~\ref{assumption-market} and \ref{assumption-agent} be in force. Then
the family of processes $\{R^p, p\in \mathcal{A}^1 \}$ defined by $\eqref{eq-R-hypo}$ with the process $Y$ as the unique bounded solution to
the BSDE $\eqref{BSDE-org}$ satisfies the Condition-R, and the process $p^*$ given by $\eqref{optimal-p}$
gives the unique $($up to $dt\otimes \mathbb{P}^{0,1}$-null set$)$ optimal trading strategy for the agent-1.
\end{thm}

\noindent
\textbf{\textit{Proof}}\\
From $\eqref{eq-R-hypo}$, we have
\begin{equation}
R_t^p=-\exp\bigl(-\gamma^1(\mathcal{W}_t^{1,p}-Y_t)\bigr)=-\exp\bigl(-{\gamma^1}\mathcal{W}_t^{1,p}+y_t\bigr), ~t\in[0,T], \nonumber
\end{equation}
and 
\begin{equation}
R_0^p=-\exp(-\gamma^1 \xi^1+y_0) \nonumber
\end{equation} 
for all $p\in\mathcal{A}^1$. 
Here, $y:=(y_t)_{t\in[0,T]}$ is the rescaled solution of $\eqref{BSDE-norm}$.
Since $y\in \mathbb{S}^\infty$, $(R_t^p, t\in[0,T])$ is clearly of class $\mathcal{D}$ for any $p\in\mathcal{A}^1$  by the definition of admissibility $\mathcal{A}^1$.

Let us choose $p=p^*$ as in $\eqref{optimal-p}$. Then we have
\begin{equation}
\begin{split}
dR_t^{p^*}&=R_t^{p^*}\bigl(-\gamma^1(p_t^*-Z_t^0)dW_t^0+\gamma^1 Z_t^1 dW_t^1\bigr)\\
&=R_t^{p^*}\bigl(-(\theta_t^\top-z_t^{0\perp}) dW_t^0+z_t^1 dW_t^1\bigr), ~~~t\in[0,T], \nonumber
\end{split}
\end{equation}
and hence, for any $t\in[0,T]$, 
\begin{equation}
\begin{split}
R_t^{p*}&=-\exp\bigl(-\gamma^1 \mathcal{W}_t^{1,p^*}+y_t\bigr)\\
&=-\exp\bigl(-\gamma^1 \xi^1+y_0\bigr)\mathcal{E}\Bigl(-\int_0^\cdot (\theta_s^\top-z_s^{0\perp}) dW_s^0+\int_0^\cdot z_s^1 dW_s^1\Bigr)_t. \nonumber
\end{split}
\end{equation}
Since $(\theta^\top-z^{0\perp}, z^1)$ are in $\mathbb{H}^2_{\rm BMO}$ and $(\gamma^1, \xi^1, y_0)$ are all bounded, $R^{p*}$ is a uniformly integrable martingale.
Uniform integrability of $R^{p*}$ and the boundedness of $y$ then imply that $(\exp(-\gamma^1 \mathcal{W}_t^{1,p^*}))_{t\in[0,T]}$
is also uniformly integrable. Therefore, we obtain the admissibility $p^*\in \mathcal{A}^1$.
The uniqueness of $p^*$ follows from the strict convexity of 
$-\gamma^1 (p \theta+f(s,z^0,z^1))+\dfrac{({\gamma^1})^2}{2}(|p-z^0|^2+|z^1|^2)$
with respect to $p$, which induces a strictly negative drift for $R^p$ if $p\neq p^*$.
Since $R^p$ is of class $\mathcal{D}$, its supermartingale property is now obvious. $\square$

\begin{rem}
It is important to note that the optimal trading strategy $\pi^*$ (or equivalently $p^*$) is independent from the initial wealth $\xi^1$.
This is a well-known characteristic of  exponential-type utility functions.
It follows, combined with the specification of $U^1(\pi)$ in $\eqref{eq-utility-functional}$, that the problem for the agent-1 and her optimal trading
strategy are invariant under the following transformation:
\begin{equation}
\begin{split}
\xi^1\longrightarrow &~(\xi^1-\mathbb{E}[F^1|\mathcal{F}_0^1]), \\
F^1\longrightarrow &~(F^1-\mathbb{E}[F^1|\mathcal{F}_0^1]). 
\end{split}
\label{duality-relation}
\end{equation}
\end{rem}

\begin{rem}
	The optimization method of \cite{huUtilityMaximizationIncomplete2005a} we used in this section can be applied exactly in the same way even if the market is complete.
	However, we are going to find a market-clearing equilibrium based on the MFG technique, and it crucially relies on the presence of idiosyncratic noise of each agent to simplify the system in the large population limit.
	Since this inevitably makes the market incomplete, our strategy in the following sections does not work in a complete market.
	We refer to Karatzas~\&~Shreve~\cite{karatzas_methods_1998}[Chapter 4] for the construction of an equilibrium in a complete market,
	which uses convex duality and Lagrange multipliers.
\end{rem}

\section{Mean-Field Equilibrium Model}
\label{sec-mean-field-BSDE}

We are now going to investigate a financial market being participated by many agents,
who are interacting each other through the price process of risky stocks.
Recall that, our final goal of this chapter is to find a risk-premium process $\theta=(\theta_t)_{t\geq 0}$ 
of the $n$ risky stocks {\it endogenously} by imposing {\it the market-clearing condition}, which requires 
the demand and supply of the risky stocks to be always balanced among the agents.
In this section, we shall propose a novel  mean-field BSDE with a quadratic-growth driver,
which is expected to provide, at least intuitively,  the characterization of the desired equilibrium in the large population limit.

\subsection{Heuristic Derivation of the Mean-Field BSDE}
Suppose that there are $N\in \mathbb{N}$ agents (agent-$i$, $1\leq i\leq N$) participating in
the same financial market given in Assumption~\ref{assumption-market}.
For each $1\leq i\leq N$, the information set of agent-$i$ is provided by the probability space $(\Omega^{0,i},\mathcal{F}^{0,i}, \mathbb{P}^{0,i})$
which is a completion of the product space $(\Omega^0,\mathcal{F}^0, \mathbb{P}^0)\otimes (\Omega^i,\mathcal{F}^i,\mathbb{P}^i)$.
The associated filtration $\mathbb{F}^{0,i}:=(\mathcal{F}_t^{0,i})_{t\in [0,T]}$ is the complete and right-continuous augmentation of 
$(\mathcal{F}_t^0\otimes \mathcal{F}_t^i)_{t\in[0,T]}$. $\mathcal{T}^{0,i}$ is the set of $\mathbb{F}^{0,i}$-measurable stopping times with values in $[0,T]$.
Here,  for each $i$, the filtered probability space $(\Omega^i, \mathcal{F}^i, \mathbb{P}^i,\mathbb{F}^i)$
is an independent copy of $(\Omega^1,\mathcal{F}^1,\mathbb{P}^1,\mathbb{F}^1)$
constructed exactly in the same way  as in Section~\ref{sec-each-optimization} 
with $(\xi^i,\gamma^i, W^i)$ instead of $(\xi^1, \gamma^1, W^1)$. 
Let us define a complete probability space $(\Omega,\mathcal{F},\mathbb{P})$ with filtration $\mathbb{F}:=(\mathcal{F}_t)_{t\in[0,T]}$ in such a way that
it allows us to model all the agents in a common  space:
$\Omega:=\Omega^0\times \prod_{i=1}^N \Omega^i$
and $(\mathcal{F},\mathbb{P})$ is the completion of $\bigl(\mathcal{F}^0\otimes \mathcal{F}^1\otimes\cdots\otimes\mathcal{F}^N,
\mathbb{P}^0\otimes \mathbb{P}^1\otimes \cdots\otimes \mathbb{P}^N\bigr)$.
$\mathbb{F}$ denotes the complete and the right-continuous augmentation of $(\mathcal{F}^0_t\otimes \mathcal{F}_t^1\otimes\cdots\otimes\mathcal{F}_t^N)_{t\in [0,T]}$.
$\mathbb{E}[\cdot]$ denotes the expectation with respect to $\mathbb{P}$. 

We also introduce the liability $F^i$ of the agent-$i$, $1\leq i\leq N$.
Each agent-$i$ is assumed to face the optimization problem 
\begin{equation}
\sup_{\pi\in \mathbb{A}^i}U^i(\pi), \nonumber
\end{equation}
where the utility functional $U^i$ is defined by 
\[
U^i(\pi):=\mathbb{E}\Bigl[-\exp\Bigl(-\gamma^i \bigl(\mathcal{W}_T^{i,\pi}-F^i\bigr)\Bigr)\Bigr], \nonumber
\]
with 
\[
\displaystyle\mathcal{W}_t^{i,\pi}:=\xi^i+\int_0^t \pi_s^\top \sigma_s(dW_s^0+\theta_s ds). \nonumber
\]
The admissible space $\mathbb{A}^i$ (and $\mathcal{A}^i)$ is defined in the same way as in Definition~\ref{def-admissible-space}
with all the superscripts $``1"$ replaced by $``i"$.
\begin{dfn}
The admissible space $\mathbb{A}^i$ is the set of all $\mathbb{R}^n$-valued, $\mathbb{F}^{0,i}$-progressively measurable trading strategies $\pi$
that satisfy
$
\displaystyle\mathbb{E}\Bigl[ \int_0^T |\pi_s^\top \sigma_s|^2ds \Bigr]<\infty, \nonumber
$
and such that
\begin{equation}
\bigl\{\exp(-\gamma^i \mathcal{W}_\tau^{i,\pi}); \tau \in \mathcal{T}^{0,i}\bigr\} \nonumber
\end{equation}
is uniformly integrable (i.e. of class $\mathcal{D}$). We also define
$
\mathcal{A}^i:=\bigl\{p=\pi^\top \sigma; \pi\in \mathbb{A}^i\bigr\}. \nonumber
$
\end{dfn}

We work under the following assumption.

\begin{asm}~\\
\label{assumption-hetero}
{\rm (i)} The statements in Assumption~\ref{assumption-agent} hold with $``1"$ replaced by $``i"$, $1\leq i\leq N$.\\
{\rm (ii)} $\{(\xi^i, \gamma^i), 1\leq i\leq N\}$ have the same distribution. In other words, they are independently and identically 
distributed (i.i.d.) on $(\Omega,\mathcal{F},\mathbb{P})$. \\
{\rm (iii)} $\{F^i, 1\leq i\leq N\}$ are $\mathcal{F}^0$-conditionally i.i.d.
\end{asm}

We want to find the risk-premium process $\theta$ that clears the market.
Let us first derive, heuristically,  the relevant mean-field BSDE, which then will be shown to characterize the market-clearing
risk-premium process in the large population limit, by the idea proposed by Fujii \& Takahashi~\cite{fujiiMeanFieldGame2022}.

Suppose that a risk-premium process $\theta\in \mathbb{H}^2_{\rm BMO}(\mathbb{P},\mathbb{F}^0,\mathbb{R}^{d_0})$ is given.
Recall that $\theta$ has values in ${\rm Range}(\sigma^\top)={\rm Ker}(\sigma)^\perp$, i.e. $\theta_s^\top\in L_s$ for every $s\in[0,T]$. 
By repeating the analysis done in Section~\ref{sec-each-optimization}, one can show that
the unique optimal strategy of agent-$i$, $1\leq i\leq N$, is given by
\begin{equation}
p^{i,*}_t~(=(\pi_t^{i,*})^\top\sigma_t)=Z_t^{i,0\|}+\frac{\theta_t^\top}{\gamma^i}, ~~~t\in[0,T], 
\label{p-i-optimal}
\end{equation}
where $Z^{i,0}$ is associated to the solution $(Y^i, Z^{i,0},Z^i)$ of the following BSDE:
\begin{equation}
Y_t^i=F^i+\int_t^T \Bigl(-Z_s^{i,0\|}\theta_s-\frac{|\theta_s|^2}{2\gamma^i}+\frac{\gamma^i}{2}(|Z_s^{i,0\perp}|^2+|Z_s^i|^2)\Bigr)ds
-\int_t^T Z_s^{i,0}dW_s^0-\int_t^T Z_s^i dW_s^i, ~~~t\in[0,T]. 
\label{BSDE-agent-i}
\end{equation}
Under Assumptions~\ref{assumption-market} and \ref{assumption-hetero}, 
we already know from Corollary~\ref{corollary-existence}
that there is a unique bounded solution $(Y^i,Z^{i,0},Z^i)\in \mathbb{S}^\infty(\mathbb{P},\mathbb{F}^{0,i})
\times \mathbb{H}^2_{\rm BMO}(\mathbb{P},\mathbb{F}^{0,i})\times \mathbb{H}^2_{\rm BMO}(\mathbb{P},\mathbb{F}^{0,i})$ to $\eqref{BSDE-agent-i}$, for each $1\leq i\leq N$.

\begin{dfn}
\label{def-market-clearing}
We say that the market-clearing condition is satisfied if
\begin{equation}
	\label{mc-eqn}
	\frac{1}{N}\sum_{i=1}^N \pi_t^{i,*}=0, ~~~ dt\otimes \mathbb{P}\text{-}\mathrm{a.e.} 
\end{equation}
where $\pi^{i,*}$ is the optimal trading strategy of the agent $i$, $1\leq i\leq N$.
\end{dfn}

\begin{rem}~\\
(1) The  market clearing condition in Definition~\ref{def-market-clearing} means that the net supply (demand) of securities is zero in the whole period
under the optimal decisions of all the agents.
For finite $N$, the factor $1/N$ in the above definition is irrelevant.
However, in order to  handle the large population limit $N\rightarrow\infty$, we need the factor $1/N$.
We are going to find a risk-premium process so that the excess demand (or supply)  per capita converges to zero in 
the large population limit.\\
(2) If we consider the case in which agents have $\alpha^i$-shares of $i$-th stock ($\alpha^i\in\mathbb{N}$, $i=\{1,\ldots,n\}$) in total, the market clearing condition should be
\begin{equation}
	\frac{1}{N}\sum_{i=1}^N \pi_t^{i,*}=\frac{1}{N}\mathrm{diag}(\alpha)S_t, ~~~dt\otimes \mathbb{P}\text{-}\mathrm{a.e.},
\end{equation}
where $\alpha:=(\alpha^1,\ldots,\alpha^n)^\top\in\mathbb{N}^n$. Note that the right hand side goes to zero as $N\to\infty$. 
Therefore, even in this case, we can set \eqref{mc-eqn} as the market-clearing condition since our goal is to characterize the equilibrium risk-premium process in the large population limit.
\end{rem}

The market-clearing condition given by Definiton~\ref{def-market-clearing} requires
that the risk-premium process $\theta$ to be
\begin{equation}
\theta_t=-\Bigl(\frac{1}{N}\sum_{j=1}^N \frac{1}{\gamma^j}\Bigr)^{-1} \frac{1}{N}\sum_{j=1}^N (Z_t^{j,0\|})^\top, \quad t\in[0,T]. 
\label{theta-naive}
\end{equation}
Unfortunately, the suggested $\theta$ by $\eqref{theta-naive}$ is inconsistent with 
our information assumption that requires $\theta$ to be $\mathbb{F}^0$-adapted i.e. being dependent only on the market-wide information.
However, at this moment, let us {\it formally} consider the $N$-coupled system of quadratic-BSDEs obtained from $\eqref{BSDE-agent-i}$
with $\theta$ replaced by the one given in $\eqref{theta-naive}$; for $1\leq i\leq N$, 
\begin{equation}
\begin{split}
Y_t^i&=F^i+\int_t^T \Bigl\{Z_s^{i,0\|}\Bigl(\frac{1}{N}\sum_{j=1}^N \frac{1}{\gamma^j}\Bigr)^{-1}\frac{1}{N}\sum_{j=1}^N (Z_s^{j,0\|})^\top
-\frac{1}{2\gamma^i}\Bigl(\frac{1}{N}\sum_{j=1}^N \frac{1}{\gamma^j}\Bigr)^{-2}\Bigl|  \frac{1}{N}\sum_{j=1}^N Z_s^{j,0\|} \Bigr|^2 \\
&\hspace{20mm}+\frac{\gamma^i}{2}(|Z_s^{i,0\perp}|^2+|Z_s^i|^2)\Bigr\}ds-\int_t^T Z_s^{i,0}dW_s^0-\int_t^T Z_s^i dW_s^i.
\end{split}
\label{qg-BSDE-system}
\end{equation}

In order to make this system well defined, we need at least additional stochastic integral terms with respect to all of the $(W^j)_{j\neq i}$.
This stems from the martingale representation theorem since $\theta$ is only $\mathbb{F}$-adapted. 
We also have to change the space of admissible controls of each agent
to make the associated optimization problem meaningful. 
In any case however, since each agent is a price taker, the interaction among them appears only through the risk-premium process
which has a symmetric form. Thus, from Assumption~\ref{assumption-hetero}, if there is a solution 
$\{(Y^i,Z^{i,0},Z^i), 1\leq i\leq N\}$ to the system $\eqref{qg-BSDE-system}$ (after appropriate modifications), then they are expected to be exchangeable.  In particular,  $(Z^{i,0}_t)_{i=1}^N$ (and hence $(Z_t^{i,0\|})_{i=1}^N$) would be  exchangeable random variables, i.e. their joint distribution is invariant under the permutation $\sigma(i)$ of their orders.
If this is the case,  De Finetti's theory of exchangeable sequence of random variables would imply that
\begin{equation}
	\label{De-Finetti-eq}
\lim_{N\rightarrow \infty}\frac{1}{N}\sum_{i=1}^N Z^{i,0\|}_t=\mathbb{E}\Bigl[Z_t^{1,0\|}|\bigcap_{k\geq 1}\sigma\{Z_t^{j,0\|}, j\geq k\}\Bigr]\quad {\text{a.s.}}
\end{equation}
See, for example, \cite{carmonaProbabilisticTheoryMean2018a}[Theorem 2.1].  We can also naturally expect that the tail $\sigma$-field
in $\eqref{De-Finetti-eq}$ converges to $\mathcal{F}^0_t$ since $(\mathcal{F}_t^i)_{i\geq 1}$ are independent by construction.
In this way, we can expect from $\eqref{theta-naive}$, at least heuristically, that the equilibrium risk-premium process $\theta$ in the large-$N$ limit may be 
given by
\begin{equation}
\theta_t=-\widehat{\gamma}\mathbb{E}[Z_t^{1,0\|}|\mathcal{F}_t^0]^\top=-\widehat{\gamma}\mathbb{E}[Z_t^{1,0\|}|\mathcal{F}^0]^\top,  
\label{eq-theta-new}
\end{equation}
which is $\mathbb{F}^0$-adapted as desired, and  $\widehat{\gamma}$ is defined by
\begin{equation}
\widehat{\gamma}:=\frac{1}{\mathbb{E}[1/\gamma^1]}. 
\end{equation}
The replacement by the conditional expectation $\eqref{eq-theta-new}$ makes the system decoupled since there is no interaction among $(Z^{0,j\|})_{j\in\mathbb{N}}$ any more as opposed to $\eqref{theta-naive}$.
Such a decoupling phenomenon is usually referred to as {\it propagation of chaos}, and it is the 
main driving force of MFG theory to transform a complex coupled problem into a simple decoupled one.

The above heuristic discussion motivates us to study the following mean-field BSDE defined
on the filtered probability space $(\Omega^{0,1}, \mathcal{F}^{0,1},\mathbb{P}^{0,1},\mathbb{F}^{0,1})$ by choosing the agent-1 as the representative;
\begin{equation}
	\begin{split}
		Y_t=&F^1+\int_t^T \Bigl(\widehat{\gamma}Z_s^{0\|}\overline{\mathbb{E}}[Z_s^{0\|}]^\top-\frac{\widehat{\gamma}^2}{2\gamma^1 }|\overline{\mathbb{E}}[Z_s^{0\|}]|^2+\frac{\gamma^1}{2}(|Z_s^{0\perp}|^2+|Z_s^1|^2)\Bigr)ds\\
		&-\int_t^T Z_s^{0} dW_s^0-\int_t^T Z_s^1 dW_s^1,~~~t\in[0,T].
\label{mfg-BSDE-org}
	\end{split}
\end{equation}
Here, we have defined, for any $X\in \mathbb{H}^2(\mathbb{P}^{0,1},\mathbb{F}^{0,1})$, that
\begin{equation}
\overline{\mathbb{E}}[X_t](\omega^0):=\begin{cases} 
\mathbb{E}[X_t|\mathcal{F}^0](\omega^0)=\mathbb{E}^{\mathbb{P}^1}[X_t(\omega^0,\cdot)] \quad {\text{if it exists}}\\
 0\hspace{20mm} {\text{otherwise}}
\end{cases}~.
\nonumber
\end{equation}
Note that we have $\mathbb{E}[X_t|\mathcal{F}_t^0]=\mathbb{E}[X_t|\mathcal{F}^0]$ a.s.~for any $\mathbb{F}^{0,1}$-adapted process.  This is because that $X_t$ is independent from  $\sigma(\{W_s^0-W_t^0, s\geq t\})$ and thus the additional information $\mathcal{F}^0=\mathcal{F}_T^0\supset \mathcal{F}_t^0$
does not affect the expectation value. 
As in \cite{carmonaProbabilisticTheoryMean2018a}[Section 4.3.1], we always choose $\mathbb{F}^0$-progressively measurable 
modification of $\overline{\mathbb{E}}[X]$. In the remainder of this section, we show that there is a solution to this mean-field BSDE under some conditions.
In Section~\ref{sec-market-clearing}, we will show that the risk-premium process defined as 
\begin{equation}
\theta_t^{\rm mfg}:=-\widehat{\gamma}\overline{\mathbb{E}}[Z_t^{0\|}]^\top,~~~ t\in[0,T]
\label{theta-mfg}
\end{equation}
by the solution of the mean-field BSDE $\eqref{mfg-BSDE-org}$ actually clears the market 
in the large population limit.
\subsection{Existence of a Solution to the Mean-Field BSDE}
We will work on the filtered probability space $(\Omega^{0,1}, \mathcal{F}^{0,1}, \mathbb{P}^{0,1},\mathbb{F}^{0,1})$.
For notational ease, we simply write $(F,\gamma)$ instead of $(F^1,\gamma^1)$ and 
investigate the well-posedness of the mean-field BSDE:
\begin{equation}
	\begin{split}
	Y_t=&F+\int_t^T \Bigl(\widehat{\gamma}Z_s^{0\|}\overline{\mathbb{E}}[Z_s^{0\|}]^\top-\frac{\widehat{\gamma}^2}{2\gamma}|\overline{\mathbb{E}}[Z_s^{0\|}]|^2+\frac{\gamma}{2}(|Z_s^{0\perp}|^2+|Z_s^1|^2)\Bigr)ds\\
	&-\int_t^T Z_s^0 dW_s^0-\int_t^T Z_s^1 dW_s^1,~~~ t\in[0,T]. 
	\label{mfg-BSDE}
	\end{split}
\end{equation}
Note that, by Assumption~\ref{assumption-agent} (ii), $\widehat{\gamma}$ is a strictly positive constant.

\begin{lem}
If there exists a bounded solution $(Y,Z^0,Z^1)\in \mathbb{S}^\infty\times \mathbb{H}^2\times \mathbb{H}^2$ to the BSDE $\eqref{mfg-BSDE}$,
then $Z:=(Z^0,Z^1)$ is in $\mathbb{H}^2_{\rm BMO}$.
\end{lem}

\noindent
\textbf{\textit{Proof}}\\
By applying Ito formula to $e^{2\gamma Y_t}$ and using $|Z_t^0|^2=|Z_t^{0\|}|^2+|Z_t^{0\perp}|^2$, we get
\begin{equation}
\begin{split}
d(e^{2\gamma Y_t})&=e^{2\gamma Y_t}\bigl(2\gamma dY_t+2\gamma^2 (|Z_t^0|^2+|Z_t^1|^2)dt\bigr)\\
&=e^{2\gamma Y_t}\bigl(-2\gamma\widehat{\gamma}Z_t^{0\|}\overline{\mathbb{E}}[Z_t^{0\|}]^\top+\widehat{\gamma}^2|\overline{\mathbb{E}}[Z_t^{0\|}]|^2
-\gamma^2(|Z_t^{0\perp}|^2+|Z_t^1|^2)+2\gamma^2(|Z_t^0|^2+|Z_t^1|^2)\bigr)dt\\
&\quad+e^{2\gamma Y_t}2\gamma (Z_t^0dW_t^0+Z_t^1 dW_t^1)\\
&\geq e^{2\gamma Y_t}\gamma^2(|Z_t^0|^2+|Z_t^1|^2)dt+e^{2\gamma Y_t}2\gamma(Z_t^0 dW_t^0+Z_t^1 dW_t^1).  \nonumber
\end{split}
\end{equation}
Hence, for any $\tau \in \mathcal{T}^{0,1}$, 
\begin{equation}
\mathbb{E}\bigl[e^{2\gamma Y_T}|\mathcal{F}_\tau^{1,0}\bigr]\geq \mathbb{E}\bigl[e^{2\gamma Y_T}-e^{2\gamma Y_\tau}|\mathcal{F}_\tau^{1,0}\bigr]\geq \mathbb{E}\Bigl[\int_\tau^T e^{2\gamma Y_s}\gamma^2 (|Z_s^0|^2+
|Z_s^1|^2)ds|\mathcal{F}_\tau^{1,0}\Bigr], \nonumber 
\end{equation}
which then gives~
$
\displaystyle
\|(Z^0,Z^1)\|_{\mathbb{H}^2_{\rm BMO}}^2\leq \frac{1}{\underline{\gamma}^2}\exp\bigl({4\overline{\gamma}\|Y\|_{\mathbb{S}^\infty}}\bigr). \nonumber
$ $\square$

Proving the well-posedness of the mean-field BSDE $\eqref{mfg-BSDE}$ is  quite difficult.
First of all, due to the presence of conditional expectations, the comparison principle is not available.
This makes  many of the existing techniques for qg-BSDEs such as in \cite{Kobylanski2000BackwardSD, briandBSDEQuadraticGrowth2006, briandQuadraticBSDEsConvex2008}
inapplicable. 
Due to the conditional McKean-Vlasov nature of the BSDE $\eqref{mfg-BSDE}$, the traditional approach of MFGs
using Schauder's fixed point theorem does not work, either.
Second, we did not succeed in obtaining any {\it a priori} bound on $\|Y\|_{\mathbb{S}^\infty}$ nor
any stability result such as \cite{Fujii-qgBSDE}[Lemma 3.3]. Thus, the strategy of constructing a compact set for the decoupling functions
of the regularized BSDEs under the Markovian setup as in \cite{Fujii-ABSDE}[Theorem 4.1] cannot be used.

Therefore, in the following, we will try the method proposed by Tevzadze~\cite{tevzadzeSolvabilityBackwardStochastic2008}[Proposition 1].
Among the existing literature on qg-BSDEs,  a unique characteristic of his work  is 
to solve a problem by a contraction mapping theorem without relying on the comparison principle. 
This method is also adopted by Hu et.al.~\cite{Hu-ABSDE} to deal with an anticipated BSDE with quadratic growth terms,  
which is related to an optimization problem with delay.

For notational simplicity, let us define the driver $f$ of the mean-field BSDE by,
for any $(z^0,z^1)\in \mathbb{H}^2(\mathbb{P}^{0,1},\mathbb{F}^{0,1},\mathbb{R}^{1\times d_0} \times \mathbb{R}^{1\times d})$, 
\begin{equation}
\begin{split}
f(z_s^0,z_s^1)&:=\widehat{\gamma}z_s^{0\|}\overline{\mathbb{E}}[z_s^{0\|}]^\top-\frac{\widehat{\gamma}^2}{2\gamma}|\overline{\mathbb{E}}[z_s^{0\|}]|^2
+\frac{\gamma}{2}(|z_s^{0\perp}|^2+|z_s^1|^2), \quad s\in[0,T]. \nonumber
\end{split}
\end{equation}
We have, by completing the square,  
\begin{equation}
\begin{split}
f(z^0,z^1)&=-\Bigl|\frac{\widehat{\gamma}}{\sqrt{2\gamma}}\overline{\mathbb{E}}[z^{0\|}]-\frac{\sqrt{\gamma}}{\sqrt{2}}z^{0\|}\Bigr|^2
+\frac{\gamma}{2}|z^{0\|}|^2+\frac{\gamma}{2}(|z^{0\perp}|^2+|z^1|^2) \\
&=-\Bigl|\frac{\widehat{\gamma}}{\sqrt{2\gamma}}\overline{\mathbb{E}}[z^{0\|}]-\frac{\sqrt{\gamma}}{\sqrt{2}}z^{0\|}\Bigr|^2
+\frac{\gamma}{2}(|z^{0}|^2+|z^1|^2). \nonumber
\end{split}
\end{equation}
Hence, for any $(z^0,z^1)$, 
\begin{equation}
\begin{split}
f^+(z^0,z^1)&\leq \frac{\overline{\gamma}}{2}(|z^0|^2+|z^1|^2), \\
f^-(z^0,z^1)&\leq 0\vee \Bigl(\Bigl|\frac{\widehat{\gamma}}{\sqrt{2\gamma}}\overline{\mathbb{E}}[z^{0\|}]-\frac{\sqrt{\gamma}}{\sqrt{2}}z^{0\|}\Bigr|^2
-\frac{\gamma}{2}|z^0|^2\Bigr) \leq \frac{\widehat{\gamma}^2}{\underline{\gamma}}|\overline{\mathbb{E}}[z^{0\|}]|^2+\frac{\overline{\gamma}}{2}|z^{0\|}|^2, \nonumber
\end{split}
\end{equation}
where $f^+(z^0,z^1):=\max(f(z^0,z^1),0)$ and $f^-(z^0,z^1):=\max(-f(z^0,z^1),0)$.
Thus, in particular, 
\begin{equation}
|f(z^0,z^1)|\leq \frac{\overline{\gamma}}{2}(|z^0|^2+|z^1|^2)+\frac{\widehat{\gamma}^2}{\underline{\gamma}}|\overline{\mathbb{E}}[z^{0}]|^2. 
\label{f-abs}
\end{equation}
Moreover,  by Assumption~\ref{assumption-agent} (ii), there is a positive constant $C_\gamma$, which depends only on $(\underline{\gamma},
\overline{\gamma},\widehat{\gamma})$,  such that, for any $(z^0,z^1), (\acute{z}^0,\acute{z}^1)\in  \mathbb{H}^2$, 
\begin{equation}
\begin{split}
&|f(z^0,z^1)-f(\acute{z}^0,\acute{z}^1)|\\
&\quad \leq C_\gamma \bigl(|z^0|+|\acute{z}^0|+|z^1|+|\acute{z}^1|+|\overline{\mathbb{E}}[z^{0}]|+|\overline{\mathbb{E}}[\acute{z}^{0}]|\bigr)\bigl(|z^0-\acute{z}^0|+|z^1-\acute{z}^1|+
|\overline{\mathbb{E}}[z^{0}-\acute{z}^{0}]|\bigr).
\end{split}
\label{f-spread}
\end{equation}

Let us observe the following simple fact.
\begin{lem}
\label{lemma-f-abs}
For any input $(z^0,z^1)\in \mathbb{H}^2_{\rm BMO}(\mathbb{P}^{0,1},\mathbb{F}^{0,1})$,
the next  inequality holds;
\begin{equation}
\sup_{\tau\in\mathcal{T}^{0,1}}\Bigl\|\mathbb{E}\Bigl[\int_\tau^T |f(z^0_s,z_s^1)|ds|\mathcal{F}_\tau^{0,1}\Bigr]\Bigr\|_{\infty}
\leq c_\gamma \|(z^0,z^1)\|^2_{\mathbb{H}^2_{\rm BMO}}, \nonumber
\end{equation}
where $c_\gamma$ is a positive constant given by $\displaystyle c_\gamma:=\frac{\overline{\gamma}}{2}+\frac{\widehat{\gamma}^2}{\underline{\gamma}}$.
\end{lem}

\noindent
\textbf{\textit{Proof}}\\
For any input $z^0\in \mathbb{H}^2_{\rm BMO}$, $(\overline{\mathbb{E}}[z^0_t], t\in[0,T])$ is an $\mathbb{F}^0$-adapted process and hence independent of $\mathcal{F}^1$.
This implies, with Jensen's inequality, 
\begin{equation}
\begin{split}
\sup_{\tau\in \mathcal{T}^{0,1}}\Bigl\|\mathbb{E}\Bigl[\int_\tau^T |\overline{\mathbb{E}}[z_s^0]|^2ds|\mathcal{F}_\tau^{0,1}\Bigr]\Bigr\|_{\infty}
&=\sup_{\tau\in\mathcal{T}^0}\Bigl\|\mathbb{E}\Bigl[\int_\tau^T |\overline{\mathbb{E}}[z_s^0]|^2ds|\mathcal{F}_\tau^0\Bigr]\Bigr\|_{\infty}  \leq \sup_{\tau\in\mathcal{T}^0}\Bigl\|\mathbb{E}\Bigl[\int_\tau^T |z_s^0|^2ds|\mathcal{F}_\tau^0\Bigr]\Bigr\|_{\infty}.  \nonumber
\end{split}
\end{equation}
Moreover, for any $\tau\in \mathcal{T}^0 ~(\subset \mathcal{T}^{0,1})$, we have
$
\displaystyle\mathbb{E}\Bigl[\int_\tau^T |z_s^0|^2ds|\mathcal{F}_\tau^0\Bigr]=\mathbb{E}\Bigl[\mathbb{E}\bigl[\int_\tau^T |z_s^0|^2ds|\mathcal{F}_\tau^{0,1}\bigr]|\mathcal{F}_\tau^0\Bigr] \nonumber
$
and hence
\begin{equation}
\Bigl\|\mathbb{E}\Bigl[\int_\tau^T |z_s^0|^2ds|\mathcal{F}_\tau^0\Bigr]\Bigr\|_{\infty}\leq \Bigl\|\mathbb{E}\Bigl[\int_\tau^T |z_s^0|^2ds|\mathcal{F}_\tau^{0,1}\Bigr]
\Bigr\|_{\infty}. \nonumber
\end{equation}
It follows that
\begin{equation}
\sup_{\tau\in \mathcal{T}^{0,1}}\Bigl\|\mathbb{E}\Bigl[\int_\tau^T |\overline{\mathbb{E}}[z_s^0]|^2ds|\mathcal{F}_\tau^{0,1}\Bigr]\Bigr\|_{\infty}
\leq \sup_{\tau\in \mathcal{T}^{0,1}}\Bigl\|\mathbb{E}\Bigl[\int_\tau^T |z_s^0|^2ds|\mathcal{F}_\tau^{0,1}\Bigr]
\Bigr\|_{\infty}.
\label{bmo-comparison}
\end{equation}
Now the conclusion immediately follows from $\eqref{f-abs}$. $\square$\\

We now define the map $\Gamma: \mathbb{H}^2_{\rm BMO}(\mathbb{P}^{0,1},\mathbb{F}^{0,1},\mathbb{R}^{1\times d_0}\times \mathbb{R}^{1\times d})\ni (z^0,z^1)\mapsto \Gamma(z^0,z^1)=(Z^0,Z^1)\in \mathbb{H}^2_{\rm BMO}(\mathbb{P}^{0,1},\mathbb{F}^{0,1},\mathbb{R}^{1\times d_0}\times \mathbb{R}^{1\times d})$
by
\begin{equation}
Y_t=F+\int_t^T f(z^0_s,z^1_s)ds-\int_t^T Z_s^0 dW_s^0-\int_t^T Z_s^1 dW_s^1,
\label{map-Gamma}
\end{equation}
where $\Gamma(z^0,z^1)=(Z^0,Z^1)$ is the stochastic integrands associated to the solution of BSDE $\eqref{map-Gamma}$.

\begin{lem}
Under Assumption~\ref{assumption-agent}, the map $\Gamma$ is well-defined.
\end{lem}

\noindent
\textbf{\textit{Proof}}\\
For any $(z^0,z^1)\in \mathbb{H}^2_{\rm BMO}$, the existence of the unique solution $(Y,Z^0,Z^1)$ to $\eqref{map-Gamma}$ is obvious.
By taking a conditional expectation, we have, from Lemma~\ref{lemma-f-abs},
\begin{equation}
\|Y\|_{\mathbb{S}^\infty}\leq \|F\|_{\infty}+c_\gamma\|(z^0,z^1)\|^2_{\mathbb{H}^2_{\rm BMO}}<\infty. \nonumber
\end{equation}
Moreover, by Ito formula applied to $|Y_t|^2$, we obtain, for any $t\in[0,T]$,
\begin{equation}
\mathbb{E}\Bigl[\int_t^T(|Z^0_s|^2+|Z_s^1|^2)ds|\mathcal{F}_t^{0,1}\Bigr]\leq \|F\|_{\infty}^2+2\|Y\|_{\mathbb{S}^\infty}
\mathbb{E}\Bigl[\int_t^T |f(z_s^0,z_s^1)|ds|\mathcal{F}_t^{0,1}\Bigr]. \nonumber
\end{equation}
Thus, 
$
\|(Z^0,Z^1)\|_{\mathbb{H}^2_{\rm BMO}}^2\leq \|F\|_{\infty}^2+2c_\gamma \|Y\|_{\mathbb{S}^\infty}\|(z^0,z^1)\|^2_{\mathbb{H}^2_{\rm BMO}}<\infty. \nonumber
$ $\square$

For each positive constant $R$, let $\mathcal{B}_R$ be a subset of $\mathbb{H}^2_{\rm BMO}$ defined by
\begin{equation}
\mathcal{B}_R:=\Bigl\{ (z^0,z^1)\in \mathbb{H}^2_{\rm BMO}(\mathbb{P}^{0,1},\mathbb{F}^{0,1},\mathbb{R}^{1\times d_0}\times \mathbb{R}^{1\times d}); \|(z^0,z^1)\|^2_{\mathbb{H}^2_{\rm BMO}}\leq R^2\Bigr\}. \nonumber
\end{equation}

\begin{prop}
\label{prop-Gamma-stability}
Let Assumption~\ref{assumption-agent} be in force. If $ \|F\|_{\infty}\leq \dfrac{1}{4\sqrt{2}c_\gamma}$,
then with $R:=2\|F\|_{\infty}$, the set $\mathcal{B}_R$ is stable under the map $\Gamma$, i.e. $(z^0,z^1)\in \mathcal{B}_R$ implies 
$\Gamma(z^0,z^1)\in \mathcal{B}_{R}$. Moreover, in this case, the $y$-component of the solution to $\eqref{map-Gamma}$ satisfies $\|Y\|_{\mathbb{S}^\infty} \leq R$.
\end{prop}

\noindent
\textbf{\textit{Proof}}\\
By Ito formula, we have for any $t\in[0,T]$,
\begin{equation}
|Y_t|^2+\int_t^T (|Z_s^0|^2+|Z_s^1|^2)ds=|F|^2+\int_t^T 2Y_s f(z_s^0,z_s^1)ds-\int_t^T 2Y_s (Z_s^0 dW_s^0+Z_s^1 dW_s^1). \nonumber
\end{equation}
Taking conditional expectations we obtain
\begin{equation}
	\begin{split}
		\label{BR-middle}
	|Y_t|^2+\mathbb{E}\Bigl[\int_t^T (|Z_s^0|^2+|Z_s^1|^2)ds|\mathcal{F}_t^{0,1}\Bigr]\leq \|F\|^2_{\infty}+2\|Y\|_{\mathbb{S}^\infty}\mathbb{E}\Bigl[\int_t^T |f(z_s^0,z_s^1)|ds|\mathcal{F}_t^{0,1}\Bigr].
	\end{split}
\end{equation}
From Lemma~\ref{lemma-f-abs}, taking the essential supremum in the both hands yields,
\begin{equation}
	\begin{split}
		\esssup_{(t,\omega)\in[0,T]\times\Omega}\Bigl(|Y_t|^2+\mathbb{E}\Bigl[\int_t^T (|Z_s^0|^2+|Z_s^1|^2)ds|\mathcal{F}_t^{0,1}\Bigr]\Bigr)\leq \|F\|_{\infty}^2+\frac{1}{2}\|Y\|_{\mathbb{S}^\infty}^2+2c_\gamma^2\|(z^0,z^1)\|^4_{\mathbb{H}^2_{\rm BMO}}. 
\end{split}
\end{equation}
Using the above result and an obvious relation
\begin{equation}
\begin{split}
	\frac{1}{2}\bigl(\|Y\|_{\mathbb{S}^\infty}^2+\|(Z^0,Z^1)\|_{\mathbb{H}^2_{\rm BMO}}^2\bigr)&\leq \max\bigl(\|Y\|_{\mathbb{S}^\infty}^2, \|(Z^0,Z^1)\|_{\mathbb{H}^2_{\rm BMO}}^2\bigr) \\
	&\leq 
	\esssup_{(t,\omega)\in[0,T]\times\Omega}\Bigl(|Y_t|^2+\mathbb{E}\Bigl[\int_t^T (|Z_s^0|^2+|Z_s^1|^2)ds|\mathcal{F}_t^{0,1}\Bigr]\Bigr),
\end{split}
\end{equation}
we obtain
\begin{equation}
	\begin{split}
\begin{split}
\|(Z^0,Z^1)\|_{\mathbb{H}^2_{\rm BMO}}^2\leq 2\|F\|_{\infty}^2+4c_\gamma^2 \|(z^0,z^1)\|^4_{\mathbb{H}^2_{\rm BMO}}.
\end{split}
\end{split}
\end{equation}

We now try to find an $R>0$ such that
\begin{equation}
	2\|F\|_{\infty}^2+4 c_\gamma^2 R^4\leq R^2 \nonumber
\end{equation}
holds. By completing the square, one sees that this is solvable if and only if  $\|F\|_{\infty}\leq \dfrac{1}{4\sqrt{2}c_\gamma}$,
and in this case, we can choose $R=2\|F\|_{\infty}$. This proves the first statement.
Moreover, by rearranging the first inequality in $\eqref{BR-middle}$, we get
\begin{equation}
	|Y_t|^2\leq \|F\|_{\infty}^2+\frac{1}{2}\|Y\|_{\mathbb{S}^\infty}^2+2c_\gamma^2 \|(z^0,z^1)\|^4_{\mathbb{H}^2_{\rm BMO}} \nonumber
\end{equation}
and hence
\begin{equation}
	\|Y\|_{\mathbb{S}^\infty}^2\leq 2\|F\|_{\infty}^2+4c_\gamma^2 \|(z^0,z^1)\|^4_{\mathbb{H}^2_{\rm BMO}}. \nonumber
\end{equation}
This yields $\|Y\|_{\mathbb{S}^\infty}^2\leq R^2$ if $(z^0,z^1)\in \mathcal{B}_R$. $\square$

Now, we provide the first main result of this chapter.
\begin{thm}
\label{th-mfg-existence}
Let Assumption~\ref{assumption-agent} be in force. 
If the terminal function $F$ is small enough in the sense that
\begin{equation}
\|F\|_{\infty}<\frac{1}{48C_\gamma}, 
\label{terminal-constraint}
\end{equation}
where $C_\gamma$ is a constant used in $\eqref{f-spread}$, then there exists a unique solution $(Y,Z^0,Z^1)$ to the mean-field BSDE $\eqref{mfg-BSDE}$
in the domain
\begin{equation}
(Z^0,Z^1)\in \mathcal{B}_R, \quad \|Y\|_{\mathbb{S}^\infty}\leq R \nonumber
\end{equation}
with $R:=2\|F\|_{\infty}$.
\end{thm}

\noindent
\textbf{\textit{Proof}}\\
Note that the requirement $\eqref{terminal-constraint}$ is more stringent than the one used in 
Proposition~\ref{prop-Gamma-stability}. Thus it suffices to prove that the map $\Gamma$ is a strict contraction.

To prove it, let us consider two arbitrary inputs $z:=(z^0,z^1),~\acute{z}:=(\acute{z}^0,\acute{z}^1)\in\mathcal{B}_R$.
We set
\begin{equation}
Z:=(Z^0,Z^1):=\Gamma(z), \quad \acute{Z}:=(\acute{Z}^0,\acute{Z}^1):=\Gamma(\acute{z}), \nonumber
\end{equation}
and $Y, \acute{Y}$ as the $y$-component of the solution of $\eqref{map-Gamma}$ with input $z$ and $\acute{z}$, respectively.
For notational simplicity, we put 
\begin{equation}
\Delta z:=z-\acute{z}, \quad \Delta Y:=Y-\acute{Y}, \quad \Delta Z:=Z-\acute{Z}.  \nonumber
\end{equation}
This gives
\begin{equation}
\Delta Y_t=\int_t^T \bigl(f(z_s^0,z_s^1)-f(\acute{z}_s^0,\acute{z}_s^1)\bigr)ds-\int_t^T \Delta Z_s^0 dW_s^0-\int_t^T \Delta Z_s^1 dW_s^1, \nonumber
\end{equation}
and hence by Ito formula, 
\begin{equation}
	\begin{split}
&|\Delta Y_t|^2+\int_t^T(|\Delta Z_s^0|^2+|\Delta Z_s^1|^2)ds\\
&=
\int_t^T 2\Delta Y_s (f(z_s^0,z_s^1)-f(\acute{z}_s^0,\acute{z}_s^1))ds-\int_t^T2\Delta Y_s(\Delta Z_s^0 dW_s^0+\Delta Z_s^1 dW_s^1). \nonumber
	\end{split}
\end{equation}
By taking the conditional expectation, we have, from H\"older inequality,
\begin{equation}
\begin{split}
&|\Delta Y_t|^2+\mathbb{E}\Bigl[\int_t^T (|\Delta Z_s^0|^2+|\Delta Z_s^1|^2)ds|\mathcal{F}_t^{0,1}\Bigr]
\leq 2\|\Delta Y\|_{\mathbb{S}^\infty}\mathbb{E}\Bigl[\int_t^T |f(z_s^0,z_s^1)-f(\acute{z}_s^0,\acute{z}_s^1)|ds|\mathcal{F}_t^{0,1}\Bigr]\\
&~ \leq\frac{1}{2} \|\Delta Y\|_{\mathbb{S}^\infty}^2+2\Bigl(\mathbb{E}\Bigl[\int_t^T |f(z_s^0,z_s^1)-f(\acute{z}_s^0,\acute{z}_s^1)|ds|\mathcal{F}_t^{0,1}\Bigr]\Bigr)^2\\
&~ \leq \frac{1}{2}\|\Delta Y\|_{\mathbb{S}^\infty}^2\\
&~~~+2 C_\gamma^2\Bigl(
\mathbb{E}\Bigl[\int_t^T (|z_s^0|+|\acute{z}_s^0|+|z_s^1|+|\acute{z}_s^1|+|\overline{\mathbb{E}}[z_s^0]|+|\overline{\mathbb{E}}[\acute{z}_s^0]|)
(|\Delta z_s^0|+|\Delta z_s^1|+|\overline{\mathbb{E}}[\Delta z_s^0]|)ds|\mathcal{F}_t^{0,1}\Bigr]\Bigr)^2  \\
&~\leq \frac{1}{2}\|\Delta Y\|_{\mathbb{S}^\infty}^2+ 2 C_\gamma^2
\Bigl(\mathbb{E}\Bigl[\int_t^T (|z_s^0|+|\acute{z}_s^0|+|z_s^1|+|\acute{z}_s^1|+|\overline{\mathbb{E}}[z_s^0]|+|\overline{\mathbb{E}}[\acute{z}_s^0]|)^2ds|\mathcal{F}_t^{0,1}\Bigr]\Bigr)\\
&\hspace{30mm} \times \Bigl(\mathbb{E}\Bigl[\int_t^T (|\Delta z_s^0|+|\Delta z_s^1|+|\overline{\mathbb{E}}[\Delta z_s^0]|)^2ds|\mathcal{F}_t^{0,1}\Bigr]\Bigr) \\
&~\leq \frac{1}{2}\|\Delta Y\|_{\mathbb{S}^\infty}^2+2 C_\gamma^2
\Bigl( 6 \mathbb{E}\Bigl[\int_t^T (|z_s^0|^2+|\acute{z}_s^0|^2+|z_s^1|^2+|\acute{z}_s^1|^2+|\overline{\mathbb{E}}[z_s^0]|^2+|\overline{\mathbb{E}}[\acute{z}_s^0]|^2)ds|\mathcal{F}_t^{0,1}\Bigr]\Bigr)\\
&\hspace{30mm} \times \Bigl(3 \mathbb{E}\Bigl[\int_t^T (|\Delta z_s^0|^2+|\Delta z_s^1|^2+|\overline{\mathbb{E}}[\Delta z_s^0]|^2)ds|\mathcal{F}_t^{0,1}\Bigr]\Bigr).
\end{split}
\end{equation}
By the same technique used in the proof of Proposition~\ref{prop-Gamma-stability} and $\eqref{bmo-comparison}$, we get
\begin{equation}
\begin{split}
\|\Delta Z\|^2_{\mathbb{H}^2_{\rm BMO}}
&\leq 
4 C_\gamma^2\Bigl(12\sup_{\tau\in\mathcal{T}^{0,1}}\Bigl\|\mathbb{E}\Bigl[\int_\tau^T (|z_s|^2+|\acute{z_s}|^2)ds|\mathcal{F}_\tau^{0,1}\Bigr]\Bigr\|_{\infty}\Bigr)\Bigl(6\sup_{\tau\in\mathcal{T}^{0,1}}\Bigl\|\mathbb{E}\Bigl[\int_\tau^T |\Delta z_s|^2 ds|\mathcal{F}_\tau^{0,1}\Bigr]\Bigr\|_{\infty}\Bigr). 
\end{split}
\end{equation}

Since $z,\acute{z}\in \mathcal{B}_R$, we get
\begin{equation}
\|\Delta Z\|^2_{\mathbb{H}^2_{\rm BMO}}\leq 576C_\gamma^2 R^2\|\Delta z\|^2_{\mathbb{H}^2_{\rm BMO}}.
\end{equation}
Hence the map $\Gamma$ on $\mathcal{B}_R$ becomes contraction if $R<\dfrac{1}{24 C_\gamma}$.
Under the choice of $R=2\|F\|_{\infty}$, this is equivalent to
\begin{equation}
\|F\|_{\infty}<\frac{1}{48 C_\gamma}. 
\end{equation}
In this case $\Delta Z \rightarrow 0$ in $\mathbb{H}^2_{\rm BMO}$ under the repeated application of the map $\Gamma$, and it is also clear that
$\Delta Y\rightarrow 0$ in $\mathbb{S}^\infty$.
The unique fixed point $Z\in \mathcal{B}_R$ of the map $\Gamma$ and the associated $Y$ gives a unique solution 
to the mean-field BSDE $\eqref{mfg-BSDE}$ in the domain $Z\in \mathcal{B}_R$. $\square$

\begin{rem}
Recall that there is the invariance of the optimal trading strategy under the transformation given by $\eqref{duality-relation}$.
Therefore, for our purposes to obtain an equilibrium model with exponential utility, 
the constraint on the terminal function $F$ in Theorem~\ref{th-mfg-existence}
is not a direct restriction on the absolute size of liability, but  on the size of deviation from its mean:
\begin{equation}
\bigl|F-\mathbb{E}[F|\mathcal{F}_0^{0,1}]\bigr|=\bigl|F-\mathbb{E}[F|\mathcal{F}_0^1]\bigr|. \nonumber
\end{equation}
\end{rem}

\subsection{Existence Under the Special Situations}
In this section, we provide special examples where the mean-field BSDE $\eqref{mfg-BSDE}$ has a solution $(Y,Z^0,Z^1)\in \mathbb{S}^\infty\times \mathbb{H}^2_{\rm BMO}\times \mathbb{H}^2_{\rm BMO}$ even when the terminal function $F$ does not satisfy the constraint $\eqref{terminal-constraint}$. 
As a result, we shall see that Theorem~\ref{th-mfg-existence} is merely an example of sufficient conditions for the well-posedness of the mean-field BSDE $\eqref{mfg-BSDE}$.
Let us rewrite the mean-field BSDE $\eqref{mfg-BSDE}$ with the rescaled variables:~\footnote{Recall that we are working under the convention $(F,\gamma)=(F^1,\gamma^1)$.}
\begin{equation}
(y,z^0,z^1):=(\gamma Y, \gamma Z^0, \gamma Z^1), \quad G=\gamma F, \nonumber
\end{equation}
which yields
\begin{equation}
	\begin{split}
		y_t=&G+\int_t^T \Bigl(\widehat{\gamma}z_s^{0\|}\overline{\mathbb{E}}\Bigl[\frac{1}{\gamma}z_s^{0\|}\Bigr]^\top-\frac{\widehat{\gamma}^2}{2}\Bigl|\overline{\mathbb{E}}\Bigl[\frac{1}{\gamma}z_s^{0\|}\Bigr]\Bigr|^2+\frac{1}{2}(|z_s^{0\perp}|^2+|z_s^1|^2)\Bigr)ds\\
		&-\int_t^T z_s^0 dW_s^0-\int_t^T z_s^1 dW_s^1,~~~t\in[0,T]. 
\label{mfg-BSDE-norm}
	\end{split}
\end{equation}
For the first example, we put the following assumption.
\begin{asm}
\label{assumption-special}
The rescaled terminal function $G\in\mathbb{L}^\infty(\mathbb{R},\mathcal{F}^{0,1}_T)$ has an additive form:
\begin{equation}
G=G^0+G^1,
\end{equation}
for $G^0\in\mathbb{L}^\infty(\mathbb{R},\mathcal{F}^{0}_T)$ and $G^1\in\mathbb{L}^\infty(\mathbb{R},\mathcal{F}^{1}_T)$.
\end{asm}

\begin{rem}
In terms of the original liability $F~(=F^1)$, the above condition is equivalent to assume that $F$ has
the following structure:
\begin{equation}
F=\frac{1}{\gamma}\widetilde{F}^0+\widetilde{F}^1, \nonumber
\end{equation}
where $\widetilde{F}^0$ (resp. $\widetilde{F}^1$) is a bounded $\mathcal{F}_T^0$ (resp. $\mathcal{F}_T^1$)-measurable random variable.
In the financial market with distribution of agents as specified by Assumption~\ref{assumption-hetero},
this implies that the part of liability dependent on the common noise are distributed as inversely proportional 
to the agents' coefficients of risk aversion $(\gamma^i)_{i\in \{1,\ldots,N\}}$.
\end{rem}

\begin{thm}
\label{th-special-existence}
Under Assumptions~\ref{assumption-agent} and \ref{assumption-special},
there is, at least, one solution $(y,z^0,z^1)\in \mathbb{S}^\infty\times \mathbb{H}^2_{\rm BMO}
\times \mathbb{H}^2_{\rm BMO}$ to $\eqref{mfg-BSDE-norm}$, or equivalently  a solution $(Y,Z^0,Z^1)\in \mathbb{S}^\infty\times \mathbb{H}^2_{\rm BMO}
\times \mathbb{H}^2_{\rm BMO}$ to $\eqref{mfg-BSDE}$.
\end{thm}

\noindent
\textbf{\textit{Proof}}\\
Consider the following two BSDEs:
\begin{equation}
\begin{split}
	\label{separated-bsde}
		y_t^0&=G^0+\int_t^T \frac{1}{2}|z_s^0|^2ds-\int_t^T z_s^0 dW_s^0, ~~~ t\in[0,T], \\
		y_t^1&=G^1+\int_t^T \frac{1}{2}|z_s^1|^2 ds-\int_t^T z_s^1 dW_s^1, ~~~ t\in[0,T]. 
\end{split}
\end{equation}
It is clear that there exists a unique solution $(y^i, z^i)\in \mathbb{S}^\infty\times\mathbb{H}^2_{\rm BMO}$ with $i=0,1$
by the standard results on qg-BSDEs in \cite{Kobylanski2000BackwardSD}. 
Since $(y^0,z^0)$ is $\mathbb{F}^0$-adapted and so is $z^{0\|}$, we have
\begin{equation}
\overline{\mathbb{E}}\Bigl[\frac{1}{\gamma}z_s^{0\|}\Bigr]=\overline{\mathbb{E}}\Bigl[\frac{1}{\gamma}\Bigr]z_s^{0\|}=\frac{1}{\widehat{\gamma}}z_s^{0\|}. \nonumber
\end{equation}
Therefore, $(y^0,z^0)$ also solves the BSDE
\begin{equation}
y_t^0=G^0+\int_t^T\Bigl(\widehat{\gamma}z^{0\|}_s\overline{\mathbb{E}}\Bigl[\frac{1}{\gamma}z_s^{0\|}\Bigr]^\top-\frac{\widehat{\gamma}^2}{2}
\Bigl|\overline{\mathbb{E}}\Bigl[\frac{1}{\gamma}z_s^{0\|}\Bigr]\Bigr|^2+\frac{1}{2}|z_s^{0\perp}|^2\Bigr)ds-\int_t^T z_s^0 dW_s^0, ~~~t\in[0,T]. \nonumber
\end{equation}
It is now clear that $(y,z^0,z^1):=(y^0+y^1, z^0,z^1)$ provides a solution to $\eqref{mfg-BSDE-norm}$. $\square$

\begin{rem}~\\
	\textup{(i)} Under the assumptions used in Theorem~\ref{th-special-existence}, we actually have a closed form solution,
				\begin{equation}
					\label{Cole-Hopf}
				y_t^j=\log\bigl(\mathbb{E}\bigl[\exp(G^j)|\mathcal{F}_t^j\bigr]\bigr), ~j=0,1. \nonumber
				\end{equation}
				This result can be easily confirmed by the Cole-Hopf transformation (i.e., considering the process $\exp({y_t^j})$).\\
	\textup{(ii)}  The boundedness of $F^1$ in Assumption \ref{assumption-agent} (iii) can be relaxed; we only need the integrability of $\exp(G^j)$ ($j=0,1$). See \cite{briandOne-dim} [Theorem 3.1].\\
	\textup{(iii)} In this case, the risk premium $\theta^{\mathrm{mfg}}$ is expressed by $\theta^{\mathrm{mfg}}_t = - (z^{0\|}_t)^\top$ for each $t\in[0,T]$, and the trading strategy for agent-$i$ reads:
	\[
		p^{i,*}_t = \frac{z^{0\|}_t}{\gamma^i}+\frac{(\theta^{\mathrm{mfg}}_t)^\top}{\gamma^i}=0,~~~t\in[0,T]
	\]
	for all $i\in\{1,\ldots,N\}$. This is the same result as for the model with a representative agent.\\
	\textup{(iv)} In this special case, the equilibrium risk-premium process $\theta$ no longer depends on $G^1$, which is the liability subject to the idiosyncratic noise.
	Qualitatively speaking, agents have no way to hedge $G^1$ through security trading when the liability is additively separated into the parts subject to the common noise $G^0$ and the idiosyncratic noise $G^1$.
	In other words, the demand to hedge $G^0$ is factored into the risk premium $\theta$, while $G^1$ does not affect $\theta$ as it cannot be hedged anyway.
	Example \ref{eg-exp-affine} below visualizes this property with an explicit form of the solution.
\end{rem}
The following example shows that the solution \eqref{Cole-Hopf} can be written explicitly by considering an explicit formulation of $G^0$.
\begin{eg}
	\label{eg-exp-affine}
	We set
	\[
		G^0 := \int_0^T (ax^2_t + bx_t)dt,
	\]	
	where $a\leq0$, $b\in\mathbb{R}$ and $(x_t)_{t\in[0,T]}\in\mathbb{S}^2(\mathbb{F}^0,\mathbb{R})$ is defined by
	\[
		dx_t = (\alpha x_t + \beta)dt + \delta dW_t^0,~~~t\in[0,T],~~~x_0\in\mathbb{R}
	\]
	for $\alpha,\beta\in\mathbb{R}$ and $\delta\in\mathbb{R}^{1\times d_0}$. $(x_t)_{t\in[0,T]}$ is explicitly expressed as
	\[
		x_t = e^{\alpha t}x_0 + \frac{\beta}{\alpha}(e^{\alpha t}-1) + \int_0^t e^{\alpha(t-s)}\delta dW_s^0,~~~t\in[0,T].
	\]
	In this case, $\exp(G^0)$ is integrable and we have
	\begin{equation}
		\mathbb{E}[\exp(G^0)|\mathcal{F}_t^0] = \exp\Bigl(A(t,T)x_t^2 + B(t,T)x_t + C(t,T) + \int_0^t (ax_s^2 + bx_s)ds\Bigr)
	\end{equation}
	for $t\in[0,T]$, where functions $A(\cdot,T)$, $B(\cdot,T)$ and $C(\cdot,T)$ are solution to the following ODEs:
	\begin{equation}
		\begin{split}
		&\dot{A}(t,T) + 2|\delta|^2A(t,T)^2 + 2\alpha A(t,T) + a = 0,\\
		&\dot{B}(t,T) + (\alpha + 2|\delta|^2A(t,T))B(t,T) + 2\beta A(t,T) + b = 0,\\
		&\dot{C}(t,T) + |\delta|^2A(t,T) + \Bigl(\beta + \frac{1}{2}|\delta|^2B(t,T)\Bigr)B(t,T) = 0,\\
		&A(T,T)=B(T,T)=C(T,T)=0.
		\end{split}
	\end{equation}
	for $t\in[0,T]$. The explicit form of $A(\cdot,T)$, $B(\cdot,T)$ and $C(\cdot,T)$ are given as follows (see \cite{carmonaProbabilisticTheoryMean2018} [Equation (2.50)]); for $t\in[0,T]$,
	\begin{equation}
		\begin{split}
		A(t,T) &= \frac{-a(e^{(\rho^+-\rho^-)(T-t)}-1)}{\rho^-e^{(\rho^+-\rho^-)(T-t)}-\rho^+},\\
		B(t,T) &= \int_t^T \exp\Bigl(\int_t^s (\alpha + 2|\delta|^2A(u,T))du\Bigr)(2\beta A(s,T)+b)ds,\\
		C(t,T) &= \int_t^T \Bigl\{|\delta|^2A(s,T) + \Bigl(\beta + \frac{1}{2}|\delta|^2B(s,T)\Bigr)B(s,T)\Bigr\} ds,
		\end{split}
	\end{equation}
	where $\rho^{\pm}:=\alpha\pm\sqrt{\alpha^2-2a|\delta|^2}$.
	Then, we obtain
	\begin{equation}
		\begin{split}
		y^0_t 
		&=
		\log \mathbb{E}[\exp(G^0)|\mathcal{F}_t^0]\\
		&=
		A(t,T)x_t^2 + B(t,T)x_t + C(t,T) + \int_0^t (ax_s^2 + bx_s)ds,~~~t\in[0,T].
		\end{split}
	\end{equation}
	By applying Ito formula to $y^0$, it is easy to see
	\[
		z^0_t = (2A(t,T)x_t + B(t,T))\delta,~~~t\in[0,T].
	\]
	Then, we have
	\[
		\theta^{\rm mfg}_t = -(z^{0\|}_t)^\top = -(2A(t,T)x_t + B(t,T))\Pi_t(\delta)^\top,~~~t\in[0,T].
	\]
	This clearly shows that the risk premium is determined by $G^0$ as noted in the previous remark. 

	When $a=0$, we have 
	\[
		A(t,T)\equiv 0,~~~B(t,T)=\frac{b}{\alpha}(e^{\alpha(T-t)}-1),
	\]
	and $\theta^{\rm mfg}_t = -B(t,T)\Pi_t(\delta)^\top$ for $t\in[0,T]$. Moreover, notice that
	\begin{equation}
		\begin{split}
		G^0 
		&=
		\int_0^T bx_t dt\\
		&=
		b\int_0^T \{e^{\alpha t}x_0 + \frac{\beta}{\alpha}(e^{\alpha t}-1)\}dt + b\int_0^T\int_0^t e^{\alpha(t-s)}\delta dW_s^0dt,
		\end{split}
	\end{equation}
	and that
	\begin{equation}
		\begin{split}
			b\int_0^T\int_0^t e^{\alpha(t-s)}\delta dW_s^0dt 
			= 
			\int_0^T \frac{b}{\alpha}(e^{\alpha(T-t)}-1)\delta dW_t^0
			= 
			\int_0^T B(t,T)\delta dW_t^0
			=
			\int_0^T z^0_t dW_t^0.
		\end{split}
	\end{equation}
	Here, we used (the stochastic version of) Fubini's theorem (see \cite{Medvegyev} [Theorem 5.25]).
	This implies that the Malliavin derivative of $G^0$ is given by $D_tG^0=z^0_t$ for $t\in[0,T]$.
	
	The excess return of securities $\mu^{\mathrm{mfg}}:=\sigma\theta^{\mathrm{mfg}}$ reads:
	\[
		\mu^{\mathrm{mfg}}_t := \sigma_t\theta^{\mathrm{mfg}}_t = -\sigma_t(z^0_t)^\top,~~~t\in[0,T].
	\]
	Equivalently, if we denote by $\mu^{\mathrm{mfg}(k)}_t$ and $\sigma_t^{(k)}$ the $k$-th element of $\mu^{\mathrm{mfg}}_t$ and the $k$-th row vector of $\sigma_t$, respectively, we have
	\[
		\mu^{\mathrm{mfg}(k)}_t = -\sigma^{(k)}_t(z^0_t)^\top,~~~t\in[0,T].
	\]
	When $\sigma^{(k)}_t(z^0_t)^\top>0$, the excess return of the $k$-th security $\mu^{\mathrm{mfg}(k)}_t$ is negative at each time $t\in[0,T]$.
	Note that $\sigma^{(k)}_t(z^0_t)^\top>0$ implies a positive correlation between $\sigma_t^{(k)}dW_t^0$ and $z^0_t dW_t^0$, meaning that a long position in the $k$-th security helps to hedge $G^0$. 
	Therefore, in this case, agents will buy the $k$-th security even if its return is low.
	Conversely, when $\sigma^{(k)}_t(z^0_t)^\top<0$, the excess return of the $k$-th security $\mu^{\mathrm{mfg}(k)}_t$ is positive at each time $t\in[0,T]$.
	From an economic perspective, this is because holding the $k$-th security has a negative effect on hedging $G^0$ in this case, and the market requires a high return on such a security.
\end{eg}

For the second example, we make the following assumption.
\begin{asm} ~\\
    \label{assumption-special2}
    \textup{(i)} The terminal condition $F\in\mathbb{L}^\infty(\mathbb{R},\mathcal{F}^{0,1}_T)$ satisfies
    \[
        F = \frac{1}{\gamma}F^{0} + F^{10} + F^{1}
    \]
    for $F^{0}\in\mathbb{L}^\infty(\mathbb{R},\mathcal{F}^0_T)$, $F^{10}\in\mathbb{L}^\infty(\mathbb{R},\mathcal{F}^{0,1}_T)$, and $F^{1}\in\mathbb{L}^\infty(\mathbb{R},\mathcal{F}^1_T)$.

    \noindent
    \textup{(ii)} The random variable $F^{10}$ satisfies
    \[
         \|F^{10}\|_\infty\leq \frac{1}{48C_\gamma},
    \]
    where $C_\gamma$ is the parameter specified in \eqref{terminal-constraint}.
\end{asm}
By setting
\[
	G^{10}:=\gamma F^{10},~~~G^1:=\gamma F^1,
\]
we have
\[
	G:=\gamma F = F^{0} + G^{10} + G^1.
\]
\begin{thm}
\label{th-special-existence2}
	Under Assumptions~\ref{assumption-agent} and \ref{assumption-special2}, the BSDE $\eqref{mfg-BSDE}$ has a solution $(Y,Z^0,Z^1)\in \mathbb{S}^\infty\times \mathbb{H}^2_{\rm BMO}\times \mathbb{H}^2_{\rm BMO}$.
\end{thm}
\noindent
\textbf{\textit{Proof}}\\
We show that the BSDE \eqref{mfg-BSDE-norm} has a solution $(y,z^0,z^1)\in \mathbb{S}^\infty\times \mathbb{H}^2_{\rm BMO}\times \mathbb{H}^2_{\rm BMO}$.
Let $(\widetilde{y}^0,\widetilde{z}^0)\in\mathbb{S}^\infty(\mathbb{P}^{0},\mathbb{F}^0)\times\mathbb{H}^2_{\mathrm{BMO}}(\mathbb{P}^{0},\mathbb{F}^0)$ and $(\widetilde{y}^1,\widetilde{z}^1)\in\mathbb{S}^\infty(\mathbb{P}^{1},\mathbb{F}^1)\times\mathbb{H}^2_{\mathrm{BMO}}(\mathbb{P}^{1},\mathbb{F}^1)$ be processes that solve the BSDEs 
\begin{equation}
    \begin{split}
    \label{MF-BSDE_case1-2-0}
        \widetilde{y}^0_t &= F^0 + \frac{1}{2}\int_t^T |\widetilde{z}^{0}_s|^2ds - \int_t^T \widetilde{z}^{0}_s dW^0_s,~~~t\in[0,T],\\
        \widetilde{y}^1_t &= G^1 + \frac{1}{2}\int_t^T |\widetilde{z}^{1}_s|^2ds - \int_t^T \widetilde{z}^{1}_s dW^1_s,~~~t\in[0,T].
\end{split}
\end{equation}
Note that processes $(\widetilde{y}^0,\widetilde{z}^0)$ and $(\widetilde{y}^1,\widetilde{z}^1)$ are well-defined due to \cite{Kobylanski2000BackwardSD}.
It is easy to see that $(\widetilde{y}^0,\widetilde{z}^0)$ satisfies
\begin{equation}
    \begin{split}
		\label{BSDE-tilde-y0}
        \widetilde{y}^0_t &= F^0 + \int_t^T \Bigl(\widehat\gamma \widetilde{z}^{0\|}_s\overline{\mathbb{E}}[\gamma^{-1}\widetilde{z}^{0\|}_s]^{\top} - \frac{\widehat\gamma^2}{2}|\overline{\mathbb{E}}[\gamma^{-1}\widetilde{z}^{0\|}_s]|^2 + \frac{1}{2}|\widetilde{z}^{0\perp}_s|^2 \Bigr)ds - \int_t^T \widetilde{z}^{0}_s dW^0_s~~~
\end{split}
\end{equation}
for $t\in[0,T]$, since $\widetilde{z}^{0}$ is an $\mathbb{F}^0$-progressively measurable process.

Define a new measure $\widetilde{\mathbb{P}}^{0,1}$ by
\begin{equation}
    \frac{d\widetilde{\mathbb{P}}^{0,1}}{d\mathbb{P}^{0,1}}\Bigr|_{\mathcal{F}^{0,1}_t} := \mathcal{E}\Bigl(\int_0^\cdot (\widetilde{z}^{0}_s dW^0_s + \widetilde{z}^{1}_sdW^1_s) \Bigr)_t.
\end{equation}
By Girsanov's theorem, the processes
\[
    \widetilde{W}_t^0:=W^0_t - \int_0^t (\widetilde{z}^{0}_s)^\top ds,~~~\widetilde{W}_t^1:=W^1_t - \int_0^t (\widetilde{z}^{1}_s)^\top ds,~~~t\in[0,T]
\]
are $(\mathbb{F}^{0,1},\widetilde{\mathbb{P}}^{0,1})$-standard Brownian motions. Consider a new BSDE:
\begin{equation}
    \begin{split}
    \label{MF-BSDE_case1-3}
        \mathcal{Y}_t =& G^{10} + \int_t^T \Bigl(\widehat\gamma \mathcal{Z}^{0\|}_s\overline{\mathbb{E}}[\gamma^{-1}\mathcal{Z}^{0\|}_s]^{\top} - \frac{\widehat\gamma^2}{2}|\overline{\mathbb{E}}[\gamma^{-1}\mathcal{Z}^{0\|}_s]|^2 + \frac{1}{2}|\mathcal{Z}^{0\perp}_s|^2 + \frac{1}{2}|\mathcal{Z}^1_s|^2 \Bigr)ds\\
         &- \int_t^T \mathcal{Z}^0_s d\widetilde{W}^0_s - \int_t^T \mathcal{Z}^1_s d\widetilde{W}^1_s.
\end{split}
\end{equation}
The BSDE \eqref{MF-BSDE_case1-3} has a solution $(\mathcal{Y},\mathcal{Z}^0,\mathcal{Z}^1)\in\mathbb{S}^\infty(\widetilde{\mathbb{P}}^{0,1},\mathbb{F}^{0,1})\times\mathbb{H}^2_{\mathrm{BMO}}(\widetilde{\mathbb{P}}^{0,1},\mathbb{F}^{0,1})\times\mathbb{H}^2_{\mathrm{BMO}}(\widetilde{\mathbb{P}}^{0,1},\mathbb{F}^{0,1})$. 
Indeed, by letting $(\mathscr{Y},\mathscr{Z}^0,\mathscr{Z}^1):=(\gamma^{-1}\mathcal{Y},\gamma^{-1}\mathcal{Z}^0,\gamma^{-1}\mathcal{Z}^1)$, we have
\begin{equation}
    \begin{split}
		\label{BSDE-scr}
        \mathscr{Y}_t =& F^{10} + \int_t^T \Bigl(\widehat\gamma \mathscr{Z}^{0\|}_s\overline{\mathbb{E}}[\mathscr{Z}^{0\|}_s]^{\top} - \frac{\widehat\gamma^2}{2\gamma}|\overline{\mathbb{E}}[\mathscr{Z}^{0\|}_s]|^2 + \frac{\gamma}{2}|\mathscr{Z}^{0\perp}_s|^2 + \frac{\gamma}{2}|\mathscr{Z}^1_s|^2 \Bigr)ds\\
		&- \int_t^T \mathscr{Z}^0_s d\widetilde{W}^0_s - \int_t^T \mathscr{Z}^1_s d\widetilde{W}^1_s
\end{split}
\end{equation}
for $t\in[0,T]$. By Theorem \ref{th-mfg-existence}, the BSDE \eqref{BSDE-scr} has a solution when $F^{10}$ is small enough in the sense that
\[
    \|F^{10}\|_\infty\leq \frac{1}{48C_\gamma}.
\]
Since $\mathbb{P}^{0,1}\sim \widetilde{\mathbb{P}}^{0,1}$, we have $\mathcal{Y}\in\mathbb{S}^\infty(\mathbb{P}^{0,1},\mathbb{F}^{0,1})$.
Moreover, by the definition of $(\widetilde{W}^0,\widetilde{W}^1)$, the process $(\mathcal{Y},\mathcal{Z}^0,\mathcal{Z}^1)$ satisfies
\begin{equation}
    \begin{split}
        \mathcal{Y}_t 
        =&
        G^{10} + \int_t^T f(s,\mathcal{Z}^0_s, \mathcal{Z}^1_s)ds- \int_t^T \mathcal{Z}^0_s dW^0_s - \int_t^T \mathcal{Z}^1_s dW^1_s,~~~t\in[0,T],
\end{split}
\end{equation}
where $f$ is defined by
\[
	f(s,\mathcal{Z}^0_s, \mathcal{Z}^1_s) := \widehat\gamma \mathcal{Z}^{0\|}_s\overline{\mathbb{E}}[\gamma^{-1}\mathcal{Z}^{0\|}_s]^{\top} - \frac{\widehat\gamma^2}{2}|\overline{\mathbb{E}}[\gamma^{-1}\mathcal{Z}^{0\|}_s]|^2 + \frac{1}{2}|\mathcal{Z}^{0\perp}_s|^2 + \frac{1}{2}|\mathcal{Z}^1_s|^2 + \widetilde{z}^{0}_s(\mathcal{Z}^{0}_s)^\top + \widetilde{z}^{1}_s(\mathcal{Z}^{1}_s)^\top
\]
for $s\in[0,T]$. By Young's inequality, we have
\begin{equation}
    \begin{split}
        f(s,\mathcal{Z}^0_s, \mathcal{Z}^1_s) 
        &=
        -\frac{1}{2}|\overline{\mathbb{E}}[\gamma^{-1}\mathcal{Z}^{0\|}_s] - \mathcal{Z}^{0\|}_s|^2 + \frac{1}{2}|\mathcal{Z}^{0}_s|^2 + \frac{1}{2}|\mathcal{Z}^1_s|^2 + \widetilde{z}^{0}_s(\mathcal{Z}^{0}_s)^\top + \widetilde{z}^{1}_s(\mathcal{Z}^{1}_s)^\top\\
        &\leq
        \frac{3}{4}|\mathcal{Z}^{0}_s|^2 + \frac{3}{4}|\mathcal{Z}^1_s|^2 + |\widetilde{z}^{0}_s|^2 + |\widetilde{z}^{1}_s|^2
\end{split}
\end{equation}
for $s\in[0,T]$. By Ito formula, we have
\begin{equation}
    \begin{split}
        e^{2G^{10}} - e^{2\mathcal{Y}_t}
        &=
        \int_t^T 2e^{2\mathcal{Y}_s}\{-f(s,\mathcal{Z}^0_s, \mathcal{Z}^1_s) +|\mathcal{Z}^0_s|^2 + |\mathcal{Z}^1_s|^2\} ds +  \int_t^T 2e^{2\mathcal{Y}_s}(\mathcal{Z}^0_s dW^0_s + \mathcal{Z}^1_s dW^1_s) \\
        &\geq
        \int_t^T 2e^{2\mathcal{Y}_s}\Bigl\{\frac{1}{4}(|\mathcal{Z}^0_s|^2 + |\mathcal{Z}^1_s|^2) - |\widetilde{z}^{0}_s|^2 - |\widetilde{z}^{1}_s|^2\Bigr\} ds + \int_t^T 2e^{2\mathcal{Y}_s}(\mathcal{Z}^0_s dW^0_s + \mathcal{Z}^1_s dW^1_s) \\
\end{split}
\end{equation}
for $t\in[0,T]$. This implies that
\[
    \mathbb{E}\Bigl[\int_t^T (|\mathcal{Z}^0_s|^2 + |\mathcal{Z}^1_s|^2) ds|\mathcal{F}^{0,1}_t\Bigr] \leq Ce^{4\|\mathcal{Y}\|_{\mathbb{S}^\infty}}(1 + \|\widetilde{z}^{0}\|_{\mathbb{H}^2_{\mathrm{BMO}}} + \|\widetilde{z}^{1}\|_{\mathbb{H}^2_{\mathrm{BMO}}}) <\infty,~~~t\in[0,T],
\]
i.e., $(\mathcal{Z}^0,\mathcal{Z}^1)\in\mathbb{H}^2_{\mathrm{BMO}}(\mathbb{P}^{0,1},\mathbb{F}^{0,1})\times\mathbb{H}^2_{\mathrm{BMO}}(\mathbb{P}^{0,1},\mathbb{F}^{0,1})$.

Finally, we show that the process $(y,z^0,z^1)\in\mathbb{S}^\infty(\mathbb{P}^{0,1},\mathbb{F}^{0,1})\times\mathbb{H}^2_{\mathrm{BMO}}(\mathbb{P}^{0,1},\mathbb{F}^{0,1})\times\mathbb{H}^2_{\mathrm{BMO}}(\mathbb{P}^{0,1},\mathbb{F}^{0,1})$, defined by
\begin{equation}
	\label{special2-sol}
	y_t := \mathcal{Y}_t + \widetilde{y}^0_t+ \widetilde{y}^1_t,~~~z^0_t:=\mathcal{Z}^0_t+\widetilde{z}^{0}_t,~~~z^1_t:=\mathcal{Z}^1_t+\widetilde{z}^{1}_t,~~~t\in[0,T],
\end{equation}
solves the BSDE \eqref{mfg-BSDE-norm}.
Indeed, we have
\begin{equation}
    \begin{split}
		\label{special-driver}
        &f(s,\mathcal{Z}^0_s, \mathcal{Z}^1_s)\\
        &=
        \widehat\gamma \mathcal{Z}^{0\|}_s\overline{\mathbb{E}}[\gamma^{-1}\mathcal{Z}^{0\|}_s]^{\top} - \frac{\widehat\gamma^2}{2}|\overline{\mathbb{E}}[\gamma^{-1}\mathcal{Z}^{0\|}_s]|^2 + \frac{1}{2}|\mathcal{Z}^{0\perp}_s|^2 + \frac{1}{2}|\mathcal{Z}^1_s|^2\\
        &~~~+\widehat\gamma \widetilde{z}^{0\|}_s\overline{\mathbb{E}}[\gamma^{-1}\mathcal{Z}^{0\|}_s]^{\top} - \widehat\gamma \widetilde{z}^{0\|}_s\overline{\mathbb{E}}[\gamma^{-1}\mathcal{Z}^{0\|}_s]^{\top}+ \widetilde{z}^{0\|}_s(\mathcal{Z}^{0\|}_s)^\top + \widetilde{z}^{0\perp}_s(\mathcal{Z}^{0\perp}_s)^\top + \widetilde{z}^{1}_s(\mathcal{Z}^{1}_s)^\top\\
        &=
        \widehat\gamma \mathcal{Z}^{0\|}_s\overline{\mathbb{E}}[\gamma^{-1}\mathcal{Z}^{0\|}_s]^{\top} - \frac{\widehat\gamma^2}{2}|\overline{\mathbb{E}}[\gamma^{-1}\mathcal{Z}^{0\|}_s]|^2 + \frac{1}{2}|\mathcal{Z}^{0\perp}_s|^2 + \frac{1}{2}|\mathcal{Z}^1_s|^2\\
        &~~~+\widehat\gamma \widetilde{z}^{0\|}_s\overline{\mathbb{E}}[\gamma^{-1}\mathcal{Z}^{0\|}_s]^{\top} - \frac{\widehat\gamma^2}{2} \overline{\mathbb{E}}[2\gamma^{-1}\widetilde{z}^{0\|}_s]\overline{\mathbb{E}}[\gamma^{-1}\mathcal{Z}^{0\|}_s]^{\top}+ \widetilde{z}^{0\|}_s(\mathcal{Z}^{0\|}_s)^\top + \widetilde{z}^{0\perp}_s(\mathcal{Z}^{0\perp}_s)^\top + \widetilde{z}^{1}_s(\mathcal{Z}^{1}_s)^\top\\
        &=
        \widehat\gamma (\mathcal{Z}^{0\|}_s + \widetilde{z}^{0\|}_s)\overline{\mathbb{E}}[\gamma^{-1}\mathcal{Z}^{0\|}_s]^{\top} - \frac{\widehat\gamma^2}{2}\overline{\mathbb{E}}[\gamma^{-1}(\mathcal{Z}^{0\|}_s+2\widetilde{z}^{0\|}_s)]\overline{\mathbb{E}}[\gamma^{-1}\mathcal{Z}^{0\|}_s]^\top\\
        &~~~+\widetilde{z}^{0\|}_s(\mathcal{Z}^{0\|}_s)^\top + \frac{1}{2}(\mathcal{Z}^{0\perp}_s + 2\widetilde{z}^{0\perp}_s)(\mathcal{Z}^{0\perp}_s)^\top + \frac{1}{2}(\mathcal{Z}^{1}_s + 2\widetilde{z}^{1}_s)(\mathcal{Z}^{1}_s)^\top\\
        &=
        \widehat\gamma z^{0\|}_s\overline{\mathbb{E}}[\gamma^{-1}(z^{0\|}_s-\widetilde{z}^{0\|}_s)]^{\top} + \widetilde{z}^{0\|}_s(z^{0\|}_s-\widetilde{z}^{0\|}_s)^\top - \frac{\widehat\gamma^2}{2}\overline{\mathbb{E}}[\gamma^{-1}(z^{0\|}_s+\widetilde{z}^{0\|}_s)]\overline{\mathbb{E}}[\gamma^{-1}(z^{0\|}_s-\widetilde{z}^{0\|}_s)]^\top  \\
        &~~~ + \frac{1}{2}(z^{0\perp}_s + \widetilde{z}^{0\perp}_s)(z^{0\|}_s-\widetilde{z}^{0\perp}_s)^\top + \frac{1}{2}(z^{1}_s + \widetilde{z}^{1}_s)(z^{1}_s-\widetilde{z}^{1}_s)^\top\\
        &=
        \widehat\gamma z^{0\|}_s\overline{\mathbb{E}}[\gamma^{-1}(z^{0\|}_s-\widetilde{z}^{0\|}_s)]^{\top} + \widehat\gamma\overline{\mathbb{E}}[\gamma^{-1}\widetilde{z}^{0\|}_s](z^{0\|}_s-\widetilde{z}^{0\|}_s)^\top \\
        &~~~- \frac{\widehat\gamma^2}{2}(|\overline{\mathbb{E}}[\gamma^{-1}z^{0\|}_s]|^2 - |\overline{\mathbb{E}}[\gamma^{-1}\widetilde{z}^{0\|}_s]|^2) + \frac{1}{2}(|z^{0\perp}_s|^2 - |\widetilde{z}^{0\perp}_s|^2) + \frac{1}{2}(|z^{1}_s|^2 - |\widetilde{z}^{1}_s|^2)\\
        &=
        \widehat\gamma z^{0\|}_s\overline{\mathbb{E}}[\gamma^{-1}z^{0\|}_s]^{\top} - \widehat\gamma\widetilde{z}^{0\|}_s\overline{\mathbb{E}}[\gamma^{-1}\widetilde{z}^{0\|}_s]^\top \\
        &~~~- \frac{\widehat\gamma^2}{2}(|\overline{\mathbb{E}}[\gamma^{-1}z^{0\|}_s]|^2 - |\overline{\mathbb{E}}[\gamma^{-1}\widetilde{z}^{0\|}_s]|^2) + \frac{1}{2}(|z^{0\perp}_s|^2 - |\widetilde{z}^{0\perp}_s|^2) + \frac{1}{2}(|z^{1}_s|^2 - |\widetilde{z}^{1}_s|^2)
\end{split}
\end{equation}
for $s\in[0,T]$. Then, by \eqref{BSDE-tilde-y0} and \eqref{special-driver}, we deduce that $(y,z^0,z^1)$ satisfies
\begin{equation}
    \begin{split}
        y_t 
        &:=
        \widetilde{y}^0_t+ \mathcal{Y}_t +  \widetilde{y}^1_t\\
        &= 
        F^0 + G^{10} + G^1 \\
		&~~~+ \int_t^T \Bigl(f(s,\mathcal{Z}^0_s, \mathcal{Z}^1_s) + \widehat\gamma \widetilde{z}^{0\|}_s\overline{\mathbb{E}}[\gamma^{-1}\widetilde{z}^{0\|}_s]^{\top} - \frac{\widehat\gamma^2}{2}|\overline{\mathbb{E}}[\gamma^{-1}\widetilde{z}^{0\|}_s]|^2 + \frac{1}{2}|\widetilde{z}^{0\perp}_s|^2 + \frac{1}{2}|\widetilde{z}^{1}_s|^2 \Bigr)ds\\
        &~~~- \int_t^T (\mathcal{Z}^0_s +\widetilde{z}^{0}_s) dW^0_s - \int_t^T (\mathcal{Z}^1_s +\widetilde{z}^{1}_s) dW^1_s\\
        &=
        G + \int_t^T \Bigl(\widehat\gamma z^{0\|}_s\overline{\mathbb{E}}[\gamma^{-1}z^{0\|}_s]^{\top} - \frac{\widehat\gamma^2}{2}|\overline{\mathbb{E}}[\gamma^{-1}z^{0\|}_s]|^2 + \frac{1}{2}|z^{0\perp}_s|^2 + \frac{1}{2}|z^{1}_s|^2 \Bigr)ds\\
        &~~~- \int_t^T z^0_s dW^0_s - \int_t^T z^1_s dW^1_s
\end{split}
\end{equation}
for $t\in[0,T]$, which is exactly the BSDE \eqref{mfg-BSDE-norm}. $\square$\\



\section{Market Clearing in the Large Population Limit}
\label{sec-market-clearing}
Finally, in this section, we shall show that the process $(\theta^{\rm mfg}_t,t\in[0,T])\in \mathbb{H}^2_{\rm BMO}$ defined by $\eqref{theta-mfg}$ in terms of the solution to the mean-field BSDE is actually a good approximation of the risk-premium process in the market-clearing equilibrium.

In order to treat the large population limit, we first enlarge the complete probability space $(\Omega,\mathcal{F},\mathbb{P})$ 
with filtration $\mathbb{F}:=(\mathcal{F}_t)_{t\geq 0}$ as follows:
$\Omega:=\Omega^0\times \prod_{i=1}^\infty \Omega^i$
and $(\mathcal{F},\mathbb{P})$ is the completion of $\bigl(\mathcal{F}^0\otimes \bigotimes_{i=1}^\infty \mathcal{F}^i ,
\mathbb{P}^0\otimes \bigotimes_{i=1}^\infty \mathbb{P}^i\bigr)$. $\mathbb{F}$ denotes the complete and the right-continuous augmentation of $(\mathcal{F}^0_t\otimes \bigotimes_{i=1}^\infty \mathcal{F}_t^i)_{t\geq 0}$.
$\mathbb{E}[\cdot]$ denotes the expectation with respect to $\mathbb{P}$\footnote{
For general results on the construction of a (countable or even uncountable) product probability space, 
see, for example,  Klenke~\cite{klenke}[Section~14.3].}.

Suppose that the financial market is defined as in Assumption~\ref{assumption-market} with the process $\mu$ given by 
$(\mu_t:=\sigma_t\theta_t^{\rm mfg}, t\in[0,T])$.
Notice that the process $\theta^{\rm mfg}$ (and hence also $\mu$) is $\mathbb{F}^0$-adapted and consistent with our assumption on the information structure. In particular, this means that each agent (agent-$i$) can implement his/her strategy based on the common 
and his/her own idiosyncratic information $\mathbb{F}^0\otimes \mathbb{F}^i$ without taking care of the idiosyncratic noise of the other agents.
Therefore, for each $i\in \mathbb{N}$, the optimal trading strategy of the agent-$i$ is provided by, as in $\eqref{p-i-optimal}$, 
\begin{equation}
p_t^{i,*}~(=(\pi_t^{i,*})^\top\sigma_t)=Z_t^{i,0\|}-\frac{\widehat{\gamma}}{\gamma^i}\overline{\mathbb{E}}[\mathcal{Z}_t^{0\|}], \quad t\in[0,T].  \nonumber
\end{equation}
Here,  $Z^{i,0}$ is associated with the solution of the BSDE (see, $\eqref{BSDE-agent-i}$)
\begin{equation}
\begin{split}
Y_t^i=&F^i+\int_t^T \Bigl(\widehat{\gamma}Z_s^{i,0\|}\overline{\mathbb{E}}[\mathcal{Z}_s^{0\|}]^\top-\frac{\widehat{\gamma}^2}{2\gamma^i}|\overline{\mathbb{E}}[\mathcal{Z}_s^{0\|}]|^2+\frac{\gamma^i}{2}(|Z_s^{i,0\perp}|^2+|Z_s^{i}|^2)\Bigr)ds \\
& -\int_t^T Z_s^{i,0}dW_s^0-\int_t^T Z_s^i dW_s^i, ~ t\in[0,T], 
\end{split}
\label{mfg-agent-i}
\end{equation}
and $\mathcal{Z}^0$  is associated to the solution of
the mean-field BSDE $\eqref{mfg-BSDE-org}$~\footnote{For a clear distinction, we have changed the symbols.}
\begin{equation}
	\begin{split}
	\mathcal{Y}_t^1
	=&
	F^1+\int_t^T \Bigl(\widehat{\gamma}\mathcal{Z}_s^{0\|}\overline{\mathbb{E}}[\mathcal{Z}_s^{0\|}]^\top-\frac{\widehat{\gamma}^2}{2\gamma^1}|\overline{\mathbb{E}}[\mathcal{Z}_s^{0\|}]|^2+\frac{\gamma^1}{2}(|\mathcal{Z}_s^{0\perp}|^2+|\mathcal{Z}_s^1|^2)\Bigr)ds\\
	&-\int_t^T \mathcal{Z}_s^0 dW_s^0-\int_t^T \mathcal{Z}_s^1 dW_s^1,~ t\in[0,T]. \nonumber
\end{split}
\end{equation}
From Theorem~\ref{th-mfg-existence} and Theorem~\ref{th-special-existence}, we already know that there 
exists a bounded solution (possibly not unique) $(\mathcal{Y}^1,\mathcal{Z}^0,\mathcal{Z}^1)\in \mathbb{S}^\infty\times \mathbb{H}^2_{\rm BMO}
\times \mathbb{H}^2_{\rm BMO}$ to the mean-field BSDE $\eqref{mfg-BSDE-org}$ under certain conditions.

Here is the last main result of this chapter.
\begin{thm}
	\label{th-mfg-clearing}
Let Assumptions~\ref{assumption-market} and \ref{assumption-hetero} be in force.
Assume in addition that there is a bounded solution $(\mathcal{Y}^1,\mathcal{Z}^0,\mathcal{Z}^1)\in \mathbb{S}^\infty\times \mathbb{H}^2_{\rm BMO}
\times \mathbb{H}^2_{\rm BMO}$  to the mean-field BSDE $\eqref{mfg-BSDE-org}$ and that
we select an arbitrary but fixed solution from it to define the risk-premium process $(\theta_t^{\rm mfg}:=-\widehat{\gamma}\overline{\mathbb{E}}[\mathcal{Z}_t^{0\|}]^\top,~t\in[0,T])$.
Then,  $\theta^{\rm mfg}$ clears the market in the large population limit in the sense that, 
the agents' optimal trading strategies $(\pi^{i,*})_{i\in \mathbb{N}}$ satisfy the estimate
\begin{equation}
\mathbb{E}\int_0^T\Bigl|\frac{1}{N}\sum_{i=1}^N \pi_t^{i,*}\Bigr|^2 dt\leq \frac{C}{N} 
\end{equation}
with some constant $C>0$ uniformly in $N\in \mathbb{N}$.
\end{thm}

\noindent
\textbf{\textit{Proof}}\\
For a given $\theta^{\rm mfg}~(=-\widehat{\gamma}\overline{\mathbb{E}}[\mathcal{Z}^{0\|}]^\top)$, which is in $\mathbb{H}^2_{\rm BMO}$ by $\eqref{bmo-comparison}$, 
it follows from Corollary~\ref{corollary-existence} that the BSDE $\eqref{mfg-agent-i}$ has a unique bounded solution 
$(Y^i,Z^{i,0},Z^i)\in \mathbb{S}^\infty\times \mathbb{H}^2_{\rm BMO}\times \mathbb{H}^2_{\rm BMO}$
for every $i\in \mathbb{N}$.
In particular, the uniqueness of the solution
implies  $(Y^1,Z^{1,0},Z^1)=(\mathcal{Y}^1,\mathcal{Z}^0,\mathcal{Z}^1)$, the latter of which is the one used to define the process $\theta^{\rm mfg}$.
Thus we have $\overline{\mathbb{E}}[\mathcal{Z}^{0\|}]=\overline{\mathbb{E}}[Z^{1,0\|}]$.

The uniqueness of the solution of $\eqref{mfg-agent-i}$ also implies, by Yamada-Watanabe Theorem~(see, for example, \cite{carmonaProbabilisticTheoryMean2018a}[Theorem 1.33]),  that there exists a measurable function $\Phi$ such that
\begin{equation}
(Y^i, Z^{i,0},Z^i)=\Phi\Bigl(W^0, \xi^i,\gamma^i,W^i,(\theta^{\rm mfg},F^i)\Bigr), \nonumber
\end{equation}
for all $i\in\mathbb{N}$, where $\Phi$ depends only on the joint distribution $\mathcal{L}\bigl(W^0, \xi^i,\gamma^i,W^i,(\theta^{\rm mfg},F^i)\bigr)$.
Since $\theta^{\rm mfg}$ is $\mathbb{F}^0$-adapted and $F^i$ is $\mathcal{F}^0$-conditionally i.i.d, 
this expression implies that the solutions $\{(Y^i, Z^{i,0},Z^i),~i\in\mathbb{N}\}$ are $\mathcal{F}^0$-conditionally i.i.d.

Since $\pi^{i,*}_t=(\sigma_t\sigma_t^\top)^{-1}\sigma_t (p^{i,*}_t)^\top$, we have,  from Assumption~\ref{assumption-market} (ii),
\begin{equation}
\mathbb{E}\int_0^T\Bigl|\frac{1}{N}\sum_{i=1}^N \pi_t^{i,*}\Bigr|^2 dt\leq C\mathbb{E}\int_0^T\Bigl|\frac{1}{N}\sum_{i=1}^N p_t^{i,*}\Bigr|^2 dt \nonumber
\end{equation}
with some constant $C$. Since $\mathcal{Z}^0=Z^{1,0}$, we have
\begin{equation}
\begin{split}
\frac{1}{N}\sum_{i=1}^N p_t^{i,*}=\frac{1}{N}\sum_{i=1}^N(Z_t^{i,0\|}-\overline{\mathbb{E}}[Z_t^{1,0\|}])+\frac{1}{N}\sum_{i=1}^N
\Bigl(1-\frac{\widehat{\gamma}}{\gamma^i}\Bigr)\overline{\mathbb{E}}[Z_t^{1,0\|}], \nonumber
\end{split}
\end{equation}
which then yields 
\begin{equation}
\begin{split}
&\mathbb{E}\int_0^T \Bigl|\frac{1}{N}\sum_{i=1}^N p_t^{i,*}\Bigr|^2 dt\\
&\leq 
2\mathbb{E}\int_0^T \Bigl| \frac{1}{N}\sum_{i=1}^N(Z_t^{i,0\|}-\overline{\mathbb{E}}[Z_t^{1,0\|}])\Bigr|^2dt+
2\widehat{\gamma}^2 \mathbb{E}\int_0^T\Bigl|\frac{1}{N}\sum_{i=1}^N\Bigl(\frac{1}{\widehat{\gamma}}-\frac{1}{\gamma^i}\Bigr)
\overline{\mathbb{E}}[Z_t^{1,0\|}]\Bigr|^2 dt \\
&=
2\mathbb{E}\int_0^T \Bigl| \frac{1}{N}\sum_{i=1}^N(Z_t^{i,0\|}-\overline{\mathbb{E}}[Z_t^{1,0\|}])\Bigr|^2dt+
2\widehat{\gamma}^2 \mathbb{E}\Bigl[\Bigl|\frac{1}{N}\sum_{i=1}^N\Bigl(\frac{1}{\widehat{\gamma}}-\frac{1}{\gamma^i}\Bigr)\Bigr|^2\Bigr]
\mathbb{E}\int_0^T |\overline{\mathbb{E}}[Z_t^{1,0\|}]|^2 dt, \nonumber
\end{split}
\end{equation}
where we used the independence of $(\gamma^i)_{i\in \mathbb{N}}$ and $\mathcal{F}^0$ in the second line.
Since $(1/\gamma^i)_{i=1}^N$ are i.i.d. and $(Z_t^{i,0\|})_{i=1}^N$ are $\mathcal{F}^0$-conditionally i.i.d., 
the cross terms of the two expectations both vanish. Hence, we obtain
\begin{equation}
\begin{split}
\mathbb{E}\int_0^T \Bigl|\frac{1}{N}\sum_{i=1}^N p_t^{i,*}\Bigr|^2 dt &\leq \frac{2}{N^2}\sum_{i=1}^N\mathbb{E}\int_0^T |Z_t^{0,i}-\overline{\mathbb{E}}[Z_t^{0,1}]|^2dt+\frac{2\widehat{\gamma}^2}{N^2}
\sum_{i=1}^N\mathbb{E}\Bigl[\Bigl|\frac{1}{\widehat{\gamma}}-\frac{1}{\gamma^i}\Bigr|^2\Bigr]\mathbb{E}\int_0^T |Z_t^{1,0}|^2 dt\\
&\leq \frac{4}{N}\Bigl(1+\frac{\widehat{\gamma}^2}{\underline{\gamma}^2}\Bigr)\|Z^{0,1}\|^2_{\mathbb{H}^2_{\rm BMO}}, \nonumber
\end{split}
\end{equation}
where we have used Jensen's inequality and $\eqref{ineq-energy}$ with $n=1$. This gives the desired result. $\square$

Combined with Theorems~\ref{th-mfg-existence} and \ref{th-special-existence}, the result of Theorem~\ref{th-mfg-clearing} implies that
$\theta^{\rm mfg}$ given by $\eqref{theta-mfg}$ achieves the market-clearing condition in the large-$N$ limit in the sense that
\begin{equation}
\lim_{N\rightarrow \infty} \frac{1}{N}\sum_{i=1}^N \pi_t^{i,*}=0, ~~~dt\otimes\mathbb{P}\text{-}\mathrm{a.e.}
\end{equation}
at least if $|F^i-\mathbb{E}[F^i|\mathcal{F}^i_0]|$ is small enough or $F^i$ has an additive form for every $i\in\mathbb{N}$.
From an economic perspective, it means that the excessive demand/supply per capita converges to zero $dt \otimes\mathbb{P}$-a.e. in the large population limit.
To fully understand the general conditions for the existence of mean-field equilibrium, we need further study 
on the mean-field BSDE $\eqref{mfg-BSDE}$. 

\begin{rem}
Since the contraction mapping approach by Tevzadze~\cite{tevzadzeSolvabilityBackwardStochastic2008}[Proposition 1] used in our proof of Theorem~\ref{th-mfg-existence}
is applicable to multi-dimensional setups, it may be used to prove the existence of  the  equilibrium 
among the finite number of  (say, $N$) agents, by directly solving the coupled system of qg-BSDEs $\eqref{qg-BSDE-system}$ (with appropriate 
modifications of equations as well as the admissible spaces). 
However, this approach requires the assumption of  perfect information as well as the smallness of the {\bf{total}} size of liabilities among the agents; 
$\|\sqrt{\sum_{i=1}^N |F^i|^2}\|_{\infty}$. This means that the constraint on the liabilities becomes much more stringent as the population grows.
\end{rem}

\begin{rem}
Unfortunately, uniqueness of the market equilibrium is not proved in this chapter. We expect that 
we may have uniqueness when the conditions for Theorem 4.1 are satisfied, i.e. small 
enough $\|F\|_{\infty}$ for which the mean-field BSDE $\eqref{mfg-BSDE}$ has a unique solution in a specified domain. 
However, in this chapter, we have just tried a ``guess-and-check" strategy to find a candidate of equilibrium.  Therefore, although it seems unlikely, we cannot exclude the existence of a market-clearing equilibrium that is not characterized by our BSDE  $\eqref{mfg-BSDE}$. 
\end{rem}

\section{Conclusion and Discussions}
\label{sec-conclusion}
In this chapter, we studied a problem of equilibrium price formation among a large number of investors with exponential utility.
We allowed the agents to be heterogeneous in their initial wealth, parameter of risk aversion, as well as stochastic liability at the terminal time.
We showed that the equilibrium risk-premium process of risky stocks is characterized by the solution to  a novel mean-field BSDE, whose driver has quadratic growth both in the stochastic integrands 
and in their conditional expectations. We proved the existence of a solution to the mean-field BSDE under several conditions
and showed that the resultant risk-premium process actually clears the market in the large population limit.

Let us point out several directions for further research. Extending the proposed technique to 
the price formation with consumptions or different interactions among the agents is an important research topic.
Applications to macroeconomic models in the presence of production may also be possible. 
Finally, the novel mean-field BSDE $\eqref{mfg-BSDE-org}$ deserves further investigations to relax its
well-posedness conditions. The same type of BSDEs are expected to appear in various applications of 
MFGs and the martingale method~\cite{huUtilityMaximizationIncomplete2005a}.

%% file: part2.tex
\section{Preliminary} \label{Section 1}
Asset pricing theory plays a crucial role in financial economics as it investigates how asset prices are determined through market interactions. 
The fundamental objective of the theory is to establish the equilibrium price at which the supply of assets matches the demand. 
See, for example, Back \cite{backAssetPricingPortfolio2017} and Munk \cite{Munk} for details. Karatzas \& Shreve \cite{karatzas_methods_1998} also offers comprehensive descriptions of the equilibrium asset pricing in complete markets. 
The continuous-time stochastic equilibrium pricing problems in incomplete markets are being actively researched as there are still many open issues. In recent years, numerous research efforts have been devoted to showing the existence of equilibrium solutions in incomplete markets under various conditions. 
See, for example, Christensen \& Larsen \cite{ChristensenLarsen}, Cuoco \& He \cite{CuocoHe} and \v{Z}itkovi\'{c} \cite{Zitkovic} and references therein. 
Furthermore, we refer to Jarrow \cite{jarrow2018continuous} [Part III] for a well-integrated review on this subject.

    The mean field game theory, first introduced by Lasry \& Lions \cite{lasryMeanFieldGames2007} and Huang, Malhame \& Caines \cite{huangLargePopulationStochastic2006}, has emerged as a powerful framework for studying multi-agent games. Traditional approaches to such games usually result in intractable problems due to complex interactions among agents. 
The mean field game theory overcomes this challenge by replacing such problems with a stochastic control problem of a single representative agent and a fixed point problem. Lasry \& Lions \cite{lasryMeanFieldGames2007} and Huang, Malhame \& Caines \cite{huangLargePopulationStochastic2006} presented an analytic approach, in which they showed that the problem can be framed as two highly coupled nonlinear partial differential equations. 
Meanwhile, Carmona \& Delarue \cite{carmonaProbabilisticAnalysisMeanField2013, carmonaForwardBackwardStochastic2015} introduced the probabilistic approach to the mean field problem employing forward-backward stochastic differential equations (FBSDEs) of McKean-Vlasov type in lieu of a system of partial differential equations.
The solution of these mean field equations is known to provide an $\varepsilon$-Nash equilibrium of the original game with a large but finite number of agents. The probabilistic approach is extensively covered in two volumes of monographs Carmona \& Delarue \cite{carmonaProbabilisticTheoryMean2018,carmonaProbabilisticTheoryMean2018a}, offering thorough details and applications.
Furthermore, the mean field game theory has been applied to various studies in the field of financial economics. For instance, Fu, Su \& Zhou \cite{Fu2020MeanFE}, Fu \& Zhou \cite{fu_Mean_field_portfolio_games} and Fu \cite{fu_Mean_field_portfolio_games_consumption} propose stochastic games among multiple agents with exponential or power utility competing in a relative performance criterion. 
Bo, Wang \& Yu \cite{mfg-ext-habit} studies the equilibrium consumption problem under external habit formation with many agents.
These examples illustrate the relevance and usefulness of mean field game theory in tackling complex interactive problems, providing valuable insights in financial economics and related fields.

    In recent years, there have been an increasing number of studies on asset pricing theory adopting the mean field game approach. They aim to determine the equilibrium price process based on the optimal behavior of the market participants under the market-clearing condition. One notable area of interest has been the investigation of price formation in electricity markets.
Shrivats, Firoozi \& Jaimungal \cite{SREC} employs FBSDEs of McKean-Vlasov type to study a pricing model in Solar Renewable Energy Certificate (SREC) markets and Firoozi, Shrivats \& Jaimungal \cite{Firoozi_et_al} studies principal agent mean field games in REC markets. Gomes \& S\'{a}ude \cite{Gomes_MFG} develops a deterministic price formation model in which agents can both store and trade electricity. Gomes, Gutierrez \& Ribeiro \cite{Gomes_rand_MFG} extends this model by considering the randomness on the supply side and \cite{Gomes_MFG_noise} deals with a price formation of commodities with stochastic production.
In the realm of financial economics, Evangelista, Saporito \& Thamsten \cite{evangelista2022price} develops a mean field game theoretic model of asset pricing with consideration of liquidity issues. Fujii \& Takahashi \cite{fujiiMeanFieldGame2022, Fujii-Takahashi_strong} present a mean field pricing model for securities under stochastic order flows and \cite{fujii2022equilibrium} provide its extension with a major player.  Fujii \cite{Fujii-equilibrium-pricing}  develops a price formation model in which the market participants consist of two groups: cooperative and non-cooperative ones.
Moreover, Fujii \& Sekine \cite{fujiiMeanFieldEquilibriumPrice2023a} studies a mean field equilibrium pricing model in an incomplete market participated by heterogeneous agents with exponential utility, but without considering agents' consumption. 

    The main contribution of this chapter is an extension of the aforementioned work \cite{fujiiMeanFieldEquilibriumPrice2023a}. This chapter aims to further explore the equilibrium pricing model in an incomplete market with heterogeneous agents, taking the agents' consumption behavior and habit formation into account. The research on consumption habit formation has been a fundamental and classical subject in financial economics. 
The existence of habit formation relaxes the assumption of time-separable utility functions by making the utility dependent not only on the current level of consumption but also on the agent's accumulated stock of past consumption. 
Early studies include, for instance,  \cite{constantinides_habit_1990,detemple_asset_1991,Pollak1970HabitFA,RyderHeal}. Our model specifically incorporates heterogeneity among agents in various aspects, including their initial wealth, initial consumption habits, liabilities and coefficients of risk aversion.
In this chapter, we start from the utility maximization problem of a single agent, which draws inspiration from the work Hu, Imkeller \& M\"{u}ller \cite{huUtilityMaximizationIncomplete2005a}, and derive the relevant BSDE. After proving its well-posedness, we construct the market risk-premium process endogenously under the market clearing condition by introducing the mean field BSDE. As we have done in \cite{fujiiMeanFieldEquilibriumPrice2023a}, we prove its well-posedness using the method proposed by Tevzadze \cite{tevzadzeSolvabilityBackwardStochastic2008} with additional assumptions on the size of the parameters. We then verify that the risk-premium process, expressed by its solution, indeed clears the market in the large population limit. Another contribution of this chapter is to offer an exponential quadratic Gaussian (EQG) formulation of the model, in which a solution to the mean field BSDE can be characterized by a system of ordinary differential equations.
Since the EQG model provides a semi-analytic solution, it would allow detailed numerical studies in future work.

    This chapter consists of five sections. In Section \ref{Section 1}, after providing the introduction, we provide the notations for frequently used sets and spaces. 
In Section \ref{Section 2}, we offer a mathematical formulation of the financial market and solve the optimal consumption-investment problem for a single agent. 
In Section \ref{Section 3}, we derive a mean field BSDE whose driver has a quadratic growth in both stochastic integrand and its conditional expectation and prove that it has a bounded solution under additional assumptions. We also verify that its solution does characterize the financial market in equilibrium in the large population limit. 
Furthermore, in Section \ref{Section 4}, we introduce the EQG framework and prove each result corresponding to Section \ref{Section 3}.
We conclude the chapter with a brief summary and suggestions for possible extensions in Section \ref{Section 5}. 

\section{Optimal Consumption-Investment Problem for a Single Agent} \label{Section 2}
In this section, we investigate the optimal consumption-investment problem for a single agent (whom we shall call ``agent-1'' hereafter). 
We basically follow the same line of arguments as in Fujii \& Sekine \cite{fujiiMeanFieldEquilibriumPrice2023a} and adopt the technique developed by Hu, Imkeller \& M\"{u}ller \cite{huUtilityMaximizationIncomplete2005a}. 
In this work, however, we take an agent's consumption and habit formation into consideration. 
As we shall see, this extension requires a clever choice of supermartingale processes that are needed to verify the optimality.

\subsection{The Market and the Utility Function}

To formulate the optimization problem for agent-1, let us first introduce the relevant probability spaces.\\

\noindent
(1) We denote by $(\Omega^0,\mathcal{F}^0,\mathbb{P}^0)$ a complete probability space with complete and right-continuous filtration $\mathbb{F}^0:=(\mathcal{F}^0_t)_{t\in[0,T]}$ generated by a $d_0$-dimensional standard Brownian motion $W^0:=(W^0_t)_{t\in[0,T]}$ with $\mathcal{F}^0 := \mathcal{F}^0_T$. $(\Omega^0,\mathcal{F}^0,\mathbb{P}^0)$ is used to describe the randomness of the financial market. 
Moreover, we denote by $(\Omega^1,\mathcal{F}^1,\mathbb{P}^1)$ a complete probability space with complete and right-continuous filtration $\mathbb{F}^1:=(\mathcal{F}^1_t)_{t\in[0,T]}$ generated by a $d$-dimensional standard Brownian motion $W^1:=(W^1_t)_{t\in[0,T]}$ and a $\sigma$-algebra $\sigma(\xi^1,\gamma^1,\beta^1,X^1_0,F^1_0)$, where the completion of the latter gives $\mathcal{F}^1_0$. We set $\mathcal{F}^1 := \mathcal{F}^1_T$.
Here, $\xi^1$, $X^1_0$ and $F^1_0$ are $\mathbb{R}$-valued, bounded random variables and $\gamma^1$ and $\beta^1$ are $\mathbb{R}_{++}$-valued bounded random variables. $(\Omega^1,\mathcal{F}^1,\mathbb{P}^1)$ is used to describe the idiosyncratic environment for agent-1.\\
 
\noindent
(2) We denote by $(\Omega^{0,1},\mathcal{F}^{0,1},\mathbb{P}^{0,1})$ a complete probability space over $\Omega^{0,1} := \Omega^0 \times \Omega^1$. 
Here, $(\mathcal{F}^{0,1},\mathbb{P}^{0,1})$ is the completion of $(\mathcal{F}^0 \otimes \mathcal{F}^1,\mathbb{P}^{0}\otimes \mathbb{P}^{1})$ and $\mathbb{F}^{0,1}:=(\mathcal{F}^{0,1}_t)_{t\in[0,T]}$ denotes the complete and right continuous augmentation of $(\mathcal{F}_t^0 \otimes \mathcal{F}_t^1)_{t\in[0,T]}$.\\ 

We set $\mathcal{T}^{0,1}:=\mathcal{T}(\mathbb{F}^{0,1})$ and $\mathcal{T}^{0}:=\mathcal{T}(\mathbb{F}^{0})$ for notational simplicity. The market dynamics and the idiosyncratic environment of agent-1 are modelled on the filtered probability space $(\Omega^{0,1},\mathcal{F}^{0,1},\mathbb{P}^{0,1},\mathbb{F}^{0,1})$. 
Whenever we introduce random variables on a marginal probability space, we identify them with their natural extension to the product space. For example, we use the same symbol $X$ for a random variable $X(\omega^0)$ defined on $(\Omega^0,\mathcal{F}^0,\mathbb{P}^0)$ and its natural extension $X(\omega^0,\omega^1):=X(\omega^0)$ defined on $(\Omega^{0,1},\mathcal{F}^{0,1},\mathbb{P}^{0,1})$.
In this section, we write $\mathbb{E}[\cdot]$ instead of $\mathbb{E}^{\mathbb{P}^{0,1}}[\cdot]$ unless otherwise stated. \par
We now introduce the market dynamics and its properties in the following assumption.
\begin{asm}~\\
    \label{asm1-habit}
    \textup{(i)} The risk-free interest rate is zero.\\
    \textup{(ii)} There are $n\in\mathbb{N}$ non-dividend paying risky stocks whose price dynamics, represented by an $n$-dimensional vector, is given by
    \begin{equation}
        \begin{split}
            \label{stock price}
            S_t&= S_0 + \int_0^t \mathrm{diag}(S_r)(\mu_rdr + \sigma_r dW^0_r),~~~t\in[0,T],
        \end{split}
    \end{equation}
    where $S_0\in\mathbb{R}^n_{++}$, $\mu := (\mu_t)_{t\in[0,T]}\in\mathbb{H}^2_{\mathrm{BMO}}(\mathbb{P}^{0},\mathbb{F}^0,\mathbb{R}^n)$ and $\sigma :=(\sigma_t)_{t\in[0,T]}\in\mathbb{L}^\infty(\mathbb{P}^0,\mathbb{F}^0,\mathbb{R}^{n\times d_0})$. $S_0$ is an $n$-dimensional vector representing the initial stock prices. Moreover, we assume that the process $\sigma$ is of full rank and satisfies
    \[
        \underline{\lambda}I_n\leq \sigma_t\sigma_t^\top \leq\overline{\lambda}I_n,~~~~dt\otimes \mathbb{P}^0\text{-}\mathrm{a.e.}
    \]
    for some positive constants $0<\underline{\lambda}<\overline{\lambda}$ and an identity matrix of size $n$, denoted by $I_n$. We set $n\leq d_0$ so that the financial market is incomplete in general.
\end{asm}

Under this assumption, the process $(\sigma_t\sigma_t^\top)_{t\in[0,T]}$ is regular and the risk-premium process $\theta:=(\theta_t)_{t\in[0,T]}$ is defined by $\theta_t = \sigma_t^\top(\sigma_t\sigma_t^\top)^{-1}\mu_t\in\mathbb{H}^2_{\mathrm{BMO}}(\mathbb{P}^{0},\mathbb{F}^0,\mathbb{R}^{d_0})$. Note that $\theta_t\in\mathrm{Range}(\sigma_t^\top)=\mathrm{Ker}(\sigma_t)^\perp$. It is worth mentioning that by having $\theta\in\mathbb{H}^2_{\mathrm{BMO}}$, we can change the probability measure $\mathbb{P}^0$ to the risk-neutral measure $\mathbb{Q}$, which is defined by 
\begin{equation}
    \label{risk-neutral-habit}
    \Bigl.\frac{d\mathbb{Q}}{d\mathbb{P}^{0,1}}\Bigr|_{\mathcal{F}_t} = \mathcal{E}\Bigl(-\int_0^\cdot \theta_s^\top dW^0_s \Bigr)_t,~~t\in[0,T].
\end{equation}
This ensures the well-posedness of the stock price process \eqref{stock price}, even though $\mu$ is unbounded (See Kazamaki \cite{kazamaki_sufficient_1979}).

\begin{dfn}
    \label{subspace-L}
    For each $s\in[0,T]$, let us denote by $L_s := \{u^\top\sigma_s ; u\in\mathbb{R}^n\}$ the linear subspace of $\mathbb{R}^{1\times d_0}$ spanned by the $n$ row vectors of $\sigma_s$. Furthermore, we define a map $\Pi_s:\mathbb{R}^{1\times d_0}\to L_s$ as an orthogonal projection onto $L_s$.
\end{dfn}

By its construction, we have $\theta_s^\top\in L_s$ for every $s\in[0,T]$.
\begin{rem}
    For notational convenience, we shall write
    \[
        Z^\|_s := \Pi_s(Z_s) ,~~~Z^\perp_s := Z_s-\Pi_s(Z_s),~~~s\in[0,T]
    \]
    for an $\mathbb{R}^{1\times d_0}$-valued progressively measurable process $Z$. Note that the process $(Z_s^\|)_{s\in[0,T]}$ is also progressively measurable by Karatzas \& Shreve \cite{karatzas_methods_1998} [Lemma 4.4].
\end{rem}

We shall model the idiosyncratic environment of agent-1 through a 5-tuple $(\xi^1,\gamma^1,\beta^1,X^1_0,F^1)$.
\begin{asm}~\\
    \label{asm2-habit}
    \textup{(i)} $\xi^1$ is an $\mathbb{R}$-valued, bounded, and $\mathcal{F}^1_0$-measurable random variable representing the initial wealth of agent-1.\\
    \textup{(ii)} $\gamma^1$ is an $\mathbb{R}$-valued, bounded, and $\mathcal{F}^1_0$-measurable random variable satisfying $\underline{\gamma}\leq\gamma^1\leq\overline{\gamma}$ with some positive constants $0<\underline{\gamma}\leq\overline{\gamma}$. $\gamma^1$ is the coefficient of absolute risk aversion of agent-1 with respect to his/her net wealth. \\
    \textup{(iii)} $\beta^1$ is an $\mathbb{R}$-valued, bounded, and $\mathcal{F}^1_0$-measurable random variable satisfying $\underline{\beta}\leq\beta^1\leq\overline{\beta}$ with some positive constants $0<\underline{\beta}\leq\overline{\beta}$. $\beta^1$ is the coefficient of absolute risk aversion of agent-1 with respect to his/her consumption level. \\
    \textup{(iv)} $X^1_0$ is an $\mathbb{R}$-valued, bounded, and $\mathcal{F}^1_0$-measurable random variable representing agent-1's initial stock of habits. \\
    \textup{(v)} $F^1:=(F^1_t)_{t\in[0,T]}$ is an $\mathbb{R}$-valued, bounded, and $\mathbb{F}^{0,1}$-progressively measurable process. For each $t\in[0,T]$, $F^1_t$ represents the amount of the liability at time $t$ of agent-1.\\
    \textup{(vi)} $\rho:=(\rho_t)_{t\in[0,T]}$ is an $\mathbb{R}$-valued, bounded, and $\mathbb{F}^{0}$-progressively measurable process. The process $\rho$ represents the habit trend influenced by the market shocks. \\
    \textup{(vii)} Agent-1 is a price taker; agent-1 must accept the prevailing prices as he/she lacks the market share to impact the market price.
\end{asm}

The trading and consumption strategies of agent-1 are denoted by $(\pi,c)$, where $\pi:=(\pi_t)_{t\in[0,T]}$ is an $\mathbb{R}^n$-valued, $\mathbb{F}^{0,1}$-progressively measurable process representing the amount of money invested in $n$ stocks and $c:=(c_t)_{t\in[0,T]}$ is an $\mathbb{R}$-valued, 
$\mathbb{F}^{0,1}$-progressively measurable process representing agent-1's (nominal) consumption\footnote{We do not forbid the process $c$ having negative values in order to make the analysis simple. The negative $c$ may be interpreted as, for example, ``net consumption'', i.e. consumption minus labour income. Moreover, this model does not consider the production side or the goods market equilibrium, unlike typical macroeconomic models.} process. 
The wealth process of agent-1 with strategy $(\pi,c)$ is then given by
\[
    \mathcal{W}^{1,(\pi,c)}_t = \xi^1 + \sum_{j=1}^n \int_0^t \frac{\pi_{j,s}}{S^j_s}dS^j_s - \int_0^t c_sds =\xi^1 + \int_0^t (\pi_s^\top\sigma_s\theta_s - c_s)ds + \int_0^t \pi_s^\top\sigma_s dW_s^0.
\]

We now formulate the utility maximization problem of agent-1 as follows: agent-1 solves
\[
    \sup_{(\pi,c)\in\mathbb{A}^1} U^1(\pi,c)
\]
subject to
\[
    \mathcal{W}^{1,(\pi,c)}_t =\xi^1 + \int_0^t (\pi_s^\top\sigma_s\theta_s - c_s)ds + \int_0^t \pi_s^\top\sigma_s dW_s^0,~~~t\in[0,T],
\]
where $\mathbb{A}^1$ is a set of admissible strategies for agent-1, whose definition is to be given and $U^1:\mathbb{A}^1\to\mathbb{R}$ is the utility function defined by
\begin{equation}
    \begin{split}
        \label{utility}
        &U^1(\pi,c) \\
        &:= 
        \mathbb{E}\Bigl[-\exp\Bigl(-\delta T-\gamma^1(\mathcal{W}^{1,(\pi,c)}_T-F^1_T)\Bigr) -a \int_0^T \exp\Bigl(-\delta t-\gamma^1(\mathcal{W}^{1,(\pi,c)}_t-F^1_t)-\beta^1(c_t-X_t^{1,c})\Bigr)dt\Bigr]
    \end{split}
\end{equation}
for some constants $a,\delta>0$ representing the weight of the running utility with respect to the terminal utility and the discount rate, respectively. Here, $X^{1,c}$ represents the agent-1's consumption habits defined by a mean-reverting process
\begin{equation}
    \begin{split}
        \label{habit}   
        X_t^{1,c} = X^1_0 + \int_0^t \{-\kappa(X^{1,c}_s-\rho_s) + b(c_s-\rho_s)\}ds,~~~t\in[0,T]
    \end{split}
\end{equation}
for some constants $b, \kappa>0$. By a simple calculation, we can write it in an explicit form as
\[
    X_t^{1,c} = e^{-\kappa t}X^1_0 + \int_0^t e^{-\kappa(t-s)} \{b c_s + (\kappa-b)\rho_s\}ds,~~~t\in[0,T].
\]

\begin{rem}
    The economic interpretation of the habit process $X^{1,c}$ and the utility function $U^1$ is as follows.\\
    \textup{(i)} The consumption habit $X^{1,c}_t$ is determined by the accumulation of past private consumption $(c_s)_{s\in[0,t]}$ and the given consumption trend $(\rho_s)_{s\in[0,t]}$ in the market. In particular, a higher level of past consumption increases the agent's current habit by having $b>0$. The size of the parameter $\kappa>0$ determines the rate at which the consumption habit decays. \\
    \textup{(ii)} The term $c-X^{1,c}$ in the running utility means that the agent evaluates the current consumption level relative to his/her habit. Importantly, the agent's preference is no longer time-separable as the past consumption level has an effect on the current consumption choice. \\
    \textup{(iii)}  The amount of net asset $\mathcal{W}^{1,(\pi,c)}-F^1$ enters both in the agent-1's running and terminal preferences. Note that the low portfolio performance $\mathcal{W}^{1,(\pi,c)}_t-F^1_t<0$ is heavily punished whereas the high performance $\mathcal{W}^{1,(\pi,c)}_t-F^1_t>0$ is only weakly valued in this type of utility function.
\end{rem}

The admissible strategy for agent-1 is defined as follows.
\begin{dfn} (Admissible space for agent-1)\\
    \label{Admissible space for agent-1}
        The admissible space $\mathbb{A}^1$ is the set of trading and consumption strategies $(\pi,c)\in \mathbb{H}^2(\mathbb{P}^{0,1},\mathbb{F}^{0,1},\mathbb{R}^{n})\times \mathbb{H}^2(\mathbb{P}^{0,1},\mathbb{F}^{0,1},\mathbb{R})$ such that a family
        \[
            \Bigl\{\exp\Bigl(-\gamma^1\mathcal{W}^{1,(\pi,c)}_\tau + \beta^1 |c_\tau| + K^1 |X_\tau^{1,c}|  \Bigr) ; \tau\in\mathcal{T}^{0,1}\Bigr\}
        \]
        is uniformly integrable for some $K^1>\gamma^1(A_1+\sqrt{A_1^2+B_1})^{-1} \lor \beta^1 $ where
        \[
            A_1:=\frac{1}{2}\Bigl(\kappa-b+\frac{\gamma^1}{\beta^1}\Bigr),~~~B_1:=\frac{\gamma^1 b}{\beta^1}.
        \]
        Moreover, we define $\mathcal{A}^1:=\{(p,c)=(\pi^\top\sigma,c) ; (\pi,c)\in\mathbb{A}^1\}$.
\end{dfn}
By writing $p_s:=\pi_s^\top\sigma_s$ for each $s\in[0,T]$, the utility maximization problem can be equivalently written as
\begin{equation}
    \begin{split}
        \label{p-utility}
        \sup_{(p,c)\in\mathcal{A}^1} \widetilde{U}^1(p,c)
    \end{split}
\end{equation}
subject to
\[
    \mathcal{W}^{1,(p,c)}_t =\xi^1 + \int_0^t (p_s\theta_s - c_s)ds + \int_0^t p_s dW_s^0,
\]
where $\widetilde{U}^1:\mathcal{A}^1\to\mathbb{R}$ is defined by
\begin{equation}
    \begin{split}
    &\widetilde{U}^1(p,c)\\
    &:=\mathbb{E}\Bigl[-\exp\Bigl(-\delta T-\gamma^1(\mathcal{W}^{1,(p,c)}_T-F^1_T)\Bigr) -a \int_0^T \exp\Bigl(-\delta t-\gamma^1(\mathcal{W}^{1,(p,c)}_t-F^1_t)-\beta^1(c_t-X_t^{1,c})\Bigr)dt\Bigr].
    \end{split}
\end{equation}
Note further that for each $s\in[0,T]$ and $(p,c)\in\mathcal{A}^1$, we have $p_s\in L_s$. 

\begin{rem}
    If $(p,c)\in\mathcal{A}^1$, we have $\widetilde{U}^1(p,c)>-\infty$. Indeed, for any $(p,c)\in\mathcal{A}^1$, the uniform integrability implies
    \[
        \sup_{t\in[0,T]}\mathbb{E}\Bigl[\exp\Bigl(-\gamma^1\mathcal{W}^{1,(p,c)}_t - \beta^1 (c_t -  X_t^{1,c}) \Bigr)\Bigr] < \infty.
    \]
    We can also see that the family
    \[
        \Bigl\{\int_0^\tau \exp \Bigl(-\gamma^1\mathcal{W}^{1,(p,c)}_s - \beta^1 (c_s -  X_s^{1,c})\Bigr)ds;\tau\in\mathcal{T}^{0,1} \Bigr\}
    \]
    is uniformly integrable.
\end{rem}

\subsection{Optimization}
Based on Hu, Imkeller \& M\"{u}ller \cite{huUtilityMaximizationIncomplete2005a}, we derive a BSDE which characterizes the optimality. To begin with, we consider a family of stochastic processes satisfying the following conditions.  
\begin{dfn} (Condition-R)\\
    \label{condition-R}
    A family of stochastic processes $\Bigl\{R^{1,(p,c)}:=(R^{1,(p,c)}_t)_{t\in[0,T]}; (p,c)\in\mathcal{A}^1\Bigr\}\subset\mathbb{L}^0(\mathbb{F}^{0,1},\mathbb{R})$ is said to satisfy the condition-R if the following properties are met.\\
    \textup{(i)} For all $(p,c)\in\mathcal{A}^1$, $R^{1,(p,c)}$ satisfies
    \begin{equation}
    \begin{split}
        R^{1,(p,c)}_T=& -\exp\Bigl(-\delta T-\gamma^1(\mathcal{W}^{1,(p,c)}_T-F^1_T)\Bigr) \\
        &-a \int_0^T \exp\Bigl(-\delta t-\gamma^1(\mathcal{W}^{1,(p,c)}_t-F^1_t)-\beta^1(c_t-X_t^{1,c})\Bigr)dt,~~\mathbb{P}^{0,1}\text{-}\mathrm{a.s.}
        \end{split}
    \end{equation}
    \textup{(ii)} There exists some $\mathcal{F}^{0,1}_0$-measurable random variable $R^1_0$ such that the equality $R^{1,(p,c)}_0 = R^1_0$ holds $\mathbb{P}^{0,1}$-almost surely for all $(p,c)\in\mathcal{A}^1$.\\
    \textup{(iii)} $R^{1,(p,c)}$ is an $(\mathbb{F}^{0,1},\mathbb{P}^{0,1})$-supermartingale for all $(p,c)\in\mathcal{A}^1$ and there exists some $(p^*,c^*)\in\mathcal{A}^1$ such that $R^{1,(p^*,c^*)}$ is an $(\mathbb{F}^{0,1},\mathbb{P}^{0,1})$-martingale.
\end{dfn}
Once such a family is identified, we have, for all $(p,c)\in\mathcal{A}^1$,
\[
    \widetilde{U}^1(p,c) = \mathbb{E}[R^{1,(p,c)}_T]\leq\mathbb{E}[R^1_0] = \mathbb{E}[R^{1,(p^*,c^*)}_T] =\widetilde{U}^1(p^*,c^*),
\]
which indicates that $(p^*,c^*)$ is an optimal strategy for agent-1. To find an appropriate family of processes $\{R^{1,(p,c)}\}$, we suppose that each $R^{1,(p,c)}$ has the following form: for $t\in[0,T]$,
\begin{equation}
    \begin{split}
        \label{ansatz}
        R^{1,(p,c)}_t 
        = &-\exp\Bigl(-\delta t-\gamma^1(\mathcal{W}^{1,(p,c)}_t-Y^1_t-\zeta^1_tX_t^{1,c})\Bigr)\\
        & -a \int_0^t \exp\Bigl(-\delta s-\gamma^1(\mathcal{W}^{1,(p,c)}_s-F^1_s)-\beta^1(c_s-X_s^{1,c})\Bigr)ds.~~~
    \end{split}
\end{equation}
Here, $\zeta^1$ is an $\mathcal{F}^{1}_0$-measurable and continuously differentiable process with $\zeta^1_T=0$ satisfying the ordinary differential equation (ODE) specified later. $Y^1$ is a solution to the following BSDE whose driver $f^1$ is to be determined:
\begin{equation}
    \begin{split}
       Y^1_t = F^1_T + \int_t^T f^1(s,Y^1_s,Z^{1,0}_s,Z^1_s) ds -\int_t^T Z^{1,0}_s dW^0_s  -\int_t^T Z^1_s dW^1_s ,~~~t\in[0,T].
    \end{split}
\end{equation}

For notational simplicity, we may suppress the superscript ``1" when obvious. By Ito formula,
\begin{equation*}
    \begin{split}
        dR^{(p,c)}_t 
        &=
        -\exp\Bigl(-\delta t-\gamma(\mathcal{W}^{(p,c)}_t-Y_t-\zeta_tX^c_t)\Bigr)\Bigl\{-\delta dt -\gamma d(\mathcal{W}^{(p,c)}_t-Y_t) + \frac{\gamma^2}{2}d\langle\mathcal{W}^{(p,c)}-Y\rangle_t\\
        &~~~~~~+\gamma \dot\zeta_t X^c_tdt +\gamma\zeta_tdX^c_t + a \exp\Bigl(-\gamma(Y_t-F_t+\zeta_tX^c_t)-\beta(c_t-X^c_t)\Bigr)dt \Bigr\}\\
        &=
        -\exp\Bigl(-\delta t-\gamma(\mathcal{W}^{(p,c)}_t-Y_t-\zeta_tX^c_t)\Bigr)\Bigl\{-\delta  -\gamma (p_t\theta_t-c_t+f(t,Y_t,Z^0_t,Z^1_t))\\
        &~~~+ \frac{\gamma^2}{2}(|p_t-Z^0_t|^2+|Z^1_t|^2) +\gamma (\dot\zeta_t-\kappa\zeta_t) X^c_t +\gamma\zeta_tb c_t +\gamma\zeta_t\rho_t(\kappa-b)\\
        &~~~+ a \exp\Bigl(-\gamma(Y_t-F_t+\zeta_tX^c_t)-\beta(c_t-X^c_t)\Bigr)\Bigr\}dt\\
        &~~~+\gamma \exp\Bigl(-\delta t-\gamma(\mathcal{W}^{(p,c)}_t-Y_t-\zeta_tX^c_t)\Bigr)\Bigl((p_t-Z^0_t)dW^0_t-Z^1_t dW^1_t\Bigr),
    \end{split}
\end{equation*}
where $\dot\zeta_t:=\dfrac{d}{dt}\zeta_t$. In order to make $R^{(p,c)}$ a supermartingale for all $(p,c)\in\mathcal{A}^1$, we need
\begin{equation*}
    \begin{split}
    f(t,Y_t,Z^0_t,Z^1_t)
    \leq 
    &-\frac{\delta}{\gamma}  -(p_t\theta_t-c_t) + \frac{\gamma}{2}(|p_t-Z^0_t|^2+|Z^1_t|^2) + (\dot\zeta_t-\kappa\zeta_t) X^c_t\\
    &~~~+\zeta_tb c_t +\zeta_t\rho_t(\kappa-b) + \frac{a}{\gamma} \exp\Bigl(-\gamma(Y_t-F_t+\zeta_tX^c_t)-\beta(c_t-X^c_t)\Bigr).
    \end{split}
\end{equation*}
Moreover, $R^{(p,c)}$ is a true martingale for some $(p^*,c^*)$ only if
\begin{equation*}
    \begin{split}
    f(t,Y_t,Z^0_t,Z^1_t)
    =
    &-\frac{\delta}{\gamma}  -(p^*_t\theta_t-c^*_t) + \frac{\gamma}{2}(|p^*_t-Z^0_t|^2+|Z^1_t|^2) + (\dot\zeta_t-\kappa\zeta_t) X^{c^*}_t\\
    &~~~ +\zeta_tb c^*_t +\zeta_t\rho_t(\kappa-b) + \frac{a}{\gamma} \exp\Bigl(-\gamma(Y_t-F_t+\zeta_tX^{c^*}_t)-\beta(c^*_t-X^{c^*}_t)\Bigr).
    \end{split}
\end{equation*}
Combining these observations, we deduce that
\begin{equation}
    \begin{split}
        \label{driver-f}
    &f(t,Y_t,Z^0_t,Z^1_t)\\
    &=
    -\frac{\delta}{\gamma} + (\dot\zeta_t-\kappa\zeta_t) X^c_t  +\zeta_t\rho_t(\kappa-b) + \inf_{p\in L_t}\Bigl\{ -p\theta_t + \frac{\gamma}{2}(|p-Z^0_t|^2+|Z^1_t|^2)\Bigr\}\\
    &~~~+ \inf_{c\in\mathbb{R}} \Bigl\{ (1+\zeta_tb) c  + \frac{a}{\gamma} \exp\Bigl(-\gamma(Y_t-F_t+\zeta_tX^c_t)-\beta(c-X^c_t)\Bigr)\Bigr\}.
    \end{split}
\end{equation}
Assuming $1+b\zeta_t>0$ for all $t\in[0,T]$ temporarily, the candidate for the optimal strategy reads: for $t\in[0,T]$,
\begin{equation}
    \begin{split}
        \label{optimal}
        p^*_t &= Z^{0\|}_t + \frac{\theta^\top_t}{\gamma},\\
        c^*_t &= X^{c^*}_t + \frac{1}{\beta}\Bigl\{\log\Bigl(\frac{a\beta}{\gamma(1+b\zeta_t)}\Bigr)-\gamma(Y_t-F_t+\zeta_tX^{c^*}_t)\Bigr\},
    \end{split}
  \end{equation}
whose admissibility, namely $(p^*,c^*)\in\mathcal{A}^1$, needs to be verified later. Now we obtain
\begin{equation*}
    \begin{split}
    f(t,Y_t,Z^0_t,Z^1_t)
    =&
    -Z^{0\|}_t\theta_t - \frac{|\theta_t|^2}{2\gamma} + \frac{\gamma}{2}(|Z^{0\perp}_t|^2 + |Z^1_t|^2) - \frac{\delta}{\gamma} + (\kappa-b)\zeta_t\rho_t \\
    &+ \frac{1 + b\zeta_t}{\beta}\Bigl\{1 + \log\Bigl(\frac{a\beta}{\gamma(1 + b\zeta_t)}\Bigr)+\gamma(F_t-Y_t)\Bigr\}\\
    &+X^c_t\Bigl\{\frac{1}{\beta}(1+b\zeta_t)(\beta-\gamma\zeta_t)+(\dot\zeta_t-\kappa\zeta_t)\Bigr\}.
    \end{split}
\end{equation*}

In order to make $R^{(p,c)}$ satisfy (ii) of the condition-R, we need
\[
    \frac{1}{\beta}(1+b\zeta_t)(\beta-\gamma\zeta_t)+(\dot\zeta_t-\kappa\zeta_t)=0
\]
for every $t\in[0,T]$ so that the process $Y$ is independent of $c$. To be specific, it is necessary to solve the following ordinary differential equation of Riccati type:
\begin{equation}
    \begin{split}
        \label{zeta-ODE}
        &\dot\zeta_t = \Bigl(\kappa-b+\frac{\gamma}{\beta}\Bigr)\zeta_t + \frac{\gamma b}{\beta}\zeta_t^2-1, ~~t\in[0,T],\\
        &\zeta_T = 0.
    \end{split}
\end{equation}
This is actually explicitly solvable (See, for example, Carmona \& Delarue \cite{carmonaProbabilisticTheoryMean2018} [Equation (2.50)]) as
\begin{equation}
    \label{zeta}
        \zeta_t = \frac{e^{(\delta^+-\delta^-)(T-t)}-1}{\delta^+-\delta^-e^{(\delta^+-\delta^-)(T-t)}},~~~t\in[0,T],
\end{equation}
where
\[
    \delta^{\pm}:=-A\pm\sqrt{A^2+B},~~~A:=\frac{1}{2}\Bigl(\kappa-b+\frac{\gamma}{\beta}\Bigr),~~~B:=\frac{\gamma b}{\beta}.
\]
Note that $\zeta$ satisfies 
\[
    0\leq \zeta_t\leq \frac{1}{\delta^+} e^{(\delta^+-\delta^-)T}\land \frac{1}{|\delta^-|},
\]
and in particular, $1+b\zeta_t>0$ for all $t\in[0,T]$. 

Consequently, we have derived a BSDE for the optimality:
\begin{equation}
    \begin{split}
        \label{qgBSDE}
        Y^1_t&= F^1_T + \int_t^T f^1(s,Y^1_s,Z^{1,0}_s,Z^1_s)ds - \int_t^T Z^{1,0}_s dW^0_s - \int_t^T Z^1_s dW^1_s ,~~~t\in[0,T]
    \end{split}
  \end{equation}
with
\begin{equation*}
    \begin{split}
        f^1(s,Y^1_s,Z^{1,0}_s,Z^1_s) = -Z^{1,0\|}_s\theta_s - \frac{|\theta_s|^2}{2\gamma^1} + \frac{\gamma^1}{2}(|Z^{1,0\perp}_s|^2 + |Z^1_s|^2) -\frac{\gamma^1(1+b\zeta^1_s)}{\beta^1}Y^1_s + g^1_s,
    \end{split}
  \end{equation*}
  where
  \[
      g^1_s := - \frac{\delta}{\gamma^1} + (\kappa-b)\zeta^1_s\rho_s + \frac{1 + b\zeta^1_s}{\beta^1}\Bigl\{1 + \log\Bigl(\frac{a\beta^1}{\gamma^1(1 + b\zeta^1_s)}\Bigr)+\gamma^1 F^1_s\Bigr\}.
  \]

\subsection{Well-Posedness and Verification}
We now study the well-posedness of \eqref{qgBSDE}. Let us begin with the \textit{a priori} estimation.
\begin{lem}
    \label{a-priori1}
    Let Assumptions \ref{asm1-habit} and \ref{asm2-habit} be in force. If the BSDE \eqref{qgBSDE} has a bounded solution $(Y,Z^{1,0},Z^1)\in\mathbb{S}^\infty(\mathbb{P}^{0,1},\mathbb{F}^{0,1},\mathbb{R})\times\mathbb{H}^2(\mathbb{P}^{0,1},\mathbb{F}^{0,1},\mathbb{R}^{1\times d_0})\times\mathbb{H}^2(\mathbb{P}^{0,1},\mathbb{F}^{0,1},\mathbb{R}^{1\times d})$, then $(Z^{1,0},Z^1)\in\mathbb{H}^2_{\mathrm{BMO}}(\mathbb{P}^{0,1},\mathbb{F}^{0,1},\mathbb{R}^{1\times d_0})\times\mathbb{H}^2_{\mathrm{BMO}}(\mathbb{P}^{0,1},\mathbb{F}^{0,1},\mathbb{R}^{1\times d})$ and such a solution is unique.
\end{lem}
\noindent
\textbf{\textit{Proof}}\\
In the proof, we may omit the superscript ``1'' when it is obvious for notational simplicity. First of all, we have
\begin{equation*}
    \begin{split}
     f(s,Y_s,Z^0_s,Z^1_s) 
     &= 
     -Z^{0\|}_s\theta_s - \frac{|\theta_s|^2}{2\gamma} + \frac{\gamma}{2}(|Z^{0\perp}_s|^2 + |Z^1_s|^2) -\frac{\gamma(1+b\zeta_s)}{\beta}Y_s + g_s \\
     &\leq
     \frac{\gamma}{2}(|Z^{0}_s|^2 + |Z^1_s|^2) + C(\|Y\|_{\mathbb{S}^\infty} + \|g\|_{\mathbb{L}^\infty}).
    \end{split}
\end{equation*}
Then, by Ito formula,
\begin{equation*}
    \begin{split}
     de^{2\gamma Y_t}
     &=
     2\gamma e^{2\gamma Y_t}( dY_t +\gamma d\langle Y \rangle_t)\\
     &=
     2\gamma e^{2\gamma Y_t}\Bigl\{(- f(t,Y_t,Z^0_t,Z^1_t) + \gamma|Z^{0}_t|^2 + \gamma|Z^1_t|^2)dt + Z^0_t dW^0_t + Z^1_t dW^1_t\Bigr\}.
    \end{split}
\end{equation*}
Hence,
\begin{equation*}
    \begin{split}
     &e^{2\gamma Y_T} - e^{2\gamma Y_t} \\
     &=
     2\gamma \int_t^T e^{2\gamma Y_s}\Bigl\{(- f(s,Y_s,Z^0_s,Z^1_s) + \gamma|Z^{0}_s|^2 + \gamma|Z^1_s|^2)ds + Z^0_s dW^0_s + Z^1_s dW^1_s\Bigr\}\\
     &\geq
     2\gamma \int_t^T e^{2\gamma Y_s}\Bigl\{- \frac{\gamma}{2}(|Z^{0}_s|^2 + |Z^1_s|^2) - C(\|Y\|_{\mathbb{S}^\infty} + \|g\|_{\mathbb{L}^\infty}) + \gamma|Z^{0}_s|^2 + \gamma|Z^1_s|^2\Bigr\}ds\\
     &~~~~ +  2\gamma \int_t^T e^{2\gamma Y_s}\Bigl\{Z^0_s dW^0_s + Z^1_s dW^1_s\Bigr\}\\
     &=
     2\gamma \int_t^T e^{2\gamma Y_s}\Bigl\{\frac{\gamma}{2}(|Z^{0}_s|^2 + |Z^1_s|^2) - C(\|Y\|_{\mathbb{S}^\infty} + \|g\|_{\mathbb{L}^\infty}) \Bigr\} ds  + 2\gamma \int_t^T e^{2\gamma Y_s}\Bigl\{Z^0_s dW^0_s + Z^1_s dW^1_s\Bigr\}.
    \end{split}
\end{equation*}
Thus, for any $t\in[0,T]$
\begin{equation*}
    \begin{split}
     \mathbb{E}\Bigl[ \int_t^T (|Z^{0}_s|^2 + |Z^1_s|^2)ds |\mathcal{F}^{0,1}_t\Bigr]
     \leq
     Ce^{4\overline{\gamma} \|Y\|_{\mathbb{S}^\infty}}(1 + \|Y\|_{\mathbb{S}^\infty} + \|g\|_{\mathbb{L}^\infty})
     <\infty.
    \end{split}
\end{equation*}
Clearly, $(Z^0,Z^1)\in\mathbb{H}^2_{\mathrm{BMO}}\times\mathbb{H}^2_{\mathrm{BMO}}$. \par
Next, suppose that there exists two solutions $(Y,Z^0,Z^1)$ and $(\acute Y, \acute Z^0, \acute Z^1)$ both of which are in $\mathbb{S}^\infty\times\mathbb{H}^2_{\mathrm{BMO}}\times\mathbb{H}^2_{\mathrm{BMO}}$. Let us write $\Delta Y = Y-\acute Y,~~~\Delta Z^i = Z^i -\acute Z^i~~~(i=0,1)$. Then we have
\begin{equation*}
    \begin{split}
    &f(s,Y_s,Z^0_s,Z^1_s) - f(s,\acute Y_s,\acute Z^0_s,\acute Z^1_s)\\
    &=
    -\Delta Z^{0\|}_s\theta_s + \frac{\gamma}{2}\Delta Z^{0\perp}_s (Z^{0\perp}_s + \acute Z^{0\perp}_s)^{\top} + \frac{\gamma}{2}\Delta Z^{1}_s (Z^{1}_s + \acute Z^{1}_s)^{\top} - \frac{\gamma(1 + b\zeta_s)}{\beta}\Delta Y_s\\
    &=
    -\Delta Z^{0}_s\theta_s + \frac{\gamma}{2}\Delta Z^{0}_s (Z^{0\perp}_s + \acute Z^{0\perp}_s)^{\top} + \frac{\gamma}{2}\Delta Z^{1}_s (Z^{1}_s + \acute Z^{1}_s)^{\top} - \frac{\gamma(1 + b\zeta_s)}{\beta}\Delta Y_s.
    \end{split}
\end{equation*}
Now, we define a new probability measure $\widetilde{\mathbb{P}}$($\sim \mathbb{P}^{0,1}$) by
\[
    \Bigl.\frac{d\widetilde{\mathbb{P}}}{d\mathbb{P}^{0,1}}\Bigr|_{\mathcal{F}^{0,1}_t} = \mathcal{E}\Bigl(\int_0^\cdot \Bigl\{-\theta_s^\top + \frac{\gamma}{2}(Z^{0\perp}_s + \acute Z^{0\perp}_s) \Bigr\}dW^0_s + \int_0^\cdot \frac{\gamma}{2}(Z^{1}_s + \acute Z^{1}_s)dW^1_s \Bigr)_t,~~t\in[0,T].
\]
By Kazamaki \cite{kazamaki_sufficient_1979} and Kazamaki \cite{kazamakiContinuousExponentialMartingales1994} [Remark 3.1], the right hand side is a martingale of class $\mathcal{D}$ and hence the new probability measure $\widetilde{\mathbb{P}}$ is well-defined. Then, the Girsanov's theorem implies that the processes
\begin{equation*}
    \begin{split}
        \widetilde W^0_t := W^0_t + \int_0^t \{\theta_s - \frac{\gamma}{2}(Z^{0\perp}_s + \acute Z^{0\perp}_s)^\top\} ds,~~~\widetilde W^1_t := W^1_t  - \int_0^t \frac{\gamma}{2}(Z^{1}_s + \acute Z^{1}_s)^\top ds,~~~t\in[0,T]
    \end{split}
\end{equation*}
are the standard $(\mathbb{F}^{0,1},\widetilde{\mathbb{P}})$-Brownian motions. Now we have:
\begin{equation}
    \begin{split}
        \label{tilde_W_BSDE}
     \Delta Y_t 
     =&
     \int_t^T \Bigl\{-\Delta Z^{0}_s\theta_s + \frac{\gamma}{2}\Delta Z^{0}_s (Z^{0\perp}_s + \acute Z^{0\perp}_s)^{\top} + \frac{\gamma}{2}\Delta Z^{1}_s (Z^{1}_s + \acute Z^{1}_s)^{\top} - \frac{\gamma(1 + b\zeta_s)}{\beta}\Delta Y_s\Bigr\}ds \\
     &- \int_t^T \Delta Z^0_s dW^0_s - \int_t^T \Delta Z^1_s dW^1_s\\
     =&
     -\int_t^T  \frac{\gamma(1 + b\zeta_s)}{\beta}\Delta Y_s ds - \int_t^T \Delta Z^0_s d\widetilde W^0_s - \int_t^T \Delta Z^1_s d\widetilde W^1_s,~~t\in[0,T].
    \end{split}
\end{equation}
Then, it follows that $\Delta Y=0,\Delta Z^0=0$ and $\Delta Z^1 = 0$ for $\widetilde{\mathbb{P}}$ (and thus $\mathbb{P}^{0,1}$)-almost surely since they obviously satisfy \eqref{tilde_W_BSDE} and the solution of \eqref{tilde_W_BSDE} is unique due to the standard result for Lipschitz BSDEs (See, e.g. Zhang \cite{ZhangBSDE} [Chapter 4]). $\square$\par

For the risk neutral measure $\mathbb{Q}(\sim\mathbb{P}^{0,1})$ defined by \eqref{risk-neutral-habit}, the Girsanov's theorem implies that the processes
\begin{equation}
    \begin{split}
    W^{0,\mathbb{Q}}_t := W^{0,\mathbb{P}}_t + \int_0^t \theta_s ds,~~~W^{1,\mathbb{Q}}_t := W^{1,\mathbb{P}}_t,~~t\in[0,T]
\end{split}
\end{equation}
form the standard $(\mathbb{F}^0,\mathbb{Q})$ and $(\mathbb{F}^1,\mathbb{Q})$-Brownian motions, respectively. Under this measure, the BSDE \eqref{qgBSDE} becomes
\begin{equation}
    \begin{split}
        \label{qgBSDE-Q}
        Y^1_t=& F^1_T + \int_t^T \Bigl\{- \frac{|\theta_s|^2}{2\gamma^1} + \frac{\gamma^1}{2}(|Z^{1,0\perp}_s|^2 + |Z^1_s|^2) -\frac{\gamma^1(1+b\zeta^1_s)}{\beta^1}Y^1_s + g^{1}_s\Bigr\}ds\\ &- \int_t^T Z^{1,0}_s dW^{0,\mathbb{Q}}_s - \int_t^T Z^1_s dW^{1,\mathbb{Q}}_s
    \end{split}
\end{equation}
for $t\in[0,T]$. Moreover, by Kazamaki \cite{kazamakiContinuousExponentialMartingales1994} [Theorem 3.3], we have $\theta\in \mathbb{H}^2_{\mathrm{BMO}}(\mathbb{Q},\mathbb{F}^{0})$. Since $\theta$ is unbounded in general, the standard technique cannot be applied directly to prove the well-posedness of the equation \eqref{qgBSDE-Q}. We adopt the same regularization used in Fujii \& Sekine \cite{fujiiMeanFieldEquilibriumPrice2023a}.

\begin{thm}
    \label{sec2-well-posed}
    Let Assumptions \ref{asm1-habit} and \ref{asm2-habit} be in force. Then, the BSDE \eqref{qgBSDE} has a unique solution $(Y,Z^{1,0},Z^1)\in\mathbb{S}^\infty(\mathbb{P}^{0,1},\mathbb{F}^{0,1},\mathbb{R})\times\mathbb{H}^2_{\mathrm{BMO}}(\mathbb{P}^{0,1},\mathbb{F}^{0,1},\mathbb{R}^{1\times d_0})\times\mathbb{H}^2_{\mathrm{BMO}}(\mathbb{P}^{0,1},\mathbb{F}^{0,1},\mathbb{R}^{1\times d})$.
\end{thm}
\noindent
\textbf{\textit{Proof}}\\
Obviously, $W^{0,\mathbb{Q}}, W^{1,\mathbb{Q}}$ are adapted to $\mathbb{F}^{0,1}$, but they do not necessarily generate $\mathbb{F}^{0,1}$. However, due to the equivalence of $\mathbb{Q}$ and $\mathbb{P}^{0,1}$, Jeanblanc, Yor \& Chesney \cite{jeanblanc_mathematical_2009} [Proposition 1.7.7.1] shows that every $(\mathbb{F}^{0,1},\mathbb{Q})$-local martingale has a representation through a stochastic integral with respect to $(W^{0,\mathbb{Q}},W^{1,\mathbb{Q}})$. 
This fact allows us to use the standard approach for BSDEs to deal with the equation \eqref{qgBSDE-Q}. 
In addition, if there exists a bounded solution $(Y,Z^0,Z^1)\in\mathbb{S}^\infty(\mathbb{Q},\mathbb{F}^{0,1})\times\mathbb{H}^2_{\mathrm{BMO}}(\mathbb{Q},\mathbb{F}^{0,1})\times\mathbb{H}^2_{\mathrm{BMO}}(\mathbb{Q},\mathbb{F}^{0,1})$ to the equation \eqref{qgBSDE-Q}, it obviously solves the BSDE \eqref{qgBSDE} under the original measure $\mathbb{P}^{0,1}$. The uniqueness follows from Lemma \ref{a-priori1}.
Thus, it suffices to find a bounded solution of the BSDE \eqref{qgBSDE-Q}. \par
For the remainder of the proof, we may omit the superscript ``1'' if obvious. We consider the next truncated BSDE:
\begin{equation}
    \begin{split}
        \label{qgBSDE-Q-trunc}
        Y^n_t =& F_T + \int_t^T \Bigl\{- \frac{|\theta_s|^2 \land n}{2\gamma} + \frac{\gamma}{2}(|Z^{n,0\perp}_s|^2 + |Z^{n,1}_s|^2) -\frac{\gamma(1+b\zeta_s)}{\beta}Y^n_s + g_s\Bigr\}ds \\&- \int_t^T Z^{n,0}_s dW^{0,\mathbb{Q}}_s - \int_t^T Z^{n,1}_s dW^{1,\mathbb{Q}}_s~~~~~
    \end{split}
\end{equation}
for $t\in[0,T]$. By the standard result of Kobylanski \cite{Kobylanski2000BackwardSD}, we deduce that the truncated BSDE \eqref{qgBSDE-Q-trunc} has a unique solution $(Y^n,Z^{n,0},Z^{n,1})\in\mathbb{S}^\infty\times\mathbb{H}^2_{\mathrm{BMO}}\times\mathbb{H}^2_{\mathrm{BMO}}$ for all $n\in\mathbb{N}$. In addition, the comparison principle presented in the same work shows that $Y^{n+1}\leq Y^n$ holds for all $n\in\mathbb{N}$.
In particular, this principle gives an estimate $\sup_{n\in\mathbb{N}}\|Y^n\|_{\mathbb{S}^\infty} <\infty$ by considering the following two BSDEs. For $t\in[0,T]$,
\begin{equation}
    \begin{split}
    \overline{Y}_t &= \|F\|_{\mathbb{L}^\infty} +\int_t^T \Bigl\{\frac{\overline{\gamma}}{2}(|\overline{Z}^{0\perp}_s|^2 + |\overline{Z}^1_s|^2)+ \frac{\overline{\gamma}(1+b\|\zeta\|_{\mathbb{L}^\infty})}{\underline{\beta}}|\overline{Y}_s| + \|g\|_{\mathbb{L}^\infty}\Bigr\} ds - \int_t^T \overline{Z}^0_s dW^{0,\mathbb{Q}}_s - \int_t^T \overline{Z}^1_s dW^{1,\mathbb{Q}}_s, \\
    \underline{Y}_t &= -\|F\|_{\mathbb{L}^\infty} -\int_t^T \Bigl\{\frac{|\theta_s|^2}{2\underline{\gamma}}+ \frac{\overline{\gamma}(1+b\|\zeta\|_{\mathbb{L}^\infty})}{\underline{\beta}}|\underline{Y}_s| + \|g\|_{\mathbb{L}^\infty}\Bigr\} ds - \int_t^T \underline{Z}^0_s dW^{0,\mathbb{Q}}_s - \int_t^T \underline{Z}^1_s dW^{1,\mathbb{Q}}_s.
\end{split}
\end{equation}
Then, $\underline{Y}_t\leq Y^n_t\leq \overline{Y}_t, ~~\mathbb{Q}$-a.s. for all $t\in[0,T]$ by the comparison principle, and it is also easy to see $\overline{Z}^0=0$ and $\overline{Z}^1=0$. The backward Gronwall's inequality (See, for example, Pardoux \& R{\u {a}}{\c s}canu \cite{pardouxStochasticDifferentialEquations2014} [Corollary 6.61]) yields $\overline{Y}_t \leq C(\|F\|_{\mathbb{L}^\infty}+\|g\|_{\mathbb{L}^\infty})$ for all $t\in[0,T]$. 
For $\underline{Y}$, it is obvious that $\underline{Y}\leq 0~~\mathbb{Q}$-a.s. and thus $|\underline{Y}_t|=-\underline{Y}_t$. Then, it is straightforward to see
\begin{equation*}
    \begin{split}
    \underline{Y}_t 
    =& 
    -\exp\Bigl(\frac{\overline{\gamma}(1+b\|\zeta\|_{\mathbb{L}^\infty})}{\underline{\beta}}(T-t)\Bigr)\|F\|_{\mathbb{L}^\infty}\\
    & - \mathbb{E}\Bigl[\int_t^T \exp\Bigl(\frac{\overline{\gamma}(1+b\|\zeta\|_{\mathbb{L}^\infty})}{\underline{\beta}}(s-t)\Bigr)\Bigl(\frac{|\theta_s|^2}{2\underline{\gamma}}+\|g\|_{\mathbb{L}^\infty}\Bigr)ds|\mathcal{F}^{0,1}_t\Bigr]\\
    \geq&
    -C(\|F\|_{\mathbb{L}^\infty}+\|g\|_{\mathbb{L}^\infty}+\|\theta\|^2_{\mathbb{H}^2_{\mathrm{BMO}}}).
    \end{split}
\end{equation*}
Therefore, $(Y^n)_{n\in\mathbb{N}}\subset\mathbb{S}^\infty$ is a bounded and monotonically decreasing sequence. 

We then define a bounded process $Y$ by $Y_t(\omega) := \lim_{n\to\infty} Y^n_t(\omega)$ for almost all $(t,\omega)\in[0,T] \times\Omega$. (for $(t,\omega)$ in $dt\otimes \mathbb{Q}$-negligible sets, we may put $Y_t(\omega) = 0$.) In addition, by following the same argument as in Lemma \ref{a-priori1}, we deduce that
\[
    \sup_{n\in\mathbb{N}}\|(Z^{n,0},Z^{n,1})\|^2_{\mathbb{H}^2_{\mathrm{BMO}}}\leq C(1 + \|F\|_{\mathbb{L}^\infty}+\|\theta\|^2_{\mathbb{H}^2_{\mathrm{BMO}}}+\|g\|_{\mathbb{L}^\infty})\exp(C(\|F\|_{\mathbb{L}^\infty}+\|\theta\|^2_{\mathbb{H}^2_{\mathrm{BMO}}}+\|g\|_{\mathbb{L}^\infty}))<\infty,
\]
which means $(Z^{n,0},Z^{n,1})_{n\in\mathbb{N}}$ is weakly relatively compact in $\mathbb{H}^2$. Choosing a subsequence if necessary, there exists $(Z^{0},Z^{1})\in\mathbb{H}^2\times\mathbb{H}^2$ such that
\[
    Z^{n,0}\rightharpoonup Z^0,~~~~~Z^{n,1}\rightharpoonup Z^1~~~(n\to\infty)
\]
in the sense of weak convergence in $\mathbb{H}^2$. Finally, we shall prove that $(Y,Z^0,Z^1)\in\mathbb{S}^\infty\times\mathbb{H}^2_{\mathrm{BMO}}\times\mathbb{H}^2_{\mathrm{BMO}}$ and that it actually solves the BSDE \eqref{qgBSDE}. 
We introduce a smooth convex function $\phi:\mathbb{R}\to\mathbb{R}_+$ satisfying $\phi(0)=\phi'(0)=0$, whose concrete form is to be determined. Let us denote $\Delta Y^{n,m}:=Y^n-Y^m$ and $\Delta Z^{n,m,i}:=Z^{n,i}-Z^{m,i}$ for $i=0,1$. Notice that $\Delta Y^{n,m}\geq 0 $ when $m\geq n$. 

Using Ito formula, we have
\begin{equation*}
    \begin{split}
        &\mathbb{E}^{\mathbb{Q}}[\phi(\Delta Y^{n,m}_0)] + \frac{1}{2}\mathbb{E}^{\mathbb{Q}}\Bigl[\int_0^T \phi''(\Delta Y^{n,m}_s)(|\Delta Z^{n,m,0}_s|^2 + |\Delta Z^{n,m,1}_s|^2) ds  \Bigr] \\
        &=
        \mathbb{E}^{\mathbb{Q}}\Bigl[\int_0^T \phi'(\Delta Y^{n,m}_s)\Bigl\{- \frac{|\theta_s|^2\land n}{2\gamma} + \frac{|\theta_s|^2\land m}{2\gamma} + \frac{\gamma}{2}(|Z^{n,0\perp}_s|^2 + |Z^{n,1}_s|^2)\\
        &~~~~ - \frac{\gamma}{2}(|Z^{m,0\perp}_s|^2 + |Z^{m,1}_s|^2) -\frac{\gamma(1+b\zeta_s)}{\beta}\Delta Y^{n,m}_s \Bigr\} ds  \Bigr]\\
        &\leq
        \mathbb{E}^{\mathbb{Q}}\Bigl[\int_0^T \phi'(\Delta Y^{n,m}_s)\Bigl\{\frac{|\theta_s|^2}{2\gamma} + \frac{\gamma}{2}(|Z^{n,0}_s|^2 + |Z^{n,1}_s|^2)\Bigr\} ds  \Bigr]\\
        &\leq
        \mathbb{E}^{\mathbb{Q}}\Bigl[\int_0^T C_0\phi'(\Delta Y^{n,m}_s)\Bigl\{|\theta_s|^2 + | Z^{n,0}_s - Z^{0}_s|^2 + |Z^{n,1}_s-Z^{1}_s|^2 + |Z^{0}_s|^2 + |Z^{1}_s|^2\Bigr\} ds  \Bigr],
    \end{split}
  \end{equation*}
where $C_0$ is a positive constant satisfying $C_0\geq \dfrac{1}{2}(\underline{\gamma}^{-1}+2\overline{\gamma})$. Then,
  \begin{equation*}
    \begin{split}
        &\frac{1}{2}\mathbb{E}^{\mathbb{Q}}\Bigl[\int_0^T \phi''(\Delta Y^{n,m}_s)(|\Delta Z^{n,m,0}_s|^2 + |\Delta Z^{n,m,1}_s|^2) ds  \Bigr]\\
        &\leq
        \mathbb{E}^{\mathbb{Q}}\Bigl[\int_0^T C_0\phi'(\Delta Y^{n,m}_s)\Bigl\{|\theta_s|^2 + | Z^{n,0}_s - Z^{0}_s|^2 + |Z^{n,1}_s-Z^{1}_s|^2 + |Z^{0}_s|^2 + |Z^{1}_s|^2\Bigr\} ds  \Bigr].
    \end{split}
  \end{equation*}
Now set $\phi$ as
  \begin{equation}
    \begin{split}
      \label{phi}
      \phi(y):=\frac{1}{2C_0^2}(e^{2C_0 y}-2C_0 y-1),
        \end{split}
  \end{equation}
  then $\phi(0)=\phi'(0)=0$ and
  \begin{equation*}
    \begin{split}
      \phi'(y)=\frac{1}{C_0}(e^{2C_0 y}-1),~~~\phi''(y)=2e^{2C_0 y}.
      \end{split}
  \end{equation*}
In particular, $\phi''(y)=2C_0\phi'(y) + 2$. With these relations, we get
  \begin{equation*}
    \begin{split}
        &\mathbb{E}^{\mathbb{Q}}\Bigl[\int_0^T (C_0\phi'(\Delta Y^{n,m}_s) + 1)(|\Delta Z^{n,m,0}_s|^2+|\Delta Z^{n,m,1}_s|^2) ds  \Bigr]\\
        &\leq
        \mathbb{E}^{\mathbb{Q}}\Bigl[\int_0^T C_0\phi'(\Delta Y^{n,m}_s)\Bigl\{|\theta_s|^2 + | Z^{n,0}_s - Z^{0}_s|^2 + |Z^{n,1}_s-Z^{1}_s|^2 + |Z^{0}_s|^2 + |Z^{1}_s|^2\Bigr\}  ds  \Bigr].
    \end{split}
\end{equation*}
Since $\sqrt{\phi'(\Delta Y^{n,m})+1}\Delta Z^{n,m,i}$ ($i=0,1$) is weakly convergent to $\sqrt{\phi'(Y^{n}-Y)+1}\Delta Z^{n,i}_s$ in $\mathbb{H}^2$ as $m\to\infty$, we obtain
  \begin{equation*}
    \begin{split}
        &\mathbb{E}^{\mathbb{Q}}\Bigl[\int_0^T (C_0\phi'( Y^{n}_s - Y_s) + 1)(|Z^{n,0}_s-Z^0_s|^2 + |Z^{n,1}_s-Z^1_s|^2) ds  \Bigr] \\
        &\leq
        \liminf_{m\to\infty}\mathbb{E}^{\mathbb{Q}}\Bigl[\int_0^T (C_0\phi'(\Delta Y^{n,m}_s) + 1) (|\Delta Z^{n,m,0}_s|^2 +|\Delta Z^{n,m,1}_s|^2) ds  \Bigr]\\
        &\leq
        \mathbb{E}^{\mathbb{Q}}\Bigl[\int_0^T C_0\phi'(Y^{n}_s - Y_s)\Bigl\{|\theta_s|^2 + | Z^{n,0}_s - Z^{0}_s|^2 + |Z^{n,1}_s-Z^{1}_s|^2 + |Z^{0}_s|^2 + |Z^{1}_s|^2\Bigr\}  ds  \Bigr].
    \end{split}
  \end{equation*}
This implies
\begin{equation*}
    \begin{split}
        &\mathbb{E}^{\mathbb{Q}}\Bigl[\int_0^T | Z^{n,0}_s-Z^0_s|^2 + | Z^{n,1}_s-Z^1_s|^2 ds  \Bigr]\\
        &\leq
        \mathbb{E}^{\mathbb{Q}}\Bigl[\int_0^T C_0\phi'(Y^{n}_s - Y_s)\Bigl\{|\theta_s|^2 + |Z^{0}_s|^2 + |Z^{1}_s|^2\Bigr\}  ds  \Bigr]\\
        &\to
        0~~~(n\to\infty)
    \end{split}
\end{equation*}
by the dominated convergence theorem. Thus,
\[
    Z^{n,0}\to Z^0,~~~~~Z^{n,1}\to Z^1,~~~~(n\to\infty)
\]
strongly in $\mathbb{H}^2$. Taking a subsequence if necessary, it is now straightforward to see, for $\mathbb{Q}$-a.s.,
\begin{equation*}
    \begin{split}
        \sup_{t\in[0,T]}|Y_t-Y^n_t| \to 0,~~~\sup_{t\in[0,T]}\Bigl|\int_0^t (Z^0_s-Z^{n,0}_s)dW^{0,\mathbb{Q}}_s\Bigr| + \sup_{t\in[0,T]}\Bigl|\int_0^t (Z^1_s-Z^{n,1}_s)dW^{1,\mathbb{Q}}_s\Bigr| \to 0
    \end{split}
\end{equation*}
as $n\to\infty$. Hence $(Y,Z^0,Z^1)\in\mathbb{S}^\infty\times\mathbb{H}^2_{\mathrm{BMO}}\times\mathbb{H}^2_{\mathrm{BMO}}$ is a bounded solution to the BSDE \eqref{qgBSDE-Q}. $\square$ \par

We now verify the admissibility of \eqref{optimal} and the condition-R.
\begin{thm} (Verification)\\
    \label{verification}
    Let Assumptions \ref{asm1-habit} and \ref{asm2-habit} be in force. Moreover, let $(Y,Z^{1,0},Z^1)\in\mathbb{S}^\infty(\mathbb{P}^{0,1},\mathbb{F}^{0,1},\mathbb{R})\times\mathbb{H}^2_{\mathrm{BMO}}(\mathbb{P}^{0,1},\mathbb{F}^{0,1},\mathbb{R}^{1\times d_0})\times\mathbb{H}^2_{\mathrm{BMO}}(\mathbb{P}^{0,1},\mathbb{F}^{0,1},\mathbb{R}^{1\times d})$ be the solution to the BSDE \eqref{qgBSDE}. Then, the process $(p^{1,*},c^{1,*})$ defined by \eqref{optimal}, that is,
    \begin{equation*}
        \begin{split}
        p^{1,*}_t &:= Z^{1,0\|}_t + \frac{\theta^\top_t}{\gamma^1},~~~t\in[0,T],\\
        c^{1,*}_t &:= X^{1,c^{1,*}}_t + \frac{1}{\beta^1}\Bigl\{\log\Bigl(\frac{a\beta^1}{\gamma^1(1+b\zeta^1_t)}\Bigr)-\gamma^1(Y^1_t-F^1_t+\zeta^1_tX^{1,c^{1,*}}_t)\Bigr\},~~~t\in[0,T]
        \end{split}
    \end{equation*}
    is a unique optimal strategy for agent-1.
\end{thm}
\noindent
\textbf{\textit{Proof}}\\
As usual, we omit the superscript ``1'' if there is no risk of confusion. We first show $(p^*,c^*)$ is admissible. It is straightforward to see that $c^*$ is bounded by using the Gronwall's inequality:
\begin{equation*}
    \begin{split}
    |c^*_t|
    \leq
    C(1+|X^{c^*}_t| )
    \leq
    C + C \int_0^t |c^*_s|ds
    \end{split}
\end{equation*}
and thus $\sup_{t\in[0,T]}|c^*_t| < \infty$. This also implies $X^{c^*}\in\mathbb{S}^\infty$. Thus, it suffices to show the uniform integrability of the family
\[
    \Bigl\{\exp\Bigl(-\gamma\mathcal{W}^{(p^*,c^*)}_\tau\Bigr); \tau\in\mathcal{T}^{0,1}\Bigr\}.
\]
Let us introduce a process $\psi$ by
\[
    \psi_t:=\exp\Bigl(-\delta t -\gamma(\mathcal{W}^{(p^*,c^*)}_t - Y_t - \zeta_t X^{c^*}_t)\Bigr),~~ t\in[0,T].
\]
By the definition of the process $R^{(p,c)}$, we have
\begin{equation*}
    \begin{split}
        R^{(p^*,c^*)}_t = -\psi_t -a\int_0^t \exp\Bigl(-\gamma(Y_s-F_s)-\gamma\zeta_s X^{c^*}_s -\beta(c^*_s-X^{c^*}_s)\Bigr)\psi_sds,
    \end{split}
\end{equation*}
then it holds that 
\begin{equation*}
    \begin{split}
        d\psi_t = -dR^{(p^*,c^*)}_t -a \exp\Bigl(-\gamma(Y_t-F_t)-\gamma\zeta_t X^{c^*}_t -\beta(c^*_t-X^{c^*}_t)\Bigr)\psi_t dt.
    \end{split}
\end{equation*}
Recalling how we have chosen $(p^*,c^*)$, we have
\begin{equation*}
    \begin{split}
    dR^{(p^*,c^*)}_t 
    &= 
    \gamma \exp\Bigl(-\delta t-\gamma(\mathcal{W}^{(p^*,c^*)}_t-Y_t-\zeta_tX^{c^*}_t)\Bigr)\Bigl((p^*_t-Z^0_t)dW^0_t-Z^1_t dW^1_t\Bigr) \\
    &= 
    \psi_t\Bigl\{\Bigl(\theta_t^\top-\gamma Z^{0,\perp}_t\Bigr)dW^0_t-\gamma Z^1_t dW^1_t\Bigr\}.
\end{split}
\end{equation*}
From these observations, we obtain
\[
    d\psi_t = -a \exp\Bigl(-\gamma(Y_t-F_t)-\gamma\zeta_t X^{c^*}_t -\beta(c^*_t-X^{c^*}_t)\Bigr)\psi_t dt - \psi_t\Bigl\{\Bigl(\theta_t^\top-\gamma Z^{0,\perp}_t\Bigr)dW^0_t-\gamma Z^1_t dW^1_t\Bigr\},
\]
and thus
\begin{equation*}
    \begin{split}
    &\psi_t = \exp\Bigl(-\gamma(\xi-Y_0-\zeta_0 X_0)-a\int_0^t\exp(-\gamma(Y_s-F_s)-\gamma\zeta_s X^{c^*}_s -\beta(c^*_s-X_s))ds \Bigr)\\
    &~~~~~~~~~~~~~~~~~~~~~~~\times\mathcal{E}\Bigl(-\int_0^\cdot \Bigl(\theta_s-\gamma Z^{0,\perp}_s\Bigr)dW^0_s + \int_0^\cdot \gamma Z^1_s dW^1_s\Bigr)_t.
    \end{split}
\end{equation*}
Since $\theta, Z^0, Z^1\in\mathbb{H}^2_{\mathrm{BMO}}$ and $\xi, Y, \zeta, F, X^{c^*}$ and $c^*$ are all bounded, we deduce that $\{\psi_\tau;\tau\in\mathcal{T}^{0,1}\}$ is uniformly integrable. Therefore, given the boundedness of $Y$ and $X^{c^*}$, so is the family $\Bigl\{\exp\Bigl(-\gamma\mathcal{W}^{(p^*,c^*)}_\tau\Bigr); \tau\in\mathcal{T}^{0,1}\Bigr\}$. Hence $(p^*,c^*)\in\mathcal{A}^1$.\par
Now we check that the family $\{R^{(p,c)} ; (p,c)\in\mathcal{A}^1\}$ defined by \eqref{ansatz} satisfies the condition-R. The first condition is obviously satisfied. Also, for all $(p,c)\in\mathcal{A}^1$, we have $R^{(p,c)}_0 =  -\exp(-\gamma(\xi-Y_0-\zeta_0 X_0))$, which is $\mathcal{F}^{0,1}_0$-measurable and clearly independent of $(p,c)$. Thus, condition (ii) is fulfilled. Now we move on to (iii). For any $(p,c)\in\mathcal{A}^1$, the family $\{R^{(p,c)}_\tau ; \tau\in\mathcal{T}^{0,1}\}$ is uniformly integrable due to the definition of the set $\mathcal{A}^1$, the boundedness of $Y$ and $|\gamma\zeta_t|\leq K$. Recalling how we have chosen the driver $f$, the process $R^{(p,c)}$ has a nonpositive drift for all $(p,c)\in\mathcal{A}^1$.
Indeed, from \eqref{driver-f} the drift term of $R^{(p,c)}$ reads, for all $(p,c)\in\mathcal{A}^1$,
\begin{equation*}
    \begin{split}
        &-\exp\Bigl(-\delta t-\gamma(\mathcal{W}^{(p,c)}_t-Y_t-\zeta_tX^c_t)\Bigr)\Bigl\{-\delta  -\gamma (p_t\theta_t-c_t+f(t,Y_t,Z^0_t,Z^1_t)) + \frac{\gamma^2}{2}(|p_t-Z^0_t|^2+|Z^1_t|^2)\Bigr.\\
        &~~~~~~\Bigl. +\gamma (\dot\zeta_t-\kappa\zeta_t) X^c_t +\gamma\zeta_tb c_t +\gamma\zeta_t\rho_t(\kappa-b) + a \exp\Bigl(-\gamma(Y_t-F_t+\zeta_tX^c_t)-\beta(c_t-X^c_t)\Bigr)\Bigr\}\\
        &=
        -\gamma \exp\Bigl(-\delta t-\gamma(\mathcal{W}^{(p,c)}_t-Y_t-\zeta_tX^c_t)\Bigr)\Bigl[\Bigl\{- p_t\theta_t+ \frac{\gamma}{2}(|p_t-Z^0_t|^2+|Z^1_t|^2)\Bigr\}-\Bigl\{- p^*_t\theta_t \\
        &~~~~~~+ \frac{\gamma}{2}(|p^*_t-Z^0_t|^2+|Z^1_t|^2)\Bigr\} + \Bigl\{(1+\zeta_tb) c_t  + \frac{a}{\gamma} \exp\Bigl(-\gamma(Y_t-F_t+\zeta_tX^c_t)-\beta(c_t-X^c_t)\Bigr)\Bigr\}\\
        &~~~~~~- \Bigl\{(1+\zeta_tb) c^*_t  + \frac{a}{\gamma} \exp\Bigl(-\gamma(Y_t-F_t+\zeta_tX^{c^*}_t)-\beta(c^*_t-X^{c^*}_t)\Bigr)\Bigr\} + (\dot\zeta_t-\kappa\zeta_t) (X^c_t - X^{c^*}_t)\Bigr]\\
        &\leq
        -\gamma \exp\Bigl(-\delta t-\gamma(\mathcal{W}^{(p,c)}_t-Y_t-\zeta_tX^c_t)\Bigr)\Bigl[\inf_{\varrho\in\mathbb{R}}\Bigl\{(1+\zeta_tb) \varrho  + \frac{a}{\gamma}\exp\Bigl(-\gamma(Y_t-F_t+\zeta_tX^c_t)-\beta(\varrho-X^c_t)\Bigr)\Bigr\} \\
        &~~~~~~- \Bigl\{(1+\zeta_tb) c^*_t  + \frac{a}{\gamma} \exp\Bigl(-\gamma(Y_t-F_t+\zeta_tX^{c^*}_t)-\beta(c^*_t-X^{c^*}_t)\Bigr)\Bigr\}+ (\dot\zeta_t-\kappa\zeta_t) (X^c_t - X^{c^*}_t)\Bigr]\\
        &=
        -\gamma \exp\Bigl(-\delta t-\gamma(\mathcal{W}^{(p,c)}_t-Y_t-\zeta_tX^c_t)\Bigr)\Bigl\{\frac{1}{\beta}(1+\zeta_tb)(\beta-\gamma\zeta_t)+ (\dot\zeta_t-\kappa\zeta_t)\Bigr\} (X^c_t - X^{c^*}_t)\\
        &=
        0.
    \end{split}
\end{equation*}
Here, we have used the equalities
\begin{equation*}
    \begin{split}
        &\inf_{\varrho\in\mathbb{R}}\Bigl\{(1+\zeta_tb) \varrho  + \frac{a}{\gamma}\exp\Bigl(-\gamma(Y_t-F_t+\zeta_tX^c_t)-\beta(\varrho-X^c_t)\Bigr)\Bigr\}\\
        &~~= \frac{1+\zeta_tb}{\beta}\Bigl\{(\beta-\gamma\zeta_t)X_t^c+1+\log\Bigl(\frac{a\beta}{\gamma(1+\zeta_tb)}\Bigr)+\gamma(F_t-Y_t)\Bigr\},\\
        &(1+\zeta_tb) c^*_t  + \frac{a}{\gamma} \exp\Bigl(-\gamma(Y_t-F_t+\zeta_tX^{c^*}_t)-\beta(c^*_t-X^{c^*}_t)\Bigr)\\
        &~~ = \frac{1+\zeta_tb}{\beta}\Bigl\{(\beta-\gamma\zeta_t)X_t^{c^*}+1+\log\Bigl(\frac{a\beta}{\gamma(1+\zeta_tb)}\Bigr)+\gamma(F_t-Y_t)\Bigr\},
    \end{split}
\end{equation*}
and the ODE \eqref{zeta-ODE} in the last equality. Together with the uniform integrability, the supermartingale property is now clear. 
In addition, since the process $R^{(p^*,c^*)}$ is uniformly integrable and follows
\[
    dR^{(p^*,c^*)}_t = \gamma \exp\Bigl(-\delta t-\gamma(\mathcal{W}^{(p^*,c^*)}_t-Y_t-\zeta_tX^{c^*}_t)\Bigr)\Bigl((p^*_t-Z^0_t)dW^0_t-Z^1_t dW^1_t\Bigr),
\]
it is a true martingale. Finally, the strict convexity of $p\mapsto  -p\theta_t + \frac{\gamma}{2}(|p-Z^0_t|^2+|Z^1_t|^2)$ and $c\mapsto (1+\zeta_tb) c  + \frac{a}{\gamma} e^{-\gamma(Y_t-F_t+\zeta_tX_t)-\beta(c-X_t)}$ shows that such $(p^*,c^*)$ is unique. Consequently, $R^{(p,c)}$ is a true martingale if and only if $(p,c)=(p^*,c^*)$ thereby satisfying (iii). $\square$

\section{Mean Field Equilibrium Model Under the Market Clearing Condition} \label{Section 3}
Based on the results of the previous section, we construct a financial market with multiple agents. 
With the help of the mean field game theory, we are going to determine the risk-premium process $\theta$ endogenously so that the resultant stock prices satisfy the market-clearing condition, i.e., buy and sell orders among the agents are always balanced for the period $[0,T]$.\par
This section first provides a heuristic derivation of a mean field BSDE that characterizes the financial market in a state of equilibrium and then proves its well-posedness under certain conditions. Finally, we verify that the solution of the mean field BSDE does indeed provide the risk-premium process in the large population limit.

\subsection{Multi-agent problem and the relevant BSDE}
Suppose there are $N\in\mathbb{N}$ agents in the common financial market. In order to study the equilibrium state, let us first introduce the relevant probability spaces as in Section \ref{Section 2}.\\

\noindent
(1) We denote by $(\Omega^0,\mathcal{F}^0,\mathbb{P}^0)$ a complete probability space with complete and right-continuous filtration $\mathbb{F}^0:=(\mathcal{F}^0_t)_{t\in[0,T]}$ generated by a $d_0$-dimensional standard Brownian motion $W^0:=(W^0_t)_{t\in[0,T]}$ with $\mathcal{F}^0:=\mathcal{F}^0_T$. The space $(\Omega^0,\mathcal{F}^0,\mathbb{P}^0)$ is used to describe the randomness of the financial market and the market-wide information common to all agents. 
Moreover, we denote by $(\Omega^i,\mathcal{F}^i,\mathbb{P}^i)$ ($i=1,\ldots,N$) a complete probability space with complete and right-continuous filtration $\mathbb{F}^i:=(\mathcal{F}^i_t)_{t\in[0,T]}$, generated by a $d$-dimensional standard Brownian motion $W^i:=(W^i_t)_{t\in[0,T]}$ and a $\sigma$-algebra $\sigma(\xi^i,\gamma^i,\beta^i,X^i_0,F^i_0)$, where the completion of the latter gives $\mathcal{F}^i_0$. We set $\mathcal{F}^i:=\mathcal{F}^i_T$. 
Here, $(\xi^i,X^i_0,F^i_0)$ are $\mathbb{R}$-valued bounded random variables and $(\gamma^i,\beta^i)$ are $\mathbb{R}_{++}$-valued bounded random variables. Each space $(\Omega^i,\mathcal{F}^i,\mathbb{P}^i)$ is used to describe the idiosyncratic environment of agent-$i$.\\

\noindent
(2) We denote by $(\Omega^{0,i},\mathcal{F}^{0,i},\mathbb{P}^{0,i})$ ($i=1,\ldots,N$) a complete probability space over $\Omega^{0,i} := \Omega^0 \times \Omega^i$. Here, $(\mathcal{F}^{0,i},\mathbb{P}^{0,i})$ is the completion of $(\mathcal{F}^0 \otimes \mathcal{F}^i,\mathbb{P}^{0}\otimes \mathbb{P}^{i})$ and $\mathbb{F}^{0,i}:=(\mathcal{F}^{0,i}_t)_{t\in[0,T]}$ denotes the complete and right-continuous augmentation of $(\mathcal{F}_t^0 \otimes \mathcal{F}_t^i)_{t\in[0,T]}$.
 Moreover, we set $\mathcal{T}^{0,i}:=\mathcal{T}(\mathbb{F}^{0,i})$ for notational simplicity.\\ 

\noindent
(3) Let $(\Omega,\mathcal{F},\mathbb{P})$ be an enlarged complete probability space defined on $\Omega:=\prod_{i=0}^N\Omega^i$. $(\mathcal{F},\mathbb{P})$ is the completion of $\Bigl(\bigotimes_{i=0}^N\mathcal{F}^i,\bigotimes_{i=0}^N\mathbb{P}^i\Bigr)$ and the filtration $\mathbb{F}=(\mathcal{F}_t)_{t\in[0,T]}$ is the complete and right-continuous augmentation of $(\bigotimes_{i=0}^N\mathcal{F}^{i}_t)_{t\in[0,T]}$. \\

In this section, we make the following assumptions on heterogeneity of agents.
\begin{asm}~\\
    \label{asm3-habit}
    \textup{(i)} For each $i\in\{1,\ldots,N\}$, all statements of Assumption \ref{asm2-habit} hold with ``1" replaced by $``i"$.\\
    \textup{(ii)} $(\xi^i,\gamma^i, \beta^i,X^i_0)_{i\in\{1,\ldots,N\}}$ have the same distribution, i.e. they are independently and identically distributed on $(\Omega,\mathbb{F},\mathbb{P})$.\\
    \textup{(iii)} The liability processes $(F^i_t;t\in[0,T])_{i\in\{1,\ldots,N\}}$ are $\mathcal{F}^0$-conditionally independent and identically distributed on $(\Omega,\mathbb{F},\mathbb{P})$.
\end{asm}
The multi-agent problem is formulated on the filtered probability space $(\Omega,\mathcal{F},\mathbb{P},\mathbb{F})$ in the following way. Each agent-$i$ solves an optimal consumption-investment problem:
\begin{equation*}
    \begin{split}
        \sup_{(\pi,c)\in\mathbb{A}^i} U^i(\pi,c)
    \end{split}
\end{equation*}
subject to
\[
    \mathcal{W}^{i,(\pi,c)}_t =\xi^i + \int_0^t (\pi_s^\top\sigma_s\theta_s - c_s)ds + \int_0^t \pi_s^\top\sigma_s dW_s^0,~~~t\in[0,T],
\]
where $\mathbb{A}^i$ is an admissible space for agent-$i$, whose definition is to be given later. $U^i:\mathbb{A}^i\to\mathbb{R}$ is a utility function of agent-$i$ defined similarly to Section \ref{Section 2} by
\begin{equation*}
    \begin{split}
    &U^i(\pi,c)\\
    &:=\mathbb{E}\Bigl[-\exp\Bigl(-\delta T-\gamma^i(\mathcal{W}^{i,(\pi,c)}_T-F^i_T)\Bigr) -a \int_0^T \exp\Bigl(-\delta t-\gamma^i(\mathcal{W}^{i,(\pi,c)}_t-F^i_t)-\beta^i(c_t-X_t^{i,c})\Bigr)dt\Bigr].
\end{split}
\end{equation*}
Here, $X^{i,c}$ is agent-$i$'s consumption habits defined by
\[
    X_t^{i,c} = X^i_0 + \int_0^t \{-\kappa X^{i,c}_s + b c_s + (\kappa-b)\rho_s\}ds,~~~t\in[0,T].
\]
In this model, the parameters $\delta,a,\kappa,b>0$ and the habit trend $\rho\in\mathbb{L}^\infty(\mathbb{P}^0,\mathbb{F}^0,\mathbb{R})$ are common to all agents.\footnote{It would be possible to make the variables $(\delta,a,\kappa,b,\rho)$ different for each agent as we have done so for $(\xi^i,\gamma^i,\beta^i,X^i_0,F^i)$. For simplicity, we assume that they are common across the agents in this work.}
In the same manner as Section \ref{Section 2}, we define the admissible space $(\mathbb{A}^i)_{i=1,\ldots,N}$ as follows.
\begin{dfn} (Admissible space: a multi-agent version)\\
    For each $i=1,\ldots,N$, the admissible space $\mathbb{A}^i$ is the set of strategies $(\pi,c)\in\mathbb{H}^2(\mathbb{P}^{0,i},\mathbb{F}^{0,i},\mathbb{R}^{n})\times\mathbb{H}^2(\mathbb{P}^{0,i},\mathbb{F}^{0,i},\mathbb{R})$ such that the family
    \[
        \Bigl\{\exp\Bigl(-\gamma^i\mathcal{W}^{i,(\pi,c)}_\tau + \beta^i |c_\tau| + K^i |X_\tau^{i,c}|  \Bigr) ; \tau\in\mathcal{T}^{0,i}\Bigr\}
    \]
    is uniformly integrable for some $K^i>\gamma^i(A_i+\sqrt{A_i^2+B_i})^{-1} \lor \beta^i $ where
    \[
        A_i:=\frac{1}{2}\Bigl(\kappa-b+\frac{\gamma^i}{\beta^i}\Bigr),~~~B_i:=\frac{\gamma^i b}{\beta^i}.
    \]
    Moreover, we define $\mathcal{A}^i:=\{(p,c)=(\pi^\top\sigma,c) ; (\pi,c)\in\mathbb{A}^i\}$.
\end{dfn}
In the same way as in Section \ref{Section 2}, we restate the problem by writing $p_t:=\pi_t^\top\sigma_t$ for $t\in[0,T]$.
\begin{equation}
    \begin{split}
        \sup_{(p,c)\in\mathcal{A}^i} \widetilde{U}^i(p,c)
    \end{split}
\end{equation}
subject to
\[
    \mathcal{W}^{i,(p,c)}_t =\xi^i + \int_0^t (p^i_s\theta_s - c^i_s)ds + \int_0^t p^i_s dW_s^0,~~~t\in[0,T],
\]
where the objective function $\widetilde{U}^i:\mathcal{A}^i\to\mathbb{R}$ is defined by
\begin{equation}
    \begin{split}
    &\widetilde{U}^i(p,c)\\
    &:=\mathbb{E}\Bigl[-\exp\Bigl(-\delta T-\gamma^i(\mathcal{W}^{i,(p,c)}_T-F^i_T)\Bigr) -a \int_0^T \exp\Bigl(-\delta t-\gamma^i(\mathcal{W}^{i,(p,c)}_t-F^i_t)-\beta^i(c_t-X_t^{i,c})\Bigr)dt\Bigr].
\end{split}
\end{equation}

We also introduce an $\mathcal{F}^{i}_0$-measurable continuously differentiable process $\zeta^i$ for $i\in\{1,\ldots,N\}$ by
\begin{equation*}
        \zeta^i_t = \frac{e^{(\delta_i^+-\delta_i^-)(T-t)}-1}{\delta_i^+-\delta_i^-e^{(\delta_i^+-\delta_i^-)(T-t)}},~~~\delta_i^{\pm}:=-A_i\pm\sqrt{A_i^2+B_i},~~~t\in[0,T].
\end{equation*}
As before, $\zeta^i$ satisfies 
\[
    0\leq \zeta^i_t\leq \frac{1}{\delta_i^+} e^{(\delta_i^+-\delta_i^-)T}\land \frac{1}{|\delta_i^-|}.
\]
By the same argument as in Section \ref{Section 2}, the unique optimal strategy for each agent-$i$ is
\begin{equation*}
    \begin{split}
    p^{i,*}_t  &:=  (\pi_t^{i,*})^\top\sigma_t  = Z^{i,0\|}_t + \frac{\theta^\top_t}{\gamma^i},~~~t\in[0,T],\\
    c^{i,*}_t &:= X^{i,c^{i,*}}_t + \frac{1}{\beta^i}\Bigl\{\log\Bigl(\frac{a\beta^i}{\gamma^i(1+b\zeta^i_t)}\Bigr)-\gamma^i(Y^i_t-F^i_t+\zeta^i_tX^{i,c^{i.*}}_t)\Bigr\},~~~t\in[0,T],
    \end{split}
\end{equation*}
where the triple $(Y^i,Z^{i,0},Z^i)\in\mathbb{S}^\infty\times\mathbb{H}^2_{\mathrm{BMO}}\times\mathbb{H}^2_{\mathrm{BMO}}$ solves the BSDE \eqref{qgBSDE} with the superscript ``1'' replaced by ``$i$''. To derive the relevant mean field BSDE, let us define the market-clearing condition.
\begin{dfn}~(Market clearing condition)\\
    \label{market-clearing-habit}
    The financial market satisfies the market-clearing condition if the equality
    \begin{equation}
        \label{MC-eqn-habit}
        \frac{1}{N}\sum_{i=1}^N \pi_t^{i,*} = 0
    \end{equation}
    holds $dt\otimes \mathbb{P}$-almost everywhere. Here, $\pi_t^{i,*}$ denotes the optimal trading strategy of the agent-$i$.
\end{dfn}

As in Fujii \& Sekine \cite{fujiiMeanFieldEquilibriumPrice2023a} [Section 4], this condition motivates us to study the following mean field BSDE defined on $(\Omega^{0,i},\mathcal{F}^{0,i},\mathbb{P}^{0,i},\mathbb{F}^{0,i})$ for each $i\in\{1,\ldots,N\}$:
\begin{equation}
    \begin{split}
        \label{MF-BSDE1-habit}
            Y^i_t&= F^i_T + \int_t^T f^i(s,Y^i_s,Z^{i,0}_s,Z^i_s)ds - \int_t^T Z^{i,0}_s dW^0_s - \int_t^T Z^i_s dW^i_s,~~~t\in[0,T]
    \end{split}
  \end{equation}
with
\begin{equation*}
    \begin{split}
        f^i(s,Y^i_s,Z^{i,0}_s,Z^i_s)
        = &
        \hat\gamma Z^{i,0\|}_s\mathbb{E}[Z^{i,0\|}_s|\mathcal{F}^0]^{\top} - \frac{\hat\gamma^2}{2\gamma^i}|\mathbb{E}[Z^{i,0\|}_s|\mathcal{F}^0]|^2 + \frac{\gamma^i}{2}(|Z^{i,0\perp}_s|^2 + |Z^i_s|^2)\\
        & -\frac{\gamma^i(1+b\zeta^i_s)}{\beta^i}Y^i_s + g^i_s,
    \end{split}
  \end{equation*}
where
\[
    g^i_s := - \frac{\delta}{\gamma^i} + (\kappa-b)\zeta^i_s\rho_s + \frac{1 + b\zeta^i_s}{\beta^i}\Bigl\{1 + \log\Bigl(\frac{a\beta^i}{\gamma^i(1 + b\zeta^i_s)}\Bigr)+\gamma^i F^i_s\Bigr\}.
\]
We shall see later that this BSDE has a bounded solution $(Y^i,Z^{i,0},Z^i)\in\mathbb{S}^\infty \times\mathbb{H}^2_{\mathrm{BMO}}\times\mathbb{H}^2_{\mathrm{BMO}}$ under some conditions and the process $\theta^{\mathrm{mfg}} \in \mathbb{H}^2_{\mathrm{BMO}}(\mathbb{F}^0,\mathbb{R}^{d_0})$ defined by
\footnote{The market clearing condition requires the risk-premium process $\theta$ to satisfy
\[
    \frac{1}{N}\sum_{i=1}^N Z_t^{i,0\|} + \Bigl(\frac{1}{N}\sum_{i=1}^N \frac{1}{\gamma^i}\Bigr)\theta_t^\top = 0.
\]
By the mutual independence of $(\mathcal{F}_t^i)_{i\geq 1}$ and symmetry among agents, it is anticipated that $\theta^{\mathrm{mfg}}$ is the market-clearing risk-premium process in the large population limit.  See Fujii \& Sekine \cite{fujiiMeanFieldEquilibriumPrice2023a} [Section 4] for details.}
\begin{equation}
    \label{mf-premium}
    \theta^{\mathrm{mfg}}_t = -\hat\gamma\mathbb{E}[Z^{1,0\|}_t|\mathcal{F}^0]^{\top},~~~t\in[0,T]
\end{equation}
with $\hat\gamma:=\mathbb{E}\Bigl[\dfrac{1}{\gamma^1}\Bigr]^{-1}$ in fact clears the financial market in the large population limit.
\subsection{Well-Posedness of the Mean Field BSDE}
We are now going to investigate the well-posedness of the equation \eqref{MF-BSDE1-habit}.
\begin{lem}
    Let Assumptions \ref{asm1-habit} and \ref{asm3-habit} be in force. If the BSDE \eqref{MF-BSDE1-habit} has a bounded solution $(Y^i,Z^{i,0},Z^i)\in\mathbb{S}^\infty(\mathbb{P}^{0,i},\mathbb{F}^{0,i},\mathbb{R})\times\mathbb{H}^2(\mathbb{P}^{0,i},\mathbb{F}^{0,i},\mathbb{R}^{1\times d_0})\times\mathbb{H}^2(\mathbb{P}^{0,i},\mathbb{F}^{0,i},\mathbb{R}^{1\times d})$, then $(Z^{i,0},Z^i)\in\mathbb{H}^2_{\mathrm{BMO}}(\mathbb{P}^{0,i},\mathbb{F}^{0,i},\mathbb{R}^{1\times d_0})\times\mathbb{H}^2_{\mathrm{BMO}}(\mathbb{P}^{0,i},\mathbb{F}^{0,i},\mathbb{R}^{1\times d})$.
\end{lem}
\noindent
\textbf{\textit{Proof}}\\
Without loss of generality, we choose the agent-1 as a representative agent and suppress the superscript ``1'' when obvious. We have, by Young's inequality,
\begin{equation*}
    \begin{split}
     f(s,Y_s,Z^0_s,Z^1_s)
     &= 
     \hat\gamma Z^{0\|}_s\mathbb{E}[Z^{0\|}_s|\mathcal{F}^0]^{\top} - \frac{\hat\gamma^2}{2\gamma}|\mathbb{E}[Z^{0\|}_s|\mathcal{F}^0]|^2 + \frac{\gamma}{2}(|Z^{0\perp}_s|^2 + |Z^1_s|^2) -\frac{\gamma(1+b\zeta_s)}{\beta}Y_s + g_s\\
     &\leq
     \frac{\gamma}{2}(|Z^{0}_s|^2 + |Z^1_s|^2) + C(\|Y\|_{\mathbb{S}^\infty} + \|g\|_{\mathbb{L}^\infty}).
    \end{split}
\end{equation*}
Then, by Ito formula,
\begin{equation*}
    \begin{split}
     de^{2\gamma Y_t}
     &=
     2\gamma e^{2\gamma Y_t}( dY_t +\gamma d\langle Y \rangle_t)\\
     &=
     2\gamma e^{2\gamma Y_t}\Bigl\{(- f(t,Y_t,Z^0_t,Z^1_t) + \gamma|Z^{0}_t|^2 + \gamma|Z^1_t|^2)dt + Z^0_t dW^0_t + Z^1_t dW^1_t\Bigr\}.
    \end{split}
\end{equation*}
This yields:
\begin{equation*}
    \begin{split}
     &e^{2\gamma Y_T} - e^{2\gamma Y_t} \\
     &=
     2\gamma \int_t^T e^{2\gamma Y_s}\Bigl\{(- f(s,Y_s,Z^0_s,Z^1_s) + \gamma|Z^{0}_s|^2 + \gamma|Z^1_s|^2)ds + Z^0_s dW^0_s + Z^1_s dW^1_s\Bigr\}\\
     &\geq
     2\gamma \int_t^T e^{2\gamma Y_s}\Bigl\{- \frac{\gamma}{2}(|Z^{0}_s|^2 + |Z^1_s|^2) - C(\|Y\|_{\mathbb{S}^\infty} + \|g\|_{\mathbb{L}^\infty}) + \gamma|Z^{0}_s|^2 + \gamma|Z^1_s|^2\Bigr\}ds \\
     &~~~~+  2\gamma \int_t^T e^{2\gamma Y_s}\Bigl\{Z^0_s dW^0_s + Z^1_s dW^1_s\Bigr\}\\
     &\geq
     2\gamma \int_t^T e^{2\gamma Y_s}\Bigl\{\frac{\gamma}{2}(|Z^{0}_s|^2 + |Z^1_s|^2) - C(\|Y\|_{\mathbb{S}^\infty} + \|g\|_{\mathbb{L}^\infty}) \Bigr\}ds  + 2\gamma \int_t^T e^{2\gamma Y_s}\Bigl\{Z^0_s dW^0_s + Z^1_s dW^1_s\Bigr\}.
    \end{split}
\end{equation*}
Thus, for all $t\in[0,T]$, we get
\begin{equation*}
    \begin{split}
     \mathbb{E}\Bigl[ \int_t^T (|Z^{0}_s|^2 + |Z^1_s|^2)ds |\mathcal{F}^{0,1}_t\Bigr]
     \leq
     Ce^{4\overline{\gamma} \|Y\|_{\mathbb{S}^\infty}}(1+\|Y\|_{\mathbb{S}^\infty} + \|g\|_{\mathbb{L}^\infty})
     <\infty,
    \end{split}
\end{equation*}
and clearly, $(Z^0,Z^1)\in\mathbb{H}^2_{\mathrm{BMO}}\times\mathbb{H}^2_{\mathrm{BMO}}$. $\square$\par
To show the well-posedness, we need to make an additional assumption on the size of the terminal liability $F^i_T$ and the process $g^i$.
\begin{asm}
    \label{asm4-habit}
    Assume that, for each $i\in\{1,\ldots,N\}$, the random variable $F^i_T$ and the process $(g^i_t)_{t\in[0,T]}$ satisfy
    \begin{equation}
        \label{small}
        \sqrt{\|F^i_T\|^2_{\infty} + 4\Bigl\|\int_0^T |g^i_s| ds\Bigr\|^2_{\infty}} <   \frac{1}{16 c_\gamma}\land \frac{1}{32 C_\gamma},
    \end{equation}  
    where
    \begin{equation*}
        \begin{split}
            c_\gamma:=\frac{\overline{\gamma}}{2} \lor \frac{\hat\gamma^2}{\underline{\gamma}},~~~C_\gamma:=\hat\gamma + \Bigl(\frac{\hat\gamma^2}{2\underline{\gamma}} \lor \frac{\overline{\gamma}}{2}\Bigr).
        \end{split}
    \end{equation*}
\end{asm}
\begin{rem}~\\
    For each $s\in[0,T]$, we have $|g^i_s|\to 0$ when, for instance, $\delta\to0$, $\|\rho\|_{\mathbb{L}^\infty}\to 0$ and $\beta^i\to\infty$. This observation allows us to find appropriate parameters that fulfils the condition \eqref{small} if $F^i_T$ is sufficiently small.
\end{rem}
Here is our first main result of this section. The method is inspired by the work Tevzadze \cite{tevzadzeSolvabilityBackwardStochastic2008}.
\begin{thm}
    \label{MF-BSDE-wellposed1}
    Under Assumptions \ref{asm1-habit}, \ref{asm3-habit}, and \ref{asm4-habit}, the mean field BSDE \eqref{MF-BSDE1-habit} has a bounded solution $(Y^i,Z^{i,0},Z^i)\in\mathbb{S}^\infty(\mathbb{P}^{0,i},\mathbb{F}^{0,i},\mathbb{R}) \times\mathbb{H}^2_{\mathrm{BMO}}(\mathbb{P}^{0,i},\mathbb{F}^{0,i},\mathbb{R}^{1\times d_0})\times\mathbb{H}^2_{\mathrm{BMO}}(\mathbb{P}^{0,i},\mathbb{F}^{0,i},\mathbb{R}^{1\times d})$. 
\end{thm}

\noindent
\textbf{\textit{Proof}}\\
Again, we choose the agent-1 as a representative agent and omit the superscript ``1'' for simplicity.\\
(Step I)\\
By completing the square, we have
\begin{equation*}
    \begin{split}
        &\hat\gamma Z^{0\|}_s\mathbb{E}[Z^{0\|}_s|\mathcal{F}^0]^{\top} - \frac{\hat\gamma^2}{2\gamma}|\mathbb{E}[Z^{0\|}_s|\mathcal{F}^0]|^2 + \frac{\gamma}{2}(|Z^{0\perp}_s|^2 + |Z^1_s|^2) \\
        &=
        -\Bigl|\frac{\hat\gamma}{\sqrt{2\gamma}}\mathbb{E}[Z^{0\|}_s|\mathcal{F}^0] - \frac{\sqrt{\gamma}}{\sqrt{2}}Z^{0\|}_s\Bigr|^2 +\frac{\gamma}{2}|Z^{0\|}_s|^2+\frac{\gamma}{2}(|Z^{0\perp}_s|^2 + |Z^1_s|^2) \\
        &=
        -\Bigl|\frac{\hat\gamma}{\sqrt{2\gamma}}\mathbb{E}[Z^{0\|}_s|\mathcal{F}^0] - \frac{\sqrt{\gamma}}{\sqrt{2}}Z^{0\|}_s\Bigr|^2 +\frac{\gamma}{2}(|Z^{0}_s|^2 + |Z^1_s|^2).
    \end{split}
\end{equation*}
It then follows that
\begin{equation*}
    \begin{split}
        \hat\gamma Z^{0\|}_s\mathbb{E}[Z^{0\|}_s|\mathcal{F}^0]^{\top} - \frac{\hat\gamma^2}{2\gamma}|\mathbb{E}[Z^{0\|}_s|\mathcal{F}^0]|^2 + \frac{\gamma}{2}(|Z^{0\perp}_s|^2 + |Z^1_s|^2) 
        \leq
        \frac{\gamma}{2}(|Z^{0}_s|^2 + |Z^1_s|^2)
    \end{split}
\end{equation*}
and
\begin{equation*}
    \begin{split}
        \hat\gamma Z^{0\|}_s\mathbb{E}[Z^{0\|}_s|\mathcal{F}^0]^{\top} - \frac{\hat\gamma^2}{2\gamma}|\mathbb{E}[Z^{0\|}_s|\mathcal{F}^0]|^2 + \frac{\gamma}{2}(|Z^{0\perp}_s|^2 + |Z^1_s|^2) 
        &\geq
        -\Bigl|\frac{\hat\gamma}{\sqrt{2\gamma}}\mathbb{E}[Z^{0\|}_s|\mathcal{F}^0] - \frac{\sqrt{\gamma}}{\sqrt{2}}Z^{0\|}_s\Bigr|^2 +\frac{\gamma}{2}|Z^{0}_s|^2 \\
        &\geq
        -\frac{\hat\gamma^2}{\gamma}|\mathbb{E}[Z^{0\|}_s|\mathcal{F}^0]|^2 - \gamma |Z^{0\|}_s|^2 +\frac{\gamma}{2}|Z^{0}_s|^2 \\
        &\geq
        -\frac{\hat\gamma^2}{\gamma}|\mathbb{E}[Z^{0\|}_s|\mathcal{F}^0]|^2 - \frac{\gamma}{2}|Z^{0}_s|^2.
    \end{split}
\end{equation*}
Putting these together, we obtain
\begin{equation*}
    \begin{split}
        \Bigl|\hat\gamma Z^{0\|}_s\mathbb{E}[Z^{0\|}_s|\mathcal{F}^0]^{\top} - \frac{\hat\gamma^2}{2\gamma}|\mathbb{E}[Z^{0\|}_s|\mathcal{F}^0]|^2 + \frac{\gamma}{2}(|Z^{0\perp}_s|^2 + |Z^1_s|^2) \Bigr|
        &\leq
        \frac{\hat\gamma^2}{\gamma}|\mathbb{E}[Z^{0\|}_s|\mathcal{F}^0]|^2 + \frac{\gamma}{2}(|Z^{0}_s|^2 + |Z^1_s|^2)\\
        &\leq
        c_\gamma (|\mathbb{E}[Z^{0\|}_s|\mathcal{F}^0]|^2 + |Z^{0}_s|^2 + |Z^1_s|^2)
    \end{split}
\end{equation*}
for
\[
    c_\gamma=\frac{\overline{\gamma}}{2} \lor \frac{\hat\gamma^2}{\underline{\gamma}}.
\]
This yields
\begin{equation*}
    \begin{split}
        &Y_s f(s,Y_s,Z^0_s,Z^1_s)\\
        &\leq
        |Y_s|\Bigl|\hat\gamma Z^{0\|}_s\mathbb{E}[Z^{0\|}_s|\mathcal{F}^0]^{\top} - \frac{\hat\gamma^2}{2\gamma}|\mathbb{E}[Z^{0\|}_s|\mathcal{F}^0]|^2 + \frac{\gamma}{2}(|Z^{0\perp}_s|^2 + |Z^1_s|^2)\Bigr|- \frac{\gamma(1+b\zeta_s)}{\beta}|Y_s|^2 + |Y_s||g_s|\\
        &\leq
        c_\gamma |Y_s|(|\mathbb{E}[Z^{0\|}_s|\mathcal{F}^0]|^2 + |Z^{0}_s|^2 + |Z^1_s|^2) + |Y_s||g_s|.
    \end{split}
\end{equation*}

Let us now consider a BSDE
\begin{equation*}
    \begin{split}
        Y_t&= F_T + \int_t^T f(s,Y_s,z^0_s,z^1_s)ds - \int_t^T Z^0_s dW^0_s - \int_t^T Z^1_s dW^1_s,~~~t\in[0,T]
    \end{split}
\end{equation*}
for an arbitrary $(z^0,z^1)\in\mathbb{H}^2_{\mathrm{BMO}}\times\mathbb{H}^2_{\mathrm{BMO}}$ as an input. From the standard result for Lipschitz BSDEs, there exists a unique solution $(Y,Z^0,Z^1)\in\mathbb{S}^\infty\times\mathbb{H}^2_{\mathrm{BMO}}\times\mathbb{H}^2_{\mathrm{BMO}}$ for every given $(z^0,z^1)\in\mathbb{H}^2_{\mathrm{BMO}}\times\mathbb{H}^2_{\mathrm{BMO}}$. 
In this manner, we define a map $\Gamma:\mathbb{H}^2_{\mathrm{BMO}}(\mathbb{P}^{0,1},\mathbb{F}^{0,1},\mathbb{R}^{1\times d_0}\times \mathbb{R}^{1\times d})\to \mathbb{H}^2_{\mathrm{BMO}}(\mathbb{P}^{0,1},\mathbb{F}^{0,1},\mathbb{R}^{1\times d_0}\times \mathbb{R}^{1\times d})$ by $\Gamma(z^0,z^1)=(Z^0,Z^1)$.
By Ito formula,
\begin{equation*}
    \begin{split}
        &|Y_t|^2 + \mathbb{E}\Bigl[\int_t^T (|Z_s^0|^2 + |Z_s^1|^2 )ds | \mathcal{F}^{0,1}_t\Bigr]\\
        &=
        \mathbb{E}\Bigl[|F_T|^2 + 2\int_t^T Y_sf(s,Y_s,z^0_s,z^1_s)ds | \mathcal{F}^{0,1}_t\Bigr]\\
        &\leq
        \|F_T\|^2_{\infty} + 2c_\gamma\mathbb{E}\Bigl[ \int_t^T |Y_s|(|\mathbb{E}[z^{0\|}_s|\mathcal{F}^0]|^2 + |z^{0}_s|^2 + |z^1_s|^2 )ds | \mathcal{F}^{0,1}_t\Bigr] + 2\mathbb{E}\Bigl[ \int_t^T |Y_s||g_s|ds | \mathcal{F}^{0,1}_t\Bigr]\\
        &\leq
        \|F_T\|^2_{\infty} + 2c_\gamma\|Y\|_{\mathbb{S}^\infty}\mathbb{E}\Bigl[ \int_t^T (|\mathbb{E}[z^{0\|}_s|\mathcal{F}^0]|^2 + |z^{0}_s|^2 + |z^1_s|^2)ds | \mathcal{F}^{0,1}_t\Bigr]+ 2\|Y\|_{\mathbb{S}^\infty}\mathbb{E}\Bigl[ \int_t^T|g_s|ds | \mathcal{F}^{0,1}_t\Bigr]\\
        &\leq
        \|F_T\|^2_{\infty} + 4c_\gamma\|Y\|_{\mathbb{S}^\infty}\|(z^0,z^1)\|^2_{\mathbb{H}^2_{\mathrm{BMO}}} + 2\|Y\|_{\mathbb{S}^\infty}\mathbb{E}\Bigl[ \int_t^T|g_s|ds | \mathcal{F}^{0,1}_t\Bigr] \\
        &\leq
        \|F_T\|^2_{\infty} + \frac{1}{2}\|Y\|^2_{\mathbb{S}^\infty} + 16c_\gamma^2\|(z^0,z^1)\|^4_{\mathbb{H}^2_{\mathrm{BMO}}} + 4 \Bigl\|\int_0^T |g_s| ds\Bigr\|^2_{\infty}.
    \end{split}
\end{equation*}
Here, we have used the fact that
\begin{equation*}
    \begin{split}
        \sup_{\tau\in\mathcal{T}^{0,1}}\Bigl\|\mathbb{E}\Bigl[ \int_\tau^T |\mathbb{E}[z^{0\|}_s|\mathcal{F}^0]|^2 ds | \mathcal{F}^{0,1}_\tau\Bigr]\Bigr\|_\infty
        &\leq
        \|z^0\|^2_{\mathbb{H}^2_{\mathrm{BMO}}},
    \end{split}
\end{equation*}
which is shown in Fujii \& Sekine \cite{fujiiMeanFieldEquilibriumPrice2023a} [Lemma 4.2]. Taking the essential supremum on both sides, we get:
\begin{equation*}
    \begin{split}
        &\esssup_{(t,\omega)\in[0,T]\times\Omega}\Bigl(|Y_t|^2 + \mathbb{E}\Bigl[\int_t^T (|Z_s^0|^2 + |Z_s^1|^2 )ds | \mathcal{F}^{0,1}_t\Bigr]\Bigr) \\
        &\leq
        \|F_T\|^2_{\infty} + \frac{1}{2}\|Y\|^2_{\mathbb{S}^\infty} + 16c_\gamma^2\|(z^0,z^1)\|^4_{\mathbb{H}^2_{\mathrm{BMO}}} + 4 \Bigl\|\int_0^T |g_s| ds\Bigr\|^2_{\infty}.
    \end{split}
\end{equation*}
Using the fact that
\begin{equation*}
    \begin{split}
        \frac{1}{2}(\|Y\|^2_{\mathbb{S}^\infty} + \|(Z^0,Z^1)\|^2_{\mathbb{H}^2_{\mathrm{BMO}}})
        &\leq
        \max\{\|Y\|^2_{\mathbb{S}^\infty} , \|(Z^0,Z^1)\|^2_{\mathbb{H}^2_{\mathrm{BMO}}} \}\\
        &\leq
        \esssup_{(t,\omega)\in[0,T]\times\Omega}\Bigl(|Y_t|^2 + \mathbb{E}\Bigl[\int_t^T (|Z_s^0|^2 + |Z_s^1|^2 )ds | \mathcal{F}^{0,1}_t\Bigr]\Bigr),
    \end{split}
\end{equation*}
we obtain
\begin{equation*}
    \begin{split}
       \|(Z^0,Z^1)\|^2_{\mathbb{H}^2_{\mathrm{BMO}}}
        &\leq
        2\|F_T\|^2_{\infty} + 8\Bigl\|\int_0^T |g_s| ds\Bigr\|^2_{\infty} + 32c_\gamma^2 \|(z^0,z^1)\|^4_{\mathbb{H}^2_{\mathrm{BMO}}}.
    \end{split}
\end{equation*}
Since we have assumed that 
\begin{equation*}
    \begin{split}
        \|F_T\|^2_{\infty} + 4\Bigl\|\int_0^T |g_s| ds\Bigr\|^2_{\infty} \leq \frac{1}{256c_\gamma^2},
    \end{split}
\end{equation*}
there exists $R>0$ such that the inequality
\begin{equation*}
    \begin{split}
        2\|F_T\|^2_{\infty} + 8\Bigl\|\int_0^T |g_s| ds\Bigr\|^2_{\infty} + 32c_\gamma^2 R^4\leq R^2
    \end{split}
\end{equation*}
holds true. We can choose, for instance,
\begin{equation}
    \label{R-choice}
    R = 2\sqrt{\|F_T\|^2_{\infty} + 4\Bigl\|\int_0^T |g_s| ds\Bigr\|^2_{\infty}}\leq \frac{1}{8c_\gamma}.
\end{equation}
(Step II)\\
From the results of (Step I), we have $\Gamma(\mathcal{B}_R)\subset \mathcal{B}_R$, where
\[
    \mathcal{B}_R:=\{(z^0,z^1)\in \mathbb{H}^2_{\mathrm{BMO}}(\mathbb{F}^{0,1},\mathbb{R}^{1\times d_0}\times \mathbb{R}^{1\times d}) ;  \|(z^0,z^1)\|_{\mathbb{H}^2_{\mathrm{BMO}}}\leq R\}.
\]
Our objective is now to prove that $\Gamma|_{\mathcal{B}_R}:\mathcal{B}_R\to \mathcal{B}_R$ is a strict contraction. For $(z^0,z^1), (\acute{z}^0,\acute{z}^1)\in \mathcal{B}_R$, we set $(Z^0,Z^1):=\Gamma(z^0,z^1)$ and $(\acute{Z}^0,\acute{Z}^1):=\Gamma(\acute{z}^0,\acute{z}^1)$. Also, let $Y$ and $\acute{Y}$ be corresponding solutions and set $\Delta Y := Y-\acute{Y}$ and $\Delta Z^i := Z^i-\acute{Z}^i$. Notice that
\begin{equation*}
    \begin{split}
        &\Delta Y_s\{f(s,Y_s,z^0_s,z^1_s)-f(s,\acute{Y}_s,\acute{z}^0_s,\acute{z}^1_s)\}\\
        &\leq
        |\Delta Y_s|\Bigl\{\hat\gamma(|\mathbb{E}[z^{0\|}_s |\mathcal{F}^0]| + |\acute{z}^{0\|}_s|)(|\Delta z^{0\|}_s| + |\mathbb{E}[\Delta z^{0\|}_s |\mathcal{F}^0]|)\\
        &~~~~~~~~~~~~~+ \frac{\hat\gamma^2}{2\underline{\gamma}}(|\mathbb{E}[z^{0\|}_s |\mathcal{F}^0]| + |\mathbb{E}[\acute{z}^{0\|}_s |\mathcal{F}^0]|)|\mathbb{E}[\Delta z^{0\|}_s |\mathcal{F}^0]|\Bigr.\\
        &~~~~~~~~~~~~~+\frac{\overline{\gamma}}{2}(|z^{0\perp}_s| + |\acute{z}^{0\perp}_s| + |z^1_s| + |\acute{z}^1_s| )(|\Delta z^{0\perp}_s| + |\Delta z^1_s|)\Bigr\} -\frac{\gamma(1+b\zeta_s)}{\beta}|\Delta Y_s|^2 \\
        &\leq
        |\Delta Y_s|\Bigl\{\Bigl(\Bigl(\hat\gamma + \frac{\hat\gamma^2}{2\underline{\gamma}}\Bigr)|\mathbb{E}[z^{0\|}_s |\mathcal{F}^0]| + \frac{\hat\gamma^2}{2\underline{\gamma}}|\mathbb{E}[\acute{z}^{0\|}_s |\mathcal{F}^0]| + \hat\gamma|\acute{z}^{0\|}_s |\Bigr) |\mathbb{E}[\Delta z^{0\|}_s |\mathcal{F}^0]|  \Bigr.\\
        &~~~~~~~~~~~~~\Bigl. +\Bigl(\hat\gamma + \frac{\overline{\gamma}}{2}\Bigr)(|\mathbb{E}[z^{0\|}_s |\mathcal{F}^0]| + |\acute{z}^{0}_s| +|z^{0}_s| + |z^1_s| + |\acute{z}^1_s| )(|\Delta z^{0}_s| + |\Delta z^1_s|)\Bigr\}.
    \end{split}
\end{equation*}
Applying Ito formula to $|\Delta Y_t|^2$ , we have
\begin{equation*}
    \begin{split}
       &|\Delta Y_t|^2 + \mathbb{E}\Bigl[\int_t^T (|\Delta Z_s^0|^2 + |\Delta Z_s^1|^2) ds | \mathcal{F}^{0,1}_t\Bigr]\\
       &=
       2\mathbb{E}\Bigl[\int_t^T \Delta Y_s\{f(s,Y_s,z^0_s,z^1_s)-f(s,\acute{Y}_s,\acute{z}^0_s,\acute{z}^1_s)\} ds | \mathcal{F}^{0,1}_t\Bigr]\\
       &\leq
       2  \|\Delta Y\|_{\mathbb{S}^\infty}\mathbb{E}\Bigl[\int_t^T \Bigl(\Bigl(\hat\gamma + \frac{\hat\gamma^2}{2\underline{\gamma}}\Bigr)|\mathbb{E}[z^{0\|}_s |\mathcal{F}^0]| + \frac{\hat\gamma^2}{2\underline{\gamma}}|\mathbb{E}[\acute{z}^{0\|}_s |\mathcal{F}^0]| + \hat\gamma|\acute{z}^{0\|}_s|\Bigr) |\mathbb{E}[\Delta z^{0\|}_s |\mathcal{F}^0]| ds | \mathcal{F}^{0,1}_t\Bigr]\\
       &~~~+2\Bigl(\hat\gamma + \frac{\overline{\gamma}}{2}\Bigr) \|\Delta Y\|_{\mathbb{S}^\infty}\mathbb{E}\Bigl[\int_t^T (|\mathbb{E}[z^{0\|}_s |\mathcal{F}^0]| + |z^{0}_s| + |\acute{z}^{0}_s| + |z^1_s| + |\acute{z}^1_s| )(|\Delta z^{0}_s| + |\Delta z^1_s|) ds | \mathcal{F}^{0,1}_t\Bigr]\\
       &\leq
       2\|\Delta Y\|_{\mathbb{S}^\infty}\mathbb{E}\Bigl[\int_t^T  \Bigl(\Bigl(\hat\gamma + \frac{\hat\gamma^2}{2\underline{\gamma}}\Bigr)|\mathbb{E}[z^{0\|}_s |\mathcal{F}^0]| + \frac{\hat\gamma^2}{2\underline{\gamma}}|\mathbb{E}[\acute{z}^{0\|}_s |\mathcal{F}^0]| + \hat\gamma|\acute{z}^{0\|}_s|\Bigr)^2ds | \mathcal{F}^{0,1}_t\Bigr]^{\frac{1}{2}}\\
       &~~~~~\times\mathbb{E}\Bigl[\int_t^T|\mathbb{E}[\Delta z^{0\|}_s |\mathcal{F}^0]|^2 ds | \mathcal{F}^{0,1}_t\Bigr]^{\frac{1}{2}}\\
       &~~~+ 2\Bigl(\hat\gamma + \frac{\overline{\gamma}}{2}\Bigr)\|\Delta Y\|_{\mathbb{S}^\infty}\mathbb{E}\Bigl[\int_t^T (|\mathbb{E}[z^{0\|}_s |\mathcal{F}^0]| + |z^{0}_s| + |\acute{z}^{0}_s| + |z^1_s| + |\acute{z}^1_s| )^2ds | \mathcal{F}^{0,1}_t\Bigr]^{\frac{1}{2}}\\
       &~~~~~\times\mathbb{E}\Bigl[\int_t^T(|\Delta z^{0}_s| + |\Delta z^1_s|)^2 ds | \mathcal{F}^{0,1}_t\Bigr]^{\frac{1}{2}}\\
       &\leq
       2\sqrt{3}\|\Delta Y\|_{\mathbb{S}^\infty}\mathbb{E}\Bigl[\int_t^T  \Bigl(\Bigl(\hat\gamma + \frac{\hat\gamma^2}{2\underline{\gamma}}\Bigr)^2|\mathbb{E}[z^{0\|}_s |\mathcal{F}^0]|^2 + \frac{\hat\gamma^4}{4\underline{\gamma}^2}|\mathbb{E}[\acute{z}^{0\|}_s |\mathcal{F}^0]|^2 + \hat\gamma^2|\acute{z}^{0\|}_s|^2\Bigr)ds | \mathcal{F}^{0,1}_t\Bigr]^{\frac{1}{2}}\\
       &~~~~~\times\mathbb{E}\Bigl[\int_t^T|\mathbb{E}[\Delta z^{0\|}_s |\mathcal{F}^0]|^2 ds | \mathcal{F}^{0,1}_t\Bigr]^{\frac{1}{2}}\\
       &~~~+ 2\sqrt{10}\Bigl(\hat\gamma + \frac{\overline{\gamma}}{2}\Bigr)\|\Delta Y\|_{\mathbb{S}^\infty}\mathbb{E}\Bigl[\int_t^T (|\mathbb{E}[z^{0\|}_s |\mathcal{F}^0]|^2  +|z^{0}_s|^2 + |\acute{z}^{0}_s|^2 + |z^1_s|^2 + |\acute{z}^1_s|^2 )ds | \mathcal{F}^{0,1}_t\Bigr]^{\frac{1}{2}}\\
       &~~~~~\times\mathbb{E}\Bigl[\int_t^T(|\Delta z^{0}_s|^2 + |\Delta z^1_s|^2) ds | \mathcal{F}^{0,1}_t\Bigr]^{\frac{1}{2}}
    \end{split}
\end{equation*}
\begin{equation*}
    \begin{split}
       &\leq
       2(\sqrt{6} + \sqrt{30})C_\gamma\|\Delta Y\|_{\mathbb{S}^\infty} R\|(\Delta z^0, \Delta z^1)\|_{\mathbb{H}^2_{\mathrm{BMO}}} \\
       &\leq 
       \frac{1}{2}\|\Delta Y\|_{\mathbb{S}^\infty}^2 + 2(\sqrt{6} + \sqrt{30})^2C_\gamma^2 R^2 \|(\Delta z^0, \Delta z^1)\|^2_{\mathbb{H}^2_{\mathrm{BMO}}} \\
       &\leq
       \frac{1}{2}\|\Delta Y\|_{\mathbb{S}^\infty}^2 + 128C_\gamma^2 R^2 \|(\Delta z^0, \Delta z^1)\|^2_{\mathbb{H}^2_{\mathrm{BMO}}},
    \end{split}
\end{equation*}
where
\[
    C_\gamma = \hat\gamma + \Bigl(\frac{\hat\gamma^2}{2\underline{\gamma}} \lor \frac{\overline{\gamma}}{2}\Bigr).
\]
Taking the essential supremum, we get
\begin{equation*}
    \begin{split}
       &\esssup_{(t,\omega)\in[0,T]\times\Omega}\Bigl(|\Delta Y_t|^2 + \mathbb{E}\Bigl[\int_t^T (|\Delta Z_s^0|^2 + |\Delta Z_s^1|^2) ds | \mathcal{F}^{0,1}_t\Bigr]\Bigr) \\
       &~~~\leq
       \frac{1}{2}\|\Delta Y\|_{\mathbb{S}^\infty}^2 + 128C_\gamma^2 R^2 \|(\Delta z^0, \Delta z^1)\|^2_{\mathbb{H}^2_{\mathrm{BMO}}}.
    \end{split}
\end{equation*}
Using the fact that
\begin{equation*}
    \begin{split}
        &\frac{1}{2}(\|\Delta Y\|^2_{\mathbb{S}^\infty} + \|(\Delta Z^0,\Delta Z^1)\|^2_{\mathbb{H}^2_{\mathrm{BMO}}})\\
        &~~\leq
        \max\{\|\Delta Y\|^2_{\mathbb{S}^\infty} , \|(\Delta Z^0,\Delta Z^1)\|^2_{\mathbb{H}^2_{\mathrm{BMO}}} \}\\
        &~~\leq
        \esssup_{(t,\omega)\in[0,T]\times\Omega}\Bigl(|\Delta Y_t|^2 + \mathbb{E}\Bigl[\int_t^T (|\Delta Z_s^0|^2 + |\Delta Z_s^1|^2 )ds | \mathcal{F}^{0,1}_t\Bigr]\Bigr),
    \end{split}
\end{equation*}
we obtain
\begin{equation*}
    \begin{split}
        \|(\Delta Z^0, \Delta Z^1)\|^2_{\mathbb{H}^2_{\mathrm{BMO}}}
       &\leq
       256C_\gamma^2R^2\|(\Delta z^0, \Delta z^1)\|^2_{\mathbb{H}^2_{\mathrm{BMO}}}.
    \end{split}
\end{equation*}
Under \eqref{small}, $\Gamma|_{\mathcal{B}_R}$ becomes a strict contraction. Indeed, having chosen $R$ by \eqref{R-choice}, we clearly have $256 C_\gamma^2R^2 < 1$. This yields that there exists a unique fixed point of $\Gamma|_{\mathcal{B}_R}$, which represents a bounded solution of the BSDE \eqref{MF-BSDE1-habit}. $\square$

\begin{rem}
    Due to the uniqueness of the fixed point, the mean field BSDE \eqref{MF-BSDE1-habit} has a unique solution if we restrict the domain to $\mathbb{S}^\infty\times\mathcal{B}_R$.
\end{rem}

\subsection{Asymptotic Equilibrium in the Large Population Limit}
We now prove that the optimal trading strategy in the market with risk-premium process $\theta^{\mathrm{mfg}}$ defined by \eqref{mf-premium} satisfies the market clearing condition \eqref{MC-eqn-habit} in the large population limit. 
To deal with the large population limit, we need to enlarge our probability space in the following way. Let $(\overline{\Omega},\overline{\mathcal{F}},\overline{\mathbb{P}})$ be an complete probability space defined on $\overline{\Omega}:=\prod_{i=0}^\infty\Omega^i$. Here, $(\overline{\mathcal{F}},\overline{\mathbb{P}})$ is the completion of $\Bigl(\bigotimes_{i=0}^\infty\mathcal{F}^i,\bigotimes_{i=0}^\infty\mathbb{P}^i\Bigr)$ and the filtration $\overline{\mathbb{F}}:=(\overline{\mathcal{F}}_t)_{t\in[0,T]}$ is the complete and right-continuous augmentation of $(\bigotimes_{i=0}^\infty\mathcal{F}^{i}_t)_{t\in[0,T]}$. In the remainder of this section, $\mathbb{E}[\cdot]$ denotes the expectation with respect to $\overline{\mathbb{P}}$.

Let us arbitrarily choose one bounded solution $(\mathcal{Y}^1,\mathcal{Z}^0,\mathcal{Z}^1)\in\mathbb{S}^\infty\times\mathbb{H}^2_{\mathrm{BMO}}\times\mathbb{H}^2_{\mathrm{BMO}}$ of the mean field BSDE
\begin{equation*}
    \begin{split}
            \mathcal{Y}^1_t
            = 
            F^1_T &+ \int_t^T \Bigl\{\hat\gamma \mathcal{Z}^{0\|}_s\mathbb{E}[\mathcal{Z}^{0\|}_s|\mathcal{F}^0]^{\top} - \frac{\hat\gamma^2}{2\gamma^1}|\mathbb{E}[\mathcal{Z}^{0\|}_s|\mathcal{F}^0]|^2 + \frac{\gamma^1}{2}(|\mathcal{Z}^{0\perp}_s|^2 + |\mathcal{Z}^1_s|^2) -\frac{\gamma^1(1+b\zeta^1_s)}{\beta^1}\mathcal{Y}^1_s + g^1_s\Bigr\}ds\\
            &- \int_t^T \mathcal{Z}^{0}_s dW^0_s - \int_t^T \mathcal{Z}^1_s dW^1_s,~~~t\in[0,T],
    \end{split}
  \end{equation*}
where
\[
    g^1_s := - \frac{\delta}{\gamma^1} + (\kappa-b)\zeta^1_s\rho_s + \frac{1 + b\zeta^1_s}{\beta^1}\Bigl\{1 + \log\Bigl(\frac{a\beta^1}{\gamma^1(1 + b\zeta^1_s)}\Bigr)+\gamma^1 F^1_s\Bigr\},~~~s\in[0,T],
\]
and fix it. Theorem \ref{MF-BSDE-wellposed1} provides one such example under the appropriate conditions. Using this solution, we define the process $\theta^{\mathrm{mfg}}\in\mathbb{H}^2_{\mathrm{BMO}}(\mathbb{F}^0,\mathbb{R}^{d_0})$ by $\theta^{\mathrm{mfg}}_t := -\hat\gamma\mathbb{E}[\mathcal{Z}^{0\|}_t | \mathcal{F}^0]^\top$ for $t\in[0,T]$ as in \eqref{mf-premium}.\footnote{Notice that the process $\theta^{\mathrm{mfg}}$ is consistent with Assumption \ref{asm1-habit} as a risk-premium process.} \par
Recalling Section 5.2.2 and 5.3.1, if the market risk-premium process is $\theta^{\mathrm{mfg}}$, the optimal trading strategy for agent-$i$ is given by
\[
    p^{i,*}_t := (\pi^{i,*}_t)^\top\sigma_t = Z^{i,0\|}_t + \frac{(\theta^{\mathrm{mfg}}_t)^\top}{\gamma^i} = Z^{i,0\|}_t - \frac{\hat\gamma}{\gamma^i}\mathbb{E}[\mathcal{Z}^{0\|}_t | \mathcal{F}^0],~~t\in[0,T].
\]
Here, $Z^{i,0}$ is a solution of the following \textit{non}-mean field BSDE:
\begin{equation}
    \begin{split}
        \label{asymp-non-MF-BSDE}
        Y^i_t
        = 
        F^i_T &+ \int_t^T \Bigl\{-Z^{i,0\|}_s\theta^{\mathrm{mfg}}_s - \frac{|\theta^{\mathrm{mfg}}_s|^2}{2\gamma^i} + \frac{\gamma^i}{2}(|Z^{i,0\perp}_s|^2 + |Z^i_s|^2) -\frac{\gamma^i(1+b\zeta^i_s)}{\beta^i}Y^i_s + g^i_s\Bigr\}ds\\
        & - \int_t^T Z^{i,0}_s dW^0_s - \int_t^T Z^i_s dW^i_s \\
        =
        F^i_T &+ \int_t^T \Bigl\{\hat\gamma Z^{i,0\|}_s\mathbb{E}[\mathcal{Z}^{0\|}_s|\mathcal{F}^0]^{\top} - \frac{\hat\gamma^2}{2\gamma^i}|\mathbb{E}[\mathcal{Z}^{0\|}_s|\mathcal{F}^0]|^2 + \frac{\gamma^i}{2}(|Z^{i,0\perp}_s|^2 + |Z^i_s|^2) -\frac{\gamma^i(1+b\zeta^i_s)}{\beta^i}Y^i_s + g^i_s\Bigr\}ds\\
            & - \int_t^T Z^{i,0}_s dW^0_s - \int_t^T Z^i_s dW^i_s,~~t\in[0,T]
    \end{split}
  \end{equation}
with
\[
    g^i_s := - \frac{\delta}{\gamma^i} + (\kappa-b)\zeta^i_s\rho_s + \frac{1 + b\zeta^i_s}{\beta^i}\Bigl\{1 + \log\Bigl(\frac{a\beta^i}{\gamma^i(1 + b\zeta^i_s)}\Bigr)+\gamma^i F^i_s\Bigr\},~~~s\in[0,T].
\]
This equation has a unique bounded solution $(Y^i,Z^{i,0},Z^i)\in\mathbb{S}^\infty\times\mathbb{H}^2_{\mathrm{BMO}}\times\mathbb{H}^2_{\mathrm{BMO}}$ by Theorem \ref{sec2-well-posed}. With the above setup, we have the main result of this section. 

\begin{thm}~(Asymptotic equilibrium)\\
    \label{Asymp}
    Let Assumptions \ref{asm1-habit} and \ref{asm3-habit} be in force. Suppose that the mean field BSDE \eqref{MF-BSDE1-habit} has a bounded solution, and that we arbitrarily choose and fix one such solution $(\mathcal{Y}^1,\mathcal{Z}^0,\mathcal{Z}^1)\in\mathbb{S}^\infty(\mathbb{P}^{0,1},\mathbb{F}^{0,1},\mathbb{R})\times \mathbb{H}^2_{\mathrm{BMO}}(\mathbb{P}^{0,1},\mathbb{F}^{0,1},\mathbb{R}^{1\times d_0})\times \mathbb{H}^2_{\mathrm{BMO}}(\mathbb{P}^{0,1},\mathbb{F}^{0,1},\mathbb{R}^{1\times d})$.
    Then, the process $\theta^{\mathrm{mfg}}$, defined by $\theta^{\mathrm{mfg}}_t := -\hat\gamma\mathbb{E}[\mathcal{Z}^{0\|}_t | \mathcal{F}^0]^\top$ for $t\in[0,T]$, clears the financial market in the large population limit in the sense that
    \begin{equation}
        \lim_{N\to\infty}\mathbb{E}\int_0^T \Bigl|\frac{1}{N}\sum_{i=1}^N \pi^{i,*}_t\Bigr|^2dt = 0,
    \end{equation}
    where $(\pi^{i,*}_t;t\in[0,T])_{i\in\mathbb{N}}$ are the agents' optimal trading strategies.
\end{thm}
\noindent
\textbf{\textit{Proof}}\\
(Step I)\\
As mentioned above, the BSDE \eqref{asymp-non-MF-BSDE} with $i=1$ has a unique solution $(Y^1,Z^{1,0},Z^1)\in\mathbb{S}^\infty\times\mathbb{H}^2_{\mathrm{BMO}}\times\mathbb{H}^2_{\mathrm{BMO}}$ by Theorem \ref{sec2-well-posed}. Since $(\mathcal{Y}^1,\mathcal{Z}^0,\mathcal{Z}^1)$ obviously solves the same equation, the uniqueness implies $(Y^1,Z^{1,0},Z^1)=(\mathcal{Y}^1,\mathcal{Z}^0,\mathcal{Z}^1)$. In particular, we have $\theta^{\mathrm{mfg}}_t = -\hat\gamma\mathbb{E}[Z^{1,0\|}_t | \mathcal{F}^0]^\top$ for $t\in[0,T]$.

Moreover, the (strong) uniqueness implies that, for each $i\in\mathbb{N}$, there exists a measurable function $\Phi$ such that 
\[
    (Y_t^i,Z_t^{i,0},Z_t^i)_{t\in[0,T]}=\Phi (W^0,W^i,\xi^i,\gamma^i,\beta^i,X^i_0,F^i,\theta^{\mathrm{mfg}}),~~~ \mathbb{P}^{0,1}\text{-}\mathrm{a.s.},
\] 
by Yamada-Watanabe's theorem (See, for example, Carmona \& Delarue \cite{carmonaProbabilisticTheoryMean2018a} [Theorem 1.33]). It then follows that $\{(Y_t^i,Z_t^{i,0},Z_t^i);t\in[0,T]\}_{i\in\mathbb{N}}$ are $\mathcal{F}^0$-conditionally independently and identically distributed.\\

\noindent
(Step II)\\
Since $\pi^{i,*}_t=(\sigma_t\sigma_t^\top)^{-1}\sigma_t (p^{i,*}_t)^\top$ for $t\in[0,T]$ and $|(\sigma_t\sigma_t^\top)^{-1}\sigma_t|\leq C$ uniformly in $t$ by Assumption \ref{asm1-habit}, we have
\begin{equation}
    \begin{split}
        \label{pi-and-p}
        \mathbb{E}\int_0^T \Bigl|\frac{1}{N}\sum_{i=1}^N \pi^{i,*}_t\Bigr|^2dt
        &\leq
        C\mathbb{E}\int_0^T \Bigl|\frac{1}{N}\sum_{i=1}^N p^{i,*}_t\Bigr|^2dt.
    \end{split}
\end{equation}
for all $N\in\mathbb{N}$. Moreover, it is clear that
\begin{equation*}
    \begin{split}
       \frac{1}{N}\sum_{i=1}^N p^{i,*}_t
       &=
       \frac{1}{N}\sum_{i=1}^N \Bigl(Z^{i,0\|}_t- \mathbb{E}[Z^{1,0\|}_t | \mathcal{F}^0]\Bigr) + \frac{1}{N}\sum_{i=1}^N \Bigl(1-\frac{\hat\gamma}{\gamma^i}\Bigr)\mathbb{E}[Z^{1,0\|}_t | \mathcal{F}^0].
    \end{split}
\end{equation*}
Then, we have the following estimate:
\begin{equation*}
    \begin{split}
       &\mathbb{E}\int_0^T \Bigl|\frac{1}{N}\sum_{i=1}^N p^{i,*}_t\Bigr|^2dt \\
       &\leq
       2\mathbb{E}\int_0^T \Bigl|\frac{1}{N}\sum_{i=1}^N \Bigl(Z^{i,0\|}_t- \mathbb{E}[Z^{1,0\|}_t | \mathcal{F}^0]\Bigr)\Bigr|^2 dt + 2\mathbb{E}\int_0^T \Bigl|\frac{1}{N}\sum_{i=1}^N \Bigl(1-\frac{\hat\gamma}{\gamma^i}\Bigr)\mathbb{E}[Z^{1,0\|}_t | \mathcal{F}^0]\Bigr|^2 dt\\
       &=
       2\mathbb{E}\int_0^T \Bigl|\frac{1}{N}\sum_{i=1}^N \Bigl(Z^{i,0\|}_t- \mathbb{E}[Z^{1,0\|}_t | \mathcal{F}^0]\Bigr)\Bigr|^2 dt + 2\mathbb{E}\Bigl[\Bigl|\frac{1}{N}\sum_{i=1}^N \Bigl(1-\frac{\hat\gamma}{\gamma^i}\Bigr)\Bigr|^2\Bigr]\mathbb{E}\int_0^T\Bigl|\mathbb{E}[Z^{1,0\|}_t | \mathcal{F}^0]\Bigr|^2 dt\\
       &\leq
       \frac{2}{N^2}\sum_{i=1}^N\mathbb{E}\int_0^T \Bigl|Z^{i,0\|}_t- \mathbb{E}[Z^{1,0\|}_t | \mathcal{F}^0]\Bigr|^2 dt + \frac{2}{N^2}\sum_{i=1}^N \mathbb{E}\Bigl[\Bigl|1-\frac{\hat\gamma}{\gamma^i}\Bigr|^2\Bigr]\mathbb{E}\int_0^T|Z^{1,0\|}_t |^2 dt\\
       &\leq
       \frac{4}{N}\Bigl(1+\frac{\hat\gamma^2}{\underline{\gamma}^2}\Bigr)\|Z^{1,0\|}\|^2_{\mathbb{H}^2}\\
       &\to
       0~~~~(N\to\infty).
    \end{split}
\end{equation*}
Here, we used the fact that $(\gamma^i)_{i\in\mathbb{N}}$ are i.i.d. random variables and that $(Z_t^{i,0})_{i\in\mathbb{N}}$ are $\mathcal{F}^0$-conditionally i.i.d. Together with \eqref{pi-and-p}, we get the desired result. $\square$

\section{Special Solution for the Exponential Quadratic Gaussian Model} \label{Section 4}
In this section, we reformulate the equilibrium model via the exponential quadratic Gaussian (EQG) framework. In the previous section, we made several strong assumptions to prove the existence of bounded solutions to the mean field BSDE \eqref{MF-BSDE1-habit}. 
The EQG framework, on the other hand, provides a good example in which unbounded solutions can be obtained under certain conditions. Since this framework allows us to have a semi-explicit representation of the solutions, it will help us to carry out detailed numerical analysis in future work.
\subsection{Reformulation of the Equilibrium Model}
Suppose there are infinitely many agents in the common financial market. In this section, we assume that the coefficients of absolute risk aversion $(\gamma^i,\beta^i)_{i\in\mathbb{N}}$ are common to all agents and we hereafter denote their common values by $(\gamma,\beta)\in\mathbb{R}_{++}\times\mathbb{R}_{++}$. 
Since they are no longer random variables, we need a slight modification to the definition of the relevant probability spaces.\\

\noindent
(1) We denote by $(\Omega^0,\mathcal{F}^0,\mathbb{P}^0)$ a complete probability space with complete and right-continuous filtration $\mathbb{F}^0:=(\mathcal{F}^0_t)_{t\in[0,T]}$ generated by a $d_0$-dimensional standard Brownian motion $W^0:=(W^0_t)_{t\in[0,T]}$ with $\mathcal{F}^0:=\mathcal{F}^0_T$.
Also, we denote by $(\Omega^i,\mathcal{F}^i,\mathbb{P}^i)$ ($i\in\mathbb{N}$) a complete probability space with complete and right-continuous filtration $\mathbb{F}^i:=(\mathcal{F}^i_t)_{t\in[0,T]}$, generated by a $d$-dimensional standard Brownian motion $W^i:=(W^i_t)_{t\in[0,T]}$ and a $\sigma$-algebra $\sigma(\xi^i,X^i_0,x^i_0)$, where the completion of the latter gives $\mathcal{F}^i_0$. We set $\mathcal{F}^i:=\mathcal{F}^i_T$. 
Here, $(\xi^i,X^i_0)$ are $\mathbb{R}$-valued random variables and $x^i_0$ is an $\mathbb{R}^d$-valued random variable. \\

\noindent
(2) We denote by $(\Omega^{0,i},\mathcal{F}^{0,i},\mathbb{P}^{0,i})$ ($i\in\mathbb{N}$) a complete probability space with $\Omega^{0,i} := \Omega^0 \times \Omega^i$ and with $(\mathcal{F}^{0,i},\mathbb{P}^{0,i})$, the completion of $(\mathcal{F}^0 \otimes \mathcal{F}^i,\mathbb{P}^{0}\otimes \mathbb{P}^{i})$. 
We denote by $\mathbb{F}^{0,i}:=(\mathcal{F}^{0,i}_t)_{t\in[0,T]}$ the complete and right-continuous augmentation of $(\mathcal{F}_t^0 \otimes \mathcal{F}_t^i)_{t\in[0,T]}$.\\

\noindent
(3) Let $(\Omega,\mathcal{F},\mathbb{P})$ be a complete probability space defined by $\Omega:=\prod_{i=0}^\infty\Omega^i$ and $(\mathcal{F},\mathbb{P})$, the completion of $\Bigl(\bigotimes_{i=0}^\infty\mathcal{F}^i,\bigotimes_{i=0}^\infty\mathbb{P}^i\Bigr)$. The filtration $\mathbb{F}:=(\mathcal{F}_t)_{t\in[0,T]}$ is the complete and right-continuous augmentation of $(\bigotimes_{i=0}^\infty\mathcal{F}^{i}_t)_{t\in[0,T]}$.\\

Let us first give a new assumption on the market as follows.
\begin{asm} 
    \label{asmEQG-market}~\\
    \textup{(i)} The risk-free interest rate is zero.\\
    \textup{(ii)} There are $n\in\mathbb{N}$ non-dividend paying risky stocks whose price dynamics, represented by an $n$-dimensional vector, are given by
    \begin{equation}
        \begin{split}
            \label{stock price-EQG-habit}
            S_t&= S_0 + \int_0^t \mathrm{diag}(S_r)(\mu_rdr + \sigma_r dW^0_r),~~t\in[0,T],
        \end{split}
    \end{equation}
    for $S_0\in\mathbb{R}^n_{++}$, $\mu := (\mu_t)_{t\in[0,T]}\in\mathbb{H}^2(\mathbb{P}^{0},\mathbb{F}^0,\mathbb{R}^n)$ and $\sigma :=(\sigma_t)_{t\in[0,T]}\in\mathbb{L}^\infty(\mathbb{P}^{0},\mathbb{F}^0,\mathbb{R}^{n\times d_0})$. We also assume $n\leq d_0$.\\
    \textup{(iii)} The process $(\sigma_t)_{t\in[0,T]}$ is of the form $\sigma_t = (\hat{\sigma}_t, \check{\sigma}_t)$ for each $t\in[0,T]$, where $(\hat{\sigma}_t)_{t\in[0,T]}\in\mathbb{L}^\infty(\mathbb{P}^{0},\mathbb{F}^0,\mathbb{R}^{n\times n})$ is a process such that $\hat{\sigma}_t$ is invertible for all $t\in[0,T]$ and $(\check{\sigma}_t)_{t\in[0,T]}\in\mathbb{L}^\infty(\mathbb{P}^{0},\mathbb{F}^0,\mathbb{R}^{n\times (d_0-n)})$. Moreover, $(\sigma_t)_{t\in[0,T]}$ satisfies
    \[
        \underline{\lambda}I_n\leq \sigma_t\sigma_t^\top\leq\overline{\lambda}I_n,~~~~dt\otimes \mathbb{P}^0\text{-}\mathrm{a.e.}
    \]
    for some positive constants $0<\underline{\lambda}<\overline{\lambda}$ and $I_n$, an identity matrix of size $n$.\\
    \textup{(iv)} The risk-premium process $\theta\in\mathbb{L}^0(\mathbb{F}^0,\mathbb{R}^{d_0})$, defined by $\theta_t = \sigma_t^\top(\sigma_t\sigma_t^\top)^{-1}\mu_t$ for $t\in[0,T]$, is a process such that the Dol\'{e}ans-Dade exponential $\displaystyle\Bigl\{\mathcal{E}\Bigl(-\int_0^\cdot \theta_s^\top dW^0_s\Bigr)_t; t\in[0,T]\Bigr\}$ is a martingale of class $\mathcal{D}$.
\end{asm}

\begin{rem}~\\
    \textup{(i)} Under Assumption \ref{asmEQG-market} (iii), the linear subspace $L_t$ defined in Definition \ref{subspace-L} is spanned by the first $n$-standard basis vectors of $\mathbb{R}^{1\times d_0}$ for all $t\in[0,T]$. We use the symbol $L$ instead of $L_t$ in this section. In addition, we denote by $\Pi$ the orthogonal projection of $\mathbb{R}^{1\times d_0}$ onto $L$.\\
    \textup{(ii)} Unlike Assumption \ref{asm1-habit}, the process $\mu$ is no longer in $\mathbb{H}^2_{\mathrm{BMO}}$ and thus so is $\theta$. Despite this, the well-posedness of the stock price process $(S_t)_{t\in[0,T]}$ can still be shown by changing the original measure $\mathbb{P}^0$ to $\mathbb{Q}$, the risk neutral measure defined by \eqref{risk-neutral-habit}, which is possible thanks to Assumption \ref{asmEQG-market} (iv).
\end{rem}

\begin{asm} ~\\
    \label{asm6-habit}
    \textup{(i)} For each $i\in\mathbb{N}$, $\xi^i$ and $X^i_0$ are $\mathbb{R}$-valued, $\mathcal{F}^i_0$-measurable, and normally-distributed random variables representing agent-$i$'s initial wealth and initial consumption habit, respectively. $x^i_0$ is an $\mathbb{R}^d$-valued, $\mathcal{F}^i_0$-measurable, and normally-distributed random variable.
    
    \noindent
    \textup{(ii)} The random variables $\xi^i,X^i_0$ and $x_0^i$ are mutually independent for each $i\in\mathbb{N}$ and $(\xi^i,X^i_0,x_0^i)_{i\in\mathbb{N}}$ have the same distribution, i.e. they are independently and identically distributed on $(\Omega,\mathcal{F},\mathbb{P})$.
    
    \noindent
    \textup{(iii)} $(\gamma, \beta)\in\mathbb{R}_{++}\times\mathbb{R}_{++}$ are the coefficients of absolute risk aversion for agents' net wealth and consumption, respectively. In particular, they are common to all agents. 

    \noindent
    \textup{(iv)} The habit trend $\rho:[0,T]\to\mathbb{R}$ is a continuous function of time.

    \noindent
    \textup{(v)} For each $i\in\mathbb{N}$, the liability process $(F^i_t;t\in[0,T])_{i\in\mathbb{N}}$ is $\mathbb{R}$-valued and $\mathbb{F}^{0,i}$-progressively measurable, which is given by a quadratic form\footnote{The symbol $\langle \cdot,\cdot\rangle$ denotes the Euclidean inner product, i.e., $\langle x, y\rangle:=x^\top y$ for $x,y\in\mathbb{R}^n$.}
    \begin{equation}
        \begin{split}
        \label{EQG-F-habit}
            F^i_t 
            :=& 
            \frac{1}{2} \langle A^{F}_{00}(t)x^0_t,x^0_t\rangle + \frac{1}{2} \langle A^{F}_{11}(t)x^i_t,x^i_t\rangle + \langle A^{F}_{10}(t)x^0_t,x^i_t\rangle\\
            & + \langle B^{F}_0(t),x^0_t\rangle + \langle B_1^{F}(t),x^i_t\rangle + C^{F}(t),~~~t\in[0,T],
        \end{split}
    \end{equation}
    for $(A^{F}_{00}, A^{F}_{11}, A^{F}_{10},B^F_0,B^F_1,C^F)\in\mathcal{C}([0,T];\mathbb{M}_{d_0})\times\mathcal{C}([0,T];\mathbb{M}_{d})\times\mathcal{C}([0,T];\mathbb{R}^{d\times d_0})\times\mathcal{C}([0,T];\mathbb{R}^{d_0})\times\mathcal{C}([0,T];\mathbb{R}^{d})\times\mathcal{C}([0,T];\mathbb{R})$ and the Gaussian factor processes $(x^0,x^i)\in\mathbb{L}^0(\mathbb{F}^0,\mathbb{R}^{d_0})\times\mathbb{L}^0(\mathbb{F}^i,\mathbb{R}^{d})$ defined by
    \begin{equation*}
        \begin{split}
            x_t^0 &= x^0_0 -\int_0^t K_0(x^0_s - m_0)ds + \Sigma_0 W_t^0,~~~t\in[0,T],\\
            x_t^i &= x_0^i -\int_0^t K(x^i_s - m)ds + \Sigma W_t^i,~~~t\in[0,T]
        \end{split}
    \end{equation*}
    with\footnote{This method is still available with time-dependent deterministic and continuous coefficients $(m_0(t),m(t),K_0(t),K(t),\Sigma_0(t),\Sigma(t))$. For simplicity, however, we only consider the constant case in this chapter.} $x^0_0\in\mathbb{R}^{d_0}$, $(K_0,K)\in\mathbb{R}_{++}\times \mathbb{R}_{++}$, $(m_0,m)\in\mathbb{R}^{d_0}\times \mathbb{R}^{d}$, and $(\Sigma_0,\Sigma)\in\mathbb{R}^{d_0\times d_0}\times \mathbb{R}^{d\times d}$.

    \noindent
    \textup{(vi)} Each agent is a price taker; agent-$i$ must accept the prevailing prices as he/she lacks the market share to impact the market price.
\end{asm}

\begin{rem}
    In this model, the agents are heterogeneous in the idiosyncratic noises $(W^i)_{i\in\mathbb{N}}$, initial wealths $(\xi^i)_{i\in\mathbb{N}}$, initial habits $(X^i_0)_{i\in\mathbb{N}}$, and initial conditions $(x^i_0)_{i\in\mathbb{N}}$ for the factor processes that affect the liabilities $(F^i)_{i\in\mathbb{N}}$.
\end{rem}

The agents' problems are modelled on the probability space $(\Omega,\mathcal{F},\mathbb{P},\mathbb{F})$. For each $i\in\mathbb{N}$, agent-$i$ solves the following utility maximization problem:
\begin{equation}
    \begin{split}
        \sup_{(\pi,c)\in\mathbb{A}^i_{\mathrm{EQG}}} {U}^i(\pi,c)
    \end{split}
\end{equation}
subject to
\[
    \mathcal{W}^{i,(\pi,c)}_t =\xi^i + \int_0^t (\pi_s^\top\sigma_s\theta_s - c_s)ds + \int_0^t \pi_s^\top\sigma_s dW_s^0,~~t\in[0,T],
\]
where $\mathbb{A}^i_{\mathrm{EQG}}$ is the admissible set for agent-$i$, whose definition is to be given. The utility function ${U}^i:\mathbb{A}^i_{\mathrm{EQG}}\to\mathbb{R}$ is defined by
\begin{equation}
    \begin{split} 
    &U^i(\pi,c)\\
    &:=\mathbb{E}\Bigl[-\exp\Bigl(-\delta T-\gamma(\mathcal{W}^{i,(\pi,c)}_T-F^i_T)\Bigr) -a \int_0^T \exp\Bigl(-\delta t-\gamma(\mathcal{W}^{i,(\pi,c)}_t-F^i_t)-\beta(c_t-X_t^{i,c})\Bigr)dt\Bigr],
\end{split}
\end{equation}  
with some common parameters $a,\delta>0$. The process $X^{i,c}$ represents the agent-$i$'s consumption habits and is defined by
\begin{equation}
    \begin{split} 
        \label{EQG-habit}
        X_t^{i,c} = X^i_0 + \int_0^t \{-\kappa(X^{i,c}_s-\rho_s) + b(c_s-\rho_s)\}ds, ~~t\in[0,T]
    \end{split}
\end{equation}
for some constants $\kappa,b>0$, which are also common to all agents. 
As usual, by setting $(p_t)_{t\in[0,T]}:=(\pi^\top_t\sigma_t)_{t\in[0,T]}$, the utility maximization problem can be equivalently written as
\begin{equation*}
    \begin{split}
        \sup_{(p,c)\in\mathcal{A}^i_{\mathrm{EQG}}} \widetilde{U}^i(p,c)
    \end{split}
\end{equation*}
subject to
\begin{equation}
    \begin{split} 
    \mathcal{W}^{i,(p,c)}_t =\xi^i + \int_0^t (p_s\theta_s - c_s)ds + \int_0^t p_s dW_s^0,~~t\in[0,T],
    \end{split}
\end{equation}
where the set $\mathcal{A}^i_{\mathrm{EQG}}$ is defined by $\mathcal{A}_{\mathrm{EQG}}^i:=\{(p,c)=(\pi^\top\sigma,c);(\pi,c)\in\mathbb{A}^i_{\mathrm{EQG}}\}$ and the objective function $\widetilde{U}^i:\mathcal{A}^i_{\mathrm{EQG}}\to\mathbb{R}$ is defined by
\begin{equation}
    \begin{split} 
    \label{tilde-U}
    &\widetilde{U}^i(p,c)\\
    &:=\mathbb{E}\Bigl[-\exp\Bigl(-\delta T-\gamma(\mathcal{W}^{i,(p,c)}_T-F^i_T)\Bigr) -a \int_0^T \exp\Bigl(-\delta t-\gamma(\mathcal{W}^{i,(p,c)}_t-F^i_t)-\beta(c_t-X_t^{i,c})\Bigr)dt\Bigr].
    \end{split}
\end{equation}

Under these assumptions, we define the process $R^{i,(p,c)}$ in analogy with Section 5.2.2 in the following way: for each $i\in\mathbb{N}$, we set
\begin{equation*}
    \begin{split}
        R^{i,(p,c)}_t := &-\exp\Bigl(-\delta t-\gamma(\mathcal{W}^{i,(p,c)}_t-y^i_t-\zeta_tX_t^{i,c})\Bigr)\\& -a \int_0^t \exp\Bigl(-\delta s-\gamma(\mathcal{W}^{i,(p,c)}_s-F^i_s)-\beta(c_s-X_s^{i,c})\Bigr)ds,~~t\in[0,T],
    \end{split}
\end{equation*}
where the process $(y^i_t)_{t\in[0,T]}$ is a solution to the BSDE \footnote{The solution is denoted by lower case letters in order to avoid confusion with the solution of the mean field BSDE \eqref{EQG-mfBSDE}, which is denoted by $(Y,Z^{i,0},Z^i)$.}:
\begin{equation}
    \begin{split}
        \label{std BSDE}
        y^i_t =& F^i_T + \int_t^T \Bigl\{-z^{i,0\|}_s\theta_s - \frac{|\theta_s|^2}{2\gamma} + \frac{\gamma}{2}(|z^{i,0\perp}_s|^2 + |z^i_s|^2) -\frac{\gamma(1+b\zeta_s)}{\beta}y^i_s + g^i_s\Bigr\}ds\\
        & - \int_t^T z^{i,0}_s dW^0_s - \int_t^T z^i_s dW^i_s
    \end{split}
  \end{equation}
with
\[
    g^i_s := - \frac{\delta}{\gamma} + (\kappa-b)\zeta_s\rho_s + \frac{1 + b\zeta_s}{\beta}\Bigl\{1 + \log\Bigl(\frac{a\beta}{\gamma(1 + b\zeta_s)}\Bigr)+\gamma F^i_s\Bigr\},
\]
and
\begin{equation*}
    \zeta_t := \frac{e^{(\delta^+-\delta^-)(T-t)}-1}{\delta^+-\delta^-e^{(\delta^+-\delta^-)(T-t)}},
\end{equation*}
where
\begin{equation*}
    \delta^{\pm}:=-A\pm\sqrt{A^2+B},~~~A:=\frac{1}{2}\Bigl(\kappa-b+\frac{\gamma}{\beta}\Bigr),~~~B:=\frac{\gamma b}{\beta},
\end{equation*}
for $t\in[0,T]$. Moreover, we say that the process $R^{i,(p,c)}$ satisfies the condition-R if all conditions in Definition \ref{condition-R} with ``1'' replaced by ``$i$'' and $(\gamma^i,\beta^i)$ replaced by $(\gamma,\beta)$ hold. In order to work within this framework, we further need to modify the notion of admissibility.
\begin{dfn} (Admissible space for an EQG model)\\
    For each $i\in\mathbb{N}$, the admissible space $\mathbb{A}^i_{\mathrm{EQG}}$ is the set of $\mathbb{F}^{0,i}$-progressively measurable strategies $(\pi,c)\in\mathbb{H}^2(\mathbb{P}^{0,i},\mathbb{F}^{0,i},\mathbb{R}^{n})\times\mathbb{H}^2(\mathbb{P}^{0,i},\mathbb{F}^{0,i},\mathbb{R})$ which make the utility function finite (namely $U^i(\pi,c)>-\infty$) and the set $\{R^{i,(p,c)}_\tau;\tau\in\mathcal{T}^{0,i}\}$ uniformly integrable. 
\end{dfn}

We shall write the admissible space by $\mathcal{A}^i_{\mathrm{EQG}}(\theta)$ when we want to emphasize its dependence on the risk-premium process $\theta$. In a similar way as in Section \ref{Section 3}, the market clearing condition motivates us to study the following mean field BSDE defined on the filtered probability space $(\Omega^{0,i},\mathcal{F}^{0,i},\mathbb{P}^{0,i},\mathbb{F}^{0,i})$ for each $i\in\mathbb{N}$:
\begin{equation}
    \begin{split}
        \label{EQG-mfBSDE}
            Y^i_t&= F^i_T + \int_t^T f^i(s,Y^i_s,Z^{i,0}_s,Z^i_s)ds - \int_t^T Z^{i,0}_s dW^0_s - \int_t^T Z^i_s dW^i_s,~~~t\in[0,T],
    \end{split}
  \end{equation}
where (note that we have $\hat\gamma = \gamma$ by Assumption \ref{asm6-habit} (iii))
\begin{equation*}
    \begin{split}
        &f^i(s,Y^i_s,Z^{i,0}_s,Z^i_s) \\
        &= 
        \gamma Z^{i,0\|}_s\mathbb{E}[Z^{i,0\|}_s|\mathcal{F}^0]^{\top} - \frac{\gamma}{2}|\mathbb{E}[Z^{i,0\|}_s|\mathcal{F}^0]|^2 + \frac{\gamma}{2}(|Z^{i,0\perp}_s|^2 + |Z^i_s|^2) -\frac{\gamma(1+b\zeta_s)}{\beta}Y^i_s + g^i_s.
    \end{split}
  \end{equation*}
By completing the square, the driver $f^i$ can be written as
\begin{equation*}
    \begin{split}
        &f^i(s,Y^i_s,Z^{i,0}_s,Z^i_s) \\
        &=
        -\frac{\gamma}{2}\Bigl|\mathbb{E}[Z^{i,0\|}_s|\mathcal{F}^0] - Z^{i,0\|}_s\Bigr|^2 + \frac{\gamma}{2}(|Z^{i,0}_s|^2 + |Z^i_s|^2) -\frac{\gamma(1+b\zeta_s)}{\beta}Y^i_s + {g}^i_s.
    \end{split}
  \end{equation*}

\subsection{Mean Field BSDE and the System of ODEs}
We now derive a system of ordinary differential equations (ODEs) which provides a solution of the mean field BSDE through the EQG modelling. 
Our approach basically follows Fujii \& Takahashi \cite{fujiiMakingMeanvarianceHedging2014} [Section 5], which proposes the method of associating the solution of the quadratic growth BSDE with the Riccati matrix equation. 
As a heuristic argument, if $Y^i$ is a quadratic form of $(x^0,x^i)$, its drift term is expected to be a quadratic form of $(x^0,x^i)$, and its diffusion terms are expected to be affine in $(x^0,x^i)$ by applying Ito formula. On the other hand, as the driver $f^i$ of the BSDE \eqref{EQG-mfBSDE} is quadratic in $(Z^{i,0},Z^i)$ and is linear in $Y^i$, it is anticipated that $f^i$ is a quadratic form of $(x^0,x^i)$ as well.
These observations imply that such an ansatz for $Y^i$ seems to be consistent, and we thus search for a solution of the form:
\begin{equation}
    \begin{split}
    \label{Y-ansatz}
    Y^i_t =& \frac{1}{2} \langle A^i_{00}(t)x^0_t,x^0_t\rangle + \frac{1}{2} \langle A^i_{11}(t)x^i_t,x^i_t\rangle + \langle A^i_{10}(t)x^0_t,x^i_t\rangle\\
    & + \langle B^i_0(t),x^0_t\rangle + \langle B^i_1(t),x^i_t\rangle + C^i(t),~~~t\in[0,T]
    \end{split}
\end{equation}
for some processes $(A^i_{00},A^i_{11},A^i_{10},B^i_0,B^i_1,C^i):[0,T]\times\Omega\to\mathbb{M}_{d_0}\times\mathbb{M}_{d}\times\mathbb{R}^{d\times d_0}\times\mathbb{R}^{d_0}\times\mathbb{R}^{d}\times\mathbb{R}$, all of which are to be determined. At this moment, let us temporarily assume that $(A^i_{00},A^i_{11},A^i_{10},B^i_0,B^i_1,C^i)$ are once continuously time-differentiable and independent of $(\xi^i,X_0^i,x^i_0,W^0,W^i)$, i.e. they are deterministic functions of time common to all agents. After deriving the relevant ODEs, we shall verify this property. Since we search for functions common to all agents, we simply write $(A_{00},A_{11},A_{10},B_0,B_1,C)$ instead of $(A^i_{00},A^i_{11},A^i_{10},B^i_0,B^i_1,C^i)$ from now on. \par
As usual, we choose agent-1 as a representative agent and omit the superscript ``1" when there is no confusion. By applying Ito formula to \eqref{Y-ansatz}, we have
\begin{equation}
    \begin{split}
        \label{Y-ito-habit}
        dY_t 
        &=
        \Bigl\{\Bigl\langle \Bigl(\frac{1}{2} \dot{A}_{00}(t) - K_0 A_{00}(t)\Bigr)x^0_t,x^0_t\Bigr\rangle + \Bigl\langle \Bigl(\frac{1}{2} \dot{A}_{11}(t) - K A_{11}(t)\Bigr)x^1_t,x^1_t\Bigr\rangle \\
        &~~~+ \Bigl\langle \Bigl( \dot{A}_{10}(t) - (K_0+K)A_{10}(t)\Bigr)x^0_t,x^1_t\Bigr\rangle + \langle \dot{B}_0(t)-K_0B_0(t) + K_0A_{00}(t)m_0 + K A_{10}(t)^\top m,x^0_t\rangle \\
        &~~~+ \langle \dot{B}_1(t)-K B_1(t) + K A_{11}(t)m + K_0A_{10}(t)m_0,x^1_t\rangle \\
        &~~~+ \Bigl.\dot{C}(t) + \langle K_0 B_0(t) ,m_0 \rangle + \langle K B_1(t) ,m \rangle +\frac{1}{2}\mathrm{tr}[A_{00}(t)\Sigma_0\Sigma_0^\top] + \frac{1}{2}\mathrm{tr}[A_{11}(t)\Sigma\Sigma^\top] \Bigr\}dt \\
        &~~~+   \langle\Sigma_0^\top (A_{00}(t)x^0_t + A_{10}(t)^\top x^1_t + B_0(t)) ,dW^0_t \rangle + \langle \Sigma^\top (A_{10}(t)x^0_t + A_{11}(t)x^1_t + B_1(t)) , dW^1_t \rangle.
    \end{split}
\end{equation}
In order for $Y$ given in \eqref{Y-ansatz} to be the solution to the mean field BSDE \eqref{EQG-mfBSDE}, we must have
\begin{equation*}
    \begin{split}
        &Z_t^0 = \Bigl\{\Sigma_0^\top (A_{00}(t)x^0_t + A_{10}(t)^\top x^1_t + B_0(t))\Bigr\}^\top,~~~t\in[0,T],\\
        &Z_t^1 = \Bigl\{\Sigma^\top (A_{10}(t)x^0_t + A_{11}(t)x^1_t + B_1(t))\Bigr\}^\top,~~~t\in[0,T].
    \end{split}
\end{equation*}
To deal with the process $Z^{0\|}$, let us write $\Sigma_0 = \hat{\Sigma}_0 + \check{\Sigma}_0$, where $\hat{\Sigma}_0, \check{\Sigma}_0 \in \mathbb{R}^{d_0\times d_0}$ are of the forms:
\[
    \hat{\Sigma}_0 = (\hat{\Sigma}^n_0~~ 0),~~\check{\Sigma}_0 = (0~~ \check{\Sigma}^{d_0-n}_0),
\]
for $\hat{\Sigma}^n_0\in\mathbb{R}^{d_0\times n}$ and $\check{\Sigma}^{d_0-n}_t\in\mathbb{R}^{d_0\times (d_0-n)}$, so that we have $\Pi(u^\top\Sigma_0)=u^\top \hat{\Sigma}_0$ for any $u\in\mathbb{R}^{d_0}$. In addition, it is easy to see
\begin{equation*}
    \begin{split}
        &\mathbb{E}[x_t^0 |\mathcal{F}^0] = x_t^0,~~~t\in[0,T],\\
        &\mu^1_t := \mathbb{E}[x_t^1 |\mathcal{F}^0] = \mathbb{E}[x_t^1] = \mathbb{E}[x_0^1]e^{-K t}+m(1-e^{-K t}),~~~t\in[0,T].
    \end{split}
\end{equation*}
Then we obtain:
\begin{equation*}
    \begin{split}
        &Z_t^{0\|} = \Bigl\{\hat{\Sigma}_0^\top (A_{00}(t)x^0_t + A_{10}(t)^\top x^1_t + B_0(t))\Bigr\}^\top,~~~t\in[0,T],\\
        &\mathbb{E}[Z_t^{0\|} |\mathcal{F}^0] = \Bigl\{\hat{\Sigma}_0^\top (A_{00}(t)x^0_t + A_{10}(t)^\top \mu^1_t + B_0(t)) \Bigr\}^\top,~~~t\in[0,T].
    \end{split}
\end{equation*}
Plugging these results into the driver $f$, we have: for $t\in[0,T]$,
\begin{equation}
    \begin{split}
        \label{f-EQG}
    &f(t,Y_t,Z^0_t,Z^1_t) \\
    &= 
    -\frac{\gamma}{2} \Bigl|\mathbb{E}[Z^{0\|}_t|\mathcal{F}^0] - Z^{0\|}_t\Bigr|^2 + \frac{\gamma}{2}(|Z^{0}_t|^2 + |Z^1_t|^2) -\frac{\gamma(1+b\zeta_t)}{\beta}(Y_t-F_t) + \widetilde{g}_t\\
    &=
    \Bigl\langle \Bigl\{\frac{\gamma}{2}\Bigl(A_{00}(t)\Sigma_0\Sigma_0^\top A_{00}(t) + A_{10}(t)^\top \Sigma\Sigma^\top A_{10}(t) \Bigr) -\frac{\gamma(1+b\zeta_t)}{2\beta}(A_{00}(t)-A^F_{00}(t))\Bigr\}x_t^0,x_t^0 \Bigr\rangle\\
    &~~~+\Bigl\langle \Bigl\{\frac{\gamma}{2}( A_{10}(t) \check{\Sigma}_0\check{\Sigma}_0^\top A_{10}(t)^\top + A_{11}(t) \Sigma\Sigma^\top A_{11}(t)) -\frac{\gamma(1+b\zeta_t)}{2\beta}(A_{11}(t)-A^F_{11}(t))\Bigr\} x_t^1 ,x_t^1 \Bigr\rangle \\
    &~~~+\Bigl\langle \Bigl\{\gamma( A_{10}(t)\Sigma_0\Sigma_0^\top A_{00}(t) + A_{11}(t) \Sigma\Sigma^\top A_{10}(t) )-\frac{\gamma(1+b\zeta_t)}{\beta}(A_{10}(t)-A^F_{10}(t)) \Bigr\}x_t^0,x_t^1 \Bigr\rangle \\
    &~~~+\Bigl\langle \gamma (A_{00}(t)\Sigma_0\Sigma_0^\top B_{0}(t) + A_{10}(t)^\top\Sigma\Sigma^\top B_{1}(t)) -\frac{\gamma(1+b\zeta_t)}{\beta}(B_{0}(t)-B^F_{0}(t)) ,x_t^0 \Bigr\rangle \\
    &~~~+\Bigl\langle \gamma (A_{10}(t)\hat{\Sigma}_0\hat{\Sigma}_0^\top A_{10}(t)^\top \mu_t^1 + A_{10}(t)\Sigma_0\Sigma_0^\top B_0(t) + A_{11}(t)\Sigma\Sigma^\top B_1(t)) - \frac{\gamma(1+b\zeta_t)}{\beta}(B_{1}(t)-B^F_{1}(t)),x_t^1 \Bigr\rangle\\
    &~~~- \frac{\gamma}{2}\langle A_{10}(t)\hat{\Sigma}_0\hat{\Sigma}_0^\top A_{10}(t)^\top \mu_t^1,\mu_t^1 \rangle + \frac{\gamma}{2}\langle \Sigma_0^\top B_0(t),\Sigma_0^\top B_0(t) \rangle + \frac{\gamma}{2}\langle \Sigma^\top B_1(t),\Sigma^\top B_1(t) \rangle \\
    &~~~-\frac{\gamma(1+b\zeta_t)}{\beta}(C(t)-C^F(t)) + \widetilde{g}_t,
    \end{split}
\end{equation}
where $\widetilde{g}$ is a deterministic and continuous function defined by:
\[
    \widetilde{g}_t := g_t - \frac{\gamma(1+b\zeta_t)}{\beta} F_t = - \frac{\delta}{\gamma} + (\kappa-b)\zeta_t\rho_t + \frac{1 + b\zeta_t}{\beta}\Bigl\{1 + \log\Bigl(\frac{a\beta}{\gamma(1 + b\zeta_t)}\Bigr)\Bigr\},~~~t\in[0,T].
\]

By matching \eqref{f-EQG} and the drift term of \eqref{Y-ito-habit} with respect to the quadratic or linear coefficients of $(x^0,x^1)$ as well as the remaining constant terms, we obtain: for $t\in[0,T]$,
\begin{equation}
    \begin{split}
        \label{Riccati eqn-habit}
        &\dot{A}_{00}(t) = -\gamma A_{00}(t)\Sigma_0\Sigma_0^\top A_{00}(t)  - \gamma A_{10}(t)^\top \Sigma\Sigma^\top A_{10}(t) + \Bigl(2K_0 + \frac{\gamma(1+b\zeta_t)}{\beta}\Bigr) A_{00}(t) - \frac{\gamma(1+b\zeta_t)}{\beta} A^F_{00}(t), \\
        &\dot{A}_{11}(t) = -\gamma A_{11}(t) \Sigma\Sigma^\top A_{11}(t)  - \gamma A_{10}(t) \check{\Sigma}_0\check{\Sigma}_0^\top A_{10}(t)^\top + \Bigl(2K + \frac{\gamma(1+b\zeta_t)}{\beta}\Bigr) A_{11}(t) - \frac{\gamma(1+b\zeta_t)}{\beta} A^F_{11}(t), \\
        &\dot{A}_{10}(t)  = -\gamma A_{10}(t)\Sigma_0\Sigma_0^\top A_{00}(t) - \gamma A_{11}(t) \Sigma\Sigma^\top A_{10}(t) + \Bigl((K_0+K)+\frac{\gamma(1+b\zeta_t)}{\beta}\Bigr)A_{10}(t) - \frac{\gamma(1+b\zeta_t)}{\beta} A^F_{10}(t),\\
        &\dot{B}_0(t)=\Bigl(- \gamma A_{00}(t)\Sigma_0\Sigma_0^\top  + \frac{\gamma(1+b\zeta_t)}{\beta} + K_0\Bigr)B_{0}(t) - \gamma A_{10}(t)^\top\Sigma\Sigma^\top B_{1}(t)\\
        &~~~~~~~~~~~~~- \frac{\gamma(1+b\zeta_t)}{\beta}B^F_{0}(t) - K_0A_{00}(t)m_0 - KA_{10}(t)^\top m,\\
        &\dot{B}_1(t)= \Bigl(-\gamma A_{11}(t)\Sigma\Sigma^\top + \frac{\gamma(1+b\zeta_t)}{\beta} + K\Bigr)B_1(t) - \gamma \Bigl(A_{10}(t)\hat{\Sigma}_0\hat{\Sigma}_0^\top A_{10}(t)^\top \mu_t^1 + A_{10}(t)\Sigma_0\Sigma_0^\top B_0(t)\Bigr) \\
        &~~~~~~~~~~~~~- \frac{\gamma(1+b\zeta_t)}{\beta}B^F_{1}(t) - KA_{11}(t)m - K_0A_{10}(t)m_0, \\
        &\dot{C}(t)= \frac{\gamma(1+b\zeta_t)}{\beta}C(t) - \frac{\gamma(1+b\zeta_t)}{\beta}C^F(t) - \frac{\gamma}{2}\langle \Sigma_0^\top B_0(t),\Sigma_0^\top B_0(t) \rangle - \frac{\gamma}{2} \langle \Sigma^\top B_1(t),\Sigma^\top B_1(t) \rangle\\
        &~~~~~~~~~~~~~- \langle K_0 B_0(t) ,m_0 \rangle - \langle K B_1(t) ,m \rangle + \frac{\gamma}{2}\langle A_{10}(t)\hat{\Sigma}_0\hat{\Sigma}_0^\top A_{10}(t)^\top \mu_t^1,\mu_t^1 \rangle \\
        &~~~~~~~~~~~~~- \frac{1}{2}\mathrm{tr}[A_{00}(t)\Sigma_0\Sigma_0^\top] - \frac{1}{2}\mathrm{tr}[A_{11}(t)\Sigma\Sigma^\top] - \widetilde{g}_t,\\
        &A_{00}(T)=A_{00}^F(T),~~ A_{11}(T)=A_{11}^F(T), ~~A_{10}(T)=A_{10}^F(T),\\
        &B_0(T)=B_0^F(T), ~~B_1(T)=B_1^F(T), ~~C(T)=C^F(T).
    \end{split}
\end{equation}
Here, the terminal conditions for $(A_{00},A_{11},A_{10},B_0,B_1,C)$ are set to satisfy $Y_T=F_T$.

\begin{rem}~\\
    \textup{(i)}  The equations for $(A_{00},A_{11},A_{10})$ are of Riccati type. In this chapter, however, we do not delve into the general well-posedness result due to its complexity.\\
    \textup{(ii)} Since the coefficients appearing in \eqref{Riccati eqn-habit} are all deterministic and in particular, independent of $(\xi^i,X_0^i,x^i_0,W^0,W^i)$, we deduce that $(A_{00},A_{11},A_{10},B_0,B_1,C)$ are deterministic functions of time and common to all agents if they exist.\\
    \textup{(iii)} By the local Lipschitz condition, the equation \eqref{Riccati eqn-habit} has a locally unique solution. Furthermore, by making $|\Sigma_0|$ and $|\Sigma|$ sufficiently small, we expect to have also a global solution since the Riccati equation for $(A_{00},A_{11},A_{10})$ becomes approximately linear.\\
    \textup{(iv)} We may possibly allow heterogeneity among the coefficients of risk aversion $(\gamma^i,\beta^i)_{i\in\mathbb{N}}$ as in Section \ref{Section 3}. However, in this case, the system of equations \eqref{Riccati eqn-habit} becomes mean field type and checking its well-posedness would be much harder.
\end{rem}

These observations result in the following theorem. 
\begin{thm}
    \label{EQG solution-habit}
    Let Assumption \ref{asmEQG-market} and \ref{asm6-habit} be in force. In addition, assume that the equation \eqref{Riccati eqn-habit} has a global solution $(A_{00},A_{11},A_{10},B_0,B_1,C)\in\mathcal{C}^1([0,T];\mathbb{M}_{d_0})\times\mathcal{C}^1([0,T];\mathbb{M}_{d})\times\mathcal{C}^1([0,T];\mathbb{R}^{d\times d_0})\times\mathcal{C}^1([0,T];\mathbb{R}^{d_0})\times\mathcal{C}^1([0,T];\mathbb{R}^{d})\times\mathcal{C}^1([0,T];\mathbb{R})$. 
    Then, for each $i\in\mathbb{N}$, the process $(Y^i,Z^{i,0},Z^i)\in\mathbb{S}^2(\mathbb{P}^{0,i},\mathbb{F}^{0,i},\mathbb{R}) \times \mathbb{S}^2(\mathbb{P}^{0,i},\mathbb{F}^{0,i},\mathbb{R}^{1\times d_0}) \times \mathbb{S}^2(\mathbb{P}^{0,i},\mathbb{F}^{0,i},\mathbb{R}^{1\times d})$, defined by
    \begin{equation}
        \begin{split}
            \label{EQG solution YZ-habit}
            &Y^i_t  = \frac{1}{2} \langle A_{00}(t)x^0_t,x^0_t\rangle + \frac{1}{2} \langle A_{11}(t)x^i_t,x^i_t\rangle + \langle A_{10}(t)x^0_t,x^i_t\rangle + \langle B_0(t),x^0_t\rangle + \langle B_1(t),x^i_t\rangle + C(t), \\
            &Z_t^{i,0} = \Bigl\{\Sigma_0^\top (A_{00}(t)x^0_t + A_{10}(t)^\top x^i_t + B_0(t))\Bigr\}^\top,~~~Z_t^i = \Bigl\{\Sigma^\top (A_{10}(t)x^0_t + A_{11}(t)x^i_t + B_1(t))\Bigr\}^\top,
        \end{split}
    \end{equation}
    for $t\in[0,T]$ solves the mean field BSDE \eqref{EQG-mfBSDE}. The solution is unique among those with the quadratic Gaussian form.
\end{thm}

\subsection{Optimality, Verification and Asymptotic Equilibrium}
Let $(Y_t^i,Z_t^{i,0},Z_t^i;t\in[0,T])_{i\in\mathbb{N}}$ be processes defined by \eqref{EQG solution YZ-habit} and suppose that they are well-defined. 
Then the process $\vartheta$, defined by\footnote{Since $(Z^{i,0}_t;t\in[0,T])_{i\in\mathbb{N}}$ have the same distribution, we can, without loss of generality, choose $Z^{1,0}$ to define $\vartheta$.}
\begin{equation}
    \begin{split}
    \label{EQG risk premium}
    \vartheta_t 
    &:=
    -\gamma \mathbb{E}[Z_t^{1,0\|} |\mathcal{F}^0]^\top,~~~t\in[0,T],
    \end{split}
\end{equation}
is expected to be the market-clearing risk-premium process in the large population limit in analogy with Section \ref{Section 3}3. However, we have
\begin{equation*}
    \begin{split}
    \vartheta_t 
    &=
    -\gamma \hat{\Sigma}_0^\top \Bigl(A_{00}(t)x^0_t + A_{10}(t)^\top \mu^1_t + B_0(t) \Bigr)\\
    &=
    -\gamma \hat{\Sigma}_0^\top \Bigl(A_{00}(t)\mathbb{E}[x^0_t] + A_{10}(t)^\top\mu^1_t + B_0(t)\Bigr)-\gamma \hat{\Sigma}_0^\top A_{00}(t)\Sigma_0\int_0^t e^{-K_0(t-s)}dW^0_s
    \end{split}
\end{equation*}
for $t\in[0,T]$, which implies that $\vartheta$ is a Gaussian process and thus $\vartheta\notin\mathbb{H}^2_{\mathrm{BMO}}$. Furthermore, since $Y^i$ and $F^i$ are given by quadratic forms of $x^0$ and $x^i$, they are unbounded processes. 
Therefore, this EQG model does not fulfil the assumptions of Section \ref{Section 3}. Despite this, if $|\mathrm{Var}(x^i_0)|$, $|\Sigma_0|$ and $|\Sigma|$ are small enough, we shall see that we can still obtain the well-posedness. Here, $\mathrm{Var}(x^i_0)$ is a covariance matrix of $x^i_0$, defined by $\mathrm{Var}(x^i_0):=\mathbb{E}[(x^i_0-\mathbb{E}[x^i_0])(x^i_0-\mathbb{E}[x^i_0])^\top]$.
The following result is well known.
\begin{lem}
    \label{crisan}
    Let $(\Omega,\mathcal{F},\mathbb{P},\mathbb{F}~(:=(\mathcal{F}_t)_{t\in[0,T]}))$ be a filtered probability space with usual conditions and $W:=(W_t)_{t\in[0,T]}$ be a standard $k$-dimensional $(\mathbb{F},\mathbb{P})$-Brownian motion. Also, let $\mathscr{X}$ be an $m$-dimensional $\mathbb{F}$-adapted process defined by
    \[
        \mathscr{X}_t = \mathscr{X}_0 + \int_0^t B(\mathscr{X}_s)ds + \int_0^t \Xi (\mathscr{X}_s)dW_s,~~t\in[0,T],
    \]
    where $B:\mathbb{R}^{m}\to\mathbb{R}^{m}$ and $\Xi :\mathbb{R}^{m}\to\mathbb{R}^{m\times k}$ are Lipschitz continuous functions and $\mathscr{X}_0\in\mathbb{L}^2(\mathbb{P},\mathcal{F}_0,\mathbb{R}^m)$. Moreover, let $h:\mathbb{R}^{m}\to\mathbb{R}^{k}$ be a Borel-measurable function satisfying $|h(x)|^2 \leq C(1+|x|^2)$ for all $x\in\mathbb{R}^{m}$ and some constant $C>0$. Then, the Dol\'{e}ans-Dade exponential $\displaystyle\Bigl\{\mathcal{E}\Bigl(\int_0^\cdot h(\mathscr{X}_s)^\top dW_s\Bigr)_t;t\in[0,T]\Bigr\}$ is a martingale of class $\mathcal{D}$.
\end{lem}
\noindent
\textbf{\textit{Proof}}\\
See Bain \& Crisan \cite{bain_fundamentals_2009} [Exercise 3.11]. $\square$

\begin{prop}
    Let Assumption \ref{asmEQG-market} and \ref{asm6-habit} be in force. In addition, assume that the equation \eqref{Riccati eqn-habit} has a global solution $(A_{00},A_{11},A_{10},B_0,B_1,C)\in\mathcal{C}^1([0,T];\mathbb{M}_{d_0})\times\mathcal{C}^1([0,T];\mathbb{M}_{d})\times\mathcal{C}^1([0,T];\mathbb{R}^{d\times d_0})\times\mathcal{C}^1([0,T];\mathbb{R}^{d_0})\times\mathcal{C}^1([0,T];\mathbb{R}^{d})\times\mathcal{C}^1([0,T];\mathbb{R})$. 
    Then, the Dol\'{e}ans-Dade exponential $\displaystyle\Bigl\{\mathcal{E}\Bigl(-\int_0^\cdot \vartheta_s^\top dW_s^0\Bigr)_t;t\in[0,T]\Bigr\}$ is a martingale of class $\mathcal{D}$, where the process $\vartheta\in\mathbb{L}^0(\mathbb{F}^0,\mathbb{R}^{d_0})$ is defined by \eqref{EQG risk premium}.
\end{prop}
\noindent
\textbf{\textit{Proof}}\\
This is a direct result of Lemma \ref{crisan}. $\square$

This proposition particularly shows that the process $\vartheta$ is consistent with Assumption \ref{asmEQG-market} as a risk-premium process. With these preparations, we can recover the corresponding result of Section \ref{Section 2}.
\begin{thm} (Optimality and verification)\\
    \label{EQG-verification-habit}
    Let Assumption \ref{asmEQG-market} and \ref{asm6-habit} be in force. Assume further that the equation \eqref{Riccati eqn-habit} has a global solution $(A_{00},A_{11},A_{10},B_0,B_1,C)\in\mathcal{C}^1([0,T];\mathbb{M}_{d_0})\times\mathcal{C}^1([0,T];\mathbb{M}_{d})\times\mathcal{C}^1([0,T];\mathbb{R}^{d\times d_0})\times\mathcal{C}^1([0,T];\mathbb{R}^{d_0})\times\mathcal{C}^1([0,T];\mathbb{R}^{d})\times\mathcal{C}^1([0,T];\mathbb{R})$. 
    Then, there exists a constant $\varsigma > 0$ such that, as long as $|\Sigma_0|^2\lor |\Sigma|^2 \lor |\mathrm{Var}(x^1_0)| < \varsigma$, the process $(p^{i,*},c^{i,*})$, defined by
    \begin{equation}
        \begin{split}
            \label{EQG-optimal-habit}
            p^{i,*}_t &:= (\pi^{i,*}_t)^\top\sigma_t := Z^{i,0\|}_t + \frac{\vartheta_t^\top}{\gamma},~~~t\in[0,T], \\
            c^{i,*}_t &= X^{i,c^*}_t + \frac{1}{\beta}\Bigl\{\log\Bigl(\frac{a\beta}{\gamma(1+b\zeta_t)}\Bigr)-\gamma(Y^i_t-F^i_t+\zeta_tX^{i,c^*}_t)\Bigr\},~~~t\in[0,T],
        \end{split}
    \end{equation}
    belongs to $\mathcal{A}^i_{\mathrm{EQG}}(\vartheta)$ and is an optimal strategy for agent-$i$ for each $i\in\mathbb{N}$. 
    Here, the process $X^{i,c}$ represents the agent's consumption habit \eqref{EQG-habit}, the process $(Y^i,Z^{i,0},Z^i)\in\mathbb{S}^2(\mathbb{P}^{0,i},\mathbb{F}^{0,i},\mathbb{R}) \times \mathbb{S}^2(\mathbb{P}^{0,i},\mathbb{F}^{0,i},\mathbb{R}^{1\times d_0}) \times \mathbb{S}^2(\mathbb{P}^{0,i},\mathbb{F}^{0,i},\mathbb{R}^{1\times d})$ is given by \eqref{EQG solution YZ-habit} and the market risk-premium process $\vartheta$ is defined by \eqref{EQG risk premium}.
\end{thm}
\begin{rem}
    Note that the strategy $(p^{i,*},c^{i,*})$ given above may not be the unique optimal strategy for agent-$i$ under the risk-premium process $\vartheta$. This is because the BSDE \eqref{std BSDE} may have a solution outside of the quadratic Gaussian form.
\end{rem}
\noindent
\textbf{\textit{Proof}}\\
In this proof, we denote the general nonnegative constant by $\widetilde{C}$ to avoid confusion with the function $C$, which is a part of the solution to the ODE \eqref{Riccati eqn-habit}. By the definition of the process $R^{i,(p,c)}$ and the argument in the proof of Theorem \ref{verification}, the process $R^{i,(p^*,c^*)}$ is a local martingale and thus the optimality follows once $(p^{i,*},c^{i,*})\in\mathcal{A}^i_{\mathrm{EQG}}(\vartheta)$ is achieved. 
It then suffices to show that the function $\widetilde{U}^i(p^{i,*},c^{i,*})$ defined in \eqref{tilde-U} has a finite value and that the process $R^{i,(p^*,c^*)}$ is of class $\mathcal{D}$.

Let us write
\begin{equation*}
    \begin{split}
        \phi^0_t := \int_0^t e^{-K_0(t-s)}dW^0_s,~~~\phi^i_t := \int_0^t e^{-K(t-s)}dW^i_s,~~t\in[0,T].
    \end{split}
\end{equation*}
Then we have
\begin{equation*}
    \begin{split}
    x^0_t =  x_0^0e^{-K_0 t}+m_0(1-e^{-K_0 t}) + \Sigma_0 \phi^0_t,~~~x^i_t =  x^i_0e^{-Kt} + m(1-e^{-Kt}) + \Sigma \phi^i_t,~~t\in[0,T],
    \end{split}
\end{equation*}
and in particular, $|x^0_t|^2 + |x^i_t|^2 \leq \widetilde{C}(1+|x^i_0|^2+|\Sigma_0|^2|\phi^0_t|^2+|\Sigma|^2|\phi^i_t|^2)$ for all $t\in[0,T]$. We shall show that there exists a constant $\eta>0$ such that
\begin{equation}
    \label{u.i.-EQG}
    \sup_{t\in[0,T]}\mathbb{E}\Bigl[\exp\Bigl(-(1+\eta)\gamma \mathcal{W}^{i,(p^*,c^*)}_t + (1+\eta)M \Bigl\{|F^i_t| + |Y^i_t| + |c^{i,*}_t| + |X^{i,c^*}_t|\Bigr\}\Bigr)\Bigr] <\infty,
\end{equation}
where $M$ is a constant satisfying $M\geq \max\{\gamma, \beta, \sup_{t\in[0,T]}|\gamma\zeta_t|\}$. If this is the case, we clearly have $\widetilde{U}^i(p^{i,*},c^{i,*})>-\infty$. 
Moreover, Jensen's inequality and Doob's submartingale inequality yields
\begin{equation}
    \begin{split}
    &\mathbb{E}\Bigl[\sup_{t\in[0,T]}|R^{i,(p^*,c^*)}_t|\Bigr]^{1+\eta}\leq \mathbb{E}\Bigl[\sup_{t\in[0,T]}|R^{i,(p^*,c^*)}_t|^{1+\eta}\Bigr]\\
    &\leq \tilde{C}\mathbb{E}\Bigl[|R^{i,(p^*,c^*)}_T|^{1+\eta}\Bigr]= \tilde{C}\sup_{t\in[0,T]}\mathbb{E}\Bigl[|R^{i,(p^*,c^*)}_t|^{1+\eta}\Bigr]<\infty.
    \end{split}
\end{equation}
This implies that the process $R^{i,(p^*,c^*)}$ is dominated by an integrable random variable $\displaystyle\sup_{t\in[0,T]}|R^{i,(p^*,c^*)}_t|$. Using Medvegyev \cite{Medvegyev} [Corollary 1.145], we deduce that $R^{i,(p^*,c^*)}$ is a martingale of class $\mathcal{D}$.

Without loss of generality, we set $i=1$ and omit the superscript ``1'' when obvious. As $(A_{00},A_{11},A_{10},B_0,B_1,C)$ is a global solution and hence is bounded, we have, from \eqref{EQG-F-habit} and \eqref{EQG solution YZ-habit},
\begin{equation}
    \begin{split}
    \label{estimate for p}
    &|\vartheta_t|\leq \widetilde{C}(1+|x^0_t|),~~~ |p^{*}_t| \leq \widetilde{C}(1+|x^0_t| + |x^1_t|),\\
    &|Y_t| \leq \widetilde{C}(1 + |x^0_t|^2 + |x^1_t|^2),~~~ |F_t| \leq \widetilde{C}(1 + |x^0_t|^2 + |x^1_t|^2)
    \end{split}
\end{equation}
for all $t\in[0,T]$. Moreover, by Gronwall's inequality, we have $\displaystyle |X^{c^*}_t| \leq  |X_0| + \widetilde{C} + \widetilde{C}\int_0^t |c^{*}_s| ds$. Using this, we get
\begin{equation*}
    \begin{split}
    |c^{*}_t|\leq \widetilde{C}(1 + |X^{c^*}_t| + |Y_t| + |F_t|) \leq  \widetilde{C}(1 + |X_0| + |Y_t| + |F_t|) + \widetilde{C}\int_0^t |c^{*}_s| ds.
\end{split}
\end{equation*}
and then $|c^{*}_t| \leq  \widetilde{C}(1 + |X_0| + |Y_t| + |F_t|)$, again, by Gronwall's inequality. Together with \eqref{estimate for p}, we obtain
\begin{equation*}
    \begin{split}
    &|c^{*}_t| \leq  \widetilde{C}(1 + |X_0| + |x^0_t|^2 + |x^1_t|^2), \\
    &|X^{c^*}_t| \leq  |X_0| + \widetilde{C} + \widetilde{C}\int_0^t |c^{*}_s| ds  \leq \widetilde{C}\Bigl(1+ |X_0| + \int_0^t (|x^0_s|^2 + |x^1_s|^2) ds\Bigr),
\end{split}
\end{equation*}
for each $t\in[0,T]$.
Using these estimates, we have
\begin{equation*}
    \begin{split}
        \mathcal{W}^{(p^*,c^*)}_t
        &=
        \xi + \int_0^t (p^{*}_s\vartheta_s - c^{*}_s) ds  + \int_0^t p^{*}_s dW^0_s\\
        &\geq
        - |\xi| - \int_0^t (|p^{*}_s||\vartheta_s| + |c^{*}_s|) ds - \gamma(1+\eta)\int_0^t |p^{*}_s|^2 ds + \gamma(1+\eta)\int_0^t |p^{*}_s|^2 ds + \int_0^t p^{*}_s dW^0_s \\
        &\geq
        - \widetilde{C} (1+|\xi|+|X_0|) - \widetilde{C}\int_0^t (|x^0_s|^2 + |x^1_s|^2) ds + \gamma(1+\eta)\int_0^t |p^{*}_s|^2 ds +  \int_0^t p^{*}_s dW^0_s,~~t\in[0,T].
    \end{split}
\end{equation*}
Putting these together, it follows that, for all $t\in[0,T]$,
\begin{equation*}
    \begin{split}
        &-\gamma \mathcal{W}^{(p^*,c^*)}_t + M (|F_t| + |Y_t| + |c^{*}_t| + |X^{c^*}_t|)\\
        &\leq
        \widetilde{C} (1+|\xi|+ |X_0|) + \widetilde{C}(|x^0_t|^2 + |x^1_t|^2) + \widetilde{C}\int_0^t (|x^0_t|^2 + |x^1_t|^2) ds - \gamma^2(1+\eta)\int_0^t |p^{*}_s|^2 ds -\gamma  \int_0^t p^{*}_s dW^0_s \\
        &\leq
        \widetilde{C} (1+|\xi|+ |X_0|) + \widetilde{C}(|\Sigma_0|^2|\phi^0_t|^2 + |\Sigma|^2|\phi^1_t|^2 + |x^1_0|^2) \\
        &~~~~+ \widetilde{C}\int_0^t (|\Sigma_0|^2|\phi^0_s|^2 + |\Sigma|^2|\phi^1_s|^2 + |x^1_0|^2) ds - \gamma^2(1+\eta)\int_0^t |p^{*}_s|^2 ds -\gamma  \int_0^t p^{*}_s dW^0_s.
    \end{split}
\end{equation*}
Then, for all $t\in[0,T]$,
\begin{equation*}
    \begin{split}
        &\mathbb{E}\Bigl[\exp\Bigl(-(1+\eta)\gamma \mathcal{W}^{(p^*,c^*)}_t + (1+\eta)M (|F_t| + |Y_t| + |c^{*}_t| + |X^{c^*}_t|)\Bigr)\Bigr]\\
        &\leq
        \widetilde{C}\mathbb{E}\Bigl[\exp\Bigl(\widetilde{C}(|\xi| + |X_0|) + \widetilde{C}\Bigl((|\Sigma_0|^2|\phi^0_t|^2 + |\Sigma|^2|\phi^1_t|^2 + |x^1_0|^2) + \int_0^T (|\Sigma_0|^2|\phi^0_s|^2 + |\Sigma|^2|\phi^1_s|^2 + |x^1_0|^2) ds\Bigr)\Bigr.\Bigr.\\
        &~~~~~~~~~~~~~~~~~~~~~~\Bigl.\Bigl. - \gamma^2(1+\eta)^2\int_0^t |p^{*}_s|^2 ds -\gamma (1+\eta) \int_0^t p^{*}_s dW^0_s \Bigr)\Bigr]\\
        &\leq
        \widetilde{C}\mathbb{E}\Bigl[\exp\Bigl(\widetilde{C}(|\xi| + |X_0|) + \widetilde{C}(|\Sigma_0|^2|\phi^0_t|^2 + |\Sigma|^2|\phi^1_t|^2 + |x^1_0|^2)\Bigr)\Bigr]^{\frac{1}{4}}\\
        &~~~~~~~~~~~~~~~~~~\times \mathbb{E}\Bigl[\exp\Bigl(\widetilde{C}\int_0^T (|\Sigma_0|^2|\phi^0_s|^2 + |\Sigma|^2|\phi^1_s|^2 + |x^1_0|^2)ds \Bigr)\Bigr]^{\frac{1}{4}} \mathbb{E}\Bigl[\mathcal{E}\Bigl(-2\gamma(1+\eta)\int_0^\cdot p^{*}_s dW^0_s\Bigr)_t\Bigr]^{\frac{1}{2}}\\
        &=
        \widetilde{C}\mathbb{E}\Bigl[\exp\Bigl(\widetilde{C}(|\xi| + |X_0|)\Bigr)\Bigr]^{\frac{1}{4}}\mathbb{E}\Bigl[\exp\Bigl(\widetilde{C}(|\Sigma_0|^2|\phi^0_t|^2 + |\Sigma|^2|\phi^1_t|^2 + |x^1_0|^2)\Bigr)\Bigr]^{\frac{1}{4}}\\
        &~~~~~~~~~~~~~~~~~~\times \mathbb{E}\Bigl[\exp\Bigl(\widetilde{C}\int_0^T (|\Sigma_0|^2|\phi^0_s|^2 + |\Sigma|^2|\phi^1_s|^2 + |x^1_0|^2)ds \Bigr)\Bigr]^{\frac{1}{4}} \mathbb{E}\Bigl[\mathcal{E}\Bigl(-2\gamma(1+\eta)\int_0^\cdot p^{*}_s dW^0_s\Bigr)_t\Bigr]^{\frac{1}{2}}
    \end{split}
\end{equation*}
by using Holder's inequality.\par 
As $\xi$ and $X_0$ are independent and normally distributed, we have 
\[
    \mathbb{E}\Bigl[\exp\Bigl(\widetilde{C}(|\xi| + |X_0|)\Bigr)\Bigr] = \mathbb{E}\Bigl[e^{\widetilde{C}|\xi|}\Bigr]\mathbb{E}\Bigl[e^{\widetilde{C}|X_0|}\Bigr] <\infty.
\]
By Lemma \ref{crisan} and \eqref{estimate for p}, we deduce
\[
    \sup_{t\in[0,T]}\mathbb{E}\Bigl[\mathcal{E}\Bigl(-2\gamma(1+\eta)\int_0^\cdot p^{*}_s dW^0_s\Bigr)_t\Bigr] < \infty.
\]
Furthermore, since the random variables $\phi^0_t$, $\phi^1_t$ and $x^1_0$ are mutually independent, we have
\begin{equation*}
    \begin{split}
        &\mathbb{E}\Bigl[\exp\Bigl(\widetilde{C}\Bigl\{(|\Sigma_0|^2|\phi^0_t|^2 + |\Sigma|^2|\phi^1_t|^2 + |x^1_0|^2)\Bigr\}\Bigr)\Bigr] \\
        &=
        \mathbb{E}\Bigl[\exp\Bigl(\widetilde{C}|\Sigma_0|^2|\phi^0_t|^2\Bigr)\Bigr]\mathbb{E}\Bigl[\exp\Bigl(\widetilde{C}|\Sigma|^2|\phi^1_t|^2\Bigr)\Bigr]\mathbb{E}\Bigl[\exp\Bigl(\widetilde{C}|x^1_0|^2\Bigr)\Bigr]\\
        &\leq
        \widetilde{C}\mathbb{E}\Bigl[\exp\Bigl(\widetilde{C}|\Sigma_0|^2v^0_t Z^2\Bigr)\Bigr]\mathbb{E}\Bigl[\exp\Bigl(\widetilde{C}|\Sigma|^2v^1_t Z^2\Bigr)\Bigr]\mathbb{E}\Bigl[\exp\Bigl(\widetilde{C}|\mathrm{Var}(x^1_0)|Z^2\Bigr)\Bigr],
    \end{split}
\end{equation*}
where $Z\sim N(0,1)$ and 
\begin{equation*}
    \begin{split}
        v^0_t := \int_0^t e^{-2K_0(t-s)} ds = \frac{1}{2K_0}(1-e^{-2K_0 t}) < \frac{1}{2K_0},~~~v^1_t := \int_0^t e^{-2K(t-s)} ds = \frac{1}{2K}(1-e^{-2K t}) < \frac{1}{2K}
    \end{split}
\end{equation*}
for $t\in[0,T]$. Therefore, we have
\[
    \mathbb{E}\Bigl[\exp\Bigl(\widetilde{C}\Bigl\{(|\Sigma_0|^2|\phi^0_t|^2 + |\Sigma|^2|\phi^1_t|^2 + |x^1_0|^2)\Bigr\}\Bigr)\Bigr]<\infty
\]
if and only if 
\[
    \widetilde{C}(|\Sigma_0|^2v^0_t\lor |\Sigma|^2v^1_t \lor |\mathrm{Var}(x^1_0)|)<\frac{1}{2}.
\]
Similarly, we have
\[
    \mathbb{E}\Bigl[\exp\Bigl(\widetilde{C}\int_0^T (|\Sigma_0|^2|\phi^0_s|^2 + |\Sigma|^2|\phi^1_s|^2 + |x^1_0|^2)ds \Bigr)\Bigr] < \infty
\]
if and only if
\[
    \widetilde{C}\Bigl(|\Sigma_0|^2\int_0^T v^0_tdt \lor |\Sigma|^2\int_0^T v^1_tdt \lor |\mathrm{Var}(x^1_0)|T\Bigr)<\frac{1}{2}.
\]

Above all, if
\begin{equation}
    \label{var-sigma}
    |\Sigma_0|^2\lor |\Sigma|^2 \lor |\mathrm{Var}(x^1_0)| < \widetilde{C}^{-1}(1\land T^{-1})(K_0\land K \land 2^{-1}) =: \varsigma
\end{equation}
holds, we get \eqref{u.i.-EQG}, which implies $(p^{*},c^{*})\in\mathcal{A}^1_{\mathrm{EQG}}(\vartheta)$. $\square$

\begin{rem}
    If the quadratic form of $(\phi^0,\phi^i)$ in the exponential function of $R^{i,(p^*,c^*)}$ happens to be negative semidefinite, we need no constraints on the diffusion coefficients.
\end{rem}

This result also recovers the corresponding asymptotic properties of Theorem \ref{Asymp}. To be specific, the process $\vartheta$ satisfies the market clearing condition in the large population limit.

\begin{thm} (Asymptotic equilibrium in the EQG model)\\
    Let Assumption \ref{asmEQG-market} and \ref{asm6-habit} be in force. Assume further that the equation \eqref{Riccati eqn-habit} has a global solution $(A_{00},A_{11},A_{10},B_0,B_1,C)\in\mathcal{C}^1([0,T];\mathbb{M}_{d_0})\times\mathcal{C}^1([0,T];\mathbb{M}_{d})\times\mathcal{C}^1([0,T];\mathbb{R}^{d\times d_0})\times\mathcal{C}^1([0,T];\mathbb{R}^{d_0})\times\mathcal{C}^1([0,T];\mathbb{R}^{d})\times\mathcal{C}^1([0,T];\mathbb{R})$ and that 
    \[
        |\Sigma_0|^2\lor |\Sigma|^2 \lor |\mathrm{Var}(x^1_0)| < \varsigma,
    \]
    where $\varsigma>0$ is a positive constant specified in \eqref{var-sigma}.
    Then, as long as each agent adopts \eqref{EQG-optimal-habit} as his/her optimal strategy, the process $\vartheta$ defined by \eqref{EQG risk premium} clears the financial market in the large population limit, i.e., the agents' optimal trading strategies $(\pi^{i,*}_t;t\in[0,T])_{i\in\mathbb{N}}$ given by \eqref{EQG-optimal-habit} satisfy
    \begin{equation}
        \lim_{N\to\infty}\mathbb{E}\int_0^T \Bigl|\frac{1}{N}\sum_{i=1}^N \pi^{i,*}_t\Bigr|^2dt = 0.
    \end{equation}
\end{thm}
\begin{rem}
    Since the optimal strategy (under market risk-premium process $\vartheta$) is not shown to be unique, this statement is only valid for the specific choice of optimal strategies \eqref{EQG-optimal-habit}.
\end{rem}
\noindent
\textbf{\textit{Proof}}\\
Again, we denote the general constant by $\widetilde{C}$ to avoid misinterpretation with the function $C$. Since $\pi^{i,*}_t=(\sigma_t\sigma_t^\top)^{-1}\sigma_t (p^{i,*}_t)^\top$ for $t\in[0,T]$ and $|(\sigma_t\sigma_t^\top)^{-1}\sigma_t|\leq \widetilde{C}$ uniformly in $t$ by Assumption \ref{asmEQG-market}, we have
\begin{equation*}
    \begin{split}
        \mathbb{E}\int_0^T \Bigl|\frac{1}{N}\sum_{i=1}^N \pi^{i,*}_t\Bigr|^2dt
        &\leq
        \widetilde{C}\mathbb{E}\int_0^T \Bigl|\frac{1}{N}\sum_{i=1}^N p^{i,*}_t\Bigr|^2dt.
    \end{split}
\end{equation*}
By \eqref{EQG solution YZ-habit}, the process $p^{i,*}$ can be written as
\begin{equation*}
    \begin{split}
        p^{i,*}_t
        :=
        Z^{i,0\|}_t + \frac{\vartheta_t^\top}{\gamma}
        =
        Z^{i,0\|}_t - \mathbb{E}[Z_t^{1,0\|} |\mathcal{F}^0]
        =
        (x^i_t-\mu^1_t)^\top A_{10}(t) \hat\Sigma_0
    \end{split}
\end{equation*}
for every $i\in\mathbb{N}$ and $t\in[0,T]$. It is then easy to see
\begin{equation*}
    \begin{split}
       \mathbb{E}\int_0^T \Bigl|\frac{1}{N}\sum_{i=1}^N \pi^{i,*}_t\Bigr|^2dt
       \leq
       \widetilde{C}\mathbb{E}\int_0^T \Bigl|\frac{1}{N}\sum_{i=1}^N \Bigl(x^i_t-\mu^1_t\Bigr)\Bigr|^2 dt \\
       \leq
       \frac{\widetilde{C}}{N^2}\sum_{i=1}^N\mathbb{E}\int_0^T \Bigl|x^i_t-\mu^1_t\Bigr|^2 dt 
       \leq
       \frac{\widetilde{C}}{N}
       \to
       0~~~~(N\to\infty),
    \end{split}
\end{equation*}
where we have used the fact that $(x^i_t;t\in[0,T])_{i\in\mathbb{N}}$ are mutually independent and that $\mathbb{E}[x^i_t]=\mu^1_t$ for every $i\in\mathbb{N}$ and $t\in[0,T]$. $\square$

\section{Some Extensions of the Model} \label{Section 5}
In this section, we propose some extensions of the preceding models.
Specifically, we examine the cases where consumption is nonnegative, the coefficients of risk aversion differ between the terminal and running utilities, and the habit trend process $(\rho_t)_{t\in[0,T]}$ is determined by the agents' average consumption.
\subsection{Nonnegative Consumption}
In the model of the previous sections, the consumption process $c$ can be negative, and it can be understood as labour income. 
Despite this, we are also interested in the case of nonnegative consumptions. We formulate the problem as follows:
\[
    \sup_{(\pi,c)\in\mathbb{A}^1_+} U^1(\pi,c)
\]
subject to
\[
    \mathcal{W}^{1,(\pi,c)}_t =\xi^1 + \int_0^t (\pi_s^\top\sigma_s\theta_s - c_s)ds + \int_0^t \pi_s^\top\sigma_s dW_s^0,~~~t\in[0,T].
\]
The utility function is
\begin{equation}
    \begin{split}
        \label{utility-nonnegative}
        &U^1(\pi,c) \\
        &:= 
        \mathbb{E}\Bigl[-\exp\Bigl(-\delta T-\gamma^1(\mathcal{W}^{1,(\pi,c)}_T-F^1_T)\Bigr) -a \int_0^T \exp\Bigl(-\delta t-\gamma^1(\mathcal{W}^{1,(\pi,c)}_t-F^1_t)-\beta^1 c_t\Bigr)dt\Bigr].
    \end{split}
\end{equation}
Note that we do not consider habit formation in this case.
We define the admissible space of this case in the following way.
\begin{dfn} (Admissible space for agent-1)\\
    \label{Admissible space for agent-1 nonnegative}
        The admissible space $\mathbb{A}^1_+$ is the set of trading and consumption strategies $(\pi,c)\in \mathbb{H}^2(\mathbb{P}^{0,1},\mathbb{F}^{0,1},\mathbb{R}^{n})\times \mathbb{H}^2(\mathbb{P}^{0,1},\mathbb{F}^{0,1},\mathbb{R}_+)$ such that the family
        \[
            \Bigl\{\exp\Bigl(-\gamma^1\mathcal{W}^{1,(\pi,c)}_\tau - \beta^1 c_\tau  \Bigr) ; \tau\in\mathcal{T}^{0,1}\Bigr\}
        \]
        is uniformly integrable. Moreover, we define $\mathcal{A}^1_+:=\{(p,c)=(\pi^\top\sigma,c) ; (\pi,c)\in\mathbb{A}^1\}$.
\end{dfn}
We set $p_t := \pi_t^\top\sigma_t$ for all $t\in[0,T]$. As we did in Section 5.2.2, we search for a process $R^{(p,c)}$ that satisfies the condition-R (i)-(iii) to find the optimal strategy.
We suppose that $R^{(p,c)}$ has the following form.
\begin{equation}
    \begin{split}
        \label{ansatz-nonnegative}
        R^{1,(p,c)}_t 
        = &-\exp\Bigl(-\delta t-\gamma^1(\mathcal{W}^{1,(p,c)}_t-Y^1_t)\Bigr)\\
        & -a \int_0^t \exp\Bigl(-\delta s-\gamma^1(\mathcal{W}^{1,(p,c)}_s-F^1_s)-\beta^1c_s\Bigr)ds,~~~t\in[0,T],
    \end{split}
\end{equation}
where $Y^1$ solves
\begin{equation}
    \begin{split}
            Y^1_t&= F^1_T + \int_t^T f^1(s,Y^1_s,Z^{1,0}_s,Z^1_s)ds - \int_t^T Z^{1,0}_s dW^0_s - \int_t^T Z^1_s dW^1_s,~~~t\in[0,T].
    \end{split}
  \end{equation}
By Ito formula,
\begin{equation}
    \begin{split}
        dR^{1,(p,c)}_t 
        = &-\gamma^1\exp\Bigl(-\delta t-\gamma^1(\mathcal{W}^{1,(p,c)}_t-Y^1_t)\Bigr)\Bigl\{-\frac{\delta}{\gamma^1}-(p_t\theta_t-c_t+f^1(t,Y^1_t,Z^{1,0}_t,Z^1_t))\\
        &+\frac{\gamma^1}{2}(|p_t-Z_t^{0,1}|^2+|Z^1_t|^2)+\frac{a}{\gamma^1}\exp(-\gamma^1(Y^1_t-F^1_t)-\beta^1c_t)\Bigr\}dt\\
        & +\gamma^1\exp\Bigl(-\delta t-\gamma^1(\mathcal{W}^{1,(p,c)}_t-Y^1_t)\Bigr)\Bigl\{(p_t-Z_t^{0,1})dW_t^{0}-Z_t^1dW_t^1\Bigr\},~~~t\in[0,T].
    \end{split}
\end{equation}
As in Section 5.2.2, we need
\begin{equation}
    \begin{split}
        \label{driver-f-nonnegative}
    &f^1(t,Y^1_t,Z^{1,0}_t,Z^1_t)\\
    &=
    -\frac{\delta}{\gamma^1} + \inf_{p\in L_t}\Bigl\{ -p\theta_t + \frac{\gamma^1}{2}(|p-Z^{1,0}_t|^2+|Z^1_t|^2)\Bigr\}+ \inf_{c\in\mathbb{R}_+} \Bigl\{ c  + \frac{a}{\gamma^1} \exp\Bigl(-\gamma^1(Y^1_t-F^1_t)-\beta c\Bigr)\Bigr\}
    \end{split}
\end{equation}
to make $R^{(p,c)}$ satisfy the condition-R (iii).
The candidate for the optimal strategy reads:
\begin{equation}
    \begin{split}
        \label{optimal-nonnegative}
        p^{1,*}_t &= Z^{1,0\|}_t + \frac{\theta^\top_t}{\gamma^1},~~~t\in[0,T],\\
        c^{1,*}_t &= 0\lor \Bigl[\frac{1}{\beta^1}\Bigl\{\log\Bigl(\frac{a\beta^1}{\gamma^1}\Bigr)-\gamma^1(Y^1_t-F^1_t)\Bigr\}\Bigr],~~~t\in[0,T].
    \end{split}
  \end{equation}
The admissibility of $(p^{1,*},c^{1,*})$ is verified later.
Note that, we have $c^{1,*}_t=0$ if
\[
  Y^1_t > F^1_t + \frac{1}{\gamma^1}\log\Bigl(\frac{a\beta^1}{\gamma^1}\Bigr) =: Y^{1,*}_t.
\]
Then,
  \begin{equation}
    \begin{split}
        &\inf_{c\in\mathbb{R}_+} \Bigl\{ c  + \frac{a}{\gamma^1} \exp\Bigl(-\gamma^1(Y^1_t-F^1_t)-\beta^1 c\Bigr)\Bigr\}\\
        &=
        \left\{
        \begin{array}{ll}
        \dfrac{a}{\gamma^1}\exp\Bigl(-\gamma^1(Y^1_t-F^1_t)\Bigr) & \mathrm{if}~~Y^1_t > Y^{1,*}_t, \\
        \dfrac{1}{\beta^1}\Bigl\{1+\log\Bigl(\dfrac{a\beta^1}{\gamma^1}\Bigr)-\gamma^1(Y^1_t-F^1_t)\Bigr\}& \mathrm{if}~~Y^1_t\leq Y^{1,*}_t
        \end{array}
        \right.
    \end{split}
\end{equation}
for each $t\in[0,T]$. Now we obtain
\begin{equation*}
    \begin{split}
    f^1(t,Y^1_t,Z^{1,0}_t,Z^1_t)
    =
    &- \frac{\delta}{\gamma^1} -Z^{1,0\|}_t\theta_t - \frac{|\theta_t|^2}{2\gamma^1} + \frac{\gamma^1}{2}(|Z^{1,0\perp}_t|^2 + |Z^1_t|^2)  + h^1(t,Y^1_t),
    \end{split}
\end{equation*}
where the function $h^1:[0,T]\times\mathbb{R}\to\mathbb{R}$ is defined by
\begin{equation}
    \begin{split}
        h^1(t,y)
        :=
        \left\{
        \begin{array}{ll}
        \dfrac{a}{\gamma^1}\exp\Bigl(-\gamma^1(y-F^1_t)\Bigr) & \mathrm{if}~~y > Y^*_t, \\
        \dfrac{1}{\beta^1}\Bigl\{1+\log\Bigl(\dfrac{a\beta^1}{\gamma^1}\Bigr)-\gamma^1(y-F^1_t)\Bigr\}& \mathrm{if}~~y\leq Y^*_t.
        \end{array}
        \right.
    \end{split}
\end{equation}
Note that $h^1$ is positive and monotone decreasing. In addition, $h^1$ is Lipschitz continuous in $y$; for all $t\in[0,T]$ and $y,y'\in\mathbb{R}$, it holds that
\[
    |h^1(t,y)-h^1(t,y')|\leq \frac{\overline{\gamma}}{\underline{\beta}}|y-y'|.
\]
The BSDE which characterizes the optimality for agent-1 is given by
\begin{equation}
    \begin{split}
        \label{BSDE-habit-nonnegative}
            Y^1_t&= F^1_T + \int_t^T f^1(s,Y^1_s,Z^{1,0}_s,Z^1_s)ds - \int_t^T Z^{1,0}_s dW^0_s - \int_t^T Z^1_s dW^1_s,~~~t\in[0,T]
    \end{split}
  \end{equation}
with
\begin{equation*}
    \begin{split}
        f^1(s,Y^1_s,Z^{1,0}_s,Z^1_s)
        =&
         -Z^{1,0\|}_t\theta_t - \frac{|\theta_t|^2}{2\gamma^1} + \frac{\gamma^1}{2}(|Z^{1,0\perp}_t|^2 + |Z^1_t|^2)  + h^1(t,Y^1_t) - \frac{\delta}{\gamma^1},~~~s\in[0,T].
    \end{split}
\end{equation*}

\begin{thm}
    Let Assumptions \ref{asm1-habit} and \ref{asm2-habit} be in force. Then, the BSDE \eqref{BSDE-habit-nonnegative} has a unique solution $(Y,Z^{1,0},Z^1)\in\mathbb{S}^\infty(\mathbb{P}^{0,1},\mathbb{F}^{0,1},\mathbb{R})\times\mathbb{H}^2_{\mathrm{BMO}}(\mathbb{P}^{0,1},\mathbb{F}^{0,1},\mathbb{R}^{1\times d_0})\times\mathbb{H}^2_{\mathrm{BMO}}(\mathbb{P}^{0,1},\mathbb{F}^{0,1},\mathbb{R}^{1\times d})$.
    Moreover, the processes $(p^{1,*},c^{1,*})$ defined by \eqref{optimal-nonnegative} is the unique optimal strategy for agent-1.
\end{thm}
\noindent
\textbf{\textit{Proof}}\\
The well-posedness can be shown by using the same method as in the proof of Theorem \ref{sec2-well-posed}.

Under the risk-neutral measure $\mathbb{Q}$, the BSDE \eqref{BSDE-habit-nonnegative} becomes:
\begin{equation}
    \begin{split}
            Y^1_t =& F^1_T + \int_t^T \Bigl\{- \frac{|\theta_t|^2}{2\gamma^1} + \frac{\gamma^1}{2}(|Z^{1,0\perp}_t|^2 + |Z^1_t|^2)  + h^1(t,Y^1_t)- \frac{\delta}{\gamma^1}\Bigr\}ds\\
            & - \int_t^T Z^{1,0}_s dW^{0,\mathbb{Q}}_s - \int_t^T Z^1_s dW^{1,\mathbb{Q}}_s,~~~t\in[0,T].
    \end{split}
\end{equation}
We consider the truncated BSDE
\begin{equation}
\begin{split}
        Y^{n,1}_t =& F^{1}_T + \int_t^T \Bigl\{ - \frac{|\theta_t|^2\land n}{2\gamma^1} + \frac{\gamma^1}{2}(|Z^{n,1,0\perp}_t|^2 + |Z^{n,1}_t|^2)  + h^1(t,Y^{n,1}_t)- \frac{\delta}{\gamma^1}\Bigr\}ds\\
        & - \int_t^T Z^{n,1,0}_s dW^{0,\mathbb{Q}}_s - \int_t^T Z^{n,1}_s dW^{1,\mathbb{Q}}_s,~~~t\in[0,T]
\end{split}
\end{equation}
for each $n\in\mathbb{N}$. Since $h^1$ is Lipschitz continuous in $y$, the result of \cite{Kobylanski2000BackwardSD} shows that this truncated BSDE has a unique solution $(Y^{n,1},Z^{n,1,0},Z^{n,1})\in\mathbb{S}^\infty\times\mathbb{H}^2_{\mathrm{BMO}}\times\mathbb{H}^2_{\mathrm{BMO}}$ for all $n\in\mathbb{N}$. In addition, the comparison principle presented in the same work shows that $Y^{n+1,1}\leq Y^{n,1}$ holds for all $n\in\mathbb{N}$.

For the rest of the proof, we omit the superscript ``1''. We now consider the following two BSDEs: for $t\in[0,T]$
\begin{equation}
    \begin{split}
        \label{overunder_BSDE-habit-nonnegative}
    \overline{Y}_t &= \|F\|_{\mathbb{L}^\infty} +\int_t^T \Bigl\{\frac{\overline{\gamma}}{2}(|\overline{Z}^{0\perp}_s|^2 + |\overline{Z}^1_s|^2)+ \overline{h}(t,\overline{Y}_t)\Bigr\} ds - \int_t^T \overline{Z}^0_s dW^{0,\mathbb{Q}}_s - \int_t^T \overline{Z}^1_s dW^{1,\mathbb{Q}}_s, \\
    \underline{Y}_t &= -\|F\|_{\mathbb{L}^\infty} -\int_t^T \Bigl(\frac{|\theta_s|^2}{2\underline{\gamma}}+ \frac{\delta}{\underline{\gamma}}\Bigr) ds - \int_t^T \underline{Z}^0_s dW^{0,\mathbb{Q}}_s - \int_t^T \underline{Z}^1_s dW^{1,\mathbb{Q}}_s,
\end{split}
\end{equation}
where
\begin{equation}
    \begin{split}
        \overline{h}(t,y)
        :=
        \dfrac{1}{\underline{\beta}}\Bigl(1+\log\Bigl(\dfrac{a\overline{\beta}}{\underline{\gamma}}\Bigr)+\overline{\gamma} |y|+\overline{\gamma}\|F\|_{\mathbb{L}^\infty}\Bigr)
        \geq h^1(t,y)
    \end{split}
\end{equation}
for all $(t,y)\in[0,T]\times\mathbb{R}$. By \cite{Kobylanski2000BackwardSD}, each of \eqref{overunder_BSDE-habit-nonnegative} has a unique solution.
The comparison principle gives an estimate $\underline{Y}_t\leq Y^n_t\leq \overline{Y}_t, ~\mathbb{Q}$-a.s.
Moreover, it is easy to see $\overline{Z}^0=0$, $\overline{Z}^1=0$ and
\begin{equation}
    \begin{split}
    \overline{Y}_t 
    &\leq
    \|F\|_{\mathbb{L}^\infty} +\int_t^T \Bigl\{ \dfrac{\overline{\gamma}}{\underline{\beta}} |\overline{Y}_t| + \dfrac{1}{\underline{\beta}}\Bigl(1+\log\Bigl(\dfrac{a\overline{\beta}}{\underline{\gamma}}\Bigr)+\overline{\gamma}\|F\|_{\mathbb{L}^\infty}\Bigr)\Bigr\} ds.
\end{split}
\end{equation}
By the backward Gronwall's inequality, we have
\begin{equation}
    \begin{split}
    \overline{Y}_t 
    &\leq
    C(1+\|F\|_{\mathbb{L}^\infty} )<\infty.
\end{split}
\end{equation}
For $\underline{Y}$, we have
\begin{equation}
    \begin{split}
    \underline{Y}_t 
    &= 
    -\|F\|_{\mathbb{L}^\infty} -\mathbb{E}\Bigl[\int_t^T \Bigl(\frac{|\theta_s|^2}{2\underline{\gamma}}+ \frac{\delta}{\underline{\gamma}}\Bigr) ds | \mathcal{F}^{0,1}_t \Bigr]\\
    &\geq
    -C(1 + \|F\|_{\mathbb{L}^\infty} + \|\theta\|^2_{\mathbb{H}^2_{\mathrm{BMO}}})\\
    &>-\infty.
\end{split}
\end{equation}
These show $\sup_{n\in\mathbb{N}}\|Y^n\|_{\mathbb{S}^\infty} <\infty$. We can now apply the same method as in the proof of Theorem \ref{sec2-well-posed} to show the well-posedness of \eqref{BSDE-habit-nonnegative}.

We show the optimality of \eqref{optimal-nonnegative}. Since $c^*$ is bounded, we need to show $\Bigl\{\exp\Bigl(-\gamma^1\mathcal{W}^{1,(p^*,c^*)}_\tau  \Bigr) ; \tau\in\mathcal{T}^{0,1}\Bigr\}$ is uniformly integrable.
We set
\[
    \psi_t:=\exp\Bigl(-\delta t-\gamma(\mathcal{W}^{(p^*,c^*)}_t-Y_t)\Bigr),~~~t\in[0,T].
\]
Then, we have
\begin{equation}
    \begin{split}
        dR^{(p^*,c^*)}_t 
        = 
        \psi_t\Bigl\{(\theta^\top_t-\gamma Z_t^{0\perp})dW_t^{0}-\gamma Z_t^1dW_t^1\Bigr\},~~~t\in[0,T].
    \end{split}
\end{equation}
Moreover, it holds that
\begin{equation}
    \begin{split}
        R^{(p^*,c^*)}_t 
        = &-\psi_t -a \int_0^t\psi_s \exp\Bigl(-\gamma(Y_s-F_s)-\beta c^*_s\Bigr)ds,~~~t\in[0,T],
    \end{split}
\end{equation}
and thus
\[
    d\psi_t = -\psi_t\Bigl\{(\theta^\top_t-\gamma Z_t^{0\perp})dW_t^{0}-\gamma Z_t^1dW_t^1\Bigr\} -a\psi_t \exp\Bigl(-\gamma(Y_t-F_t)-\beta c^*_t\Bigr)dt,~~~t\in[0,T].
\]
Therefore, we obtain 
\begin{equation}
    \begin{split}
        \psi_t
        = &
        \exp\Bigl(-\gamma(\xi-Y_0)-a\int_0^t \exp\Bigl(-\gamma(Y_s-F_s)-\beta c^*_s\Bigr)ds\Bigr)\\
        &\times\mathcal{E}\Bigl(-\int_0^\cdot (\theta^\top_s-\gamma Z_s^{0\perp})dW_s^{0} + \int_0^\cdot \gamma Z_s^1dW_s^1\Bigr),~~~t\in[0,T],
    \end{split}
\end{equation}
which indicates that $\{\psi_\tau;\tau\in\mathcal{T}^{0,1}\}$ is uniformly integrable since $\xi,Y,F$ and $c^*$ are bounded and $\theta,Z^0,Z^1$ belongs to $\mathbb{H}^2_{\mathrm{BMO}}$. This also shows $\{R^{(p^*,c^*)}_\tau;\tau\in\mathcal{T}^{0,1}\}$ and $\Bigl\{\exp\Bigl(-\gamma^1\mathcal{W}^{1,(p^*,c^*)}_\tau  \Bigr) ; \tau\in\mathcal{T}^{0,1}\Bigr\}$ are uniformly integrable.
Following the same argument as in the proof of Theorem \ref{verification}, it is straightforward to see that $\{R^{(p,c)};(p,c)\in\mathcal{A}^1_+\}$ satisfies the condition-R in Definition \ref{condition-R} and that the optimal strategy is unique. $\square$

Following the argument of Section 5.3, we are also motivated to study a mean field BSDE:
\begin{equation}
    \begin{split}
        \label{MF-BSDE-habit-nonnegative}
            Y^i_t&= F^i_T + \int_t^T f^i(s,Y^i_s,Z^{i,0}_s,Z^i_s)ds - \int_t^T Z^{i,0}_s dW^0_s - \int_t^T Z^i_s dW^i_s,~~~t\in[0,T]
    \end{split}
  \end{equation}
with
\begin{equation*}
    \begin{split}
        f^i(s,Y^i_s,Z^{i,0}_s,Z^i_s)
        = 
        \hat\gamma Z^{i,0\|}_s\mathbb{E}[Z^{i,0\|}_s|\mathcal{F}^0]^{\top} - \frac{\hat\gamma^2}{2\gamma^i}|\mathbb{E}[Z^{i,0\|}_s|\mathcal{F}^0]|^2 + \frac{\gamma^i}{2}(|Z^{i,0\perp}_s|^2 + |Z^i_s|^2)  + h^i(s,Y^i_s) - \frac{\delta}{\gamma^i}.
    \end{split}
  \end{equation*}
for $s\in[0,T]$. Note that $f^i(s,y,z^0,z^i)$ is uniformly Lipschitz continuous in $y$:
\begin{equation*}
    \begin{split}
        |f^i(s,y,z^0,z^i)-f^i(s,y',z^0,z^i)|
        =
        |h^i(s,y) - h^i(s,y')|
        \leq
        \frac{\overline\gamma}{\underline\beta}|y-y'|
    \end{split}
  \end{equation*}
  for $y,y'\in\mathbb{R}$.
To show the well-posedness of \eqref{MF-BSDE-habit-nonnegative}, we again need to make an additional assumption on the size of the terminal liability $F^i_T$ and the process $g^i$.
  \begin{asm}
      \label{asm4-habit-nonnegative}
      Assume that, for each $i\in\{1,\ldots,N\}$, the random variable $F^i_T$ and the process $(g^i_t)_{t\in[0,T]}$ satisfy
      \begin{equation}
          \label{small-nonnegative}
          \sqrt{\|F^i_T\|^2_{\infty} + 4\Bigl\|\int_0^T |\tilde h^i_s| ds\Bigr\|^2_{\infty}} <   \frac{1}{16 c_\gamma}\land \frac{1}{32 C_\gamma},
      \end{equation}  
      where
      \[
        \tilde h^i_s :=  \dfrac{1}{\beta^i}\Bigl(1+\Bigl|\log\Bigl(\dfrac{a\beta^i}{\gamma^i}\Bigr)\Bigr|+\gamma^i |F^i_s|\Bigr) + \frac{\delta}{\gamma^i}
      \]
      \begin{equation*}
          \begin{split}
              c_\gamma:=\frac{\overline{\gamma}}{2} \lor \frac{\hat\gamma^2}{\underline{\gamma}},~~~C_\gamma:=\hat\gamma + \Bigl(\frac{\hat\gamma^2}{2\underline{\gamma}} \lor \frac{\overline{\gamma}}{2}\Bigr).
          \end{split}
      \end{equation*}
  \end{asm}

We have a similar result to Theorem \ref{MF-BSDE-wellposed1}.
  \begin{thm}
      \label{MF-BSDE-wellposed1-nonnegative}
      Under Assumptions \ref{asm1-habit}, \ref{asm3-habit}, and \ref{asm4-habit-nonnegative}, the mean field BSDE \eqref{MF-BSDE-habit-nonnegative} has a bounded solution $(Y^i,Z^{i,0},Z^i)\in\mathbb{S}^\infty(\mathbb{P}^{0,i},\mathbb{F}^{0,i},\mathbb{R}) \times\mathbb{H}^2_{\mathrm{BMO}}(\mathbb{P}^{0,i},\mathbb{F}^{0,i},\mathbb{R}^{1\times d_0})\times\mathbb{H}^2_{\mathrm{BMO}}(\mathbb{P}^{0,i},\mathbb{F}^{0,i},\mathbb{R}^{1\times d})$. 
  \end{thm}
  \noindent
  \textbf{\textit{Proof}}\\
  We consider the case of agent-1 and omit the superscript ``1'' when obvious. 
  In the same way as in the proof of Theorem \ref{MF-BSDE-wellposed1}, we have
  \begin{equation*}
    \begin{split}
        \Bigl|\hat\gamma Z^{0\|}_s\mathbb{E}[Z^{0\|}_s|\mathcal{F}^0]^{\top} - \frac{\hat\gamma^2}{2\gamma}|\mathbb{E}[Z^{0\|}_s|\mathcal{F}^0]|^2 + \frac{\gamma}{2}(|Z^{0\perp}_s|^2 + |Z^1_s|^2) \Bigr|
        &\leq
        \frac{\hat\gamma^2}{\gamma}|\mathbb{E}[Z^{0\|}_s|\mathcal{F}^0]|^2 + \frac{\gamma}{2}(|Z^{0}_s|^2 + |Z^1_s|^2)\\
        &\leq
        c_\gamma (|\mathbb{E}[Z^{0\|}_s|\mathcal{F}^0]|^2 + |Z^{0}_s|^2 + |Z^1_s|^2)
    \end{split}
\end{equation*}
for
\[
    c_\gamma=\frac{\overline{\gamma}}{2} \lor \frac{\hat\gamma^2}{\underline{\gamma}}.
\]
Moreover, we have
\begin{equation*}
    \begin{split}
        y h(t,y)
        =
        -\frac{\gamma}{\beta}y^2 + \dfrac{y}{\beta}\Bigl(1+\log\Bigl(\dfrac{a\beta}{\gamma}\Bigr)+\gamma F_t\Bigr)
        \leq
        \dfrac{|y|}{\beta}\Bigl(1+\Bigl|\log\Bigl(\dfrac{a\beta}{\gamma}\Bigr)\Bigr|+\gamma |F_t|\Bigr)
    \end{split}
\end{equation*}
for $y\leq Y^*_t$, and
\begin{equation*}
    \begin{split}
        y h(t,y)
        =
        \dfrac{ay}{\gamma}\exp\Bigl(-\gamma(y-F_t)\Bigr)
        \leq
        \frac{|y|}{\beta}
    \end{split}
\end{equation*}
for $y> Y^*_t$. Thus, in either case, it holds that
\begin{equation*}
    \begin{split}
        y h(t,y)
        &\leq
        \dfrac{|y|}{\beta}\Bigl(1+\Bigl|\log\Bigl(\dfrac{a\beta}{\gamma}\Bigr)\Bigr|+\gamma |F_t|\Bigr).
    \end{split}
\end{equation*}
These yield:
\begin{equation*}
    \begin{split}
        &Y_s f(s,Y_s,Z^0_s,Z^1_s)\\
        &\leq
        |Y_s|\Bigl|\hat\gamma Z^{0\|}_s\mathbb{E}[Z^{0\|}_s|\mathcal{F}^0]^{\top} - \frac{\hat\gamma^2}{2\gamma}|\mathbb{E}[Z^{0\|}_s|\mathcal{F}^0]|^2 + \frac{\gamma}{2}(|Z^{0\perp}_s|^2 + |Z^1_s|^2)\Bigr| + Y_s h(s,Y_s) - \frac{\delta}{\gamma}Y_s\\
        &\leq
        c_\gamma |Y_s|(|\mathbb{E}[Z^{0\|}_s|\mathcal{F}^0]|^2 + |Z^{0}_s|^2 + |Z^1_s|^2) + \dfrac{|Y_s|}{\beta}\Bigl(1+\Bigl|\log\Bigl(\dfrac{a\beta}{\gamma}\Bigr)\Bigr|+\gamma |F_s|\Bigr) + \frac{\delta}{\gamma}|Y_s|\\
        &=
        c_\gamma |Y_s|(|\mathbb{E}[Z^{0\|}_s|\mathcal{F}^0]|^2 + |Z^{0}_s|^2 + |Z^1_s|^2) + |Y_s||\tilde h_s|.
    \end{split}
\end{equation*}

Let us now consider a BSDE
\begin{equation*}
    \begin{split}
        Y_t&= F_T + \int_t^T f(s,Y_s,z^0_s,z^1_s)ds - \int_t^T Z^0_s dW^0_s - \int_t^T Z^1_s dW^1_s,~~~t\in[0,T]
    \end{split}
\end{equation*}
for an arbitrary $(z^0,z^1)\in\mathbb{H}^2_{\mathrm{BMO}}\times\mathbb{H}^2_{\mathrm{BMO}}$ as an input. 
From the standard result for Lipschitz BSDEs, there exists a unique solution $(Y,Z^0,Z^1)\in\mathbb{S}^\infty\times\mathbb{H}^2_{\mathrm{BMO}}\times\mathbb{H}^2_{\mathrm{BMO}}$ for every given $(z^0,z^1)\in\mathbb{H}^2_{\mathrm{BMO}}\times\mathbb{H}^2_{\mathrm{BMO}}$. 
In this manner, we define a map $\Gamma:\mathbb{H}^2_{\mathrm{BMO}}(\mathbb{P}^{0,1},\mathbb{F}^{0,1},\mathbb{R}^{1\times d_0}\times \mathbb{R}^{1\times d})\to \mathbb{H}^2_{\mathrm{BMO}}(\mathbb{P}^{0,1},\mathbb{F}^{0,1},\mathbb{R}^{1\times d_0}\times \mathbb{R}^{1\times d})$ by $\Gamma(z^0,z^1)=(Z^0,Z^1)$.

We can follow the same argument of the proof of Theorem \ref{MF-BSDE-wellposed1} to show
\begin{equation*}
    \begin{split}
       \|(Z^0,Z^1)\|^2_{\mathbb{H}^2_{\mathrm{BMO}}}
        &\leq
        2\|F_T\|^2_{\infty} + 8\Bigl\|\int_0^T |\tilde h_s| ds\Bigr\|^2_{\infty} + 32c_\gamma^2 \|(z^0,z^1)\|^4_{\mathbb{H}^2_{\mathrm{BMO}}}.
    \end{split}
\end{equation*}
Since we have assumed that 
\begin{equation*}
    \begin{split}
        \|F_T\|^2_{\infty} + 4\Bigl\|\int_0^T |\tilde h_s| ds\Bigr\|^2_{\infty} \leq \frac{1}{256c_\gamma^2},
    \end{split}
\end{equation*}
there exists $R>0$ such that the inequality
\begin{equation*}
    \begin{split}
        2\|F_T\|^2_{\infty} + 8\Bigl\|\int_0^T |\tilde h_s| ds\Bigr\|^2_{\infty} + 32c_\gamma^2 R^4\leq R^2
    \end{split}
\end{equation*}
holds true. We can choose, for instance,
\begin{equation}
    \label{R-choice-nonnegative}
    R = 2\sqrt{\|F_T\|^2_{\infty} + 4\Bigl\|\int_0^T |\tilde h_s| ds\Bigr\|^2_{\infty}}\leq \frac{1}{8c_\gamma}.
\end{equation}
(Step II)\\
From the results of (Step I), we have $\Gamma(\mathcal{B}_R)\subset \mathcal{B}_R$, where
\[
    \mathcal{B}_R:=\{(z^0,z^1)\in \mathbb{H}^2_{\mathrm{BMO}}(\mathbb{F}^{0,1},\mathbb{R}^{1\times d_0}\times \mathbb{R}^{1\times d}) ;  \|(z^0,z^1)\|_{\mathbb{H}^2_{\mathrm{BMO}}}\leq R\}.
\]
Our objective is now to prove that $\Gamma|_{\mathcal{B}_R}:\mathcal{B}_R\to \mathcal{B}_R$ is a strict contraction. For $(z^0,z^1), (\acute{z}^0,\acute{z}^1)\in \mathcal{B}_R$, we set $(Z^0,Z^1):=\Gamma(z^0,z^1)$ and $(\acute{Z}^0,\acute{Z}^1):=\Gamma(\acute{z}^0,\acute{z}^1)$. Also, let $Y$ and $\acute{Y}$ be corresponding solutions and set $\Delta Y := Y-\acute{Y}$ and $\Delta Z^i := Z^i-\acute{Z}^i$. 
Since $h$ is monotone decreasing in $y$, we have
\begin{equation*}
    \begin{split}
        (y-y')\{h(s,y)-h(s,y')\}
        &\leq
        0
    \end{split}
\end{equation*}
for all $y,y'\in\mathbb{R}$. Then,
\begin{equation*}
    \begin{split}
        &\Delta Y_s\{f(s,Y_s,z^0_s,z^1_s)-f(s,\acute{Y}_s,\acute{z}^0_s,\acute{z}^1_s)\}\\
        &\leq
        |\Delta Y_s|\Bigl\{\hat\gamma(|\mathbb{E}[z^{0\|}_s |\mathcal{F}^0]| + |\acute{z}^{0\|}_s|)(|\Delta z^{0\|}_s| + |\mathbb{E}[\Delta z^{0\|}_s |\mathcal{F}^0]|) + \frac{\hat\gamma^2}{2\underline{\gamma}}(|\mathbb{E}[z^{0\|}_s |\mathcal{F}^0]| + |\mathbb{E}[\acute{z}^{0\|}_s |\mathcal{F}^0]|)|\mathbb{E}[\Delta z^{0\|}_s |\mathcal{F}^0]|\Bigr.\\
        &~~~~~~~~~~~~~\Bigl. +\frac{\overline{\gamma}}{2}(|z^{0\perp}_s| + |\acute{z}^{0\perp}_s| + |z^1_s| + |\acute{z}^1_s| )(|\Delta z^{0\perp}_s| + |\Delta z^1_s|)\Bigr\} + \Delta Y_s\{h(s,Y_s)-h(s,\acute{Y}_s)\}\\
        &\leq
        |\Delta Y_s|\Bigl\{\Bigl(\Bigl(\hat\gamma + \frac{\hat\gamma^2}{2\underline{\gamma}}\Bigr)|\mathbb{E}[z^{0\|}_s |\mathcal{F}^0]| + \frac{\hat\gamma^2}{2\underline{\gamma}}|\mathbb{E}[\acute{z}^{0\|}_s |\mathcal{F}^0]| + \hat\gamma|\acute{z}^{0\|}_s |\Bigr) |\mathbb{E}[\Delta z^{0\|}_s |\mathcal{F}^0]|  \Bigr.\\
        &~~~~~~~~~~~~~\Bigl. +\Bigl(\hat\gamma + \frac{\overline{\gamma}}{2}\Bigr)(|\mathbb{E}[z^{0\|}_s |\mathcal{F}^0]| + |\acute{z}^{0}_s| +|z^{0}_s| + |z^1_s| + |\acute{z}^1_s| )(|\Delta z^{0}_s| + |\Delta z^1_s|)\Bigr\}
    \end{split}
\end{equation*}
and we can show
\begin{equation*}
    \begin{split}
       |\Delta Y_t|^2 + \mathbb{E}\Bigl[\int_t^T (|\Delta Z_s^0|^2 + |\Delta Z_s^1|^2) ds | \mathcal{F}^{0,1}_t\Bigr]
       \leq
       \frac{1}{2}\|\Delta Y\|_{\mathbb{S}^\infty}^2 + 128C_\gamma^2 R^2 \|(\Delta z^0, \Delta z^1)\|^2_{\mathbb{H}^2_{\mathrm{BMO}}},
    \end{split}
\end{equation*}
where
\[
    C_\gamma = \hat\gamma + \Bigl(\frac{\hat\gamma^2}{2\underline{\gamma}} \lor \frac{\overline{\gamma}}{2}\Bigr)
\]
in the same way. This yields
\begin{equation*}
    \begin{split}
        \|(\Delta Z^0, \Delta Z^1)\|^2_{\mathbb{H}^2_{\mathrm{BMO}}}
       &\leq
       256C_\gamma^2R^2\|(\Delta z^0, \Delta z^1)\|^2_{\mathbb{H}^2_{\mathrm{BMO}}}.
    \end{split}
\end{equation*}
Under \eqref{small-nonnegative}, $\Gamma|_{\mathcal{B}_R}$ becomes a strict contraction. Indeed, having chosen $R$ by \eqref{R-choice-nonnegative}, we clearly have $256 C_\gamma^2R^2 < 1$. 
This shows that there exists a unique fixed point of $\Gamma|_{\mathcal{B}_R}$, which represents a bounded solution of the BSDE \eqref{MF-BSDE-habit-nonnegative}. $\square$

\subsection{Generalization of Parameters}
For the first extension, we consider a case in which the coefficient of absolute risk aversion with respect to the agent's wealth differs between the terminal utility and the running utility.
For the filtered probability space $(\Omega^i,\mathcal{F}^i,\mathbb{P}^i,\mathbb{F}^i)$ ($i\in\{1,\ldots,N\})$, we set $\mathcal{F}^i_0$ by the completion of $\sigma(\xi^i,\gamma^i,\eta^i,\beta^i,X^i_0,F^i_0)$, where $\eta^i$ is an $\mathbb{R}_{++}$-valued random variable. 

The utility maximization problem (for agent-$i$) is formulated by
\[
    \sup_{(\pi,c)\in\mathbb{A}^i_{\mathrm{gen}}} U^i(\pi,c)
\]
subject to
\[
    \mathcal{W}^{i,(\pi,c)}_t =\xi^i + \int_0^t (\pi_s^\top\sigma_s\theta_s - c_s)ds + \int_0^t \pi_s^\top\sigma_s dW_s^0,~~~t\in[0,T].
\]
The admissible space $\mathbb{A}^i_{\mathrm{gen}}$ is defined later. Instead of \eqref{utility}, we consider the following utility function
\begin{equation}
    \begin{split}
        \label{utility-gen}
        &U^i(\pi,c) \\
        &:= 
        \mathbb{E}\Bigl[-\exp\Bigl(-\delta T-\gamma^i(\mathcal{W}^{i,(\pi,c)}_T-F^i_T)\Bigr) -a \int_0^T \exp\Bigl(-\delta t-\eta^i(\mathcal{W}^{i,(\pi,c)}_t-F^i_t)-\beta^i(c_t-X_t^{i,c})\Bigr)dt\Bigr].
    \end{split}
\end{equation}
Here, $\eta^i$ is an $\mathbb{R}_{++}$-valued, bounded, and $\mathcal{F}^i_0$-measurable random variable satisfying $\underline{\eta}\leq \eta^i\leq \overline{\eta}$ with some constants $0<\underline{\eta}\leq \overline{\eta}$. We assume $(\eta^i)_{i\in\{1,\ldots,N\}}$ have the same distribution.
\begin{asm}
    \label{asm7-habit}
    For each $i\in\{1,\ldots,N\}$, the system of ODEs
    \begin{equation}
        \begin{split}
            \label{h-zeta}
            \dot{h}^i_t &= \frac{1}{\beta^i}(h^i_t + b\zeta^i_t)(\gamma^i h^i_t-\eta^i),~~~t\in[0,T],\\
            \dot{\zeta}^i_t &= -\frac{1}{\beta^i}(h^i_t + b\zeta^i_t)(\beta^i-\gamma^i\zeta^i_t) + \kappa\zeta^i_t,~~~t\in[0,T]
        \end{split}
    \end{equation}
    with terminal condition $(h^i_T, \zeta^i_T)=(1,0)$ has a unique pair of solutions $(h^i_t,\zeta^i_t)_{t\in[0,T]}\in\mathcal{C}^1([0,T],\mathbb{R})\times\mathcal{C}^1([0,T],\mathbb{R})$ satisfying $h^i_t>0$ and $h^i_t + b\zeta^i_t>0$ for all $t\in[0,T]$.
\end{asm}
To solve this utility maximization problem, we set agent-1 as a representative agent. As usual, we set $p_t:=\pi_t^\top\sigma_t$ for $t\in[0,T]$. The supermartingale candidate $R^{1,(p,c)}(=R^{1,(\pi^\top\sigma,c)})$ is given by
\begin{equation}
    \begin{split}
        \label{ansatz-gen}
        R^{1,(p,c)}_t 
        = &-\exp\Bigl(-\delta t-\gamma^1\{h^1_t(\mathcal{W}^{1,(p,c)}_t-Y^1_t)-\zeta^1_tX_t^{1,c}\}\Bigr)\\
        & -a \int_0^t \exp\Bigl(-\delta s-\eta^1(\mathcal{W}^{1,(p,c)}_s-F^1_s)-\beta^1(c_s-X_s^{1,c})\Bigr)ds,~~~t\in[0,T],
    \end{split}
\end{equation}
where $(h_t,\zeta_t)_{t\in[0,T]}$ is the solution of \eqref{h-zeta} and $Y^1$ solves
\begin{equation}
    \begin{split}
            Y^1_t&= F^1_T + \int_t^T f^1(s,Y^1_s,Z^{1,0}_s,Z^1_s)ds - \int_t^T Z^{1,0}_s dW^0_s - \int_t^T Z^1_s dW^1_s,~~~t\in[0,T].
    \end{split}
\end{equation}
\begin{dfn} (Admissible space for agent-1)\\
    \label{Admissible space for agent-1 general}
        The admissible space $\mathbb{A}^1_{\mathrm{gen}}$ is the set of trading and consumption strategies $(\pi,c)\in \mathbb{H}^2(\mathbb{P}^{0,1},\mathbb{F}^{0,1},\mathbb{R}^{n})\times \mathbb{H}^2(\mathbb{P}^{0,1},\mathbb{F}^{0,1},\mathbb{R})$ such that the family
        \[
            \Bigl\{R^{1,(p,c)}_\tau ; \tau\in\mathcal{T}^{0,1}\Bigr\}
        \]
        is uniformly integrable. Moreover, we define $\mathcal{A}^1_{\mathrm{gen}}:=\{(p,c)=(\pi^\top\sigma,c) ; (\pi,c)\in\mathbb{A}^1_{\mathrm{gen}}\}$.
\end{dfn}
As usual, we omit the superscript ``1'' when obvious. By Ito formula,
\begin{equation}
    \begin{split}
        dR_t^{(p,c)}
        &=
        -\exp\Bigl(-\delta t - \gamma h_t (\mathcal{W}^{(p,c)}-Y_t) + \gamma\zeta_t X_t\Bigr)\\
        &~~~\times\Bigl\{-\delta dt -\gamma h_t d(\mathcal{W}^{(p,c)}_t-Y_t) + \frac{\gamma^2 h_t^2}{2}d\langle \mathcal{W}^{(p,c)}-Y\rangle_t - \gamma\dot{h}_t(\mathcal{W}^{(p,c)}_t-Y_t)dt\\
        &~~~~~~~~+ \gamma\dot{\zeta}_tX^c_t dt + \gamma\zeta_t (-\kappa X^c_t + bc_t + (\kappa-b)\rho_t)dt \\
        &~~~~~~~~+ a\exp\Bigl(\gamma h_t(\mathcal{W}^{(p,c)}_t-Y_t)-\gamma\zeta_tX^c_t-\eta(\mathcal{W}^{(p,c)}_t-F_t)-\beta(c_t-X^c_t)\Bigr)dt\Bigr\}\\
        &=
        -\gamma h_t\exp\Bigl(-\delta t - \gamma h_t (\mathcal{W}^{(p,c)}-Y_t) + \gamma\zeta_t X_t\Bigr)\\
        &~~~\times\Bigl\{-(p_t\theta_t - c_t) - f(t,Y_t,Z^0_t,Z^1_t) + \frac{\gamma h_t}{2}(|p_t-Z^0_t|^2 + |Z^1_t|^2)-\frac{\delta}{\gamma h_t}-\frac{\dot{h}_t}{h_t}(\mathcal{W}^{(p,c)}_t-Y_t)\\
        &~~~~~+\frac{b\zeta_t}{h_t}c_t + \frac{\dot{\zeta}_t-\kappa\zeta_t}{h_t}X^c_t + \frac{\kappa-b}{h_t}\zeta_t\rho_t \\
        &~~~~~+ \frac{a}{\gamma h_t}\exp\Bigl(\gamma h_t(\mathcal{W}^{(p,c)}_t-Y_t)-\gamma\zeta_tX^c_t-\eta(\mathcal{W}^{(p,c)}_t-F_t)-\beta(c_t-X^c_t)\Bigr)\Bigr\}dt\\
        &~~~+\gamma h_t\exp\Bigl(-\delta t - \gamma h_t (\mathcal{W}^{(p,c)}-Y_t) + \gamma\zeta_t X^c_t\Bigr) \{(p_t-Z_t^0)dW_t^0 - Z_t^1 dW_t^1\}.
    \end{split}
\end{equation}
This implies for $t\in[0,T]$,
\begin{equation}
    \begin{split}
        &f(t,Y_t,Z^0_t,Z^1_t) \\
        &=
        \inf_{p\in L_t}\Bigl\{-p\theta_t + \frac{\gamma h_t}{2}(|p-Z^0_t|^2 + |Z_t^1|^2)\Bigr\}\\
        &+\inf_{c\in\mathbb{R}}\Bigl\{\Bigl(1+\frac{b\zeta_t}{h_t}\Bigr)c + \frac{a}{\gamma h_t}\exp\Bigl(\gamma h_t(\mathcal{W}^{(p,c)}_t-Y_t)-\gamma\zeta_tX^c_t-\eta(\mathcal{W}^{(p,c)}_t-F_t)-\beta(c_t-X^c_t)\Bigr)\Bigr\}\\
        &+\frac{1}{h_t}\Bigl(-\frac{\delta}{\gamma} - \dot{h}_t(\mathcal{W}^{(p,c)}_t-Y_t) + (\dot{\zeta}_t - \kappa\zeta_t)X^c_t + (\kappa-b)\zeta_t\rho_t\Bigr).
    \end{split}
\end{equation}
Then, under the condition $h^1_t + b\zeta^1_t >0$ for all $t\in[0,T]$, we obtain:
\begin{equation}
    \begin{split}
        \label{optimal-gen}
        p^{1,*}_t &= Z^{1,0\|}_t + \frac{\theta_t^\top}{\gamma^1 h^1_t},\\
        c^{1,*}_t &= \frac{1}{\beta^1}\Bigl\{\log\Bigl(\frac{a\beta^1}{\gamma^1(h^1_t+b\zeta^1_t)}\Bigr)+\eta^1 F^1_t - \gamma^1 h^1_t Y^1_t + (\gamma^1 h^1_t - \eta^1)\mathcal{W}^{1,(p^{1,*,}c^{1,*})}_t + (\beta^1-\gamma^1\zeta^1_t)X^{c^{1,*}}_t\Bigr\}.
    \end{split}
\end{equation}
for $t\in[0,T]$.
The driver $f$ becomes
\begin{equation}
    \begin{split}
        f(t,Y_t,Z^0_t,Z^1_t)
        =&
        -Z^{0\|}_t\theta_t - \frac{|\theta_t|^2}{2\gamma h_t} + \frac{\gamma h_t}{2}(|Z^{0\perp}_t|^2 + |Z^1_t|^2) - \frac{\delta}{\gamma h_t} + \frac{(\kappa-b)}{h_t}\zeta_t\rho_t\\
        &- \frac{\dot{h}_t}{h_t}(\mathcal{W}^{(p,c)}_t-Y_t) + \frac{\dot{\zeta}_t-\kappa \zeta_t}{h_t}X^{c^*}_t +\frac{h_t + b\zeta_t}{\beta h_t}(\beta c^*_t + 1) \\
        =&
        -Z^{0\|}_t\theta_t - \frac{|\theta_t|^2}{2\gamma h_t} + \frac{\gamma h_t}{2}(|Z^{0\perp}_t|^2 + |Z^1_t|^2) - \frac{\delta}{\gamma h_t} + \frac{(\kappa-b)}{h_t}\zeta_t\rho_t\\
        &- \frac{\dot{h}_t}{h_t}(\mathcal{W}^{(p,c)}_t-Y_t) + \frac{\dot{\zeta}_t-\kappa \zeta_t}{h_t}X^{c^*}_t \\
        &+\frac{h_t + b\zeta_t}{\beta h_t}\Bigl\{1 + \log\Bigl(\frac{a\beta}{\gamma(h_t+b\zeta_t)}\Bigr)+\eta F_t - \gamma h_t Y_t + (\gamma h_t - \eta)\mathcal{W}^{(p^*,c^*)}_t + (\beta-\gamma\zeta_t)X^{c^*}_t\Bigr\}\\
        =&
        -Z^{0\|}_t\theta_t - \frac{|\theta_t|^2}{2\gamma h_t} + \frac{\gamma h_t}{2}(|Z^{0\perp}_t|^2 + |Z^1_t|^2) - \frac{\delta}{\gamma h_t} + \frac{(\kappa-b)}{h_t}\zeta_t\rho_t\\
        &+ \frac{h_t + b\zeta_t}{\beta h_t}\Bigl\{1 + \log\Bigl(\frac{a\beta}{\gamma(h_t+b\zeta_t)}\Bigr)+\eta F_t - \gamma h_t Y_t\Bigr\} + \frac{\dot{h}_t}{h_t} Y_t\\
        &+ \frac{1}{h_t}\Bigl\{\frac{1}{\beta}(h_t + b\zeta_t)(\gamma h_t-\eta)-\dot{h}_t\Bigr\}\mathcal{W}^{(p^*,c^*)}_t + \frac{1}{h_t}\Bigl\{\frac{1}{\beta}(h_t + b\zeta_t)(\beta-\gamma\zeta_t)-\kappa\zeta_t+\dot{\zeta}_t\Bigr\}X^{c^*}_t\\
        =&
        -Z^{0\|}_t\theta_t - \frac{|\theta_t|^2}{2\gamma h_t} + \frac{\gamma h_t}{2}(|Z^{0\perp}_t|^2 + |Z^1_t|^2) - \frac{\delta}{\gamma h_t} + \frac{(\kappa-b)}{h_t}\zeta_t\rho_t\\
        &+ \frac{h_t + b\zeta_t}{\beta h_t}\Bigl\{1 + \log\Bigl(\frac{a\beta}{\gamma(h_t+b\zeta_t)}\Bigr)+\eta (F_t - Y_t)\Bigr\}.
    \end{split}
\end{equation}
for $t\in[0,T]$. Here, we used the ODE \eqref{h-zeta} in the last equality. Note that, in order for $R^{(p,c)}$ to satisfy the condition-R (ii), the processes $h$ and $\zeta$ were chosen to make $f$ independent of the control variable $(p,c)$.
\begin{rem}
If $\eta=\gamma$, the ODE for $\zeta$ becomes
\begin{equation}
    \begin{split}
        h_t &\equiv 1,\\
        \dot{\zeta}_t &= \Bigl(\kappa-b+\frac{\gamma}{\beta}\Bigr)\zeta_t + \frac{\gamma b}{\beta}\zeta^2_t-1,~~~t\in[0,T],~~~\zeta_T = 0,
    \end{split}
\end{equation}
which is consistent with \eqref{zeta-ODE}.
\end{rem}
The BSDE which characterizes the optimality for this case is as follows:
\begin{equation}
    \begin{split}
        \label{BSDE-habit-gen}
            Y^1_t&= F^1_T + \int_t^T f^1(s,Y^1_s,Z^{1,0}_s,Z^1_s)ds - \int_t^T Z^{1,0}_s dW^0_s - \int_t^T Z^1_s dW^1_s,~~~t\in[0,T]
    \end{split}
  \end{equation}
with
\begin{equation*}
    \begin{split}
        f^1(s,Y^1_s,Z^{1,0}_s,Z^1_s)
        =&
        -Z^{1,0\|}_s\theta_s - \frac{|\theta_s|^2}{2\gamma^1 h^1_s} + \frac{\gamma^1 h^1_s}{2}(|Z^{1,0\perp}_s|^2 + |Z^1_s|^2)-\frac{\eta^1(h^1_s + b\zeta^1_s)}{\beta^1 h^1_s}Y^1_s + g^1_s,
    \end{split}
\end{equation*}
where
\begin{equation*}
    \begin{split}
        g^1_s
        =
        - \frac{\delta}{\gamma^1 h^1_s} + \frac{(\kappa-b)}{h^1_s}\zeta^1_s\rho_s + \frac{h^1_s + b\zeta^1_s}{\beta^1 h^1_s}\Bigl\{1 + \log\Bigl(\frac{a\beta^1}{\gamma^1(h^1_s+b\zeta^1_s)}\Bigr)+ \eta^1 F^1_s \Bigr\}.
    \end{split}
\end{equation*}
for $s\in[0,T]$.
Under Assumptions \ref{asm1-habit}, \ref{asm2-habit}, and \ref{asm7-habit}, we can show the well-posedness of the BSDE \eqref{BSDE-habit-gen} in the same way as in Theorem \ref{sec2-well-posed}.
Moreover, the market-clearing condition motivates us to study a mean field BSDE:
\begin{equation}
    \begin{split}
            Y^1_t&= F^1_T + \int_t^T f^1(s,Y^1_s,Z^{1,0}_s,Z^1_s)ds - \int_t^T Z^{1,0}_s dW^0_s - \int_t^T Z^1_s dW^1_s,~~~t\in[0,T]
    \end{split}
  \end{equation}
with
\begin{equation*}
    \begin{split}
        f^1(s,Y^1_s,Z^{1,0}_s,Z^1_s)
        =&
        \mathbb{E}\Bigl[\dfrac{1}{\gamma^1h^1_t}\Bigr]^{-1}Z^{1,0\|}_s\mathbb{E}[Z^{1,0\|}_t|\mathcal{F}^0]^{\top} -\mathbb{E}\Bigl[\dfrac{1}{\gamma^1h^1_t}\Bigr]^{-2} \frac{|\mathbb{E}[Z^{1,0\|}_t|\mathcal{F}^0]|^2}{2\gamma^1 h^1_s} \\
        &+ \frac{\gamma^1 h^1_s}{2}(|Z^{1,0\perp}_s|^2 + |Z^1_s|^2)-\frac{\eta^1(h^1_s + b\zeta^1_s)}{\beta^1 h^1_s}Y^1_s + g^1_s,
    \end{split}
\end{equation*}
where
\begin{equation*}
    \begin{split}
        g^1_s
        =
        - \frac{\delta}{\gamma^1 h^1_s} + \frac{(\kappa-b)}{h^1_s}\zeta^1_s\rho_s + \frac{h^1_s + b\zeta^1_s}{\beta^1 h^1_s}\Bigl\{1 + \log\Bigl(\frac{a\beta^1}{\gamma^1(h^1_s+b\zeta^1_s)}\Bigr)+ \eta^1 F^1_s \Bigr\}.
    \end{split}
\end{equation*}
We can also show that this mean field BSDE has a solution under Assumptions \ref{asm1-habit}, \ref{asm3-habit}, \ref{asm7-habit}, and an assumption that is similar to Assumption \ref{asm4-habit}, in the same way as in Theorem \ref{MF-BSDE-wellposed1}.

\subsection{Additional Condition on the Habit Trend}
In Chapter 5, the process $\rho$ represents the trend of the habit determined exogenously in the market. In this section, we consider the following assumption on $\rho$:
\begin{equation}
    \label{rho_average}
    \rho_t = \mathbb{E}[c_t^{1,*}|\mathcal{F}^0],~~~t\in[0,T],
\end{equation}
so that agents interact not only through the trading strategy but also through the consumption strategy due to \eqref{rho_average}.
By \eqref{optimal}, we have
\begin{equation*}
    \begin{split}
        \rho_t 
        &=
        \mathbb{E}\Bigl[ \frac{1}{\beta^1}\log\Bigl(\frac{a\beta^1}{\gamma^1(1+b\zeta^1_t)}\Bigr)\Bigr] - \mathbb{E}\Bigl[\frac{\gamma^1}{\beta^1}(Y^1_t-F^1_t)|\mathcal{F}^0\Bigr]  + \mathbb{E}\Bigl[X^{1,c^{1,*}}_t\Bigl(1-\frac{\gamma^1\zeta^1_t}{\beta^1}\Bigr) |\mathcal{F}^0\Bigr]
    \end{split}
\end{equation*}

If we search for an equilibrium satisfying both market clearing condition \eqref{MC-eqn-habit} and \eqref{rho_average}, we are motivated to study the mean field forward-backward SDEs:
\begin{equation}
    \begin{split} 
        &X_t^{i,c^{i,*}} = X^i_0 + \int_0^t \{-\kappa(X^{i,c^{i,*}}_s-\mathbb{E}[c_s^{i,*}|\mathcal{F}^0]) + b(c^{i,*}_s-\mathbb{E}[c_s^{i,*}|\mathcal{F}^0])\}ds,~~~t\in[0,T]\\
        &Y^i_t=F^i_T + \int_t^T f^i(s,X^{i,c^{i,*}}_s,Y^i_s,Z^{i,0}_s,Z^i_s)ds - \int_t^T Z^{i,0}_sdW_s^0 - \int_t^T Z^i_sdW_s^i,~~~t\in[0,T],
    \end{split}
\end{equation}
with
\begin{equation*}
    \begin{split}
        &f^i(s,X^{i,c^{i,*}}_s,Y^i_s,Z^{i,0}_s,Z^i_s)\\
        &= 
        \hat\gamma Z^{i,0\|}_s\mathbb{E}[Z^{i,0\|}_s|\mathcal{F}^0]^{\top} - \frac{\hat\gamma^2}{2\gamma^i}|\mathbb{E}[Z^{i,0\|}_s|\mathcal{F}^0]|^2 + \frac{\gamma^i}{2}(|Z^{i,0\perp}_s|^2 + |Z^i_s|^2) -\frac{\gamma^i(1+b\zeta^i_s)}{\beta^i}Y^i_s + g^i_s,
    \end{split}
  \end{equation*}
where
\[
    g^i_s := - \frac{\delta}{\gamma^i} + (\kappa-b)\zeta^i_s-\mathbb{E}[c_s^{i,*}|\mathcal{F}^0] + \frac{1 + b\zeta^i_s}{\beta^i}\Bigl\{1 + \log\Bigl(\frac{a\beta^i}{\gamma^i(1 + b\zeta^i_s)}\Bigr)+\gamma^i F^i_s\Bigr\}.
\]
Here, $c^{i,*}$ is given by \eqref{optimal}.

\section{Conclusion and Discussions} \label{Section 6}
In this chapter, we have studied a theoretical model of asset pricing among heterogeneous agents with habit formation in consumption preferences using the mean field game theory. In Section \ref{Section 3}, the market clearing condition has motivated us to study the mean field BSDE \eqref{MF-BSDE1-habit}, which 
was shown to have a bounded solution under additional assumptions on the size of the parameters. Furthermore, we have proved that the solution of this equation does indeed characterize the market equilibrium in the large population limit.
In addition, Section \ref{Section 4} introduced an exponential Gaussian model, in which an unbounded solution of the mean field BSDE can be obtained in a semi-analytic form, characterized by a system of ODEs, under appropriate assumptions. Subsequently, we have verified an optimal strategy and the asymptotic market equilibrium in the large population limit within the EQG framework. 

    The result of Section \ref{Section 4} helps us to conduct a numerical analysis in future studies. The solutions of \eqref{Riccati eqn-habit} can be calculated by the standard Euler method. A solution to the mean field BSDE can be then obtained by \eqref{EQG solution YZ-habit} with pathwise simulations. 
The numerical analysis can provide a visualization of our equilibrium model, allowing us to investigate the distribution of wealth and the effect of the habit formation. We can also expect an application to the empirical study using the market data.
See also Carmona \& Delarue \cite{carmonaProbabilisticTheoryMean2018} [Section 3.5] for linear quadratic mean field games and [Section 3.6] for numerical results.

%% file: part3.tex
\vspace{0mm}

\section{Preliminary}
The theory of asset pricing is one of the major interests in financial economics. It examines the formulation of prices in the market at equilibrium, the state where the demand for securities matches the supply. 
Comprehensive overviews of this topic can be found in, for example, Back \cite{backAssetPricingPortfolio2017} and Munk \cite{Munk}. Additionally, we refer to Karatzas \& Shreve \cite{karatzas_methods_1998} [Section 4] for details of the asset pricing in a complete market and Jarrow \cite{jarrow2018continuous} [Part III] for an organized review of the asset pricing in an incomplete market.

Investors generally do not have full access to market information, which necessitates them to infer the risk premium from the observable security price in order to make decisions about their trading strategies. 
This type of problem has intrigued researchers and has led to numerous studies, including the mean-variance hedging (MVH) problem and the utility maximization problem. See, for instance, \cite{pham_meanvariance,MTT_MVH,fujiiMakingMeanvarianceHedging2014} for the MVH problem and \cite{lanker_optim_trading, pham_quenez_OP, ManiaSantacroce} for the utility maximization problem.
The key theory behind these partially observable market problems is the stochastic filtering theory. The objective of this theory is to provide the ``best estimate" of the state process based on the observations. 
As a particular case, the linear filtering, developed by Kalman \& Bucy \cite{kalman-bucy}, has been widely used to address these problems. 
Comprehensive literature on the stochastic filtering theory can be found in, for example, Bain \& Crisan \cite{bain_fundamentals_2009} and Liptser \& Shiryayev \cite{LiptserShiryayev}.

Mean field game theory was independently formulated by Lasry \& Lions \cite{lasryMeanFieldGames2007} and Huang, Malham\'{e} \& Caines \cite{huangLargePopulationStochastic2006}, providing a powerful framework for analyzing the problem of multi-agent games. 
Traditional approaches to such problems typically become intractable because of complex interactions among agents, whereas the mean field game theory overcomes this issue by replacing these problems with a stochastic control problem of a single representative agent and a fixed-point problem. 
Carmona \& Delarue \cite{carmonaProbabilisticAnalysisMeanField2013, carmonaForwardBackwardStochastic2015} proposed the probabilistic approach to the mean field problem, involving forward-backward stochastic differential equations (FBSDEs) of McKean-Vlasov type.
The solution of the mean field game is known to yield an $\varepsilon$-Nash equilibrium of the original multi-agent game. Their theory is extensively covered in the two-volume monographs Carmona \& Delarue \cite{carmonaProbabilisticTheoryMean2018,carmonaProbabilisticTheoryMean2018a} with thorough details and applications.

For research on the mean field game theory under partial observation, we refer to Huang, Caines \& Malham\'{e} \cite{huangPO} for an early study of mean field linear quadratic Gaussian (LQG) games with partial information, in which each agent has a local noisy measurement of its own state. Huang, Wang \& Wu \cite{HWW_LQGPO_2016} originally developed a backward mean field LQG game under partial information.
Bensoussan, Feng \& Huang \cite{Bensoussan_et_al_2021} offers an extension for mean field LQG games under partial observation with common noise. 
Huang \& Wang \cite{huangdynamic_optim_2016} investigates dynamic optimization problems of a large-population system and \c{S}en \& Caines \cite{SenCaines_2019} studies a partially observed mean field game with nonlinear cost functionals and dynamics. 
Recent contributions include Li, Nie \& Wu \cite{Li_et_al2023} for a stochastic large-population problem with partial information, where the diffusion term depends also on the control variable, and Li, Li \& Wu \cite{Li_et_al_2024} for problems where agents are coupled through the control average term.

In recent years, there has been an increasing number of studies on asset pricing in financial markets employing the mean field game theory. Evangelista, Saporito \& Thamsten \cite{evangelista2022price} developed an asset pricing model considering liquidity issues using the mean field game theory. 
Fujii \& Takahashi \cite{fujiiMeanFieldGame2022, Fujii-Takahashi_strong, fujii2022equilibrium} presented a mean field price formation model under stochastic order flow.  
Fujii \cite{Fujii-equilibrium-pricing} developed a price formation model that considers market participants of two groups: cooperative and non-cooperative ones.

The main contribution of this chapter is an extension of the model of Chapter 4 to the case of partial observation under the exponential quadratic Gaussian framework. 
As mentioned above, we assume that investors can only observe the security price but cannot distinguish between the risk-premium process and the Brownian noise. 
Our objective is to derive the market risk-premium process, which cannot be directly observed by agents, endogenously from the optimal behavior of agents and the market clearing condition by using the linear filtering theory. 
As in the previous chapters, we assume that agents are characterized by exponential-type preferences and adopt self-financing strategies. 
In addition, we allow agents to have heterogeneity in initial wealth and terminal liability, in contrast to the traditional asset pricing theory which considers a single representative agent.
We employ an exponential quadratic Gaussian framework, which was introduced in Chapter 5 to carry out the analysis of the partially observable case. 
This framework not only provides a semi-explicit solution of the mean field equilibrium but also allows us to conduct numerical simulations. 

This chapter is organized as follows. 
Section 6.2 presents a formulation of the partially observable market and the utility maximization problem of an agent, along with the derivation of the conditions for an optimal strategy. 
In Section 6.3, we introduce the asymptotic market clearing condition and consider the relevant mean field BSDE. By associating the BSDE with a system of ordinary differential equations (ODEs), we show that the solution of the BSDE allows a semi-explicit solution. Furthermore, we verify that this solution does indeed characterize the market clearing condition in the large population limit. We then construct the risk-premium process under the Kalman-Bucy framework.
Section 6.4 provides a numerical simulation to visualize this model. The chapter concludes in Section 6.5 with suggestions for possible extensions.

\section{Market With Partial Observation} 
This section studies an optimal investment problem for a single agent in a partially observable market. It basically follows Fujii \& Sekine \cite{fujiiMeanFieldEquilibriumPrice2023a} by adopting the technique of Hu, Imkeller \& M\"{u}ller \cite{huUtilityMaximizationIncomplete2005a}. 
To deal with partial observation, we shall mention some results of the filtering theory for completeness.\par

We denote by $(\Omega^0,\mathcal{F}^0,\mathbb{P}^0)$ a complete probability space with a complete and right-continuous filtration $\mathbb{F}^0:=(\mathcal{F}^0_t)_{t\in[0,T]}$ generated by a $d_0$-dimensional standard Brownian motion $W^0:=(W^0_t)_{t\in[0,T]}$, a $k$-dimensional standard Brownian motion $B^0:=(B^0_t)_{t\in[0,T]}$ and an $\mathbb{R}^{d_0}$-valued random variable $\mu_0$. Here, we assume that $W^0$ and $B^0$ are independent. $\mathcal{F}^0_0$ is the completion of $\sigma(\mu_0)$. We set $\mathcal{F}^0 := \mathcal{F}^0_T$. $(\Omega^0,\mathcal{F}^0,\mathbb{P}^0)$ is used to describe the randomness of the financial market. \par
\subsection{Market Set-up}
The market dynamics and its properties are given in the following assumption.
\begin{asm}~\\
    \label{asm1}
    \textup{(i)} The risk-free interest rate is zero.\\
    \textup{(ii)} There are $d_0$ non-dividend paying risky stocks with price dynamics
    \begin{equation}
        \begin{split}
            \label{stock price-EQG}
            S_t&= S_0 + \int_0^t \mathrm{diag}(S_r)(\mu_rdr + \sigma_r dW^0_r),~~t\in[0,T],
        \end{split}
    \end{equation}
    for $S_0\in\mathbb{R}^{d_0}_{++}$, $\mu := (\mu_t)_{t\in[0,T]}\in\mathbb{H}^2(\mathbb{P}^{0},\mathbb{F}^0,\mathbb{R}^{d_0})$ with $\mu_0\in\mathbb{L}^2(\mathbb{P}^0,\mathcal{F}^0_0,\mathbb{R}^{d_0})$ and a measurable function $\sigma:[0,T]\to \mathbb{R}^{d_0\times d_0}$. \\
    \textup{(iii)} $\sigma_t$ is invertible for all $t\in[0,T]$ and satisfies
    \[
        \underline{\lambda}I_{d_0}\leq \sigma_t\sigma_t^\top \leq\overline{\lambda}I_{d_0},~~~~dt\text{-}\mathrm{a.e.}
    \]
    for some positive constants $0<\underline{\lambda}\leq\overline{\lambda}$ and $I_{d_0}$, an identity matrix of size $d_0$.\\
    \textup{(iv)} The risk-premium process $\theta\in\mathbb{H}^2(\mathbb{P}^{0},\mathbb{F}^0,\mathbb{R}^{d_0})$, defined by $\theta_t = \sigma_t^{-1}\mu_t$ for $t\in[0,T]$, is a process such that the Dol\'{e}ans-Dade exponential $\displaystyle\Bigl\{\mathcal{E}\Bigl(-\int_0^\cdot \theta_s^\top dW^0_s\Bigr)_t; t\in[0,T]\Bigr\}$ is a martingale.
\end{asm}

\begin{rem}~\\
    \label{well-posedness_S}
    Although $\mu$ is unbounded, the well-posedness of the stock price process $(S_t)_{t\in[0,T]}$ can be shown by changing the original measure $\mathbb{P}^0$ to the risk neutral measure $\mathbb{Q}$, defined by
    \begin{equation}
        \label{risk-neutral}
        \Bigl.\frac{d\mathbb{Q}}{d\mathbb{P}^{0}}\Bigr|_{\mathcal{F}^{0}_t} = \mathcal{E}\Bigl(-\int_0^\cdot \theta_s^\top dW^0_s \Bigr)_t,~~t\in[0,T],
    \end{equation}
    which is well-defined thanks to Assumption \ref{asm1}(iv).
\end{rem}

In this model, we consider a case in which agents can observe the stock prices but cannot identify their drifts and Brownian shocks independently. The available market information for agents is modelled by a filtration $\mathbb{G}^0$.

\begin{dfn}~\\
    $\mathbb{G}^0:=(\mathcal{G}_t^0)_{t\in[0,T]}$ is a complete and right-continuous filtration generated by the stock price process $(S_t)_{t\in[0,T]}$.
\end{dfn}
\begin{rem}
    Since $S_0\in\mathbb{R}^{d_0}_{++}$, $\mathcal{G}^0_0$ is trivial, unlike $\mathcal{F}^0_0$.
\end{rem}
We set $\mathcal{G}^0:=\mathcal{G}^0_T$. It is obvious that $\mathcal{G}^0_t\subset\mathcal{F}^0_t$ for all $t\in[0,T]$. Define a process $\widetilde{W}^0$ by
\begin{equation}
    \begin{split}
        \label{tilde_W}
        \widetilde{W}^0_t := \int_0^t \sigma_r^{-1} \mathrm{diag}(S_r)^{-1}dS_r = W^0_t + \int_0^t \theta_s ds,~~~t\in[0,T].
    \end{split}
\end{equation}
We have the following property.
\begin{lem}~\\
    \label{G=FW}
    Let Assumption \ref{asm1} be in force. Moreover, let $\mathbb{F}^{\widetilde{W}^0}$ be a complete and right-continuous filtration generated by $(\widetilde{W}^0_t)_{t\in[0,T]}$. Then, we have $\mathbb{G}^0=\mathbb{F}^{\widetilde{W}^0}$.
\end{lem}
\noindent
\textbf{\textit{Proof}}\\
Notice that the dynamics of $S$ is given by
\begin{equation}
    \begin{split}
        \label{S_with_tilde_W}
        S_t = S_0 + \int_0^t \mathrm{diag}(S_r)\sigma_r d\widetilde{W}_r^0,~~~t\in[0,T].
    \end{split}
\end{equation} 
Since $\sigma$ is bounded, the standard result for Lipschitz SDEs implies that \eqref{S_with_tilde_W} has a unique $\mathbb{F}^{\widetilde{W}^0}$-adapted solution. (cf. Remark \ref{well-posedness_S}.) This shows $\mathbb{G}^0\subset\mathbb{F}^{\widetilde{W}^0}$. Conversely, $\mathbb{G}^0\supset\mathbb{F}^{\widetilde{W}^0}$ is obvious by \eqref{tilde_W}. $\square$

By Girsanov's theorem, $\widetilde{W}^0$ is a $(\mathbb{G}^0,\mathbb{Q})$-Brownian motion, where $\mathbb{Q}$ is the risk-neutral measure defined in Remark \ref{risk-neutral}. We denote the expectation of the risk-premium process $\theta$ conditionally on $\mathcal{G}^0_t$ by
\begin{equation}
    \begin{split}
        \label{theta_expected}
        \widehat\theta_t := \mathbb{E}[\theta_t|\mathcal{G}^0_t],~~~t\in[0,T].
    \end{split}
\end{equation}
Moreover, we introduce a process $\widehat{W}^0$ by
\begin{equation}
    \begin{split}
        \widehat{W}^0_t := \widetilde{W}^0_t - \int_0^t \widehat\theta_s ds =  W^0_t + \int_0^t (\theta_s-\widehat\theta_s) ds,~~~t\in[0,T],
    \end{split}
\end{equation}
which is called ``innovation process'' in the filtering theory. The dynamics of $S$ can be written as
\begin{equation}
    \begin{split}
        S_t&= S_0 + \int_0^t \mathrm{diag}(S_r)\sigma_r(\widehat\theta_rdr + d\widehat{W}^0_r),~~t\in[0,T].
    \end{split}
\end{equation}
The following property is well-known.
\begin{lem}
    Under Assumption \ref{asm1}, the process $\widehat{W}^0$ is a $(\mathbb{G}^0,\mathbb{P}^0)$-Brownian motion.
\end{lem}
\noindent
\textbf{\textit{Proof}}\\
This is a consequence of L\'{e}vy's theorem. See, e.g., Pardoux \cite{pardoux_filtrage} [Proposition 2.2.7]. $\square$\\
\begin{rem}~\\
    Although the filtration $\mathbb{G}^{0}$ is larger than the augmented filtration generated by $\widehat{W}^0$ in general, we can show that every $(\mathbb{G}^{0},\mathbb{P}^{0})$-local martingale has a representation through a stochastic integral with respect to $\widehat{W}^0$. (See, e.g., Jeanblanc, Yor \& Chesney \cite{jeanblanc_mathematical_2009} [Proposition 1.7.7.1].)
\end{rem}

\subsection{Optimal Investment Problem With Exponential Utility}
Suppose there are countably infinitely many agents in the common financial market. The relevant probability spaces are defined as follows.\\

\noindent
(1) We denote by $(\Omega^i,\mathcal{F}^i,\mathbb{P}^i)$ ($i\in\mathbb{N}$) a complete probability space with a complete and right-continuous filtration $\mathbb{F}^i:=(\mathcal{F}^i_t)_{t\in[0,T]}$, generated by a $d$-dimensional standard Brownian motion $W^i:=(W^i_t)_{t\in[0,T]}$ and a $\sigma$-algebra $\sigma(\xi^i,x^i_0)$. $\mathcal{F}^i_0$ is the completion of $\sigma(\xi^i,x^i_0)$. 
Here, $\xi^i$ is an $\mathbb{R}$-valued random variable and $x^i_0$ is an $\mathbb{R}^d$-valued random variable. We set $\mathcal{F}^i:=\mathcal{F}^i_T$. \\

\noindent
(2) We denote by $(\Omega^{0,i},\mathcal{F}^{0,i},\mathbb{P}^{0,i})$ ($i\in\mathbb{N}$) a complete probability space with $\Omega^{0,i} := \Omega^0 \times \Omega^i$ and with $(\mathcal{F}^{0,i},\mathbb{P}^{0,i})$, which is the completion of $(\mathcal{F}^0 \otimes \mathcal{F}^i,\mathbb{P}^{0}\otimes \mathbb{P}^{i})$. Also, we define a $\sigma$-algebra $\mathcal{G}^{0,i}$ by the completion of $\mathcal{G}^0 \otimes \mathcal{F}^i$.
We denote by $\mathbb{F}^{0,i}:=(\mathcal{F}^{0,i}_t)_{t\in[0,T]}$ the complete and right-continuous augmentation of $(\mathcal{F}_t^0 \otimes \mathcal{F}_t^i)_{t\in[0,T]}$ and by $\mathbb{G}^{0,i}:=(\mathcal{G}^{0,i}_t)_{t\in[0,T]}$ the complete and right-continuous augmentation of $(\mathcal{G}_t^0 \otimes \mathcal{F}_t^i)_{t\in[0,T]}$.\\

\noindent
(3) We denote by $(\Omega,\mathcal{F},\mathbb{P})$ a complete probability space with $\Omega:=\prod_{i=0}^\infty\Omega^i$ and with $(\mathcal{F},\mathbb{P})$, which is the completion of $\Bigl(\bigotimes_{i=0}^\infty\mathcal{F}^i,\bigotimes_{i=0}^\infty\mathbb{P}^i\Bigr)$. The $\sigma$-algebra $\mathcal{G}$ is defined by the completion of  $\bigotimes_{i=1}^\infty\mathcal{F}^i\otimes\mathcal{G}^0$.
The filtration $\mathbb{F}:=(\mathcal{F}_t)_{t\in[0,T]}$ is the complete and right-continuous augmentation of $(\bigotimes_{i=0}^\infty\mathcal{F}^{i}_t)_{t\in[0,T]}$ and $\mathbb{G}:=(\mathcal{G}_t)_{t\in[0,T]}$ is the complete and right-continuous augmentation of $(\bigotimes_{i=1}^\infty\mathcal{F}^{i}_t\otimes \mathcal{G}^0_t)_{t\in[0,T]}$.\\

We denote by $\mathbb{E}[\cdot]$ the expectation with respect to $\mathbb{P}$ unless otherwise noted. In this chapter, the heterogeneity of agents is characterized by $(W^i,\xi^i,x^i_0)_{i\in\mathbb{N}}$. The economy is modelled through an exponential quadratic Gaussian framework.
\begin{asm} ~\\
    \label{asm2}
    \textup{(i)} For each $i\in\mathbb{N}$, $\xi^i$ is an $\mathbb{R}$-valued, $\mathcal{F}^i_0$-measurable, and normally-distributed random variable representing agent-$i$'s initial wealth. $x^i_0$ is an $\mathbb{R}^d$-valued, $\mathcal{F}^i_0$-measurable, and normally-distributed random variable.
    
    \noindent
    \textup{(ii)} The random variables $\xi^i$ and $x_0^i$ are mutually independent for each $i\in\mathbb{N}$ and $(\xi^i,x_0^i)_{i\in\mathbb{N}}$ have the same distribution.

    \noindent
    \textup{(iii)} For each $i\in\mathbb{N}$, $(F^i)_{i\in\mathbb{N}}$ is an $\mathbb{R}$-valued and $\mathcal{G}_T^{0,i}$-measurable random variable, which represents the amount of liability of agent-$i$ at time $T$. Each $F^i$  is given by a quadratic form\footnote{The symbol $\langle \cdot,\cdot\rangle$ denotes the Euclidean inner product, i.e. $\langle x, y\rangle:=x^\top y$ for $x,y\in\mathbb{R}^n$.}
    \begin{equation}
        \begin{split}
        \label{EQG-F}
            F^i := \frac{1}{2} \langle A^{F}_{00}x^0_T,x^0_T\rangle + \frac{1}{2} \langle A^{F}_{11}x^i_T,x^i_T\rangle + \langle A^{F}_{10}x^0_T,x^i_T\rangle + \langle B^{F}_0,x^0_T\rangle + \langle B_1^{F},x^i_T\rangle + C^{F},
        \end{split}
    \end{equation}
    for $(A^{F}_{00}, A^{F}_{11}, A^{F}_{10},B^F_0,B^F_1,C^F)\in\mathbb{M}_{d_0}\times\mathbb{M}_{d}\times\mathbb{R}^{d\times d_0}\times\mathbb{R}^{d_0}\times\mathbb{R}^{d}\times\mathbb{R}$ and Gaussian factor processes $(x^0,x^i)\in\mathbb{L}^0(\mathbb{G}^0,\mathbb{R}^{d_0})\times\mathbb{L}^0(\mathbb{F}^i,\mathbb{R}^{d})$ defined by
    \begin{equation*}
        \begin{split}
            x_t^0 = x^0_0 -\int_0^t K_0(x^0_s - m_0)ds + \Sigma_0 \widehat{W}_t^0,~~~x_t^i = x_0^i -\int_0^t K(x^i_s - m)ds + \Sigma W_t^i,~~~t\in[0,T]
        \end{split}
    \end{equation*}
    for $x^0_0\in\mathbb{R}^{d_0}$, $(K_0,K)\in\mathbb{R}_{++}\times \mathbb{R}_{++}$, $(m_0,m)\in\mathbb{R}^{d_0}\times \mathbb{R}^{d}$, and $(\Sigma_0,\Sigma)\in\mathbb{R}^{d_0\times d_0}\times \mathbb{R}^{d\times d}$.

    \noindent
    \textup{(iv)} Each agent is a price taker; agent-$i$ must accept the prevailing prices as he/she has no market share to influence the price.
\end{asm}
\begin{rem}
    In this model, the agent-$i$'s liability $F^i$ is subject to both common noise and idiosyncratic noise. As an example of financial interpretation, suppose that the agents are financial firms and have derivative liability at time $T$. 
    In this case, $F^i$ denotes the total amount of payoff, which usually depends on the price of securities and the idiosyncratic information, such as the corporate size and the number of contracts or clients the agent-$i$ has.
\end{rem}

The trading strategy of agent-$i$ is denoted by an $\mathbb{R}^{d_0}$-valued, $\mathbb{G}^{0,i}$-progressively measurable process $\pi^i:=(\pi^i_t)_{t\in[0,T]}$. Each element of $\pi^i_t$ represents the amount of money invested in each stock at time $t$. The wealth process of agent-$i$ with strategy $\pi$ is denoted by $\mathcal{W}^{i,\pi}\in\mathbb{L}^0(\mathbb{G}^{0,i},\mathbb{R})$ and its dynamics is given by
\begin{equation}
    \begin{split}
        \mathcal{W}^{i,\pi}_t 
        &:= 
        \xi^i + \int_0^t \pi_r^\top \mathrm{diag}(S_r)^{-1}dS_r \\
        &= 
        \xi^i + \int_0^t \pi_s^\top\sigma_s\widehat\theta_sds + \int_0^t \pi_s^\top\sigma_s d\widehat{W}_s^0
    \end{split}
\end{equation}
for $t\in[0,T]$. The agents' problems are modelled on the probability space $(\Omega,\mathcal{G},\mathbb{P},\mathbb{G})$; for each $i\in\mathbb{N}$, agent-$i$ solves
\begin{equation}
    \begin{split}
        \sup_{\pi\in\mathbb{A}^i} \mathbb{E}\Bigl[-\exp\Bigl(-\gamma(\mathcal{W}^{i,\pi}_T-F^i)\Bigr) \Bigr]
    \end{split}
\end{equation}
subject to
\[
    \mathcal{W}^{i,\pi}_t =\xi^i + \int_0^t \pi_s^\top\sigma_s\widehat\theta_sds + \int_0^t \pi_s^\top\sigma_s d\widehat{W}_s^0,~~t\in[0,T].
\]
Here, $\gamma\in\mathbb{R}_{++}$ is the coefficient of absolute risk aversion, and $\mathbb{A}^i$ is the admissible set for agent-$i$, whose definition is to be given. By writing $p_t:=\pi_t^\top\sigma_t$ for each $t\in[0,T]$, the problem can equivalently be written as
\begin{equation}
    \begin{split}
        \sup_{p\in\mathcal{A}^i} \mathbb{E}\Bigl[-\exp\Bigl(-\gamma(\mathcal{W}^{i,p}_T-F^i)\Bigr) \Bigr]
    \end{split}
\end{equation}
subject to
\[
    \mathcal{W}^{i,p}_t =\xi^i + \int_0^t p_s\widehat\theta_sds + \int_0^t p_s d\widehat{W}_s^0,~~~t\in[0,T],
\]
where the set $\mathcal{A}^i$ is defined by $\mathcal{A}^i:=\{p=\pi^\top\sigma;\pi\in\mathbb{A}^i\}$.\\

To deal with the optimal control problem, let us introduce a BSDE: for each $i\in\mathbb{N}$,
\begin{equation}
    \begin{split}
        \label{BSDE-optim}
        Y^i_t = F^i + \int_t^T \Bigl(-Z^{i,0}_s\widehat\theta_s - \frac{|\widehat\theta_s|^2}{2\gamma} + \frac{\gamma}{2}|Z^i_s|^2\Bigr)ds - \int_t^T Z^{i,0}_s d\widehat{W}^0_s - \int_t^T Z^{i}_s dW^i_s,~~~t\in[0,T].
    \end{split}
\end{equation}
Suppose that the BSDE \eqref{BSDE-optim} has a solution $(Y^i,Z^{i,0},Z^i)\in\mathbb{S}^2(\mathbb{G}^{0,i},\mathbb{R})\times\mathbb{H}^2(\mathbb{G}^{0,i},\mathbb{R}^{1\times d_0})\times\mathbb{H}^2(\mathbb{G}^{0,i},\mathbb{R}^{1\times d})$. 
Then, define a process $R^{i,p}\in\mathbb{L}^0(\mathbb{G}^{0,i},\mathbb{R})$ by
\begin{equation}
    \begin{split}
        R^{i,p}_t := -\exp\Bigl(-\gamma(\mathcal{W}^{i,p}_t-Y^i_t)\Bigr),~~~t\in[0,T],~~~i\in\mathbb{N}.
    \end{split}
\end{equation}

\begin{dfn} (Admissible space)\\
    The admissible space $\mathbb{A}^i$ is the set of trading strategies $\pi\in\mathbb{H}^2(\mathbb{P}^{0,i},\mathbb{G}^{0,i},\mathbb{R}^{d_0})$ such that a family $\{R^{i,p}_\tau;\tau\in\mathcal{T}(\mathbb{G}^{0,i})\}$ is uniformly integrable.
\end{dfn}

\begin{rem}~\\
    \textup{(i)} If the BSDE \eqref{BSDE-optim} has no solution, we set $\mathbb{A}^i=\emptyset$.\\
    \textup{(ii)} Since $F^i$ and $\theta$ are unbounded, the method of Kobylanski \cite{Kobylanski2000BackwardSD} cannot be applied to show the well-posedness of \eqref{BSDE-optim}. 
    The property of quadratic growth BSDE with unbounded generator and terminal value is studied by Briand \& Hu \cite{briandBSDEQuadraticGrowth2006, briandQuadraticBSDEsConvex2008}. 
    In this chapter, however, we do not delve into the general well-posedness result of \eqref{BSDE-optim}, as we are going to search for a special solution in the exponential quadratic Gaussian framework.\\
    \textup{(iii)} The motivation of considering the BSDE \eqref{BSDE-optim} and the process $R^{i,p}$ is explained in Fujii \& Sekine \cite{fujiiMeanFieldEquilibriumPrice2023a} [Section 3.2]. 
    This method was originally proposed by Hu, Imkeller \& M\"{u}ller \cite{huUtilityMaximizationIncomplete2005a}.
\end{rem}

\begin{thm}~\\
    \label{optimality}
    Let Assumption \ref{asm1} and \ref{asm2} be in force. 
    For each $i\in\mathbb{N}$, assume further that the BSDE \eqref{BSDE-optim} has a solution $(Y^i,Z^{i,0},Z^i)\in\mathbb{S}^2(\mathbb{G}^{0,i},\mathbb{R})\times\mathbb{H}^2(\mathbb{G}^{0,i},\mathbb{R}^{1\times d_0})\times\mathbb{H}^2(\mathbb{G}^{0,i},\mathbb{R}^{1\times d})$
    and that the process $p^{i,*}:=(p^{i,*}_t)_{t\in[0,T]}$ defined by
        \[
            p^{i,*}_t := Z^{i,0}_t + \frac{\widehat{\theta}^\top_t}{\gamma},~~~t\in[0,T]
        \]
        belongs to $\mathcal{A}^i$. Then, $p^{i,*}$ is an optimal strategy for agent-$i$.
\end{thm}
\noindent
\textbf{\textit{Proof}}\\
To begin with, notice that $R^{i,p}_0=-e^{-\gamma(\xi^i-Y_0^i)}$ is independent of the control variable $p$. By Ito formula, we have
\begin{equation*}
    \begin{split}
        dR^{i,p}_t 
        &=
        R^{i,p}_t \Bigl(-\gamma d(\mathcal{W}^{i,p}_t-Y^i_t) + \frac{\gamma^2}{2}d\langle\mathcal{W}^{i,p}-Y^i\rangle_t\Bigr)\\
        &=
        \frac{\gamma^2}{2} R^{i,p}_t \Bigl|p_t- Z^{i,0}_t - \frac{\widehat{\theta}^\top_t}{\gamma} \Bigr|^2 dt + R^{i,p}_t (-\gamma(p_t-Z^{i,0}_t)d\widehat{W}^0_t + \gamma Z^i_t dW^i_t ).
    \end{split}
\end{equation*}
Then, for any $p\in\mathcal{A}^i$, we have
\[
    \frac{\gamma^2}{2} R^{i,p}_t \Bigl|p_t- Z^{i,0}_t - \frac{\widehat{\theta}^\top_t}{\gamma} \Bigr|^2 \leq 0,~~~dt\otimes \mathbb{P}\text{-}\mathrm{a.e.}
\]
Together with the definition of admissibility, this clearly implies that the process $R^{i,p}$ is a $(\mathbb{G}^{0,i},\mathbb{P}^{0,i})$-supermartingale for every $p\in\mathcal{A}^i$. Moreover, if we choose $p=p^{i,*}$, it holds that
\[
    dR^{i,p^{i,*}}_t = R^{i,p^{i,*}}_t (- \widehat{\theta}^\top_t d\widehat{W}^0_t + \gamma Z^i_t dW^i_t ).
\]
Having assumed $p^{i,*}\in\mathcal{A}^i$, we deduce that the process $R^{i,p^{i,*}}$ is a martingale. With these observations, we obtain a relation
\begin{equation*}
    \begin{split}
    \mathbb{E}\Bigl[-\exp\Bigl(-\gamma(\mathcal{W}^{i,p}_T-F^i)\Bigr) \Bigr] 
    &= \mathbb{E}[R^{i,p}_T]\\
    &\leq \mathbb{E}\Bigl[-\exp\Bigl(-\gamma(\xi^i-Y_0^i)\Bigr) \Bigr]\\
    &= \mathbb{E}[R^{i,p^{i,*}}_T]\\
    &= \mathbb{E}\Bigl[-\exp\Bigl(-\gamma(\mathcal{W}^{i,p^{i,*}}_T-F^i)\Bigr) \Bigr]
    \end{split}
\end{equation*}
for any $p\in\mathcal{A}^i$. This indicates the optimality of $p^{i,*}$. $\square$

\section{Mean Field Equilibrium Model Under Partial Observation}
In this section, we construct the risk-premium process endogenously under the asymptotic market clearing condition, whose definition is given below. 
Section 6.3.1 introduces a relevant mean field BSDE and finds its solution in a semi-analytical form by deriving an associated system of ordinary differential equations. 
Section 6.3.2 verifies that the solution obtained in Section 6.3.1 does indeed characterize the optimal strategy and the asymptotic market clearing. 
In Section 6.3.3, we derive the dynamics of the market risk-premium process endogenously using the Kalman-Bucy filtering theory.

\subsection{The Mean Field BSDE}
\begin{dfn}~(Asymptotic market clearing condition)\\
    \label{market-clearing}
    The financial market satisfies the asymptotic market clearing condition (or the market clearing condition in the large population limit) if
    \begin{equation}
        \label{MC-eqn}
        \lim_{N\to\infty}\frac{1}{N}\sum_{i=1}^N \pi_t^{i,*} = 0,~~~dt\otimes \mathbb{P}\text{-}\mathrm{a.e.}
    \end{equation}
    holds. Here, $\pi_t^{i,*}$ denotes the optimal trading strategy of the agent-$i$.
\end{dfn}
From an economic perspective, this condition means that the excess demand (or supply) per capita converges to zero (in the sense of $dt\otimes \mathbb{P}$-almost everywhere) as the population of investors tends to infinity.
For each $i\in\mathbb{N}$, if all assumptions in Theorem \ref{optimality} hold, 
\[
    p^{i,*}_t :=(\pi_t^{i,*})^\top\sigma_t = Z^{i,0}_t + \frac{\widehat{\theta}^\top_t}{\gamma},~~~t\in[0,T]
\]
is an optimal strategy for agent-$i$. In this case, the asymptotic market clearing condition \eqref{MC-eqn} requires $\widehat{\theta}$ to satisfy
\[
    \lim_{N\to\infty}\frac{1}{N}\sum_{i=1}^N Z^{i,0}_t + \frac{\widehat\theta_t^\top}{\gamma} = 0,~~~dt\otimes \mathbb{P}\text{-}\mathrm{a.e.},
\] 
which is inconsistent with the assumption that $\widehat\theta$ is $\mathbb{G}^0$-adapted. Nevertheless, since the interactions among agents are symmetric and made only through $\widehat{\theta}$, the random variables $(Z^{i,0}_t)_{i\in\mathbb{N}}$ are expected to be exchangeable for each $t\in[0,T]$. 
Moreover, $\mathcal{F}^i_t$ and $\mathcal{F}^j_t$ being independent for $i\neq j$, we can expect, at least heuristically, that
\[
    \lim_{N\to\infty}\frac{1}{N}\sum_{i=1}^N Z^{i,0}_t = \mathbb{E}[Z^{i,0}_t|\mathcal{G}^0],~~~\mathbb{P}\text{-}\mathrm{a.s.}
\]
for each $t\in[0,T]$\footnote{
    For a $\mathbb{G}^{0,i}$-adapted process $X$, we have $\mathbb{E}[X_t|\mathcal{G}^0]=\mathbb{E}[X_t|\mathcal{G}^0_t],~\mathbb{P}^0$-a.s. for each $t\in[0,T]$ since $X_t$ is independent of $(\widehat{W}^0_s-\widehat{W}^0_t)_{s\in[t,T]}$.
}. 
For these reasons, we expect that the risk-premium process $\theta\in\mathbb{H}^2(\mathbb{P}^0,\mathbb{F}^0,\mathbb{R}^{d_0})$ satisfying
\begin{equation}
    \label{hat_MRP_eqbm}
    \widehat{\theta}_t = -\gamma\mathbb{E}[Z^{i,0}_t|\mathcal{G}^0]^\top,~~~t\in[0,T]
\end{equation}
achieves the asymptotic market clearing condition. Such an observation motivates us to study the following mean field BSDE defined on $(\Omega^{0,i},\mathcal{G}^{0,i},\mathbb{P}^{0,i},\mathbb{G}^{0,i})$ for each $i\in\mathbb{N}$:
\begin{equation}
    \begin{split}
        \label{MF-BSDE1}
            Y^i_t =& F^i+ \int_t^T  \Bigl(\gamma Z^{i,0}_s\mathbb{E}[Z^{i,0}_s|\mathcal{G}^0]^{\top} - \frac{\gamma}{2}|\mathbb{E}[Z^{i,0}_s|\mathcal{G}^0]|^2 + \frac{\gamma}{2}|Z^i_s|^2 \Bigr)ds\\
            & - \int_t^T Z^{i,0}_s d\widehat{W}^0_s - \int_t^T Z^i_s dW^i_s,~~~t\in[0,T].
    \end{split}
  \end{equation}

The mean field BSDE \eqref{MF-BSDE1} can be shown to have a semi-analytical solution under certain assumptions. See also Fujii \& Sekine \cite{fujiiMeanFieldEquilibriumPrice2023a} [Section 4.1].

\begin{thm}~\\
    \label{EQG solution}
    Let Assumption \ref{asm1} and \ref{asm2} be in force. In addition, assume that the system of ordinary differential equations
    \begin{equation}
        \begin{split}
            \label{Riccati eqn}
            &\dot{A}_{00}(t) = -\gamma A_{00}(t)\Sigma_0\Sigma_0^\top A_{00}(t)  - \gamma A_{10}(t)^\top \Sigma\Sigma^\top A_{10}(t) + 2K_0  A_{00}(t), \\
            &\dot{A}_{11}(t) = -\gamma A_{11}(t) \Sigma\Sigma^\top A_{11}(t)   + 2K  A_{11}(t), \\
            &\dot{A}_{10}(t)  = -\gamma A_{10}(t)\Sigma_0\Sigma_0^\top A_{00}(t) - \gamma A_{11}(t) \Sigma\Sigma^\top A_{10}(t) + (K_0+K)A_{10}(t) ,\\
            &\dot{B}_0(t)=\Bigl(- \gamma A_{00}(t)\Sigma_0\Sigma_0^\top  + K_0\Bigr)B_{0}(t) - \gamma A_{10}(t)^\top\Sigma\Sigma^\top B_{1}(t) - K_0A_{00}(t)m_0 - KA_{10}(t)^\top m,\\
            &\dot{B}_1(t)= \Bigl(-\gamma A_{11}(t)\Sigma\Sigma^\top + K\Bigr)B_1(t) - \gamma \Bigl(A_{10}(t){\Sigma}_0{\Sigma}_0^\top A_{10}(t)^\top \mu_t^1 + A_{10}(t)\Sigma_0\Sigma_0^\top B_0(t)\Bigr)\\
            &~~~~~~~~~~~- KA_{11}(t)m - K_0A_{10}(t)m_0, \\
            &\dot{C}(t)=  - \frac{\gamma}{2}\langle \Sigma_0^\top B_0(t),\Sigma_0^\top B_0(t) \rangle - \frac{\gamma}{2} \langle \Sigma^\top B_1(t),\Sigma^\top B_1(t) \rangle - \langle K_0 B_0(t) ,m_0 \rangle - \langle K B_1(t) ,m \rangle \\
            &~~~~~~~~~~~+ \frac{\gamma}{2}\langle A_{10}(t){\Sigma}_0{\Sigma}_0^\top A_{10}(t)^\top \mu_t^1,\mu_t^1 \rangle - \frac{1}{2}\mathrm{tr}[A_{00}(t)\Sigma_0\Sigma_0^\top] - \frac{1}{2}\mathrm{tr}[A_{11}(t)\Sigma\Sigma^\top],\\
            &A_{00}(T)=A_{00}^F,~~ A_{11}(T)=A_{11}^F, ~~A_{10}(T)=A_{10}^F, ~~B_0(T)=B_0^F, ~~B_1(T)=B_1^F, ~~C(T)=C^F
        \end{split}
    \end{equation}
    for $t\in[0,T]$ has a global solution $(A_{00},A_{11},A_{10},B_0,B_1,C)\in\mathcal{C}^1([0,T];\mathbb{M}_{d_0})\times\mathcal{C}^1([0,T];\mathbb{M}_{d})\times\mathcal{C}^1([0,T];\mathbb{R}^{d\times d_0})\times\mathcal{C}^1([0,T];\mathbb{R}^{d_0})\times\mathcal{C}^1([0,T];\mathbb{R}^{d})\times\mathcal{C}^1([0,T];\mathbb{R})$. Here, $\mu^1_t:=\mathbb{E}[x^1_t] = \mathbb{E}[x_0^1]e^{-K t}+m(1-e^{-K t})$ for $t\in[0,T]$.
    Then, for each $i\in\mathbb{N}$, the processes $(Y^i,Z^{i,0},Z^i)\in\mathbb{S}^2(\mathbb{P}^{0,i},\mathbb{G}^{0,i},\mathbb{R}) \times \mathbb{H}^2(\mathbb{P}^{0,i},\mathbb{G}^{0,i},\mathbb{R}^{1\times d_0}) \times \mathbb{H}^2(\mathbb{P}^{0,i},\mathbb{G}^{0,i},\mathbb{R}^{1\times d})$, defined by
    \begin{equation}
        \begin{split}
            \label{EQG solution YZ}
            &Y^i_t  := \frac{1}{2} \langle A_{00}(t)x^0_t,x^0_t\rangle + \frac{1}{2} \langle A_{11}(t)x^i_t,x^i_t\rangle + \langle A_{10}(t)x^0_t,x^i_t\rangle + \langle B_0(t),x^0_t\rangle + \langle B_1(t),x^i_t\rangle + C(t), \\
            &Z_t^{i,0} := \Bigl\{\Sigma_0^\top (A_{00}(t)x^0_t + A_{10}(t)^\top x^i_t + B_0(t))\Bigr\}^\top,\\
            &Z_t^i := \Bigl\{\Sigma^\top (A_{10}(t)x^0_t + A_{11}(t)x^i_t + B_1(t))\Bigr\}^\top
        \end{split}
    \end{equation}
    for $t\in[0,T]$, solve the mean field BSDE \eqref{MF-BSDE1}.
\end{thm}

\noindent
\textbf{\textit{Proof}}\\
By the terminal condition of \eqref{Riccati eqn} and Assumption \ref{asm2} (iii), it follows that $Y^i_T = F^i$. Applying Ito formula to \eqref{EQG solution YZ}, we have
\begin{equation*}
    \begin{split}
        \label{Y-ito}
        dY^i_t 
        &=
        \Bigl\{\Bigl\langle \Bigl(\frac{1}{2} \dot{A}_{00}(t) - K_0 A_{00}(t)\Bigr)x^0_t,x^0_t\Bigr\rangle + \Bigl\langle \Bigl(\frac{1}{2} \dot{A}_{11}(t) - K A_{11}(t)\Bigr)x^i_t,x^i_t\Bigr\rangle\\
        &~~~+ \Bigl\langle \Bigl( \dot{A}_{10}(t) - (K_0+K)A_{10}(t)\Bigr)x^0_t,x^i_t\Bigr\rangle \Bigr.~~~~ \\
        &~~~+ \langle \dot{B}_0(t)-K_0B_0(t) + K_0A_{00}(t)m_0 + K A_{10}(t)^\top m,x^0_t\rangle \\
        &~~~+ \langle \dot{B}_1(t)-K B_1(t) + K A_{11}(t)m + K_0A_{10}(t)m_0,x^i_t\rangle \\
        &~~~+ \Bigl.\dot{C}(t) + \langle K_0 B_0(t) ,m_0 \rangle + \langle K B_1(t) ,m \rangle +\frac{1}{2}\mathrm{tr}[A_{00}(t)\Sigma_0\Sigma_0^\top] + \frac{1}{2}\mathrm{tr}[A_{11}(t)\Sigma\Sigma^\top] \Bigr\}dt \\
        &~~~+   \langle\Sigma_0^\top (A_{00}(t)x^0_t + A_{10}(t)^\top x^i_t + B_0(t)) ,d\widehat{W}^0_t \rangle + \langle \Sigma^\top (A_{10}(t)x^0_t + A_{11}(t)x^i_t + B_1(t)) , dW^i_t \rangle.\\
        &=
        \Bigl\{\Bigl\langle \Bigl(\frac{1}{2} \dot{A}_{00}(t) - K_0 A_{00}(t)\Bigr)x^0_t,x^0_t\Bigr\rangle + \Bigl\langle \Bigl(\frac{1}{2} \dot{A}_{11}(t) - K A_{11}(t)\Bigr)x^i_t,x^i_t\Bigr\rangle \\
        &~~~+ \Bigl\langle \Bigl( \dot{A}_{10}(t) - (K_0+K)A_{10}(t)\Bigr)x^0_t,x^i_t\Bigr\rangle \Bigr.~~~~ \\
        &~~~+ \langle \dot{B}_0(t)-K_0B_0(t) + K_0A_{00}(t)m_0 + K A_{10}(t)^\top m,x^0_t\rangle \\
        &~~~+ \langle \dot{B}_1(t)-K B_1(t) + K A_{11}(t)m + K_0A_{10}(t)m_0,x^i_t\rangle \\
        &~~~+ \Bigl.\dot{C}(t) + \langle K_0 B_0(t) ,m_0 \rangle + \langle K B_1(t) ,m \rangle +\frac{1}{2}\mathrm{tr}[A_{00}(t)\Sigma_0\Sigma_0^\top] + \frac{1}{2}\mathrm{tr}[A_{11}(t)\Sigma\Sigma^\top] \Bigr\}dt \\
        &~~~+   Z^{i,0}_t d\widehat{W}^0_t + Z^i_t dW^i_t
    \end{split}
\end{equation*}
for $t\in[0,T]$. By the ODE \eqref{Riccati eqn}, it holds that
\begin{equation}
    \begin{split}
        &\Bigl\langle \Bigl(\frac{1}{2} \dot{A}_{00}(t) - K_0 A_{00}(t)\Bigr)x^0_t,x^0_t\Bigr\rangle + \Bigl\langle \Bigl(\frac{1}{2} \dot{A}_{11}(t) - K A_{11}(t)\Bigr)x^i_t,x^i_t\Bigr\rangle\\
        &~~~+ \Bigl\langle \Bigl( \dot{A}_{10}(t) - (K_0+K)A_{10}(t)\Bigr)x^0_t,x^i_t\Bigr\rangle \Bigr.~~~~ \\
        &~~~+ \langle \dot{B}_0(t)-K_0B_0(t) + K_0A_{00}(t)m_0 + K A_{10}(t)^\top m,x^0_t\rangle \\
        &~~~+ \langle \dot{B}_1(t)-K B_1(t) + K A_{11}(t)m + K_0A_{10}(t)m_0,x^i_t\rangle \\
        &~~~+ \Bigl.\dot{C}(t) + \langle K_0 B_0(t) ,m_0 \rangle + \langle K B_1(t) ,m \rangle +\frac{1}{2}\mathrm{tr}[A_{00}(t)\Sigma_0\Sigma_0^\top] + \frac{1}{2}\mathrm{tr}[A_{11}(t)\Sigma\Sigma^\top]  \\
        &=
        -\Bigl\langle \frac{\gamma}{2}\Bigl(A_{00}(t)\Sigma_0\Sigma_0^\top A_{00}(t) + A_{10}(t)^\top \Sigma\Sigma^\top A_{10}(t) \Bigr) x_t^0,x_t^0 \Bigr\rangle - \Bigl\langle \frac{\gamma}{2}A_{11}(t) \Sigma\Sigma^\top A_{11}(t) x_t^i ,x_t^i \Bigr\rangle \\
        &~~~-\Bigl\langle \gamma (A_{10}(t)\Sigma_0\Sigma_0^\top A_{00}(t) + A_{11}(t) \Sigma\Sigma^\top A_{10}(t) )x_t^0,x_t^i \Bigr\rangle \\
        &~~~-\Bigl\langle \gamma (A_{00}(t)\Sigma_0\Sigma_0^\top B_{0}(t) + A_{10}(t)^\top\Sigma\Sigma^\top B_{1}(t))  ,x_t^0 \Bigr\rangle \\
        &~~~-\Bigl\langle \gamma (A_{10}(t){\Sigma}_0{\Sigma}_0^\top A_{10}(t)^\top \mu_t^1 + A_{10}(t)\Sigma_0\Sigma_0^\top B_0(t) + A_{11}(t)\Sigma\Sigma^\top B_1(t)) ,x_t^i \Bigr\rangle\\
        &~~~+ \frac{\gamma}{2}\langle A_{10}(t){\Sigma}_0{\Sigma}_0^\top A_{10}(t)^\top \mu_t^1,\mu_t^1 \rangle - \frac{\gamma}{2}\langle \Sigma_0^\top B_0(t),\Sigma_0^\top B_0(t) \rangle - \frac{\gamma}{2}\langle \Sigma^\top B_1(t),\Sigma^\top B_1(t) \rangle\\
        &=
        -\gamma Z^{i,0}_t\mathbb{E}[Z^{i,0}_t|\mathcal{G}^0]^{\top} + \frac{\gamma}{2}|\mathbb{E}[Z^{i,0}_t|\mathcal{G}^0]|^2 - \frac{\gamma}{2}|Z^i_t|^2
    \end{split}
\end{equation}
for $t\in[0,T]$. Here, we used
\begin{equation}
    \begin{split}
        \mathbb{E}[Z^{i,0}_t|\mathcal{G}^0] = \Bigl\{\Sigma_0^\top (A_{00}(t)x^0_t + A_{10}(t)^\top \mu^1_t + B_0(t))\Bigr\}^\top,~~~t\in[0,T]
    \end{split}
\end{equation}
in the last equality, since $x^0_t$ is $\mathcal{G}^0$-measurable and $x^i_t$ ($i\in\mathbb{N}$) is independent of $\mathcal{G}^0$ for every $t\in[0,T]$.\footnote{By the construction of the probability spaces, $W^i$ ($i\in\mathbb{N}$) is independent of ${W}^0$ on $(\Omega, \mathcal{F}, \mathbb{P})$.} These observations show
\begin{equation}
    \begin{split}
        dY^i_t &= -\Bigl(\gamma Z^{i,0}_t\mathbb{E}[Z^{i,0}_t|\mathcal{G}^0]^{\top} - \frac{\gamma}{2}|\mathbb{E}[Z^{i,0}_t|\mathcal{G}^0]|^2 + \frac{\gamma}{2}|Z^i_t|^2 \Bigr)dt + Z^{i,0}_t d\widehat{W}^0_t + Z^i_t dW^i_t,~~~t\in[0,T]
    \end{split}
\end{equation}
with $Y^i_T = F^i$, $\mathbb{P}^{0,i}\text{-}\mathrm{a.s.}$. This shows that $(Y^i,Z^{i,0},Z^i)$ solve the mean field BSDE \eqref{MF-BSDE1}. $\square$

\begin{rem}~\\
    \label{rem1}
    \textup{(i)} Notice that, for each $i\in\mathbb{N}$, if $A^F_{10}=0$, namely $F^i$ has no cross-term of $x^i_T$ and $x^0_T$, we can write $F^i = \widetilde{F}^0 + \widetilde{F}^i$, where 
    \begin{equation}
        \begin{split}
        \widetilde{F}^0 := \frac{1}{2} \langle A^{F}_{00}x^0_T,x^0_T\rangle + \langle B^{F}_0,x^0_T\rangle  + C^{F},~~~
        \widetilde{F}^i :=\frac{1}{2} \langle A^{F}_{11}x^i_T,x^i_T\rangle  + \langle B_1^{F},x^i_T\rangle.
        \end{split}
    \end{equation}  
    It is clear that $\widetilde{F}^0$ (resp. $\widetilde{F}^i$) is an $\mathcal{G}^0_T$-measurable (resp. $\mathcal{F}^i_T$-measurable) random variable. In this case, we can find a solution to the mean field BSDE \eqref{MF-BSDE1} by considering the following two non-mean field BSDEs:
    \begin{equation}
        \begin{split}
            \label{separated-BSDE}
            \widetilde{Y}^0_t &= \widetilde{F}^0 + \int_t^T \frac{\gamma}{2}|\widetilde{Z}^0_s|^2ds - \int_t^T \widetilde{Z}^0_s d\widehat{W}^0_s,~~~t\in[0,T],\\
            \widetilde{Y}^i_t &= \widetilde{F}^i + \int_t^T \frac{\gamma}{2}|\widetilde{Z}^i_s|^2ds - \int_t^T \widetilde{Z}^i_s dW^i_s,~~~t\in[0,T].
        \end{split}
    \end{equation}
    Indeed, if the BSDEs $\eqref{separated-BSDE}$ have solutions $(\widetilde{Y}^0,\widetilde{Z}^0)$ and $(\widetilde{Y}^i,\widetilde{Z}^i)$, we deduce that $\widetilde{Y}^0$ also solves
    \[
        \widetilde{Y}^0_t = \widetilde{F}^0 + \int_t^T \Bigl(\gamma Z^{0}_s\mathbb{E}[Z^{0}_s|\mathcal{G}^0]^{\top} - \frac{\gamma}{2}|\mathbb{E}[Z^{0}_s|\mathcal{G}^0]|^2\Bigr)ds - \int_t^T \widetilde{Z}^0_s d\widehat{W}^0_s,~~~t\in[0,T]
    \]
    since $\mathbb{E}[Z^{0}_t|\mathcal{G}^0] = Z^{0}_t$. Then, it is clear that $(\widetilde{Y}^0 + \widetilde{Y}^i, \widetilde{Z}^0, \widetilde{Z}^i)$ solves the mean field BSDE \eqref{MF-BSDE1}. \\
    \textup{(ii)} If $\exp(\widetilde{F}^0)$ and $\exp(\widetilde{F}^i)$ are integrable, $\eqref{separated-BSDE}$ have closed-form solutions: 
    \[
        \widetilde{Y}^0_t=\log\mathbb{E}[\exp(\widetilde{F}^0)|\mathcal{G}^0_t],~~~\widetilde{Y}^i_t=\log\mathbb{E}[\exp(\widetilde{F}^i)|\mathcal{F}^i_t],~~~t\in[0,T].
    \]
    \textup{(iii)} For instance, $F^i$ with $A^F_{10}=0$ can be interpreted as a liability that is additively separated into the performance of a benchmark portfolio $\widetilde{F}^0$ quoted in the market and an additional gain $\widetilde{F}^i$ required by the manager or clients of agent-$i$.\\  
    \textup{(iv)} By the local Lipschitz property, the solution of \eqref{Riccati eqn} is locally unique if it exists. As a result, \eqref{EQG solution YZ} is the unique solution to the BSDE \eqref{MF-BSDE1} among those of the form \eqref{EQG solution YZ}.\\
    \textup{(v)} In this model, agents are assumed to be homogeneous in the risk-aversion parameter in order to simplify the mathematical analysis. 
    We may possibly allow heterogeneity in the risk-aversion parameter, namely agent-$i$'s risk-aversion parameter is expressed by $\mathcal{F}^i_0$-measurable positive random variable $\gamma^i$ for each $i\in\mathbb{N}$
    \footnote{The definition of the filtered probability space $(\Omega^i,\mathcal{F}^i,\mathbb{P}^i,\mathbb{F}^i)$ should be modified to make $\gamma^i$ measurable: $\mathcal{F}^i_0$ is set to be a completion of $\sigma(\xi^i,x_0^i,\gamma^i)$. 
    We also assume that $0<\underline{\gamma}\leq\gamma^i\leq \overline{\gamma}$ ($i\in\mathbb{N}$) for some constants $0<\underline{\gamma}\leq\overline{\gamma}$ and that $(\gamma^i)_{i\in\mathbb{N}}$ are i.i.d. on $(\Omega,\mathcal{F},\mathbb{P})$. See Chapter 4 for the general settings.}.
    In such a case, however, we need to consider a system of mean field type ODEs, whose well-posedness is much more difficult to prove than \eqref{Riccati eqn}.
\end{rem}

\subsection{Optimal Control and Asymptotic Market Clearing Condition}
Suppose that the equation \eqref{Riccati eqn} has a global solution $(A_{00},A_{11},A_{10},B_0,B_1,C)\in\mathcal{C}^1([0,T];\mathbb{M}_{d_0})\times\mathcal{C}^1([0,T];\mathbb{M}_{d})\times\mathcal{C}^1([0,T];\mathbb{R}^{d\times d_0})\times\mathcal{C}^1([0,T];\mathbb{R}^{d_0})\times\mathcal{C}^1([0,T];\mathbb{R}^{d})\times\mathcal{C}^1([0,T];\mathbb{R})$
and define the processes $(Y^i,Z^{i,0},Z^i)$ by \eqref{EQG solution YZ}.
From \eqref{hat_MRP_eqbm}, if the market risk-premium process $\theta$ satisfies
\begin{equation}
    \label{eqbm-theta}
    \widehat\theta_t := \mathbb{E}[\theta_t|\mathcal{G}^0_t] = -\gamma\mathbb{E}[Z^{i,0}_t|\mathcal{G}^0]^\top = -\gamma {\Sigma}_0^\top \Bigl(A_{00}(t)x^0_t + A_{10}(t)^\top \mu^1_t + B_0(t) \Bigr),~~~t\in[0,T],
\end{equation}
we expect that the asymptotic market clearing condition is satisfied with strategies
\begin{equation}
    \begin{split}
        \label{EQG-optimal}
        p^{i,*}_t &:= (\pi^{i,*}_t)^\top\sigma_t := Z^{i,0}_t + \frac{\widehat\theta_t^\top}{\gamma},~~~t\in[0,T],~~~i\in\mathbb{N}.
    \end{split}
\end{equation}
The following theorem proves this observation under additional assumptions.
\begin{thm} ~\\
    \label{EQG-verification}
    Let Assumptions \ref{asm1} and \ref{asm2} be in force. Assume further that the equation \eqref{Riccati eqn} has a global solution $(A_{00},A_{11},A_{10},B_0,B_1,C)\in\mathcal{C}^1([0,T];\mathbb{M}_{d_0})\times\mathcal{C}^1([0,T];\mathbb{M}_{d})\times\mathcal{C}^1([0,T];\mathbb{R}^{d\times d_0})\times\mathcal{C}^1([0,T];\mathbb{R}^{d_0})\times\mathcal{C}^1([0,T];\mathbb{R}^{d})\times\mathcal{C}^1([0,T];\mathbb{R})$
    and that $\mathrm{Var}(x^1_0)^{-1}-\gamma A_{11}(0)$ is a positive definite matrix. Then, if the market risk-premium process $\theta$ satisfies \eqref{eqbm-theta}, the following statements hold. \\
    (1) For each $i\in\mathbb{N}$, the process $p^{i,*}$, defined by \eqref{EQG-optimal}, is an optimal strategy for agent-$i$. \\
    (2) The asymptotic market clearing condition \eqref{MC-eqn} is satisfied as long as each agent adopts \eqref{EQG-optimal} as his/her optimal strategy.\\
    Here, $\mathrm{Var}(x^1_0)$ is the covariance matrix of $x^1_0$, i.e. $\mathrm{Var}(x^1_0):=\mathbb{E}[(x^1_0-\mathbb{E}[x^1_0])(x^1_0-\mathbb{E}[x^1_0])^\top]$ and the processes $(Y^i,Z^{i,0},Z^i)\in\mathbb{S}^2(\mathbb{P}^{0,i},\mathbb{G}^{0,i},\mathbb{R}) \times \mathbb{H}^2(\mathbb{P}^{0,i},\mathbb{G}^{0,i},\mathbb{R}^{1\times d_0}) \times \mathbb{H}^2(\mathbb{P}^{0,i},\mathbb{G}^{0,i},\mathbb{R}^{1\times d})$ are given by \eqref{EQG solution YZ}. 
\end{thm}

\noindent
\textbf{\textit{Proof}}\\
For (1), it suffices to show $p^{i,*}\in\mathcal{A}^i$ by Theorem \ref{optimality} and \ref{EQG solution}.
In this proof, we use $\widetilde{C}>0$ as a general constant, whose value may change line by line. \par
Recall that $(x^i_0)_{i\in\mathbb{N}}$ are Gaussian random variables and are independently and identically distributed (i.i.d.) on $(\Omega,\mathcal{G},\mathbb{P})$. If $\mathrm{Var}(x^1_0)^{-1}-\gamma A_{11}(0)$ is positive definite, we have
\begin{equation}
    \begin{split}
    \mathbb{E}[e^{\gamma Y^i_0}]
    &=
    \mathbb{E}\Bigl[\exp\Bigl(\frac{\gamma}{2} \langle A_{00}(0)x^0_0,x^0_0\rangle + \frac{\gamma}{2} \langle A_{11}(0)x^i_0,x^i_0\rangle + \gamma\langle A_{10}(0)x^0_0,x^i_0\rangle\\
    &~~~~+ \gamma\langle B_0(0),x^0_0\rangle + \gamma\langle B_1(0),x^i_0\rangle + \gamma C(0)\Bigr)\Bigr]\\
    &=
    \widetilde{C}\mathbb{E}\Bigl[\exp\Bigl(\frac{\gamma}{2} \langle A_{11}(0)x^i_0,x^i_0\rangle + \gamma\langle A_{10}(0)x^0_0 + B_1(0),x^i_0\rangle\Bigr)\Bigr]\\
    &=
    \widetilde{C} \int_{\mathbb{R}^{d_0}} \exp\Bigl(-\frac{1}{2}(\bm{x}-\mu^1_0)^\top \mathrm{Var}(x^1_0)^{-1} (\bm{x}-\mu^1_0) + \frac{\gamma}{2} \bm{x}^\top A_{11}(0)\bm{x} + \gamma\langle A_{10}(0)x^0_0 + B_1(0),\bm{x}\rangle\Bigr) d\bm{x}\\
    &\leq
    \widetilde{C} \int_{\mathbb{R}^{d_0}} \exp\Bigl(-\frac{1}{2}\bm{x}^\top (\mathrm{Var}(x^1_0)^{-1}-\gamma A_{11}(0))\bm{x} + \gamma\langle A_{10}(0)x^0_0 + B_1(0) + \mathrm{Var}(x^1_0)^{-1}\mu^1_0,\bm{x}\rangle \Bigr) d\bm{x}\\
    &<
    \infty.
    \end{split}
\end{equation}
Since $\xi^i$ is independent of $x^i_0$ and normally distributed, this observation implies
\[
    \mathbb{E}[e^{-\gamma (\xi^i-Y^i_0)}] = \mathbb{E}[e^{-\gamma \xi^i}] \mathbb{E}[e^{\gamma Y^i_0}]  < \infty.
\]

We have
\[
    R^{i,p^{i,*}}_t = -\exp\Bigl(-\gamma(\xi^i-Y_0^i)\Bigr)\mathcal{E}\Bigl(-\int_0^\cdot \widehat\theta^\top_s d\widehat{W}_s^0 + \int_0^\cdot \gamma Z^i_s dW^i_s\Bigr)_t,~~~t\in[0,T]
\]
by the definition of $R^{i,p}$ and $p^{i,*}$. Define a process $V^i\in\mathbb{L}^0(\mathbb{G}^{0,i},\mathbb{R}_{++})$ by
\[
    V^i_t := \mathcal{E}\Bigl(-\int_0^\cdot \widehat\theta^\top_s d\widehat{W}_s^0 + \int_0^\cdot \gamma Z^i_s dW^i_s\Bigr)_t,~~~t\in[0,T].
\]
By writing $\Theta^i := (-\widehat\theta^\top, \gamma Z^i) \in\mathbb{S}^2(\mathbb{P}^{0,i},\mathbb{G}^{0,i},\mathbb{R}^{1\times (d_0 + d)})$ and $\bm{W}^{0,i}:= \begin{pmatrix}
    \widehat{W}^0 \\
    W^i \\
    \end{pmatrix}
    $, $V^i$ can be written as
\[
    V^i_t = \mathcal{E}\Bigl(\int_0^\cdot \Theta^i_s d\bm{W}^{0,i}_s\Bigr)_t,~~~t\in[0,T].
\]
We set $\bm{x}^{0,i}:= \begin{pmatrix}
    x^0 \\
    x^i \\
    \end{pmatrix}
    $. Then, $\bm{x}^{0,i}$ follows the dynamics
\[
    d\bm{x}^{0,i}_t = -\bm{K} \Bigl(\bm{x}^{0,i}_t - \bm{m}\Bigr) dt + \bm{\Sigma} d\bm{W}^{0,i}_s,
\]
where
\[
    \bm{K} := \begin{pmatrix} K_0 I_{d_0} & 0 \\ 0 & K I_{d} \\ \end{pmatrix},~~~\bm{m}:=\begin{pmatrix} m_0 \\ m \\ \end{pmatrix},~~~\bm{\Sigma}:= \begin{pmatrix} \Sigma_0 & 0 \\ 0 & \Sigma \\ \end{pmatrix}.
\]
Note that $\bm{x}^{0,i}_0\in\mathbb{L}^2(\mathbb{P}^{0,i},\mathcal{G}^{0,i}_0,\mathbb{R}^{d_0 + d})$, $|\Theta^i_t|^2\leq \widetilde{C}(|\widehat\theta_t|^2 + |Z^i_t|^2) \leq \widetilde{C}(1+|\bm{x}^{0,i}_t|^2)$ for all $t\in[0,T]$ and that $\bm{W}^i$ is a $(d_0 + d)$-dimensional standard $(\mathbb{G}^{0,i},\mathbb{P}^{0,i})$-Brownian motion. 
Then, by Bain \& Crisan \cite{bain_fundamentals_2009} [Exercise 3.11], $V^i$ is a $(\mathbb{G}^{0,i},\mathbb{P}^{0,i})$-martingale. It is now easy to see that
\[
    \mathbb{E}[|R^{i,p^{i,*}}_t|] = \mathbb{E}[e^{-\gamma (\xi^i-Y^i_0)}V^i_t] = \mathbb{E}[e^{-\gamma (\xi^i-Y^i_0)} \mathbb{E}[V^i_t|\mathcal{G}_0^{0,i}]] = \mathbb{E}[e^{-\gamma (\xi^i-Y^i_0)}] < \infty,~~~t\in[0,T],
\]
and that, for all $0\leq s\leq t\leq T$,
\[
    \mathbb{E}[R^{i,p^{i,*}}_t|\mathcal{G}^{0,i}_s]=  -e^{-\gamma (\xi^i-Y^i_0)}\mathbb{E}[V^i_t|\mathcal{G}^{0,i}_s] = -e^{-\gamma (\xi^i-Y^i_0)}V^i_s = R^{i,p^{i,*}}_s,~~~\mathbb{P}^{0,i}\text{-}\mathrm{a.s.}
\]
This clearly shows that $R^{i,p^{i,*}}$ is a martingale. By the optional sampling theorem (see, e.g. \cite{Medvegyev} [Theorem 1.86]) and $\mathbb{E}[|R^{i,p^{i,*}}_T|]<\infty$, we conclude that the family $\{R^{i,p^{i,*}}_\tau;\tau\in\mathcal{T}(\mathbb{G}^{0,i})\}$ is uniformly integrable, i.e. $p^{i,*}\in\mathcal{A}^i$. 

We now verify (2). Notice that $\pi^{i,*}$ can be written as
\[
    \pi^{i,*}_t = (\sigma_t^\top)^{-1} (p_t^{i,*})^\top  = (\sigma_t^\top)^{-1} \Sigma_0^\top A_{10}(t)^\top (x^i_t-\mu^1_t),~~~t\in[0,T].
\]
Since, for each $t\in[0,T]$, $(x^i_t)_{i\in\mathbb{N}}$ are i.i.d. and $\mathbb{E}[x^i_t]=\mu^1_t$ for all $i\in\mathbb{N}$, we have
\[
    \mathbb{E}\int_0^T \left|\frac{1}{N} \sum_{i=1}^N \pi_t^{i,*}\right|^2 dt \leq \widetilde{C} \mathbb{E}\int_0^T \left|\frac{1}{N} \sum_{i=1}^N (x^i_t-\mu^1_t)\right|^2 dt \leq \frac{\widetilde{C}}{N^2} \sum_{i=1}^N \mathbb{E}\int_0^T \left|x^i_t-\mu^1_t\right|^2 dt \leq \frac{\widetilde{C}}{N}\to 0
\]
as $N\to\infty$, which implies \eqref{MC-eqn}. $\square$

\begin{rem}
    If the matrix $A^F_{11}$ is negative semidefinite, the ODE
    \[
        \dot{A}_{11}(t) = -\gamma A_{11}(t) \Sigma\Sigma^\top A_{11}(t)   + 2K  A_{11}(t),~~~t\in[0,T],~~~A_{11}(T)=A^F_{11}
    \]
    has a unique solution on $[0,T]$ for any $T>0$ (See, e.g., \cite{mat_riccati} [Theorem 8]), and the solution $A_{11}(t)$ is negative semidefinite for all $t\in[0,T]$ (See, e.g., \cite{mat_riccati} [Theorem 9].) In such a case, the condition that $\mathrm{Var}(x^1_0)^{-1}-\gamma A_{11}(0)$ is positive definite is satisfied.
\end{rem}

\subsection{Market Risk-Premium Process}
In this section, we assume that the risk-premium process $\theta$ follows a linear Gaussian dynamics on $(\mathbb{P}^0,\mathbb{F}^0)$. Using the Kalman-Bucy filtering theory, we construct a semi-explicit formulation of the risk-premium process.
\begin{asm}~\\
    \label{asm_kalman}
    \textup{(i)} The market risk-premium process $\theta$ follows
                \begin{equation}
                    \begin{split}
                        \label{LG_MRP}
                        \theta_t = \theta_0 + \int_0^t (\alpha_s\theta_s + \beta_s)ds + \int_0^t \zeta_s dW^0_s + \int_0^t \eta_s dB^0_s,~~~t\in[0,T],
                    \end{split}
                \end{equation}
                for $\alpha\in\mathcal{C}([0,T];\mathbb{R}^{d_0\times d_0})$, $\beta\in\mathcal{C}([0,T];\mathbb{R}^{d_0})$, $\zeta\in\mathcal{C}^1([0,T];\mathbb{M}_{d_0})$ and $\eta\in\mathcal{C}([0,T];\mathbb{R}^{d_0\times k})$. The initial condition $\theta_0\in\mathbb{L}^2(\mathbb{P}^0,\mathcal{F}^0_0,\mathbb{R}^{d_0})$ is normally distributed: $\theta_0\sim N(m,v)$ for $(m,v)\in\mathbb{R}^{d_0}\times \mathbb{M}_{d_0}$. \\
    \textup{(ii)} $\Sigma_0$ is invertible.\\
    \textup{(iii)} The system of ordinary differential equations \eqref{Riccati eqn} has a global solution $(A_{00},A_{11},A_{10},B_0,B_1,C)\in\mathcal{C}^1([0,T];\mathbb{M}_{d_0})\times\mathcal{C}^1([0,T];\mathbb{M}_{d})\times\mathcal{C}^1([0,T];\mathbb{R}^{d\times d_0})\times\mathcal{C}^1([0,T];\mathbb{R}^{d_0})\times\mathcal{C}^1([0,T];\mathbb{R}^{d})\times\mathcal{C}^1([0,T];\mathbb{R})$ and $A_{00}(t)$ is invertible for all $t\in[0,T]$.\\
    \textup{(iv)} $\mathrm{Var}(x^1_0)^{-1}-\gamma A_{11}(0)$ is a positive definite matrix.
\end{asm}
Recall that $B^0:=(B^0_t)_{t\in[0,T]}$ is a $k$-dimensional $(\mathbb{F}^0,\mathbb{P}^0)$-standard Brownian motion independent of $W^0$. The SDE \eqref{LG_MRP} is well-posed due to the standard result for Lipschitz SDEs. The objective of this section is to find appropriate coefficients $(\alpha,\beta,\zeta,\eta)$ in \eqref{LG_MRP} with which $\theta$ satisfies \eqref{eqbm-theta}.
The following lemma shows that Assumption \ref{asm_kalman} (i) is consistent with Assumption \ref{asm1} (iv).

\begin{lem}~\\
    Under Assumption \ref{asm_kalman} (i), the Dol\'{e}ans-Dade exponential $\Bigl\{\mathcal{E}\Bigl(-\displaystyle\int_0^\cdot \theta_s^\top dW^0_s\Bigr)_t;t\in[0,T]\Bigr\}$ is a martingale.
\end{lem}
\noindent
\textbf{\textit{Proof}}\\
We write:
\[
    \Lambda_t := \mathcal{E}\Bigl(-\displaystyle\int_0^\cdot \theta_s^\top dW^0_s\Bigr)_t,~~~t\in[0,T].
\]
By Bain \& Crisan \cite{bain_fundamentals_2009} [Lemma 3.9.], it suffices to show
\[
    \mathbb{E}\Bigl[\int_0^T |\theta_s|^2 ds\Bigr] <\infty,~~~\mathbb{E}\Bigl[\int_0^T  \Lambda_s|\theta_s|^2 ds\Bigr] <\infty.
\]
The first condition is obvious by the standard result for Lipschitz SDEs. The second condition can be shown similarly by following Bain \& Crisan \cite{bain_fundamentals_2009} [Exercise 3.11] and its solution in [Section 3.9].
$\square$

The observation is made according to the stock price process $(S_t)_{t\in[0,T]}$. By Lemma \ref{G=FW}, we can set
\begin{equation}
    \begin{split}
        \widetilde{W}^0_t = W^0_t + \int_0^t \theta_s ds
    \end{split}
\end{equation}
as an observation process. The dynamics of $\widehat\theta$ is given as follows.
\begin{lem}~\\
    \label{theta_hat_dynamics}
    Let Assumptions \ref{asm_kalman} (i) be in force. Then, the process $\widehat\theta$, defined by \eqref{theta_expected}, satisfies the following SDE:
    \begin{equation}
        \begin{split}
            \label{hat_theta}
            d\widehat\theta_t = (\alpha_t\widehat\theta_t + \beta_t)dt + (\zeta_t + \varrho_t) d\widehat{W}^0_t,~~~t\in[0,T],~~~\widehat\theta_0 = m,
        \end{split}
    \end{equation}
    where $\varrho \in\mathcal{C}^1([0,T];\mathbb{M}_{d_0})$ is a function which satisfies the following Riccati equation:
    \begin{equation}
        \begin{split}
            \label{ODE_for_vrho}
            \dot{\varrho }_t=\eta_t\eta_t^\top + \alpha_t \varrho_t + \varrho_t\alpha_t^\top - \zeta_t\varrho_t - \varrho_t\zeta_t - \varrho_t^2,~~~t\in[0,T],~~~\varrho_0 = v.
        \end{split}
    \end{equation}
\end{lem}
\noindent
\textbf{\textit{Proof}}\\
See Liptser \& Shiryayev \cite{LiptserShiryayev} [Theorem 10.3] and set
\begin{equation}
    \begin{split}
    &a_0=\beta,~~~a_1=\alpha,~~~a_2\equiv 0,~~~b_1=\zeta,~~~b_2=\eta,\\
    &A_0\equiv 0,~~~A_1\equiv I_{d_0},~~~A_2\equiv 0,~~~B_1\equiv I_{d_0},~~~B_2\equiv 0
    \end{split}
\end{equation}
therein. $\square$

In addition to Assumptions \ref{asm1} and \ref{asm2}, let Assumption \ref{asm_kalman} be in force. If $\widehat\theta$ satisfies
\begin{equation}
    \label{eqbm_hat_theta}
    \widehat\theta_t = -\gamma {\Sigma}_0^\top \Bigl(A_{00}(t)x^0_t + A_{10}(t)^\top \mu^1_t + B_0(t) \Bigr),~~~t\in[0,T],
\end{equation}
the processes $(p^{i,*})_{i\in\mathbb{N}}$ defined by
\[
    p^{i,*}_t := (\pi^{i,*}_t)^\top\sigma_t := Z^{i,0}_t + \frac{\widehat\theta_t^\top}{\gamma},~~~t\in[0,T],~~~i\in\mathbb{N},
\]
are optimal strategies and satisfy the asymptotic market clearing condition by Theorem \ref{EQG-verification}. Plugging \eqref{eqbm_hat_theta} into \eqref{hat_theta}, we get: for $t\in[0,T]$,
\begin{equation}
    \begin{split}
        \label{eq_3.3-1}
    d\widehat\theta_t 
    &=
    (\alpha_t \widehat\theta_t + \beta_t)dt + (\zeta_t + \varrho_t)d\widehat{W}^0_t\\
    &=
    \{-\gamma\alpha_t {\Sigma}_0^\top A_{00}(t)x_t^0 - \gamma\alpha_t {\Sigma}_0^\top (A_{10}(t)^\top \mu^1_t + B_0(t)) + \beta_t\}dt + (\zeta_t + \varrho_t)d\widehat{W}^0_t,\\
    \widehat\theta_0 &= m = -\gamma {\Sigma}_0^\top \Bigl(A_{00}(0)x^0_0 + A_{10}(0)^\top \mathbb{E}[x^1_0] + B_0(0) \Bigr).
    \end{split}
\end{equation}
On the other hand, by applying Ito formula to \eqref{eqbm_hat_theta}, we have:
\begin{equation}
    \begin{split}
    \label{eq_3.3-2}
    d\widehat\theta_t 
    &=
    -\gamma {\Sigma}_0^\top \{\dot{A}_{00}(t)x^0_tdt + A_{00}(t)dx^0_t +  \dot{A}_{10}(t)^\top \mu^1_t dt + A_{10}(t)^\top \dot{\mu}^1_t dt + \dot{B}_0(t) dt\}\\
    &=
    -\gamma {\Sigma}_0^\top\{(\dot{A}_{00}(t) - K_0 A_{00}(t))x^0_t\\
    &~~~~~+ (K_0A_{00}(t)m_0 + \dot{A}_{10}(t)^\top \mu^1_t + A_{10}(t)^\top \dot{\mu}^1_t + \dot{B}_0(t) )\}dt\\
    &~~~~~-\gamma {\Sigma}_0^\top A_{00}(t)\Sigma_0 d\widehat{W}_t^0,~~~t\in[0,T].
    \end{split}
\end{equation}
Comparing the coefficients of \eqref{eq_3.3-1} and \eqref{eq_3.3-2} with respect to the $x^0$-term and the constant term in the drift term, as well as the diffusion term, we obtain
\begin{equation}
    \begin{split}
    &-\gamma\alpha_t {\Sigma}_0^\top A_{00}(t) = -\gamma {\Sigma}_0^\top(\dot{A}_{00}(t) - K_0 A_{00}(t)),\\
    &- \gamma\alpha_t {\Sigma}_0^\top (A_{10}(t)^\top \mu^1_t + B_0(t)) + \beta_t = -\gamma {\Sigma}_0^\top (K_0A_{00}(t)m_0 + \dot{A}_{10}(t)^\top \mu^1_t + A_{10}(t)^\top \dot{\mu}^1_t + \dot{B}_0(t) ),\\
    &\zeta_t + \varrho_t = -\gamma {\Sigma}_0^\top A_{00}(t)\Sigma_0
    \end{split}
\end{equation}
for each $t\in[0,T]$. Rearranging the terms, we get, for $t\in[0,T]$,
\begin{equation}
    \begin{split}
    &\alpha_t = \Sigma_0^\top \dot{A}_{00}(t)A^{-1}_{00}(t)(\Sigma_0^\top)^{-1} - K_0 I_{d_0}, \\
    &\beta_t = \gamma\alpha_t {\Sigma}_0^\top (A_{10}(t)^\top \mu^1_t + B_0(t)) -\gamma {\Sigma}_0^\top (K_0A_{00}(t)m_0 + \dot{A}_{10}(t)^\top \mu^1_t + A_{10}(t)^\top \dot{\mu}^1_t + \dot{B}_0(t) ),\\
    &\varrho_t = -\gamma {\Sigma}_0^\top A_{00}(t)\Sigma_0 - \zeta_t.
    \end{split}
\end{equation}

It is easy to see, for $t\in[0,T]$,
\begin{equation}
    \begin{split}
    &\dot\varrho_t = -\gamma {\Sigma}_0^\top \dot{A}_{00}(t)\Sigma_0 - \dot\zeta_t,\\
    &\alpha_t\varrho_t = -\gamma\Sigma^\top_0 (\dot A_{00}(t) - K_0 A_{00}(t))\Sigma_0 - \alpha_t\zeta_t,\\
    &\varrho_t\alpha_t^\top = -\gamma\Sigma^\top_0 (\dot A_{00}(t) - K_0 A_{00}(t))\Sigma_0 - \zeta_t\alpha_t^\top,\\
    &\zeta_t\varrho_t = -\gamma\zeta_t\Sigma^\top A_{00}\Sigma_0 - \zeta_t^2,\\
    &\varrho_t \zeta_t = -\gamma\Sigma^\top A_{00}\Sigma_0 \zeta_t - \zeta_t^2,\\
    &\varrho_t^2 = \gamma^2 \Sigma_0^\top A_{00}(t)\Sigma_0\Sigma_0^\top A_{00}(t)\Sigma_0 + \gamma\zeta_t\Sigma^\top A_{00}\Sigma_0 + \gamma\Sigma^\top A_{00}\Sigma_0 \zeta_t + \zeta_t^2.
    \end{split}
\end{equation}
By \eqref{ODE_for_vrho}, we have
\begin{equation}
    \begin{split}
        0
        &=
        \dot{\varrho }_t - \eta_t\eta_t^\top - \alpha_t \varrho_t - \varrho_t\alpha_t^\top + \zeta_t\varrho_t + \varrho_t\zeta_t + \varrho_t^2 \\
        &=
        -\gamma {\Sigma}_0^\top \dot{A}_{00}(t)\Sigma_0 - \dot\zeta_t - \eta_t\eta_t^\top \\
        &~~~+ \gamma\Sigma^\top_0 (\dot A_{00}(t) - K_0 A_{00}(t))\Sigma_0 + \alpha_t\zeta_t + \gamma\Sigma^\top_0 (\dot A_{00}(t) - K_0 A_{00}(t))\Sigma_0 + \zeta_t\alpha_t^\top \\
        &~~~-\gamma\zeta_t\Sigma^\top A_{00}\Sigma_0 - \zeta_t^2 -\gamma\Sigma^\top A_{00}\Sigma_0 \zeta_t - \zeta_t^2 \\
        &~~~+\gamma^2 \Sigma_0^\top A_{00}(t)\Sigma_0\Sigma_0^\top A_{00}(t)\Sigma_0 + \gamma\zeta_t\Sigma^\top A_{00}\Sigma_0 + \gamma\Sigma^\top A_{00}\Sigma_0 \zeta_t + \zeta_t^2\\
        &=
        - \eta_t\eta_t^\top - \dot\zeta_t + \alpha_t\zeta_t + \zeta_t\alpha_t^\top - \zeta_t^2\\
        &~~~+ \gamma\Sigma_0^\top (\dot A_{00}(t) - 2K_0 A_{00}(t) + \gamma  A_{00}(t)\Sigma_0\Sigma_0^\top A_{00}(t))\Sigma_0\\
        &=
        - \eta_t\eta_t^\top - \dot\zeta_t + \alpha_t\zeta_t + \zeta_t\alpha_t^\top - \zeta_t^2 - \gamma^2\Sigma_0^\top A_{10}(t)^\top \Sigma\Sigma^\top A_{10}(t)\Sigma_0.\\
    \end{split}
\end{equation}
Here, we used the ODE \eqref{Riccati eqn} in the last equality.\par
Above all, in order to make \eqref{hat_theta} consistent with \eqref{eqbm_hat_theta}, the initial condition $\widehat\theta_0$ and the coefficients $(\alpha_t,\beta_t,\rho_t)$ must be given by
\begin{equation}
    \begin{split}
    &\widehat\theta_0 = -\gamma {\Sigma}_0^\top \Bigl(A_{00}(0)x^0_0 + A_{10}(0)^\top \mathbb{E}[x^1_0] + B_0(0) \Bigr),\\
    &\alpha_t = \Sigma_0^\top \dot{A}_{00}(t)A^{-1}_{00}(t)(\Sigma_0^\top)^{-1} - K_0 I_{d_0},\\
    &\beta_t = \gamma\alpha_t {\Sigma}_0^\top (A_{10}(t)^\top \mu^1_t + B_0(t)) -\gamma {\Sigma}_0^\top (K_0A_{00}(t)m_0 + \dot{A}_{10}(t)^\top \mu^1_t + A_{10}(t)^\top \dot{\mu}^1_t + \dot{B}_0(t) ),\\
    &\varrho_t = -\gamma {\Sigma}_0^\top A_{00}(t)\Sigma_0 - \zeta_t,
    \end{split}
\end{equation}
for $t\in[0,T]$, where $\zeta_t$ needs to satisfy the Riccati equation
\begin{equation}
    \begin{split}
        \label{coeff_ze_eta}
    &\dot\zeta_t = - \zeta_t^2  + \alpha_t\zeta_t + \zeta_t\alpha_t^\top  - \eta_t\eta_t^\top - \gamma^2\Sigma_0^\top A_{10}(t)^\top \Sigma\Sigma^\top A_{10}(t)\Sigma_0,~~~t\in[0,T],\\
    &\zeta_0 = -\gamma {\Sigma}_0^\top A_{00}(0)\Sigma_0 - v.
    \end{split}
\end{equation}
for $\eta\in\mathcal{C}([0,T];\mathbb{R}^{d_0\times k})$. These observations result in the following theorem.
\begin{thm} ~ \\
    Let Assumptions \ref{asm1}, \ref{asm2} and \ref{asm_kalman} be in force. Furthermore, assume that the mean $m(:=\mathbb{E}[\theta_0])$ and the coefficients $(\alpha,\beta,\zeta,\eta)\in\mathcal{C}([0,T];\mathbb{R}^{d_0\times d_0})\times \mathcal{C}([0,T];\mathbb{R}^{d_0}) \times \mathcal{C}^1([0,T];\mathbb{M}_{d_0})\times \mathcal{C}([0,T];\mathbb{R}^{d_0\times k})$ satisfy
    \begin{equation}
        \begin{split}
            \label{coeff_al_be}
        &m = -\gamma {\Sigma}_0^\top \Bigl(A_{00}(0)x^0_0 + A_{10}(0)^\top \mathbb{E}[x^1_0] + B_0(0) \Bigr),\\
        &\alpha_t = \Sigma_0^\top \dot{A}_{00}(t)A^{-1}_{00}(t)(\Sigma_0^\top)^{-1} - K_0 I_{d_0},~~~t\in[0,T], \\
        &\beta_t = \gamma\alpha_t {\Sigma}_0^\top (A_{10}(t)^\top \mu^1_t + B_0(t)) -\gamma {\Sigma}_0^\top (K_0A_{00}(t)m_0 + \dot{A}_{10}(t)^\top \mu^1_t + A_{10}(t)^\top \dot{\mu}^1_t + \dot{B}_0(t) ),~~~t\in[0,T],~~~~\\
        &\dot\zeta_t = - \zeta_t^2  + \alpha_t\zeta_t + \zeta_t\alpha_t^\top  - \eta_t\eta_t^\top - \gamma^2\Sigma_0^\top A_{10}(t)^\top \Sigma\Sigma^\top A_{10}(t)\Sigma_0,~~~t\in[0,T],\\
        &\zeta_0 = -\gamma {\Sigma}_0^\top A_{00}(0)\Sigma_0 - v.
        \end{split}
    \end{equation}
    and that such $\zeta$ is well-defined. Then, the asymptotic market clearing condition \eqref{MC-eqn} is satisfied as long as each agent adopts 
    \begin{equation}
       p^{i,*}_t := (\pi^{i,*}_t)^\top \sigma_t := Z^{i,0}_t + \frac{\widehat{\theta}_t^\top}{\gamma}, ~~~t\in[0,T], ~~~i\in\mathbb{N},
    \end{equation}
    as his/her optimal strategy. Here, $Z^{i,0}$ is given by \eqref{EQG solution YZ} and $\widehat{\theta}_t := \mathbb{E}[\theta_t|\mathcal{G}^0_t]$.
\end{thm}
\noindent
\textbf{\textit{Proof}}\\
By Lemma \ref{theta_hat_dynamics}, $\widehat{\theta}$ follows \eqref{hat_theta}, where $\varrho \in\mathcal{C}^1([0,T];\mathbb{M}_{d_0})$ satisfies \eqref{ODE_for_vrho}. 
By the observation above, $\varrho_t = -\gamma {\Sigma}_0^\top A_{00}(t)\Sigma_0 - \zeta_t$ for $t\in[0,T]$ solves \eqref{ODE_for_vrho}, and the local Lipschitz property implies that it is the unique solution. 
Then, the dynamics of $\widehat{\theta}$ reads
\begin{equation}
    \label{theta-hat-sde}
    d\widehat\theta_t = (\alpha_t\widehat\theta_t + \beta_t)dt -\gamma {\Sigma}_0^\top A_{00}(t)\Sigma_0 d\widehat{W}^0_t,~~~t\in[0,T].
\end{equation}
Notice that the process $\widehat{\theta}$ satisfying \eqref{theta-hat-sde} is unique due to the standard result for Lipschitz SDEs. By \eqref{eq_3.3-1} and \eqref{eq_3.3-2}, this clearly shows that $\widehat\theta$ is given by \eqref{eqbm_hat_theta}. 
The statement follows immediately from Theorem \ref{EQG-verification}. $\square$
\section{Numerical Analysis}
In this section, we provide a numerical simulation to visualize the dynamics of our model. We consider an economy with $N=5000$ agents with time horizon $T=1$ and set $d_0 = d = k = 1$ for simplicity. Moreover, we set: \\
\begin{center}
\begin{tabular}{|c|c|c|c|c|c|c|c|c|c|c|c|c|c|c|c|} \hline
    $\gamma$ & $K_0$ & $K$ & $m_0$ & $m$ & $\Sigma_0$ & $\Sigma$ & $A^F_{00}$ & $A^F_{11}$ & $A^F_{10}$ & $B^F_0$ & $B^F_1$ & $C^F$  \\ \hline
    1.5 & 0.05 & 0.05 & $-0.5$ & $-0.5$ & 0.3 & 0.3 & 0.7 & 0.2 & 0.3 & $-1.3$ & $-0.7$ & 1.2 \\ \hline
\end{tabular}
\end{center}
Furthermore, we set $x^i_0 \sim N(-0.7,0.5)$, $\xi^i\sim N(2,0.3)$ for each $i=1,\ldots, 5000$, $v=0.1$, $x_0^0 = 0$, $\sigma_t\equiv 0.2$ and $\eta_t = (t-0.6)\mathbbm{1}_{[0.6,1]}(t)$. 
\begin{figure}[H] 
    \centering
    \includegraphics[width=16cm]{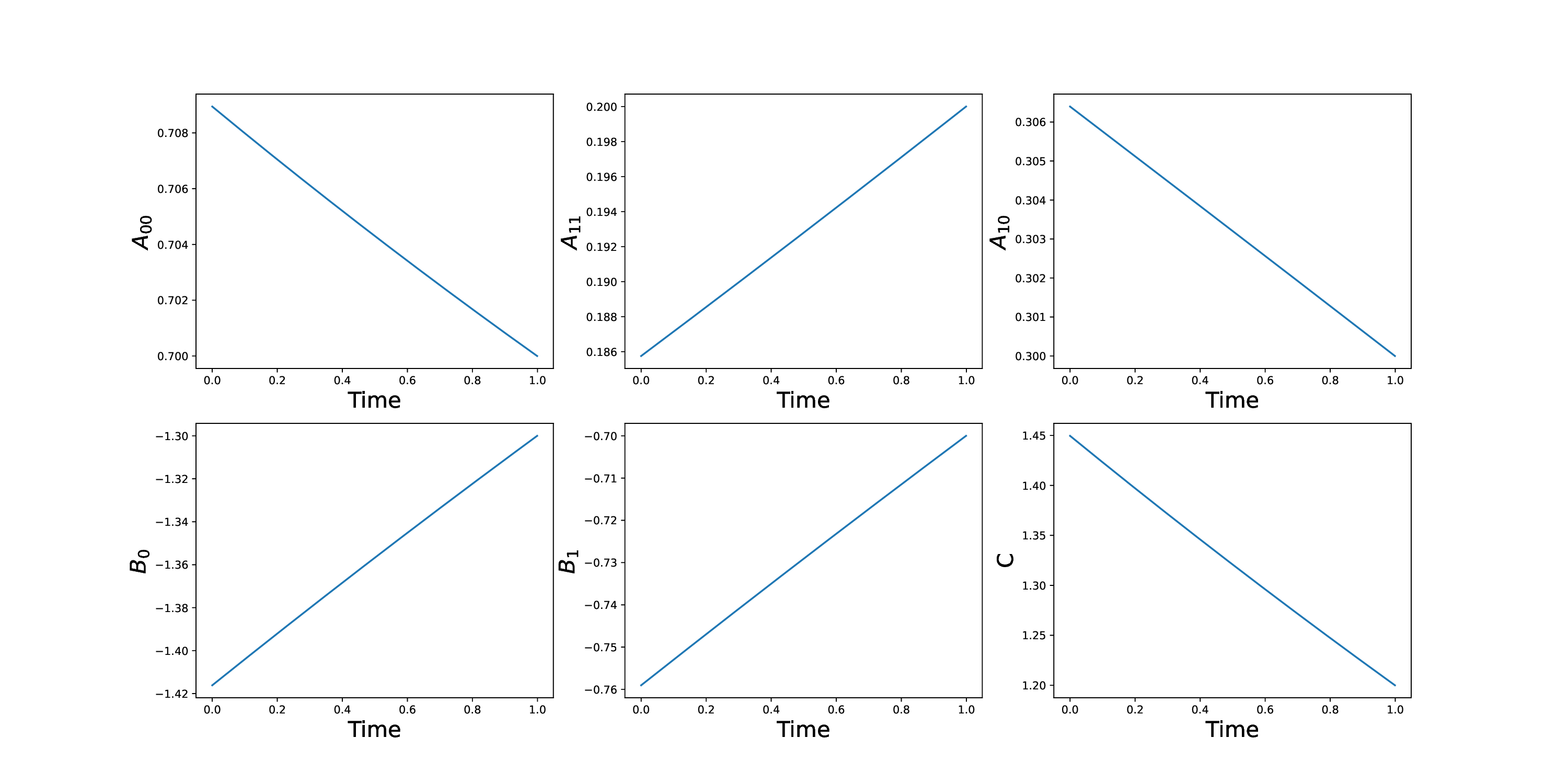}
    \caption{Solutions of Eq.\eqref{Riccati eqn}}
\end{figure}
\begin{figure}[H]
    \centering
    \begin{minipage}[b]{0.49\columnwidth}
        \centering
        \includegraphics[width=7.2cm]{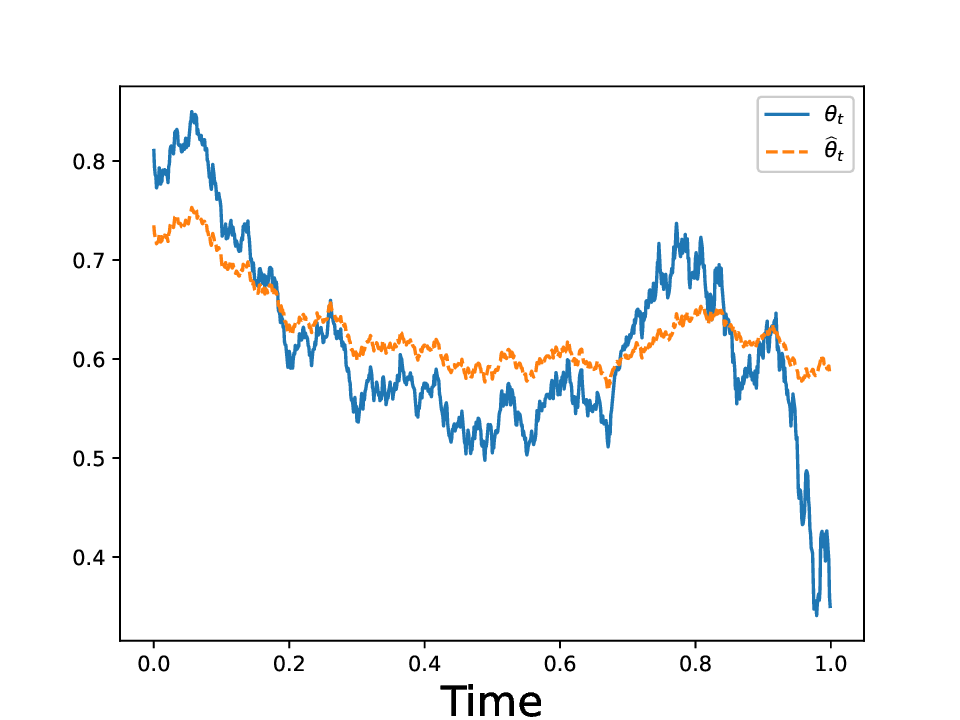}
        \caption{Market risk-premium process}
    \end{minipage}
    \begin{minipage}[b]{0.49\columnwidth}
        \centering
        \includegraphics[width=7.2cm]{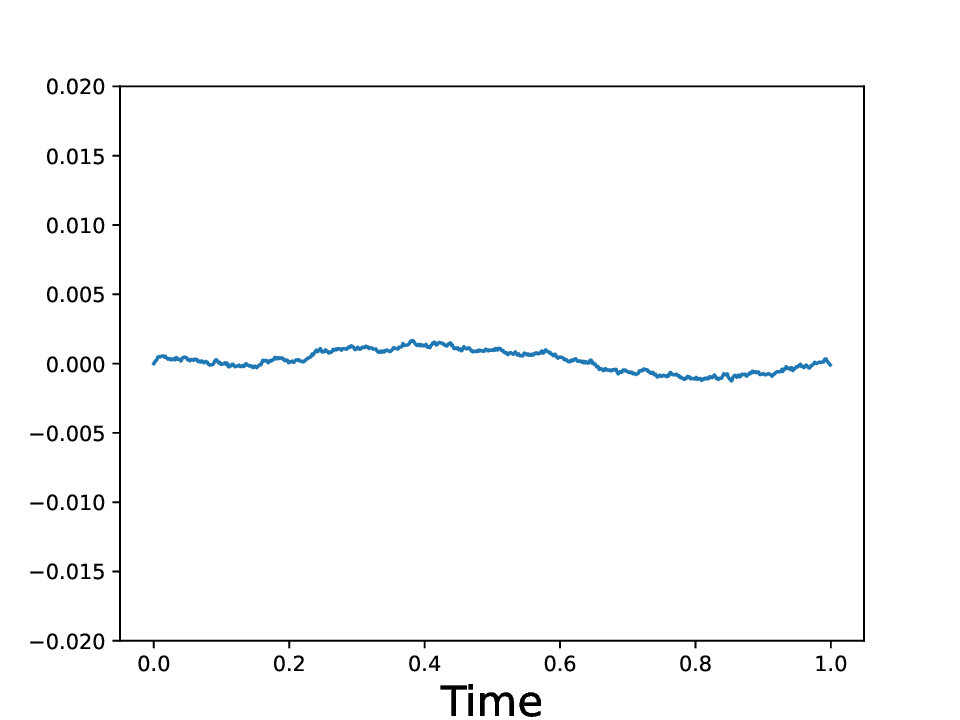}
        \caption{Asymptotic market clearing}
    \end{minipage}
\end{figure}
Figure 6.1 illustrates the numerical solution of the ODEs \eqref{Riccati eqn}.
In Figure 6.2, the sample paths for the risk-premium process $\theta$ (blue solid line) and the estimated one $\widehat\theta$ (orange dashed line) are depicted. Moreover, Figure 6.3 draws $\displaystyle\frac{1}{N}\sum_{i=1}^N \pi^{i,*}_t$ and illustrates the asymptotic market clearing property.
\begin{figure}[H] 
    \centering
    \includegraphics[width=16cm]{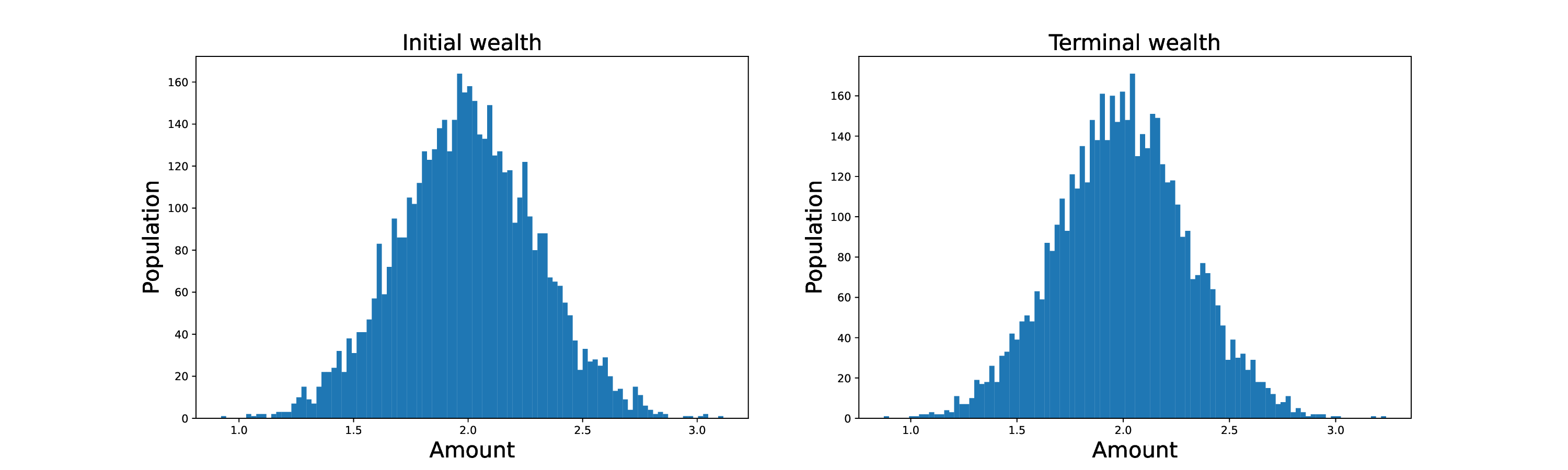}
    \caption{Initial and terminal wealth}
\end{figure}

\begin{figure}[H] 
    \centering
    \includegraphics[width=16cm]{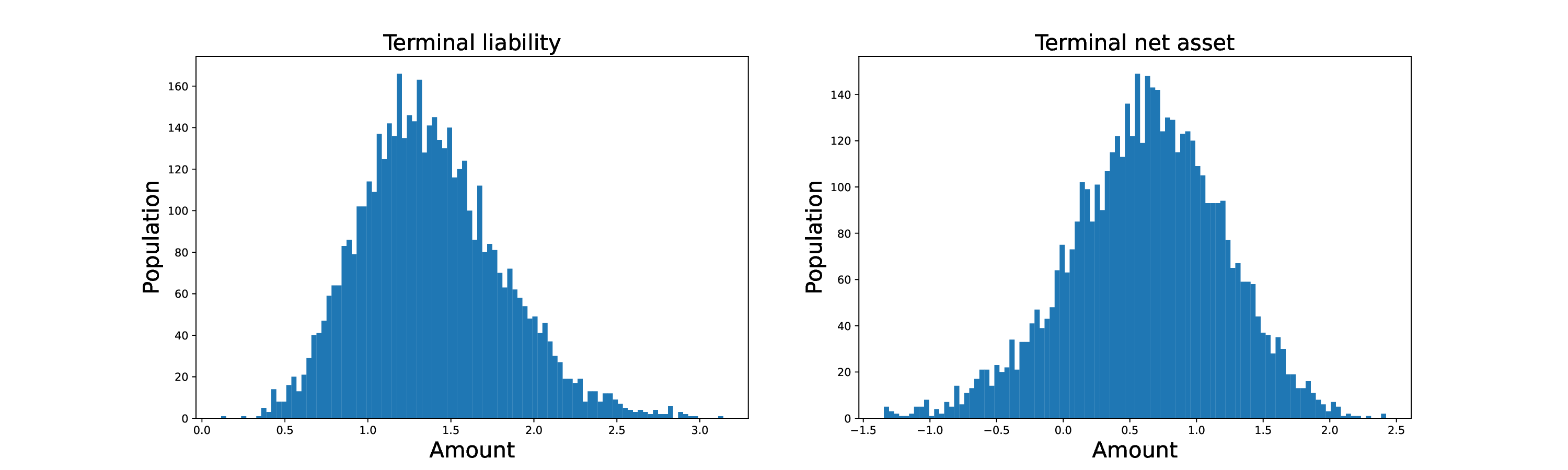}
    \caption{Terminal liability and net asset}
\end{figure}
Figure 6.4 provides the distributions of agents' initial wealths $(\xi^i)_{i=1,\ldots,5000}$ and terminal wealths $(\mathcal{W}^{i,p^{i,*}}_T)_{i=1,\ldots,5000}$. Figure 6.5 draws the distribution of terminal liabilities $(F^i)_{i=1,\ldots,5000}$ and terminal net assets $(\mathcal{W}^{i,p^{i,*}}_T-F^i)_{i=1,\ldots,5000}$.
\section{Conclusion and Discussions}
In this chapter, we studied the mean field equilibrium asset pricing model in a partially observable market. 
In Section 6.2, we formulated the utility maximization problem under partial observation and derived the condition for optimal strategies. 
In Section 6.3, within the exponential quadratic Gaussian framework, we associated the solution of the mean field BSDE with matrix ODEs and verified the asymptotic market clearing condition in the large population limit. 
We then constructed the risk-premium process endogenously using the Kalman-Bucy filtering theory.
Section 6.4 presented a simple numerical example that visualizes a sample path of the risk-premium process as well as distributions of agents' wealth.\par
As a direction for future research, we may possibly generalize the dynamics \eqref{LG_MRP} by, for example, adding a jump process to formulate the possibility of default. This may lead us to consider the non-linear filtering for jump-diffusion processes.

%% file: conclusion.tex
\vspace{0mm}
In this thesis, we have studied three asset pricing models using the mean field game theory.
The first model, provided in Chapter 4, is the fundamental model, which provides an important insight into the equilibrium risk-premium process in an incomplete market. 
The second and third models are successful extensions of the first model to general cases, allowing for semi-analytic solutions and numerical simulations.
These visualized the impact of changes in each market parameter on the equilibrium risk premium.
These models succeed in describing how risk premium is formed through the interaction of many heterogeneous investors, allowing for analysis in a setting closer to the real market.

The key to these models is the mean field BSDEs such as \eqref{mfg-BSDE-org}.
These are a novel form of equations in this field and an analysis of its properties forms the basis of this study.
However, the well-posedness of such an equation can be shown in quite limited cases; when parameters of the equation are sufficiently small or when the terminal condition has a special structure.
A detailed analysis of this equation and the discovery of further general properties will be an important challenge for future studies.

In concluding the thesis, we would like to discuss the limitations, as well as the potential for future development and applications of this study.
Firstly, the optimal martingale method, proposed by \cite{huUtilityMaximizationIncomplete2005a}, can be applied to other classes of utility functions, including power utility and Epstein-Zin-type recursive utility.
However, the resultant mean field (F)BSDEs for such classes of utility functions become much more complex than \eqref{mfg-BSDE-org}. 
One major reason for this complexity is that the optimal trading strategy depends on the wealth process for these utilities, unlike the case of exponential utility.
Currently, it is difficult to prove the solvability of such equations, even with strong assumptions.

Moreover, as mentioned in Chapter 3, the return vector $\mu$ and the volatility matrix $\sigma$ cannot be derived endogenously in the model.
It is explained in \cite{karatzas_methods_1998} that this is due to the possibility of replacing stocks by self-financing portfolios.
We will delve a little deeper into this fact.
In the models of this thesis, the financial market is characterized by the parameters of excess return $\mu$, volatility $\sigma$, and initial price $S_0$. (See Assumption \ref{assumption-market}.)
We use the symbol $\mathcal{M}=(\mu,\sigma,S_0)$ to specify the market.
For $1\leq k \leq n$, let $(\widetilde{S}^{(k)}_t)_{t\in[0,T]}$ be an $\mathbb{R}$-valued $\mathbb{F}^0$-progressively measurable process representing the value of the self-financing portfolio managed by the strategy $(q_k(t))_{t\in[0,T]}\in\mathbb{L}^\infty(\mathbb{F}^0,\mathbb{R}^n)$ with initial value $\widetilde{S}^{(k)}_0\in\mathbb{R}_{++}$ in the market $\mathcal{M}$.
For each $1\leq l \leq n$ and $t\in[0,T]$, $q_{kl}(t)$ denotes the $l$-th element of the vector $q_{k}(t)$ and represents the proportion of $\widetilde{S}^{(k)}_t$ invested in the $l$-th security at time $t$. 
In other words, the amount of money $\widetilde{S}^{(k)}_t q_{kl}(t)$ is invested in the $l$-th security at time $t$.
Then, for each $1\leq k \leq n$, the dynamics of $(\widetilde{S}^{(k)}_t)_{t\in[0,T]}$ is given by
\begin{equation}
		\widetilde{S}^{(k)}_t=\widetilde{S}^{(k)}_0+\int_0^t \widetilde{S}^{(k)}_rq_{k}(r)^\top(\mu_r dr+\sigma_r dW_r^0), ~~~t\in[0,T].
\end{equation}
We assume that the vectors $q_1(t),\ldots,q_n(t)$ are linearly independent for all $t\in[0,T]$.
Define an $\mathbb{R}^{n\times n}$-valued $\mathbb{F}^0$-progressively measurable process $Q:=(Q_t)_{t\in[0,T]}$ by $Q_t:=(q_1(t),\ldots,q_n(t))^\top$ for $t\in[0,T]$. Then, $Q$ is bounded and non-singular for all $t\in[0,T]$. We further assume that $(Q^{-1}_t)_{t\in[0,T]}$ is bounded.
We set $\widetilde{S}_t:=(\widetilde{S}^{(1)}_t,\ldots,\widetilde{S}^{(n)}_t)^\top$ for $t\in[0,T]$. Then, the dynamics of $\widetilde{S}$ is given by
\begin{equation}
		\widetilde{S}_t=\widetilde{S}_0+\int_0^t {\rm diag}(\widetilde{S}_r)Q_r(\mu_r dr+\sigma_r dW_r^0), ~~~t\in[0,T] 
\label{eq-stock-price-widetilde}
\end{equation}
with $\widetilde{S}_0\in\mathbb{R}^n_{++}$. 
Define $\widetilde\mu_t:=Q_t\mu_t$ and $\widetilde\sigma_t:=Q_t\sigma_t$ for $t\in[0,T]$. We can write 
\begin{equation}
		\widetilde{S}_t=\widetilde{S}_0+\int_0^t {\rm diag}(\widetilde{S}_r)(\widetilde\mu_r dr+\widetilde\sigma_r dW_r^0), ~~~t\in[0,T].
\end{equation}

We can then consider another financial market $\widetilde{\mathcal{M}}=(\widetilde\mu,\widetilde\sigma,\widetilde{S}_0)$.
The risk premium of the market $\widetilde{\mathcal{M}}$, denoted by $\widetilde\theta$, is given by
\[
    \widetilde\theta_t := \widetilde\sigma_t^\top(\widetilde\sigma_t\widetilde\sigma_t^\top)^{-1}\widetilde\mu_t=\sigma_t^\top Q_t^\top (Q_t\sigma_t\sigma_t^\top Q^\top_t)^{-1}Q_t\mu_t=\sigma_t^\top (\sigma_t\sigma_t^\top)^{-1}\mu_t = \theta_t
\]
for $t\in[0,T]$, which implies that the market risk premium of $\widetilde{\mathcal{M}}$ is the same as that of $\mathcal{M}$.
In the market $\widetilde{\mathcal{M}}$, the agents' utility maximization problem is formulated as
\begin{equation}
    \begin{split}
        \sup_{\widetilde\pi\in\widetilde{\mathbb{A}}^i} \mathbb{E}\Bigl[-\exp\Bigl(-\gamma^i(\widetilde{\mathcal{W}}^{i,\widetilde\pi}_T-F^i)\Bigr) \Bigr]
    \end{split}
\end{equation}
subject to
\[
    \widetilde{\mathcal{W}}^{i,\widetilde\pi}_t =\xi^i + \int_0^t \widetilde\pi_s^\top\widetilde\sigma_s(\theta_sds +  dW_s^0),~~~t\in[0,T],
\]
where the admissible space $\widetilde{\mathbb{A}}^i$ is defined as the set of strategies $\widetilde\pi\in\mathbb{H}^2(\mathbb{F}^{0,i},\mathbb{R}^n)$ such that the family $\{\exp(-\gamma^i \widetilde{\mathcal{W}}^{i,\widetilde\pi}_\tau);\tau\in\mathcal{T}^{0,i} \}$ is uniformly integrable.
Set $\pi_t:=Q_t^\top\widetilde\pi_t$ for $t\in[0,T]$. 
It is straightforward to see that
\begin{equation}
    \begin{split}
    \widetilde{\mathcal{W}}^{i,\widetilde\pi}_t 
    &=
    \xi^i + \int_0^t \widetilde\pi_s^\top\widetilde\sigma_s(\theta_sds +  dW_s^0)\\
    &=
    \xi^i + \int_0^t (Q_s^\top\widetilde\pi_s)^\top\sigma_s(\theta_sds +  dW_s^0)\\
    &=
    \xi^i + \int_0^t \pi_s^\top\sigma_s(\theta_sds +  dW_s^0)\\
    &=
    \mathcal{W}^{i,\pi}_t.
    \end{split}
\end{equation}
Since $Q_t$ is non-singular and both $Q_t$ and $Q^{-1}_t$ are bounded for $t\in[0,T]$, we deduce that $\pi\in\mathbb{A}^i$ if and only if $\widetilde\pi\in\widetilde{\mathbb{A}}^i$.
Thus, the utility maximization problem in the market $\widetilde{\mathcal{M}}$ can be written as
\begin{equation}
    \begin{split}
        \sup_{\pi\in\mathbb{A}^i} \mathbb{E}\Bigl[-\exp\Bigl(-\gamma^i(\mathcal{W}^{i,\pi}_T-F^i)\Bigr) \Bigr]
    \end{split}
\end{equation}
subject to
\[
    \mathcal{W}^{i,\pi}_t =\xi^i + \int_0^t \pi_s^\top\sigma_s(\theta_sds +  dW_s^0),~~~t\in[0,T],
\]
which is exactly the same problem as in the market $\mathcal{M}$. 
If the markets $\mathcal{M}$ and $\widetilde{\mathcal{M}}$ satisfy $\theta^{\rm mfg}_t=\sigma_t^\top (\sigma_t\sigma_t^\top)^{-1}\mu_t=\widetilde\sigma_t^\top(\widetilde\sigma_t\widetilde\sigma_t^\top)^{-1}\widetilde\mu_t$ for $t\in[0,T]$, the asymptotic market clearing condition can be achieved in both markets.
Here, $\theta^{\rm mfg}$ is given by \eqref{theta-mfg}. In other words, the difference between $\mu$ and $\widetilde\mu$ or $\sigma$ and $\widetilde\sigma$ does not affect the equilibrium state in these cases. 
Therefore, the excess return $\mu$ and the volatility $\sigma$ cannot be determined by considering the market clearing condition in the models of this thesis.

Intuitively, one might think that this result is attributed to the fact that the market is completely liquid and trades can be made without transaction costs in these models.
When agents can trade an arbitrary amount of securities at any time without incurring additional costs in either market $\mathcal{M}$ or $\widetilde{\mathcal{M}}$, the optimal trajectory of the wealth dynamics should inevitably be identical in both markets.
This is due to the uniqueness of the optimal strategy in each market, as shown in Theorem \ref{verification-part1}.
Then, one possible direction for future research would be to develop a price formation model for markets with transaction costs and liquidity restrictions.
The optimal investment problem with transaction costs is being actively researched, and in recent years, Kallsen \& Muhle-Karbe \cite{shadow-price} has developed an approach to address the problem with transaction costs using the concept of a shadow price.
Furthermore, modelling limit order books is another possible way to formulate the liquidity restriction.
Recently, Cont \& M\"{u}ller \cite{spde-lob} has developed a model for dynamic limit order books, employing the stochastic partial differential equations (SPDEs).
In any case, extending such models to an equilibrium model may contribute to the discovery of new results for asset pricing theory, but at the same time, it will inevitably lead to more advanced mathematical problems.